\documentclass[11pt]{article}
\usepackage{amsmath,amssymb,amsfonts,amsthm}
\usepackage{dsfont}
\usepackage{latexsym}
\usepackage{subcaption}
\usepackage{graphicx}
\usepackage{tabu}
\usepackage{fullpage}
\usepackage[backref=page]{hyperref}
\usepackage{cleveref}
\usepackage{color}
\usepackage{wrapfig}
\usepackage{tikz}
\usetikzlibrary{decorations.pathreplacing}
\usepackage{setspace}
\usepackage{algorithm}
\usepackage[noend]{algpseudocode}
\usepackage[framemethod=tikz]{mdframed}
\usepackage{xspace}
\usepackage{pgfplots}
\usepackage{framed}
\usepackage{todonotes}
\usepackage{thmtools}
\usepackage{thm-restate}
\pgfplotsset{compat=1.5}

\newtheorem{theorem}{Theorem}[section]
\newtheorem{corollary}[theorem]{Corollary}
\newtheorem{lemma}[theorem]{Lemma}

\newtheorem{definition}[theorem]{Definition}

\newtheorem{claim}[theorem]{Claim}

\newtheorem{observation}[theorem]{Observation}

\newtheorem{framework}[theorem]{Framework}

\newenvironment{proofof}[1]{\begin{trivlist} \item {\bf Proof
#1:~~}}
  {\qed\end{trivlist}}

\newcommand{\namedref}[2]{\hyperref[#2]{#1~\ref*{#2}}}
\newcommand{\thmlab}[1]{\label{thm:#1}}
\newcommand{\thmref}[1]{\namedref{Theorem}{thm:#1}}
\newcommand{\lemlab}[1]{\label{lem:#1}}
\newcommand{\lemref}[1]{\namedref{Lemma}{lem:#1}}
\newcommand{\claimlab}[1]{\label{claim:#1}}
\newcommand{\claimref}[1]{\namedref{Claim}{claim:#1}}
\newcommand{\corlab}[1]{\label{cor:#1}}
\newcommand{\corref}[1]{\namedref{Corollary}{cor:#1}}
\newcommand{\seclab}[1]{\label{sec:#1}}
\newcommand{\secref}[1]{\namedref{Section}{sec:#1}}
\newcommand{\applab}[1]{\label{app:#1}}

\newcommand{\figlab}[1]{\label{fig:#1}}
\newcommand{\figref}[1]{\namedref{Figure}{fig:#1}}
\newcommand{\alglab}[1]{\label{alg:#1}}
\renewcommand{\algref}[1]{\namedref{Algorithm}{alg:#1}}

\newcommand{\deflab}[1]{\label{def:#1}}
\newcommand{\defref}[1]{\namedref{Definition}{def:#1}}

\newcommand{\obslab}[1]{\label{obs:#1}}
\newcommand{\obsref}[1]{\namedref{Observation}{obs:#1}}

\newcommand{\linlab}[1]{\label{line:#1}}

\newcommand{\frameref}[1]{\namedref{Framework}{frame:#1}}
\newcommand{\framelab}[1]{\label{frame:#1}}

\newcommand{\PPr}[1]{\ensuremath{\mathbf{Pr}\left[#1\right]}}

\newcommand{\Ex}[1]{\ensuremath{\mathbb{E}\left[#1\right]}}

\renewcommand{\O}[1]{\ensuremath{\mathcal{O}\left(#1\right)}}
\newcommand{\tO}[1]{\ensuremath{\tilde{\mathcal{O}}\left(#1\right)}}
\newcommand{\eps}{\varepsilon}

\def \countsketch    {\mdef{\textsc{CountSketch}}}

\def \zeroestimate    {\mdef{\textsc{F0Estimate}}}
\def \ghss    {\mdef{\textsc{GHSS}}}

\def \psketch   {\mdef{\textsc{PStable}}}

\def \bptree   {\mdef{\textsc{BPTree}}}
\def \heavyhitters    {\mdef{\textsc{HeavyHitters}}}

\def \sampler    {\mdef{\textsc{Sampler}}}
\def \sdiffest    {\mdef{\textsc{SDiffEst}}}

\def \hhest    {\mdef{\textsc{HHEst}}}
\def \guessupdatesw    {\mdef{\textsc{GuessAndUpdate}}}
\def \updatesw    {\mdef{\textsc{UpdateSW}}}
\def \mergesw    {\mdef{\textsc{MergeSW}}}
\def \stitchsw    {\mdef{\textsc{StitchSW}}}

\def \estimator    {\mdef{\textsc{Estimator}}}
\def \estimateF    {\mdef{\textsc{EstimateF}}}

\def \calA    {\mdef{\mathcal{A}}}
\def \calB    {\mdef{\mathcal{B}}}
\def \calC    {\mdef{\mathcal{C}}}
\def \calD    {\mdef{\mathcal{D}}}
\def \calE    {\mdef{\mathcal{E}}}
\def \calH    {\mdef{\mathcal{H}}}
\def \calL    {\mdef{\mathcal{L}}}
\def \calS    {\mdef{\mathcal{S}}}

\def \bfone    {\mdef{\mathbf{1}}}

\newcommand{\mdef}[1]{{\ensuremath{#1}}\xspace}  % Math Def which can also be used in normal text.
\DeclareMathOperator*{\polylog}{polylog}
\DeclareMathOperator*{\poly}{poly}

\DeclareMathOperator*{\lsb}{lsb}
\DeclareMathOperator*{\numbits}{numbits}
\DeclareMathOperator*{\Var}{Var}

\newcommand{\flr}[1]{\mdef{\left\lfloor#1\right\rfloor}}              % Absolute value
\newcommand{\ceil}[1]{\mdef{\left\lceil#1\right\rceil}}               % Absolute value
\newcommand{\ip}[2]{\langle #1,#2 \rangle} % Expected value

\newcommand{\ignore}[1]{}

\newif\ifnotes\notestrue %set this to true if notes are visible and to false (next line) if they should be hidden
\ifnotes
\newcommand{\samson}[1]{\textcolor{purple}{{\bf (Samson:} {#1}{\bf ) }} \marginpar{\tiny\bf
             \begin{minipage}[t]{0.5in}
               \raggedright S:
            \end{minipage}}}
\newcommand{\zhili}[1]{\textcolor{red}{{\bf (Zhili:} {#1}{\bf ) }} \marginpar{\tiny\bf
             \begin{minipage}[t]{0.5in}
               \raggedright Z:
            \end{minipage}}}  
\else
\newcommand{\samson}[1]{}
\newcommand{\zhili}[1]{}
\fi

\hypersetup{
     colorlinks   = true,
     citecolor    = forestgreen,
		 urlcolor			= mahogany,
		 linkcolor		= mahogany
}

\makeatletter
\renewcommand*{\@fnsymbol}[1]{\textcolor{purple}{\ensuremath{\ifcase#1\or *\or \dagger\or \ddagger\or
 \mathsection\or \triangledown\or \mathparagraph\or \|\or **\or \dagger\dagger
   \or \ddagger\ddagger \else\@ctrerr\fi}}}
\makeatother

\providecommand{\email}[1]{\href{mailto:#1}{\nolinkurl{#1}\xspace}}

\definecolor{mahogany}{rgb}{0.75, 0.25, 0.0}
\definecolor{darkblue}{rgb}{0.0, 0.0, 0.55}
\definecolor{bleudefrance}{rgb}{0.19, 0.55, 0.91}
\definecolor{darkpastelgreen}{rgb}{0.01, 0.75, 0.24}
\definecolor{darkgreen}{rgb}{0.0, 0.2, 0.13}
\definecolor{forestgreen}{rgb}{0.13, 0.55, 0.13}
\begin{document}

\title{Tight Bounds for Adversarially Robust Streams and Sliding Windows via Difference Estimators}
%\title{Adversarially Robust and Sliding Window Streaming Algorithms without the Overhead}
\author{David P. Woodruff\thanks{Carnegie Mellon University. 
E-mail: \email{dwoodruf@cs.cmu.edu}}\\
\and
Samson Zhou\thanks{Carnegie Mellon University. 
E-mail: \email{samsonzhou@gmail.com}}}
\date{}
\begin{titlepage}
\maketitle
\thispagestyle{empty}

\begin{abstract}
In the adversarially robust streaming model, a stream of elements is presented to an algorithm and is allowed to depend on the output of the algorithm at earlier times during the stream. In the classic insertion-only model of data streams, Ben-Eliezer \emph{et al.} (PODS 2020, best paper award) show how to convert a non-robust algorithm into a robust one with a roughly $1/\varepsilon$ factor overhead. This was subsequently improved to a $1/\sqrt{\varepsilon}$ factor overhead by Hassidim \emph{et al.} (NeurIPS 2020, oral presentation), suppressing logarithmic factors. For general functions the latter is known to be best-possible, by a result of Kaplan \emph{et al.} (CRYPTO 2021). We show how to bypass this impossibility result by developing data stream algorithms for a large class of streaming problems, {\it with no overhead in the approximation factor}. Our class of streaming problems includes the most well-studied problems such as the $L_2$-heavy hitters problem, $F_p$-moment estimation, as well as empirical entropy estimation. We substantially improve upon all prior work on these problems, giving the first optimal dependence on the approximation factor.

As in previous work, we obtain a general transformation that applies to any non-robust streaming algorithm and depends on the so-called twist number. However, the key technical innovation is that we apply the transformation to what we call a {\it difference estimator} for the streaming problem, rather than an estimator for the streaming problem itself. We then develop the first difference estimators for a wide range of problems. Our difference estimator methodology is not only applicable to the adversarially robust model, but to other streaming models where temporal properties of the data play a central role. To demonstrate the generality of our technique, we additionally introduce a general framework for the related sliding window model of data streams and resolve longstanding open questions in that model, obtaining a drastic improvement from the previous $1/\varepsilon^{2+p}$ dependence for $F_p$-moment estimation for $p \in [1,2]$ and integer $p > 2$ of Braverman and Ostrovsky (FOCS, 2007), to the optimal $1/\varepsilon^2$ bound. We also improve the prior $1/\varepsilon^3$ bound for $p \in [0,1)$, and the prior $1/\varepsilon^4$ bound for empirical entropy, obtaining the first optimal $1/\varepsilon^2$ dependence for both of these problems as well. Qualitatively, our results show there {\it is no separation} between the sliding window model and the standard data stream model in terms of the approximation factor. 
\end{abstract}
\end{titlepage}

\allowdisplaybreaks
\section{Introduction}
Efficient computation of statistics over large datasets is increasingly important. 
Such datasets include logs generated from internet traffic, IoT sensors, financial markets, and scientific observations. 
To capture these applications, the streaming model defines an underlying dataset through updates that arrive sequentially and describe the evolution of the dataset over time.
The goal is to approximate statistics of the input using memory, i.e., space complexity, that is significantly sublinear in the input size $n$, while only making a single pass over the data.

%{\bf Adversarially robust streaming model:}
\paragraph{Adversarially robust streaming model.} 
In the adversarially robust streaming model, the input is adaptively chosen by an adversary who is given unlimited computational resources and may view the outputs of the streaming algorithm at previous times in the stream. 
The goal of the adversary is to design the input to the streaming algorithm so that the algorithm eventually outputs an incorrect answer. 
One application of the model is to recommendation systems, where a large set of possible items arrives in a data stream and the goal is to produce a list of fixed size, i.e., a cardinality constraint, so as to maximize a predetermined function, e.g., a submodular function representing a user's utility~\cite{BogunovicMSC17, MitrovicBNTC17, AvdiukhinMYZ19}. 
However, the set of items might subsequently be modified by an honest user based on their personal preferences, e.g., to avoid items they already have. 
Similar notions of adversarial robustness have been the recent focus of a line of work~\cite{KrauseMGG08, MironovNS11, HardtU14, OrlinSU18, BoyleLV19, NaorY19, Ben-EliezerY20, Ben-EliezerJWY20, HassidimKMMS20}.

The work of \cite{HardtW13} showed that no linear sketch can approximate the $L_2$-norm within a polynomial multiplicative factor against such an adversary. 
Linear sketching is a widely used technique for turnstile streams, in which positive and negative updates are possible. 
The work of \cite{HardtW13} also connected their result to differential privacy, showing that they rule out algorithms for answering more than a polynomial number of queries while satisfying differential privacy. 

On the positive side, \cite{Ben-EliezerJWY20} gave a general framework in the adversarially robust setting when the updates are in the standard insertion-only model of data streams, meaning that previous stream updates cannot later be deleted.  
In this model, there is an underlying vector $x$ that is initialized to $0^n$ and at the $t$-th time step, it receives an update of the form $x_i \leftarrow x_i + \Delta_t$ for some positive integer $\Delta_t$. 
It is promised that all cooordinates of $x$ are bounded by a polynomial in $n$ throughout the stream. 
The framework of \cite{Ben-EliezerJWY20} gives algorithms for a number of problems, and for this discussion, we illustrate the shortcomings of their results for the $F_p$-moment estimation problem. 
Here, the goal is to output a number that is a multiplicative $1 \pm \varepsilon$ approximation to $F_p(x) = \sum_{i=1}^n x_i^p$. 
The work of \cite{Ben-EliezerJWY20} achieves an adversarially robust algorithm with $\tO{\frac{\log n}{\eps^3}}$ bits\footnote{We use the $\tO{\cdot}$ notation to omit logarithmic terms, i.e., $\tO{f}=\O{f\polylog f}$ for any function $f$ of $n$ and $\frac{1}{\eps}$.} of memory for $p\in[0,2]$ and $\tO{\frac{n^{1-2/p}}{\eps^3}}$ bits for $p > 2$. 
%the former of which is tight in $n$ up to $\log\log n$ factors while the latter is tight in $n$ up to $\log n$ factors
Subsequently, \cite{HassidimKMMS20} improved the space dependence to $\tO{\frac{\log^4 n}{\eps^{2.5}}}$ for $F_p$-moment estimation for $p\in[0,2]$ and $\tO{\frac{n^{1-2/p}}{\eps^{2.5}}}$ for $p>2$. 
Their technique showed that differential privacy can be used to protect the internal randomness of algorithms against the adversary to a useful extent, despite the general negative results concerning differential privacy pointed out by \cite{HardtW13} for linear sketches.
% which improved upon the $\frac{1}{\eps}$ factors at the cost of $\log n$ factors. 

The algorithms above should be compared to the best known algorithms in the standard non-robust data stream model, which achieve $\O{\frac{\log n}{\eps^2}}$ memory for $p \in [0,2]$ \cite{KaneNW10b,KaneNW10,KaneNPW11} and $\tO{\frac{n^{1-2/p}}{\varepsilon^2}}$ memory for $p > 2$ \cite{g15,gw18}. 
Notably, the dependence on the approximation factor $\varepsilon$, is $1/\eps^{2.5}$ in the robust setting but only $1/\eps^2$ in the non-robust setting. 
Similar gaps exist for the important $L_2$-heavy hitters problem, where the space bounds of \cite{Ben-EliezerJWY20} and \cite{HassidimKMMS20} are $\tO{\frac{\log n}{\varepsilon^3}}$ and $\tO{\frac{\log^4 n}{\varepsilon^{2.5}}}$ bits, respectively, while the non-adversarial insertion-only streaming space complexity is $\tO{\frac{\log n}{\varepsilon^2}}$ bits \cite{BravermanCIW16,BravermanCINWW17}. 
The study on the dependence of the approximation parameter $\varepsilon$ for problems in the data stream literature has a long and rich history, including work on frequency moments and cascaded norms \cite{lw13,wz21}, multipass and distributed functional monitoring \cite{WZ12,WZ14,GWWZ15,WZ18}, low rank approximation \cite{CW09,w14,BWZ16}, sparse recovery \cite{PW11,NSW018}, subspace embeddings \cite{NN14}, and support vector machines \cite{AB0MW20}. 
As discussed in \cite{IW03,W04}, for reasonable values of $n$ and $\varepsilon$ (e.g., $n = 2^{32}$ and $\varepsilon = 10\%$), the storage bound of a streaming algorithm is often dominated by the polynomial dependence on $1/\varepsilon$. 
If one wants even better approximation quality, e.g., $\varepsilon = 1\%$, then a large polynomial dependence on $1/\varepsilon$ constitutes a severe drawback of existing algorithms. 
In \cite{BJKST02}, for distinct element estimation it was explicitly posed if one could obtain a $\frac{1}{\varepsilon} \cdot \textrm{polylog}(n)$ space algorithm or if an $\Omega(1/\varepsilon^2)$ space algorithm was best possible. 
This was later resolved in \cite{IW03,W04} for the non-robust data stream model. 
However, we are now faced with the analogous question for adversarially robust algorithms, not just for distinct element estimation, but for almost all streaming problems of interest:
\begin{center}
{\it Do we need to pay a polynomial in $\frac{1}{\varepsilon}$ overhead in the adversarially robust model for all of the well-studied streaming problems above?}
\end{center}
Recent work by~\cite{KaplanMNS21} gave a further separation between adversarial and oblivious streams by introducing a streaming problem for which there exists an oblivious streaming algorithm that uses only polylogarithmic space but any adversarially-robust streaming algorithm must use polynomial space. 
In fact, they show that in general there must be an $\Omega(\sqrt{\lambda})$ blowup in the memory required of an adversarially robust algorithm over its non-robust counterpart, where $\lambda$ is the so-called flip number and indicates roughly how many times a function can change by a $(1+\varepsilon)$-factor; for a more precise description see \secref{sec:prelim}. 
As the flip number is $\Omega(1/\varepsilon)$ for insertion streams, while the optimal non-robust streaming algorithm has a $\Theta(1/\varepsilon^2)$ dependence, this suggests that one does need to pay $\Omega(1/\varepsilon^{2.5})$ space to solve the above streaming problems in the adversarially robust model. 
Surprisingly, we will show {\it this is not the case}, and the lower bound of \cite{KaplanMNS21} does not apply to all of the well-studied streaming problems above. 

%Since insertion-only streaming algorithms for $F_p$-moment estimation use $\tO{\frac{\log n}{\eps^2}}$ bits of space for $p \in [0,2]$ and $\tO{\frac{n^{1-2/p}}{\eps^2}}$ bits of space for $p > 2$, the natural open question is whether there exists a gap in $\frac{1}{\eps}$ factors in the adversarially robust streaming model. 
%Observe that, for example even though roughly $300$ billion e-mails are sent every day, for $1\%$ error, we would have $\frac{1}{\eps}=100>39\log_2(3\times10^{11})$, so understanding the $\frac{1}{\eps}$ dependencies should be just as important as understanding the $\log n$ dependencies. 
%{\bf Sliding window model:} 

\paragraph{Sliding window model.} 
A related data stream model where temporal properties play a central role is the {\it sliding window model}. 
The streaming model does not capture applications in which recent data is considered more accurate and important than data that arrived prior to a certain time. 
For a number of applications~\cite{BabcockBDMW02,MankuM12,PapapetrouGD15,WeiLLSDW16}, the unbounded streaming model has performance inferior to the sliding window model~\cite{DatarGIM02}, where the underlying dataset consists of only the $W$ most recent updates in the stream, for a parameter $W>0$ that denotes the window size of the active data. 
All updates before the $W$ most recent updates are expired, and the goal is to aggregate information about the active data using space sublinear in $W$. 
Observe that for $W>m$ for a stream $u_1, u_2, \ldots, u_m$ of items of length $m$, the active data is the entire stream, and for an underlying vector $x$ of the data, we have $x_i=|\{t\,:\,u_t=i \textrm{ and } t\ge m-W+1\}|$ for each $i\in[n] = \{1, 2, \ldots, n\}$. 
Thus, the sliding window model is a generalization of the streaming model where $W$ can be less than $m$. 
It is especially relevant in time-sensitive applications such as network monitoring~\cite{CormodeM05a,CormodeG08,Cormode13}, event detection in social media~\cite{OsborneEtAl2014}, and data summarization~\cite{ChenNZ16,EpastoLVZ17}, and has been subsequently studied in a number of additional settings~\cite{LeeT06, LeeT06b, BravermanO07, DatarM07, BravermanOZ12, BravermanLLM15, BravermanLLM16, BravermanGLWZ18, BravermanDMMUWZ18, BorassiELVZ20}.

%The smooth histogram approach of \cite{BravermanO07} provides a sliding window algorithm for $F_p$-estimation that uses 
Returning to our running example of $F_p$-moment estimation, the best known algorithms in the sliding window model have a $1/\eps^{2+p}$ dependence on the approximation factor $\eps$, for every $p \geq 1$, in their space complexity \cite{BravermanO07}.
Surprisingly, there has been no progress on this problem since then. 
Since the sliding window model generalizes insertion-only streams, which have $F_p$-moment estimation algorithms with a $\O{\frac{1}{\eps^2}}$ dependence on $\eps$, the longstanding open question is:
\begin{center}
{\it Do we need to pay a polynomial in $\frac{1}{\varepsilon}$ overhead in the sliding window model for well-studied streaming problems, such as estimating the frequency moments?}
\end{center}
%\begin{center}
%{\it Does there exist a gap in the $\varepsilon$-dependence for sliding windows and for data streams for frequency moment estimation?}
%\end{center}
The work of \cite{BravermanO07} introduces the {\it smooth histogram} approach for solving problems in the sliding window model. 
For a function $f$, given adjacent substreams $A$, $B$, and $C$, an $(\alpha,\beta)$-smooth function demands that if $(1-\beta)f(A\cup B)\le f(B)$, then $(1-\alpha)f(A\cup B\cup C)\le f(B\cup C)$ for some parameters $0<\beta\le\alpha<1$. 
Intuitively, once a suffix of a data stream becomes a $(1\pm\beta)$-approximation for a smooth function, then it is \emph{always} a $(1\pm\alpha)$-approximation, regardless of the subsequent updates that arrive in the stream. 
Here, a suffix of length $t$ of a stream is defined to be the last $t$ updates of the stream. 

The workof \cite{BravermanO07} develops a general algorithm for smooth functions and shows, for example, that $F_p$-moment estimation is sufficiently smooth, resulting in a $1/\varepsilon^{2+p}$ dependence for $F_p$-moment estimation for every $p \geq 1$, a $1/\varepsilon^3$ bound for $0 < p < 1$, and a $1/\varepsilon^4$ bound for estimating the empirical entropy. 
Unfortunately there are examples where this smoothness is best possible for all of these functions, suggesting that there is a gap between the sliding window model and data stream models. 
Surprisingly, we will show {\it this is not the case}, and we will give the first general technique for the sliding window model that goes beyond the smooth histogram approach, which has been the standard paradigm for nearly 15 years. 

\subsection{Our Contributions}
We show that there is no loss in $\frac{1}{\eps}$ factors, up to logarithmic factors, over the standard model of data streams for all of the aforementioned central data stream problems, in either the adversarially robust streaming model or the sliding window model. 
Our results hold for $F_p$-moment estimation for $p \in [0,2]$ and integers $p > 2$, $L_2$-heavy hitters, and empirical entropy estimation, and we give a general framework that can be applied to other problems as well.    
%which is known to be optimal \cite{IW03,W04}.
%Thus despite recent separation results~\cite{KaplanMNS21}, we obtain optimal bounds for fundamental problems, introducing a new yet natural algorithmic design that seems useful for applications in the streaming model in which temporal properties of the data play a central role. 
Our techniques introduce the following crucial concept, which surprisingly had not been considered for data streams before:

\begin{definition}[Difference Estimator]
\deflab{def:diff:est}
Given frequency vectors $u$ and $v$, an accuracy parameter $\eps>0$, a failure probability $\delta\in(0,1)$, and a ratio parameter $\gamma\in(0,1]$, a $(\gamma,\eps,\delta)$-difference estimator for a function $F$ outputs an additive $\eps\cdot F(u)$ approximation to $F(u+v)-F(u)$ with probability at least $1-\delta$, given $F(u+v)-F(u)\le\gamma\cdot F(u)$ and $F(v)\le\gamma F(u)$.  
\end{definition}

A difference estimator provides a path for obtaining an approximation algorithm, as follows: 
\begin{restatable}{framework}{framestitch}
\framelab{frame:stitch}
Given a monotonic function $F$, frequency vectors $u,v_1,\ldots,v_k$, a $\left(1+\frac{\eps}{2}\right)$ multiplicative approximation to $F(u)$, and $\left(\gamma_i,\eps_i,\frac{1}{3k}\right)$-difference estimators to $F(u+v_1+\ldots+v_i)-F(u+v_1+\ldots+v_{i-1})$ for $i\in\{1,\ldots,k\}$ (and $v_0=0$) such that $\prod_{i=1}^k(1+\gamma_i)\le 2$ and $\sum_{i=1}^k\eps_i\le\frac{\eps}{4}$, then their sum gives a $(1+\eps)$-approximation to $F(u+v_1+\ldots+v_i)$ with probability at least $\frac{2}{3}$. 
\end{restatable}
Moreover, even if multiple copies of the difference estimators are run simultaneously, the space dependency can still be $\frac{1}{\eps^2}$:
\begin{restatable}{theorem}{framestitchspace}
Suppose each $\left(\gamma_i,\eps,\frac{1}{3k}\right)$-difference estimator uses space $\frac{\gamma_i}{\eps^2}\cdot S$ and $S$ is a monotonic function in $k:=\O{\log n}$, $\frac{1}{\eps_i}$, and the input size. 
If $\gamma_i\le\frac{1}{2^i}$ and there are $N_i\le 2^i$ instances of $\left(\gamma_i,\eps,\frac{1}{3k}\right)$-difference estimators, then the total space is at most $\O{\frac{S\log n}{\eps^2}}$. 
\end{restatable}

While one can construct a difference estimator for $F(v+w_t)-F(v)$ by using separate sketches for $F(v+u)$ and $F(u)$, this defeats the whole point! 
Indeed, while approximation to $F(v+u)$ and to $F(u)$ with additive error $\frac{\eps}{2}\cdot F(v+u)$ can be used to provide an approximation to $F(v+u)-F(u)$ with additive error $\frac{\eps}{2}\cdot F(v+u)$, the space required for these approximations is simply the space required for $(1+\eps)$-approximations, so we are not utilizing the property that $F(v+u)-F(u)\le\gamma\cdot F(u)$ to obtain more space efficient algorithms. 
Thus, we first need to develop difference estimators for all of the above problems. It turns out that difference estimators
for the frequency moments $F_p$ can be used as building blocks for many other streaming problems, so these will be our focus. We show: 
 
%
%the $F_p$ frequency moment problem for $p\in[0,2]$ and integers $p>2$ with space dependent on $\gamma$, and these offer marked improvements since $\gamma \ll 1$: 
%
%In general, $F$ is non-linear and thus an estimate for $F(v+w_t)-F(v)$ is not necessarily recoverable from separate sketches for $F(v+w_t)$ and $F(v)$, so it does not seem possible to use a streaming algorithm $\calA$ for $F$ as a black box to obtain a difference estimator $\calB$ for $F$ in general. 
%For example, if $F(v)=\poly\left(\frac{1}{\eps}\right)$ and $F(v+w_t)-F(v)=1$ and $\calA$ gives $(1+\eps)$-approximations to $F(v+w_t)$ and $F(v)$, then the difference of their estimates can also be as large as $\poly\left(\frac{1}{\eps}\right) \gg 1$, and so does not satisfy our desired error guarantee. 
%Thus, we first develop difference estimators for the $F_p$-moment problem for $p\in[0,2]$ and integers $p>2$. 
\begin{theorem}
\thmlab{thm:diff:est}
There exist difference estimators for the $F_p$-moment problem for $p\in[0,2]$ and integers $p>2$. 
In particular, the difference estimator uses:
\begin{enumerate}
\item
$\O{\frac{\gamma}{\eps^2}\left(\log\frac{1}{\eps}+\log\log n+\log\frac{1}{\delta}\right)+\log n}$ bits of space for the distinct elements problem, $F_0$. 
%$\O{\frac{\gamma\log n}{\eps^2}\left(\log\frac{1}{\eps}+\log\frac{1}{\delta}\right)}$ bits of space for the distinct elements problem, $F_0$.
(See \lemref{lem:diff:est:F0}.)  
\item
$\O{\frac{\gamma\log n}{\eps^2}\left(\log\frac{1}{\eps}+\log\frac{1}{\delta}\right)}$ bits of space for $F_2$. 
(See \lemref{lem:diff:est:F2}.) 
\item
$\O{\frac{\gamma^{2/p}\log n}{\eps^2}(\log\log n)^2\left(\log\frac{1}{\eps}+\log\frac{1}{\delta}\right)}$ bits of space for $F_p$ with $p\in(1,2)$. 
(See \lemref{lem:diff:est:Fp:smallp}.)
\item
$\O{\frac{\gamma}{\eps^2}n^{1-2/p}\log^3 n\log\frac{n}{\delta}}$ bits of space for $F_p$ for integer $p>2$. 
(See \lemref{lem:diff:est:Fp:largep}).  
\end{enumerate}
\end{theorem}

We define the quantity $1+\frac{\eps}{\gamma}$ to be the \emph{effective accuracy} of the difference estimator, since it effectively serves as the multiplicative error if $F(v+u)-F(u)=\gamma\cdot F(u)$. 
The key property of our $(\gamma,\eps,\delta)$-difference estimators that we shall exploit is that their space complexity has dependency $\frac{\gamma^C}{\eps^2}$ for some constant $C\ge 1$. 
For example, for $\gamma=\O{\eps}$, the space dependence is only $\frac{1}{\eps}$, which allows us to use $\O{\frac{1}{\eps}}$ copies of the $(\gamma,\eps,\delta)$-difference estimators. 
Intuitively, because the difference is small, e.g., $\eps\cdot F(v)$, we can avoid using a $\frac{1}{\eps^2}$ space dependence to get an additive $\eps\cdot F(v)$ approximation to the difference. 
We develop several unrelated corollaries along the way which may be of independent interest, e.g., we show in \thmref{thm:smallp:strong:track} that Li's geometric mean estimator gives strong-tracking for $F_p$-moment estimation for $p\in(0,2)$, offering an alternative to the algorithm of~\cite{BlasiokDN17}. We also derandomize our algorithm using a generalization of a pseudorandom generator that fools half-space queries~\cite{GopalanKM18, JayaramW18}, and the argument may be useful for other models where temporal properties play a role. 

Using our concept of difference estimators, we develop quite general frameworks for both the adversarially robust streaming model and the sliding window model. 

%{\bf Norm estimation on adversarially robust streams:} 
%\paragraph{Norm estimation on adversarially robust streams.} 
\subsubsection{Our Results for Adversarially Robust Streams.}
We first present a space-efficient framework for adversarially robust streaming algorithms, provided there exists a corresponding difference estimator and strong tracker, i.e., a streaming algorithm that is correct at all times in the stream (see, e.g., \cite{BravermanCIW16,BravermanCINWW17,BlasiokDN17,Blasiok20} for examples of strong trackers). 
\begin{framework}
\framelab{frame:robust}
Let $\eps,\delta\in(0,1)$ be given constants and $F$ be a monotonic function with $(\eps,m)$-twist number $\lambda$. 
Suppose there exists a $(\gamma,\eps,\delta)$-difference estimator that uses space $\frac{\gamma}{\eps^2} S_F(m,\delta,\eps)$ and a strong tracker for $F$ that uses space $\frac{1}{\eps^2} S_F(m,\delta,\eps)$, where $S_F$ is a monotonic function in $m$, $\frac{1}{\delta}$, and $\frac{1}{\eps}$. 
Then there exists an adversarially robust streaming algorithm that outputs a $(1+\eps)$ approximation to $F$ that succeeds with constant probability and uses $\tO{\frac{\lambda}{\eps}\cdot S_F(n,\delta',\eps)}$ space, where $\delta'=\O{\frac{1}{\poly\left(\lambda,\frac{1}{\eps}\right)}}$. 
(Informal, see \thmref{thm:framework:flip}.)
\end{framework}
The space required for our algorithms is parameterized by the \emph{$(\eps,m)$-twist number} $\lambda$, which captures the number of times a given function on a subsequent substream is at an $\epsilon$ fraction of the given function on the prefix (see \defref{def:twistnumber} for a formal definition).  
Thus the twist number is related to the flip number $\psi$ used in \cite{Ben-EliezerJWY20,HassidimKMMS20}, which captures the number of times a given function on a stream changes by a factor of $(1+\eps)$ (see \defref{def:flipnumber} for a formal definition) -- indeed, for insertion-only streams, these two quantities are equivalent. 
\cite{KaplanMNS21} recently showed a lower bound of a $\sqrt{\lambda}$ space blow-up in insertion-only streams for particular problems; our results are the first to show that this $\sqrt{\lambda}$ overhead can be bypassed for a large class of important streaming problems, such as moment estimation, heavy hitters, and entropy, since $\lambda=\O{\frac{1}{\eps}\log n}$ in these problems. 
\frameref{frame:robust} breaks the lower bounds of \cite{KaplanMNS21} by requiring the existence of a difference estimator, which does not hold for their lower bound instance.

Often, $\lambda$ is bounded by $\O{\frac{1}{\eps}\log n}$ for insertion-only streams of length $m \leq \poly(n)$. 
For a variety of specific problems, we can further optimize our results to altogether avoid the $\log n$ overhead implied by \frameref{frame:robust}:
%Using the difference estimators from \thmref{thm:diff:est}, we obtain a framework for adversarially robust streaming algorithms:

%Let $\eps,\delta>0$ be given constants. 
%There is a general framework for adversarially robust streaming algorithms that output a $(1+\eps)$-approximation to any monotonic function $F$ with $(\eps,m)$-flip number $\lambda$ on a stream of length $m$, with $\log m=\O{\log n}$, given a $(\gamma,\eps,\delta)$-difference estimator and a streaming algorithm for $F$. 
%(Informal, see \thmref{thm:framework}.)
%\lambda=\poly\left(\frac{1}{\eps},\log n\right)$

\begin{theorem}
\thmlab{thm:main:adv:robust}
\frameref{frame:robust} can be further optimized to give an adversarially robust streaming algorithm that outputs a $(1+\eps)$-approximation to:
\begin{enumerate}
\item
The distinct elements problem, $F_0$, on insertion-only streams, using $\tO{\frac{1}{\eps^2}+\frac{1}{\eps}\log n}$ bits of space. 
%The distinct elements problem, $F_0$, on insertion-only streams, using $\tO{\frac{1}{\eps^2}\log n}$ bits of space. 
(See \thmref{thm:robust:opt:F0}.)  
\item
The $F_p$-moment estimation problem for $p\in(0,2]$ on insertion-only streams, using $\tO{\frac{1}{\eps^2}\log n}$ bits of space. 
%\newline\noindent $\O{\frac{1}{\eps^2}\log n\log^4\frac{1}{\eps}\left(\log\log n+\log\frac{1}{\eps}\right)}$ bits of space. 
(See \thmref{thm:robust:opt:F2} and \thmref{thm:robust:opt:Fp:smallp}.) 
\item
The Shannon entropy estimation problem on insertion-only streams, using $\tO{\frac{1}{\eps^2}\log^3 n}$ bits of space. 
(See \thmref{thm:robust:entropy}.)   
\item
The $F_p$-moment estimation problem for integer $p>2$ on insertion-only streams, using $\tO{\frac{1}{\eps^2}n^{1-2/p}}$ bits of space. 
%The $F_p$-moment estimation problem for integer $p>2$ on insertion-only streams, using $\O{\frac{1}{\eps^2}n^{1-2/p}\log^5 n\log^3\frac{1}{\eps}}$ bits of space. 
(See \thmref{thm:robust:opt:Fp:largep}.)
\item
The $L_2$-heavy hitters problem on insertion-only streams, using $\tO{\frac{1}{\eps^2}\log n}$ bits of space. 
%The $L_2$-heavy hitters problem on insertion-only streams, using $\O{\frac{1}{\eps^2}\log n\log^4\frac{1}{\eps}\left(\log\log n+\log\frac{1}{\eps}\right)}$ bits of space. 
(See \thmref{thm:robust:opt:HH}.)  
\item
The $F_p$-moment estimation problem on turnstile streams, using $\tO{\frac{\lambda}{\eps}\log^2 n}$ bits of space for $p\in[0,2]$.
%The $F_p$-moment estimation problem on turnstile streams, using $\O{\frac{\lambda}{\eps}\log^2 n\log^3\frac{1}{\eps}\left(\log\log n+\log\frac{1}{\eps}\right)}$ bits of space for $p=2$ and 
%$\O{\frac{\lambda}{\eps}\log^2 n\log^3\frac{1}{\eps}\min\left(\frac{1}{\eps^{o(1)}}\left(\log\log n+\log\frac{1}{\eps}\right), \log n\right)}$ for $p\in(0,2)$. 
%The space can further be improved to $\O{\frac{\lambda}{\eps}\log^2 n\log^3\frac{1}{\eps}\left(\log\log n+\log\frac{1}{\eps}\right)}$ bits of space for $p\in(0,2)$ under suitable cryptographic assumptions.  
(See \thmref{thm:robust:turnstile:Fp}.)
\end{enumerate}
\end{theorem}

\begin{figure*}[!htb]
\begin{center}
\resizebox{\columnwidth}{!}{
{\tabulinesep=1.2mm
\begin{tabu}{|c|c|c|c|}\hline
Problem & \cite{Ben-EliezerJWY20} Space & \cite{HassidimKMMS20} Space & Our Result \\\hline\hline
Distinct Elements & $\tO{\frac{\log n}{\eps^3}}$ & $\tO{\frac{\log^4 n}{\eps^{2.5}}}$ & $\tO{\frac{1}{\eps^2}+\frac{\log n}{\eps}}$ \\\hline
$F_p$ Estimation, $p\in(0,2]$ & $\tO{\frac{\log n}{\eps^3}}$ & $\tO{\frac{\log^4 n}{\eps^{2.5}}}$ & $\tO{\frac{\log n}{\eps^2}}$ \\\hline
%$F_p$ Estimation, $p\in(0,1]$ & $\tO{\frac{\log n}{\eps^3}}$ & $\tO{\frac{\log^4 n}{\eps^{2.5}}}$ & $\tO{\frac{1}{\eps^2}+\frac{\log n}{\eps}}$ \\\hline
Shannon Entropy & $\tO{\frac{\log^6 n}{\eps^5}}$ & $\tO{\frac{\log^4 n}{\eps^{3.5}}}$ & $\tO{\frac{\log^3 n}{\eps^2}}$ \\\hline
$L_2$-Heavy Hitters & $\tO{\frac{\log n}{\eps^3}}$ & $\tO{\frac{\log^4 n}{\eps^{2.5}}}$ & $\tO{\frac{\log n}{\eps^2}}$ \\\hline
$F_p$ Estimation, integer $p>2$ & $\tO{\frac{n^{1-2/p}}{\eps^3}}$ & $\tO{\frac{n^{1-2/p}}{\eps^{2.5}}}$ & $\tO{\frac{n^{1-2/p}}{\eps^2}}$ \\\hline
$F_p$ Estimation, $p\in(0,2]$, dynamic streams & $\tO{\frac{\psi\log^2 n}{\eps^2}}$ & $\tO{\frac{\log^3 n\sqrt{\psi\log n}}{\eps^2}}$ & $\tO{\frac{\lambda\log^2 n}{\eps}}$ \\\hline
\end{tabu}
}}
\end{center}
\vspace{-0.2in}
\caption{Adversarially robust streaming algorithms using optimized version of \frameref{frame:robust}. $\psi$ is the flip number for \cite{Ben-EliezerJWY20,HassidimKMMS20} and $\lambda$ is the twist number for our result.}
\figlab{fig:robust:results}
\end{figure*}
%Moreover, \frameref{frame:robust} gives a general approach for obtaining adversarially robust streaming algorithms, using a corresponding oblivious streaming algorithm and difference estimator. 
%
%\thmref{thm:main:adv:robust} gives adversarially robust streaming algorithms for estimating the number of distinct elements and $F_p$-moments on insertion-only s with an $\O{\frac{1}{\eps^2}}$ dependence. 
%By contrast, \cite{Ben-EliezerJWY20} achieved an $\tO{\frac{1}{\eps^3}}$ dependence for these problems while \cite{HassidimKMMS20} achieved an $\tO{\frac{1}{\eps^{2.5}}}$ dependence at the cost of $\polylog(n)$ factors.
%We use a novel approach based on the new concept of a difference estimator, described below. 
%for $p\in[0,2]$ and tight in $\poly(n)$ factors for integer $p>2$. 
The results achieved by \thmref{thm:main:adv:robust} achieve the optimal $\tO{\frac{1}{\eps^2}}$ dependence for these fundamental problems, and match the corresponding best non-adversarial algorithms in the insertion-only streaming model up to polylogarithmic factors.  
In particular, our results show that no loss in $\frac{1}{\eps}$ factors is necessary in the adversarially robust model.

For turnstile streams with twist number $\lambda$, we achieve dependence $\tO{\frac{\lambda\log^2 n}{\eps}}$, which improves upon the result of \cite{Ben-EliezerJWY20} if $\lambda=\psi$ for the flip number $\psi$. 
For $\lambda=\psi=o\left(\frac{\log^3 n}{\eps^2}\right)$, our result also improves upon the algorithm of \cite{HassidimKMMS20}. 
We summarize these results in \figref{fig:robust:results}. 

%{\bf Norm estimation on sliding windows:} 
\subsubsection{Our Results for the Sliding Window Model.} 
We next modify the difference estimators from \thmref{thm:diff:est} to develop a general framework for algorithms in the sliding window model, substantially improving upon the smooth histogram framework, and resolving longstanding questions on moment and entropy estimation algorithms in this model. 
Existing moment and entropy estimation sliding window algorithms crucially rely on the concept of smoothness, discussed above, which quantifies the ability of a suffix of a stream to diverge from the entire stream under worst-case updates to the stream; see \secref{sec:sliding} for a formal definition. 
A framework of \cite{BravermanO07} converts streaming algorithms into sliding window algorithms with an overhead in terms of the smoothness of the function to be approximated, provided that the function is monotonic, polynomially bounded, and smooth. 
We first give a framework showing that the space complexity of these functions need not depend on the smoothness parameter. 

Our framework again makes use of difference estimators, but in this case they are what we call 
{\it suffix pivoted difference estimators}. These difference estimators handle streams in which the ``larger'' 
frequency vector arrives after the ``smaller'' frequency vector, 
whereas the opposite is true for the difference estimators required for the adversarially robust streaming model. 
This is because in the sliding window model, we use difference estimators to 
``subtract off'' terms that have expired, whereas in the adversarially
robust streaming model we use difference estimators to ``add in'' terms that have recently appeared. 
Consequently, we must now develop suffix pivoted difference estimators for a wide range of problems, though fortunately
they turn out to be related to the difference estimators we have already developed. A more formal
description is given below, but the general theorem for our framework is given here:

%We develop specific suffix-pivoted difference estimators for the $F_p$ frequency moment problem. 

\begin{framework}
\framelab{frame:sw}
Let $\eps,\delta\in(0,1)$ be constants and $F$ be a monotonic and polynomially bounded function that is $(\eps,\eps^q)$-smooth for some constant $q\ge 0$. 
Suppose there exists a $(\gamma,\eps,\delta)$-suffix pivoted difference estimator that uses space $\frac{\gamma}{\eps^2} S_F(m,\delta,\eps)$ and a streaming algorithm for $F$ that uses space $\frac{1}{\eps^2} S_F(m,\delta,\eps)$, where $S_F$ is a monotonic function in $m$, $\frac{1}{\delta}$, and $\frac{1}{\eps}$. 
Then there exists a sliding window algorithm that outputs a $(1+\eps)$ approximation to $F$ that succeeds with constant probability and uses $\frac{1}{\eps^2}\cdot S_F(m,\delta',\eps)\cdot\poly\left(\log m,\log\frac{1}{\eps}\right)$ space, where $\delta'=\O{\frac{1}{\poly(m)}}$. 
(Informal, see \thmref{thm:sw:framework}.)
\end{framework}

%adapt \frameref{frame:sw} to handle functions with a relaxed form of subadditivity, as well as 
We can also optimize \frameref{frame:sw} to improve logarithmic factors, and obtain the following results for important streaming problems:
\begin{theorem}
\thmlab{thm:main:sw}
Let $\eps,\delta>0$ be given. 
Then \frameref{frame:sw} can be optimized to obtain sliding window algorithms that output a $(1+\eps)$-approximation to:
\begin{enumerate}
\item
The $F_p$-moment estimation problem for $p\in(0,2]$, using $\tO{\frac{1}{\eps^2}\log^3 n}$ bits of space. 
(See \thmref{thm:sliding:main}.)   
\item
The $F_p$-moment estimation problem for integers $p>2$, using $\tO{\frac{1}{\eps^2}\,n^{1-2/p}}$ bits of space. 
(See \thmref{thm:sliding:largep}.)  
\item
The Shannon entropy estimation problem, using $\tO{\frac{1}{\eps^2}\log^5 n}$ bits of space. 
(See \thmref{thm:sw:entropy}.)   
\end{enumerate}
\end{theorem}
Thus, our results show that no loss in $\frac{1}{\eps}$ factors is necessary in the sliding window model, bypassing previous
algorithms that were limited by the smooth histogram framework.  
The previous framework of \cite{BravermanO07} has an $\tO{\frac{\log^3 n}{\eps^3}}$ space dependence for $p\in(0,1]$, a $\tO{\frac{\log^3 n}{\eps^{2+p}}}$ space dependence for $p\in(1,2]$, an $\tO{\frac{n^{1-2/p}}{\eps^{2+p}}}$ space dependence for $p>2$, and an $\tO{\frac{\log^5 n}{\eps^4}}$ space dependence for entropy estimation using known techniques~\cite{HarveyNO08}. 
We note that there are specialized algorithms for the distinct elements and $L_2$-heavy hitters prolems in the sliding window model that also achieve the optimal dependence on the approximation factor \cite{BravermanGLWZ18}, so we do not state our results for those problems in this model. We summarize our sliding window model results in \figref{fig:sliding:results}. 

\begin{figure*}[!htb]
\begin{center}
{\tabulinesep=1.2mm
\begin{tabu}{|c|c|c|}\hline
Problem & \cite{BravermanO07} Space & Our Result \\\hline\hline
$L_p$ Estimation, $p\in(0,1)$ & $\tO{\frac{\log^3 n}{\eps^3}}$ & $\tO{\frac{\log^3 n}{\eps^{2}}}$ \\\hline
$L_p$ Estimation, $p\in(1,2]$ & $\tO{\frac{\log^3 n}{\eps^{2+p}}}$ & $\tO{\frac{\log^3 n}{\eps^{2}}}$ \\\hline
$L_p$ Estimation, integer $p>2$ & $\tO{\frac{n^{1-2/p}}{\eps^{2+p}}}$ & $\tO{\frac{n^{1-2/p}}{\eps^{2}}}$ \\\hline
Entropy Estimation & $\tO{\frac{\log^5 n}{\eps^4}}$ & $\tO{\frac{\log^5 n}{\eps^{2}}}$ \\\hline
\end{tabu}
}
\end{center}
\vspace{-0.2in}
\caption{Sliding window algorithms using optimized version of \frameref{frame:sw}.}
\figlab{fig:sliding:results}
\end{figure*}

\subsection{Overview of Our Techniques}
Given a function $F$ and a frequency vector $v$ implicitly defined through a data stream, we first suppose there exists a $(\gamma,\eps,\delta)$-difference estimator $\calB$ for $F$ and a streaming algorithm $\calA$ that gives a $(1+\eps)$-approximation to $F(u)$.  
For the purposes of the following discussion, and to build intuition, it suffices to assume that the space dependence on $\eps$ for $\calB$ and $\calA$ is $\tO{\frac{1}{\eps^2}}$; we shall thoroughly describe our difference estimators for specific functions $F$ below.  

%{\bf Sketch stitching:}
\paragraph{Stitching together sketches.} 
We first describe our new sketch stitching paradigm that breaks down the stream into contiguous blocks of updates, which can then be combined to approximate $F(v)$ as in the proof of \frameref{frame:stitch}. 
%Although approximating $F(v)$ is trivial given the algorithm $\calA$, for the purposes of intuition, we first describe our new sketch stitching approach that breaks down the stream into contiguous blocks of updates, which can then be combined to approximate $F(v)$; ultimately we shall alter the sketch stitching approach to achieve guarantees unobtainable by $\calA$ alone.  
Observe that $F(v)=F(u)+(F(v)-F(u))$ for any frequency vector $u$. 
Similarly, $F(v)=F(u_1)+\sum_{k=1}^{\beta} (F(u_{k+1})-F(u_k))$ for any frequency vectors $u_1,\ldots,u_{\beta}$, provided that $v=u_{\beta+1}$. 
Hence, if we had $(1+\eps)$-approximations to $F(u_1)$ and each of the differences $F(u_{k+1})-F(u_k)$, then we could add these estimates to obtain a $(1+\eps)$-approximation to $F(v)$. 
In particular, if we view $u_1,\ldots,u_{\beta}$ as the frequency vectors induced by the prefixes of the stream of length $t_1<\ldots<t_{\beta}$, respectively, then we can view the approximation to $F(u_1)$ and each $F(u_{k+1})-F(u_k)$ as a sketch of how much a consecutive block of updates in the stream contributes to $F$. 
Here, recall that a prefix and suffix of length $t$ of a stream are respectively defined to be the first and last $t$ updates of the stream. 
The difference estimator $\calB$ will be used to sketch {\it each difference} $F(u_{k+1})-F(u_k)$ and $\calA$ will be used to sketch $F(u_1)$. 
We can then ``stitch'' the value of $F$ on the stream by adding the estimates of each sketch, as in the proof of \frameref{frame:stitch}, where a much weaker additive error suffices for approximating each of these differences. 
Observe that the frequency vectors $u_1,\ldots,u_{\beta}$ are unconstrained in this setting; we now describe a natural choice for these frequency vectors based on their contributions to the value of $F$ on the stream. 

%{\bf Granularity change:}
\paragraph{Granularity change.} 
Suppose $v \succeq u$, that is, for all $j$, $v_j \geq u_j \geq 0$. Let $i$ be the smallest positive integer such that $F(v)\le 2^i$, and let $t$ be the first time for which $F(u)\ge 2^{i-1}$ for the frequency vector $u$ induced by the prefix of the stream of length $t$. 
Since $F(v)-F(u)\le\frac{1}{2}\cdot F(v)$, the sum of a $\left(1+\frac{\eps}{2}\right)$-approximation to $F(u)$ and a $(1+\eps)$-approximation to $F(v)-F(u)$ is a $(1+\eps)$-approximation to $F(v)$. 

More generally, suppose we break down the stream into the earliest times $t_1,t_2,\ldots,t_\beta$ for which the corresponding frequency vectors $u_1,u_2,\ldots,u_\beta$ induced by the prefix of the stream of lengths $t_1,t_2,\ldots,t_\beta$, respectively, satisfy $F(u_1+\ldots+u_i)\ge\left(\frac{1}{2}+\ldots+\frac{1}{2^i}\right)\cdot F(v)$ for each $i\in[\beta]$. 
Then if we have a $\left(1+\frac{2^k\eps}{\beta}\right)$-approximation to $F(u_{k+1})-F(u_k)$ for each $k\in[\beta]$, where we use the convention that $v=u_{\beta+1}$, then their sum is a $(1+\eps)$-approximation to $F(v)$. 
Intuitively, $F(u_{k+1})-F(u_k)$ is at most $\frac{1}{2^k}\cdot F(v)$, so we only require a $\left(1+\frac{2^k\eps}{\beta}\right)$-approximation to $F(u_{k+1})-F(u_k)$. 
For the purposes of this discussion, we informally call the difference estimator with effective accuracy $\left(1+\frac{2^k\eps}{\beta}\right)$ a \emph{level $k$ estimator}. 
We can thus maintain sketches of different qualities for each of these blocks of the stream and stitch together the outputs to estimate $F(v)$. 

Finally, note that for $\beta$ roughly equal to $\log\frac{1}{\eps}$, we have $F(u_{\beta+1})-F(u_\beta)=\O{\eps}\cdot F(v)$. 
Thus even if we always output zero for the level $\beta$ estimator, we can incur at most an additive error of $\O{\eps}\cdot F(v)$. 
Hence, it suffices to set $\beta$ to be roughly $\log\frac{1}{\eps}$. 
We stress this is not a hierarchical blocking strategy to organize existing sketches - the sketches in the hierarchy are {\it not sketches of the function $F$ applied to a vector}. That is, we are not sketching the difference $v-u$ and estimating $F(v-u)$, but rather we are sketching $u$ and $v$ to estimate $F(v)-F(u)$; sketching such differences in the required amount of space was unknown and already makes filling in a sketch at a single node in the hierarchy challenging.    
%Observe that despite its simplicity, the sketch stitching and granularity changing approach is a fundamentally new technique that may be of independent interest in the design of data stream algorithms. 

We now describe applications to the adversarially robust streaming and sliding window models. 

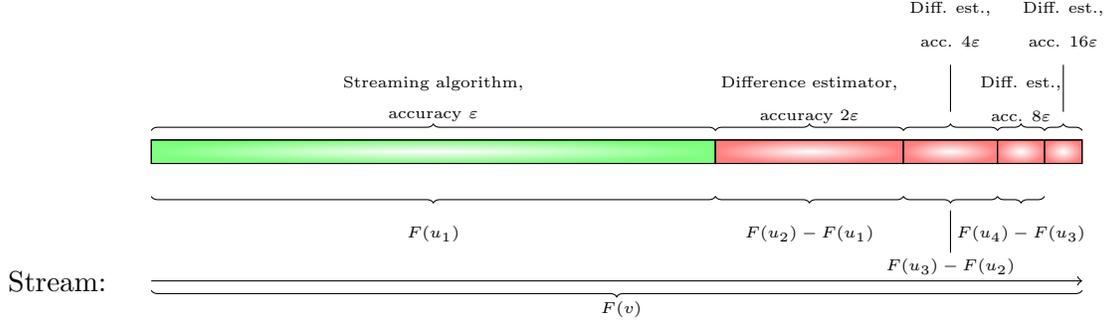
\begin{figure*}[!htb]
\centering
\begin{tikzpicture}[scale=1.25]
\draw [->] (0,-1) -- (9.9,-1);
\node at (-1,-1){Stream:};
\draw [decorate,decoration={brace}] (9.9,-1.1) -- (0,-1.1);
\node at (5,-1.3){\tiny{$F(v)$}};

\filldraw[shading=radial, inner color = white, outer color = green!50!, opacity=1] (0,0.25) rectangle+(6,0.25);
\draw (0,0.25) rectangle+(6,0.25);
\draw [decorate,decoration={brace}] (0,0.6) -- (6,0.6);
\node at (3,1.1){\tiny{Streaming algorithm,}};
\node at (3,0.75){\tiny{accuracy $\eps$}};
\draw [decorate,decoration={brace}] (6,-0.1) -- (0,-0.1);
\node at (3,-0.5){\tiny{$F(u_1)$}};

\filldraw[shading=radial, inner color = white, outer color = red!50!, opacity=1] (6,0.25) rectangle+(2,0.25);
\draw (6,0.25) rectangle+(2,0.25);
\draw [decorate,decoration={brace}] (6,0.6) -- (8,0.6);
\node at (7,1.1){\tiny{Difference estimator,}};
\node at (7,0.75){\tiny{accuracy $2\eps$}};
\draw [decorate,decoration={brace}] (8,-0.1) -- (6,-0.1);
\node at (7,-0.5){\tiny{$F(u_2)-F(u_1)$}};

\filldraw[shading=radial, inner color = white, outer color = red!50!, opacity=1] (8,0.25) rectangle+(1,0.25);
\draw (8,0.25) rectangle+(1,0.25);
\draw [decorate,decoration={brace}] (8,0.6) -- (9,0.6);
\node at (8.5,1.9){\tiny{Diff. est.,}};
\node at (8.5,1.55){\tiny{acc. $4\eps$}};
\draw (8.5,1.3) -- (8.5,0.8);
\draw [decorate,decoration={brace}] (9,-0.1) -- (8,-0.1);
\node at (8.5,-0.85){\tiny{$F(u_3)-F(u_2)$}};
\draw (8.5,-0.7) -- (8.5,-0.25);

\filldraw[shading=radial, inner color = white, outer color = red!50!, opacity=1] (9,0.25) rectangle+(0.5,0.25);
\draw (9,0.25) rectangle+(0.5,0.25);
\draw [decorate,decoration={brace}] (9,0.6) -- (9.5,0.6);
\node at (9.25,1.1){\tiny{Diff. est.,}};
\node at (9.25,0.75){\tiny{acc. $8\eps$}};
\draw [decorate,decoration={brace}] (9.5,-0.1) -- (9,-0.1);
\node at (9.25,-0.5){\tiny{$F(u_4)-F(u_3)$}};

\filldraw[shading=radial, inner color = white, outer color = red!50!, opacity=1] (9.5,0.25) rectangle+(0.4,0.25);
\draw (9.5,0.25) rectangle+(0.4,0.25);
\draw [decorate,decoration={brace}] (9.5,0.6) -- (9.9,0.6);
\node at (9.7,1.9){\tiny{Diff. est.,}};
\node at (9.7,1.55){\tiny{acc. $16\eps$}};
\draw (9.7,1.3) -- (9.7,0.8);
\end{tikzpicture}
\caption{Stitching together sketches and changing granularities. 
The difference estimators $\calB$ are in red and the streaming algorithm $\calA$ is in green. 
Observe that if $F(u_i)-F(u_{i-1})\approx\frac{1}{2^i}\cdot F(u_i)$, then the sum of the estimates is a $(1+\O{\eps})$-approximation to $F(v)$. 
}
\figlab{fig:stitch}
\end{figure*}

%{\bf Framework for adversarial robustness:}
\paragraph{Robust framework challenges.} 
The challenge in implementing the described approach in the adversarially robust streaming model is that once the output of a subroutine is revealed to the adversary, the input can now depend on previous inputs as well as the internal randomness of any subroutine possibly learned by the adversary. 
To circumvent this issue, the framework of \cite{Ben-EliezerJWY20} only reveals a new output when the internal estimate formed by various subroutines has increased by a power of roughly $(1+\eps)$. 
Thus, \cite{Ben-EliezerJWY20} requires maintaining $\O{\frac{1}{\eps}\log n}$ separate algorithms and repeatedly switching to new, previously unused sketches, which is a technique called sketch switching. 
Each sketch has a $\frac{1}{\eps^2}$ space dependence, for an overall $\frac{1}{\eps^3}$ space dependence, which is prohibitive. 
Another technique, called computation paths in \cite{Ben-EliezerJWY20}, instead sets the failure probability of a non-adaptive streaming algorithm to be $n^{-\O{(\log n)/\varepsilon}}$, but suffers a similar overall $\frac{1}{\eps^3}$ space dependence.

Adversarial inputs are problematic for several reasons. First, the correctness of a difference estimator can only be analyzed and ensured on a non-adaptively chosen stream. Second, for constant $k$, the level $k$ estimator handles a block of the stream that contributes a constant factor to the value of $F(v)$, and hence, the level $k$ estimator will cause the overall estimate to increase by a factor of $(1+\eps)$ multiple times, thus potentially compromising its internal randomness. We now describe how we overcome these issues and our general framework.

\paragraph{Framework for adversarial robustness.} 
We only estimate $F(u)$ to high accuracy, where $u$ is a frequency vector induced by the prefix of the stream at the first time $F(u)\ge 2^{i-1}$, for an integer $i$, and where the vector $v$ of subsequent stream updates satisfies $F(v)\le 2^i$. 
We observe that before $u$ can grow to $u_1$, where recall $F(u_1) \geq \frac{1}{2} \cdot F(v)$, $u$ must first grow to a frequency vector $w_1$ for which $F(w_1)\approx(1+\eps) F(u)$. 
Instead of maintaining a level $1$ estimator for $w_1$, it suffices to maintain a level $\beta$ estimator for the difference $F(w_1)-F(u)$, where $\beta=\O{\log\frac{1}{\eps}}$ as before. 
That is, since $F(w_1)-F(u)$ is small, it suffices to maintain only a constant factor approximation to the difference. 
Once the stream reaches the time with underlying frequency vector $w_1$, we approximate $F(w_1)$ by stitching together sketches and then discarding our sketch for $F(w_1)-F(u)$. 

By similar reasoning, before $w_1$ can grow to $u_1$, it must first grow to a frequency vector $w_2$ for which $F(w_2)\approx(1+2\eps) F(u)$. 
We maintain a level $\beta-1$ estimator for $F(w_2)-F(u)$, reveal the output once the stream reaches the time that induces $w_2$, and then discard our sketch since its randomness has been compromised. 

It is crucial that we next consider $w_3$ to be the first frequency vector for which $F(w_3)\approx(1+3\eps) F(u)$, rather than say, for which $F(w_3)\approx(1+4\eps) F(u)$. 
%, i.e., that we consider an arithmetic increase rather than a geometric increase. 
Indeed, if we were to use a single difference estimator to track a block of the stream in which $F$ increases from $(1+2\eps)F(u)$ to $(1+4\eps)F(u)$, then we cannot reveal the output of the difference estimator when the stream hits $(1+3\eps)F(u)$ without compromising its randomness. 
However, if we do not reveal the output, then we also might no longer have a $(1+\eps)$-approximation to the value of $F$ when the prefix of the stream reaches value $(1+3\eps)F(u)$. 

Therefore, let $w_3$ be the first frequency vector for which $F(w_3)\approx(1+3\eps) F(u)$. 
Instead of estimating $F(w_3)-F(u)$ directly, we note that $F(w_3)-F(w_2)\approx\eps F(u)$. 
Hence, we estimate $F(w_3)$ using a level $\beta$ estimator for $F(w_3)-F(w_2)$ and the previous level $\beta-1$ estimator for $F(w_2)-F(u)$. 
Note we have not changed any coordinates of $w_2$, so that the level $\beta-1$ estimator remains correct for $F(w_2)-F(u)$ even though its output has been revealed. 

We now see that if $w_k$ is the first frequency vector for which $F(w_k)\approx(1+k\eps) F(u)$, then we use level $j$ estimators corresponding to the nonzero bits in the \emph{binary representation} of $k$. 
Intuitively, the binary representation can be thought of as encoding into the correct path to stitch together sketches from a binary tree on the stream, where the stream is partitioned according to the value of each difference estimator rather than the length of the stream. 
See \figref{fig:tree} for an example. 
It follows that we need roughly $2^{k}$ instances of the level $k$ estimator, for $k\in[\beta]$, where $\beta$ is roughly $\log\frac{1}{\eps}$. 
As a level $k$ estimator requires effective accuracy roughly $\left(1+2^k\eps\right)$, we will not quite achieve a geometric series, but the total space will increase by at most $\polylog\frac{1}{\eps}$ factors. 
%But because a level $k$ estimator requires effective accuracy roughly $\left(1+2^k\eps\right)$, then the space dependence required from $\calB$ is roughly $\frac{1}{2^{k}\eps^2}$, so that the total space forms a geometric series that does not incur additional $\O{\frac{1}{\eps}}$ factors. 
Hence, the total dependence on $\eps$ in the space is $\tO{\frac{1}{\eps^2}}$, as desired.

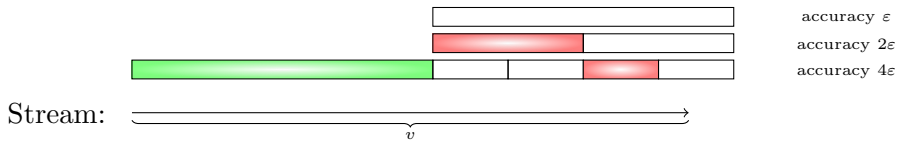
\begin{figure*}[!htb]
\centering
\begin{tikzpicture}[scale=1]
\draw [->] (0,-0.2) -- (7.4,-0.2);
\node at (-1,-0.2){Stream:};
\draw [decorate,decoration={brace}] (7.4,-0.3) -- (0,-0.3);
\node at (3.7,-0.5){\tiny{$v$}};

\filldraw[shading=radial, inner color = white, outer color = green!50!, opacity=1] (0,0.25) rectangle+(4,0.25);
%\draw [decorate,decoration={brace}] (0,0.6) -- (4,0.6);
%\node at (2,1.1){\tiny{Streaming algorithm,}};
%\node at (2,0.75){\tiny{accuracy $\eps$}};
%\draw [decorate,decoration={brace}] (4,-0.1) -- (0,-0.1);
%\node at (2,-0.5){\tiny{$F(u_1)$}};

%\filldraw[shading=radial, inner color = white, outer color = red!50!, opacity=1] (4,0.25) rectangle+(1,0.25);
\draw (4,0.25) rectangle+(1,0.25);
%\filldraw[shading=radial, inner color = white, outer color = red!50!, opacity=1] (5,0.25) rectangle+(1,0.25);
\draw (5,0.25) rectangle+(1,0.25);
\filldraw[shading=radial, inner color = white, outer color = red!50!, opacity=1] (6,0.25) rectangle+(1,0.25);
%\filldraw[shading=radial, inner color = white, outer color = red!50!, opacity=1] (7,0.25) rectangle+(1,0.25);
\draw (7,0.25) rectangle+(1,0.25);
\node at (9.5,0.25+0.1){\tiny{accuracy $4\eps$}};

\filldraw[shading=radial, inner color = white, outer color = red!50!, opacity=1] (4,0.6) rectangle+(2,0.25);
%\filldraw[shading=radial, inner color = white, outer color = red!50!, opacity=1] (6,0.6) rectangle+(2,0.25);
\draw (6,0.6) rectangle+(2,0.25);
\node at (9.5,0.6+0.1){\tiny{accuracy $2\eps$}};

%\filldraw[shading=radial, inner color = white, outer color = red!50!, opacity=1] (4,0.95) rectangle+(4,0.25);
\draw (4,0.95) rectangle+(4,0.25);
\node at (9.5,0.95+0.1){\tiny{accuracy $\eps$}};
%\draw [decorate,decoration={brace}] (6,0.6) -- (8,0.6);
%\node at (7,1.1){\tiny{Difference estimator,}};
%\node at (7,0.75){\tiny{accuracy $2\eps$}};
%\draw [decorate,decoration={brace}] (8,-0.1) -- (6,-0.1);
%\node at (7,-0.5){\tiny{$F(u_2)-F(u_1)$}};
\end{tikzpicture}
\caption{The outputs of the difference estimators in red are stitched together to form an estimate for $F(v)$. The partitions of the stream are determined by the output of the difference estimators (rather than the length of the stream).}
\figlab{fig:tree}
\end{figure*}

%{\bf Optimized space:}
\paragraph{Optimized space.} 
It is possible to optimize the space usage of our framework using standard ideas for streams, used also in \cite{Ben-EliezerJWY20}:
(1) instead of maintaining $\O{\log n}$ instances of $\calA$, it suffices to maintain only $\O{\log\frac{1}{\eps}}$ instances at a given time, since we can drop a prefix of the stream if it only contributes a $\poly(\varepsilon)$ fraction towards the value of $F$  
(2) if $\calB$ and $\calA$ have the \emph{strong-tracking} property, then we can avoid a union bound over $\O{m}$ possible frequency vectors $u$ and $v$ and instead union bound over $\poly\left(\log n,\log\frac{1}{\eps}\right)$ instances of the algorithm. 
To utilize (2), since our framework cannot use existing sketches, we need to show strong tracking holds for our difference estimators. 

\paragraph{Sliding window framework challenges.} 
Recall that only the most recent $W$ updates in a stream of length $m$ form the underlying dataset in the sliding window model. 
Because expirations of old updates are performed implicitly and the algorithm does not know when the stream will end, it is challenging to design algorithms in the sliding window model. 

A natural approach is to adapt the smooth histogram approach discussed above \cite{BravermanO07}. 
%Given adjacent substreams $A$, $B$, and $C$, an $(\alpha,\beta)$-smooth function demands that if $(1-\beta)f(A\cup B)\le f(B)$, then $(1-\alpha)f(A\cup B\cup C)\le f(B\cup C)$ for some parameters $0<\beta\le\alpha<1$. 
%Intuitively, once a suffix of a data stream becomes a $(1\pm\beta)$-approximation for a smooth function, then it is \emph{always} a $(1\pm\alpha)$-approximation, regardless of the subsequent updates that arrive in the stream. 
The smooth histogram maintains a number of timestamps throughout the stream, along with sketches for the suffixes of the stream starting at each timestamp. 
The timestamps maintain the invariant that at most three suffixes produce values that are within a $(1-\beta)$ factor of each other, since any two of the sketches always output values that are within a $(1-\alpha)$ factor afterwards. 
Unfortunately, the smooth histogram incurs a very large overhead due to the smoothness parameter of the function. 
For example, \cite{BravermanO07} showed that $F_p$ is $(\eps,\O{\eps^p})$-smooth for $p\ge 1$, so $\O{\frac{1}{\eps^p}\log n}$ instances of $F_p$ streaming algorithms must be maintained. 
Each algorithm has space dependency $\frac{1}{\eps^2}$, for a total of $\frac{1}{\eps^{2+p}}$ dependency, which is prohibitive
for large $p$. 

Another approach specifically for the $F_p$-moment estimation problem that one could try is to generalize the streaming algorithm of \cite{IndykW05}, which is based on subsampling the universe at multiple rates and using $(1+\eps)$-approximations to the heavy hitters in each subsampled level. 
The challenge with this approach is that efficient heavy hitter algorithms in the sliding window model~\cite{BravermanGLWZ18} do not give $(1+\eps)$-approximations to the frequency of the heavy hitters and are {\it inherently biased} due to items falling outside of a window, and so it is not clear how to produce $(1+\eps)$-approximations without incurring additional large $\O{\frac{1}{\eps}}$ factor overheads.

\paragraph{Framework for sliding windows.} 
To adapt our techniques to the sliding window model, we first observe that since prefixes of the stream may expire, we need a good approximation to the suffix of the stream rather than to the prefix. 
Thus, we separately run streaming algorithms $\calA$ on various suffixes of the stream similar to the smooth histogram framework. 
However, instead of maintaining separate instances of the streaming algorithm $\calA$ each time the value of $F$ on a suffix increases by a smoothness parameter, we maintain instances each time the value roughly doubles. 
By a standard sandwiching argument, we maintain an instance of $\calA$ starting at some time $t_0\le m-W+1$, whose output is within a factor $2$ of the value of $F$ of the sliding window. 
That is, the value of $F$ induced by the suffix of the stream starting at $t_0$ is at most twice the value of $F$ on the sliding window. 
Since $\calA$ gives a $(1+\eps)$-approximation to the value of $F$ on the suffix starting at $t_0$, in order to obtain a $(1+\eps)$-approximation we then need to remove the additional contribution of the updates between times $t_0$ and $m-W+1$, the latter being the starting time of the sliding window. 
Here we can again partition these elements into separate blocks based on their contribution to the value of $F$.  
We can maintain separate sketches for these blocks, with varying granularities, and stitch these sketches together at the end. 
See \figref{fig:sliding} for intuition. 

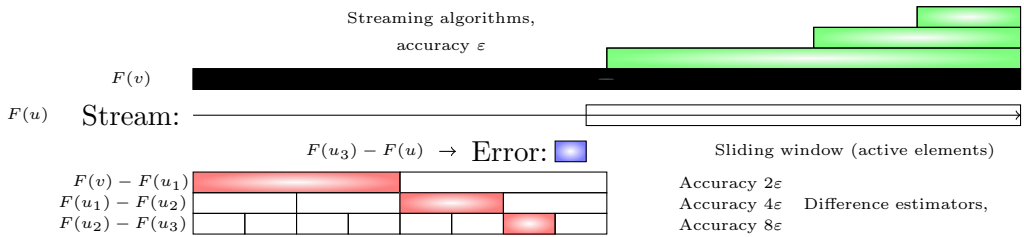
\begin{figure*}[!htb]
\centering
\begin{tikzpicture}[scale=1.10] %.25
\draw [->] (0,-0.25+0.125+0.0625) -- (10,-0.25+0.125+0.0625);
\node at (-0.75,-0.25+0.125+0.0625){Stream:};
\node at (-2,-0.25+0.125+0.0625){\tiny{$F(u)$}};

\draw (4.75,-0.5+0.25+0.0625) rectangle+(5.25,0.25);
\node at (8,-0.375-0.125){\tiny{Sliding window (active elements)}};

\filldraw[shading=radial, inner color = white, outer color = green!50!, opacity=1] (0,0.25) rectangle+(10,0.25);
\draw (0,0.25) rectangle+(10,0.25);
\node at (-0.75,0.25+0.125){\tiny{$F(v)$}};

\filldraw[shading=radial, inner color = white, outer color = green!50!, opacity=1] (5,0.5) rectangle+(5,0.25);
\draw (5,0.5) rectangle+(5,0.25);
\filldraw[shading=radial, inner color = white, outer color = green!50!, opacity=1] (7.5,0.75) rectangle+(2.5,0.25);
\draw (7.5,0.75) rectangle+(2.5,0.25);
\filldraw[shading=radial, inner color = white, outer color = green!50!, opacity=1] (8.75,1) rectangle+(1.25,0.25);
\draw (8.75,1) rectangle+(1.25,0.25);

\node at (3,1.1){\tiny{Streaming algorithms,}};
\node at (3,0.75){\tiny{accuracy $\eps$}};

\node at (-0.75,-1+0.125){\tiny{$F(v)-F(u_1)$}};
\filldraw[shading=radial, inner color = white, outer color = red!50!, opacity=1] (0,-1) rectangle+(2.5,0.25);
%\filldraw[shading=radial, inner color = white, outer color = blue!50!, opacity=1] (2.5,-1) rectangle+(2.5,0.25);
\draw (2.5,-1) rectangle+(2.5,0.25);

\node at (-0.85,-1.25+0.125){\tiny{$F(u_1)-F(u_2)$}};
%\filldraw[shading=radial, inner color = white, outer color = blue!50!, opacity=1] (0,-1.25) rectangle+(1.25,0.25);
\draw (0,-1.25) rectangle+(1.25,0.25);
%\filldraw[shading=radial, inner color = white, outer color = blue!50!, opacity=1] (1.25,-1.25) rectangle+(1.25,0.25);
\draw (1.25,-1.25) rectangle+(1.25,0.25);
\filldraw[shading=radial, inner color = white, outer color = red!50!, opacity=1] (2.5,-1.25) rectangle+(1.25,0.25);
%\filldraw[shading=radial, inner color = white, outer color = blue!50!, opacity=1] (3.75,-1.25) rectangle+(1.25,0.25);
\draw (3.75,-1.25) rectangle+(1.25,0.25);

\node at (-0.85,-1.5+0.125){\tiny{$F(u_2)-F(u_3)$}};
%\filldraw[shading=radial, inner color = white, outer color = blue!50!, opacity=1] (0,-1.5) rectangle+(0.625,0.25);
\draw (0,-1.5) rectangle+(0.625,0.25);
%\filldraw[shading=radial, inner color = white, outer color = blue!50!, opacity=1] (0.625,-1.5) rectangle+(0.625,0.25);
\draw (0.625,-1.5) rectangle+(0.625,0.25);
%\filldraw[shading=radial, inner color = white, outer color = blue!50!, opacity=1] (0.625*2,-1.5) rectangle+(0.625,0.25);
\draw (0.625*2,-1.5) rectangle+(0.625,0.25);
%\filldraw[shading=radial, inner color = white, outer color = blue!50!, opacity=1] (0.625*3,-1.5) rectangle+(0.625,0.25);
\draw (0.625*3,-1.5) rectangle+(0.625,0.25);
%\filldraw[shading=radial, inner color = white, outer color = blue!50!, opacity=1] (0.625*4,-1.5) rectangle+(0.625,0.25);
\draw (0.625*4,-1.5) rectangle+(0.625,0.25);
%\filldraw[shading=radial, inner color = white, outer color = blue!50!, opacity=1] (0.625*5,-1.5) rectangle+(0.625,0.25);
\draw (0.625*5,-1.5) rectangle+(0.625,0.25);
\filldraw[shading=radial, inner color = white, outer color = red!50!, opacity=1] (0.625*6,-1.5) rectangle+(0.625,0.25);
%\filldraw[shading=radial, inner color = white, outer color = blue!50!, opacity=1] (0.625*7,-1.5) rectangle+(0.625,0.25);
\draw (0.625*7,-1.5) rectangle+(0.625,0.25);

\node at (8.5,-1.25+0.1){\tiny{Difference estimators,}};
\node at (6.5,-1+0.1){\tiny{Accuracy $2\eps$}};
\node at (6.5,-1.25+0.1){\tiny{Accuracy $4\eps$}};
\node at (6.5,-1.5+0.1){\tiny{Accuracy $8\eps$}};

\filldraw[shading=radial, inner color = white, outer color = blue!50!, opacity=1] (0.625*7,-0.375-0.25) rectangle+(4.75-0.625*7,0.25);
\node at (0.625*7-0.55,-0.375-0.125){Error:};
\node at (0.625*7-2.3,-0.375-0.125){\tiny{$F(u_3)-F(u)$}};
\draw[<-] (0.625*7-1.2,-0.375-0.125) -- (0.625*7-1.4,-0.375-0.125);
\end{tikzpicture}
\caption{
The streaming algorithm (in green) and difference estimators (in white/red) can be used to estimate the $F_p$ value of the sliding window up to a small error caused by the elements in blue by subtracting the red estimates of the difference estimators from the estimate of the smallest suffix that contains the active elements in the sliding window.}
%Note, we must always work with difference estimators rather than sketches of $F_p$ itself, since for $p \neq 1$, $F_p$ is non-linear.}
\figlab{fig:sliding}
\end{figure*}

\paragraph{Suffix-pivoted difference estimators.} 
A barrier in the sliding window model is that the difference estimator partitions the prefix of the stream rather than the suffix. 
For example, compare the locations of the difference estimators between \figref{fig:stitch} and \figref{fig:sliding}. 
In both the adversarially robust setting and the sliding window model, the difference estimator must approximate $F(v)-F(u)$ for frequency vectors $v\succeq u$. 
However in the robust setting, $u$ is a fixed frequency vector induced by the prefix of the stream while in the sliding window model, $u$ is a growing frequency vector induced by the suffix of the stream, where we recall that a prefix and suffix of length $t$ are respectively defined to be the first and last $t$ updates of a stream. 
Thus, we define a \emph{suffix-pivoted} difference estimator to handle the case where the vector $u$ can change along with updates in the stream. 
Fortunately, the fixed-prefix difference estimators that we have developed also function as suffix-pivoted difference estimators under this definition (we shall shortly demand a stronger definition). 
We partition the stream to induce frequency vectors $v$ and $u_1$ so that $F(v+u_1)-F(u_1)\approx 2^k\cdot F(v+u_1)$ for a level $k$ suffix-pivoted difference estimator. 
We maintain $\beta$ granularities, where $\beta$ is roughly a logarithmic function of the smoothness parameter that is generally $\O{\log\frac{1}{\eps}}$. 
We thus obtain an additive $\O{\frac{\eps}{\beta}}\cdot F(v+u_1)$ error at each level $i\in[\beta]$, and summing the errors across all $\beta$ granularities, gives $\eps\cdot F(v+u_1)$ total additive error, i.e., $(1+\eps)$ multiplicative approximation to $F(v+u_1)$. 

\paragraph{Issues with dynamic differences.}
There is a subtle but significant issue with na\"{i}vely applying the $(\gamma,\eps,\delta)$-difference estimators in the sliding window model. 
Recall that a level $k$ difference estimator guarantees an additive $\eps\cdot F(v+u_1)$ error to the difference $F(v+u_1)-F(u_1)$ \emph{provided that} $F(v+u_1)-F(u_1)\le\frac{1}{2^k}\cdot F(v+u_1)$. 
In the sliding window model, additional updates to the stream may arrive and induce a frequency vector $u_2$ for the same splitting time with $F(v+u_2)-F(u_2)\gg\frac{1}{2^k}\cdot F(v+u_2)$, which is too large for the accuracy of a level $k$ difference estimator. 

Because of the smoothness parameter of an $(\eps,\eps^q)$-smooth function $F$, we observe that if the difference $F(v+u_1)-F(u_1)$ at some point is roughly $\eps^q\cdot F_2(v+u_1)$, then the difference $F(v+u_2)-F(u_2)\le\eps\cdot F(v+u_2)$ at \emph{any} later point in the stream, regardless of subsequent updates. 
Now the issue is that we maintain $\frac{1}{\eps^q}$ difference estimators corresponding to additive error $\eps^q\cdot F(v+u_1)$, but since each of these differences can eventually be as large as $\eps\cdot F(v+u_2)$, we potentially need higher accuracy for these differences, using space dependence $\frac{1}{\eps}$ for each difference estimator. 
Unfortunately, this results in space dependence $\frac{1}{\eps^{1+q}}$. 

For example, consider the $F_2$ moment function in \figref{fig:sw:bad} and suppose $v=(0,1)$ and $u_1=\left(\frac{1}{\eps},0\right)$ so that $F_2(v+u_1)=\frac{1}{\eps^2}+1$ and $F_2(v+u_1)-F_2(u_1)=1\le\eps^2\cdot F_2(v+u_1)$. 
Hence, we have $\gamma=\eps^2$ and allocate space proportional to $\frac{\gamma}{\eps^2}$ for the difference estimator for $v$. 
We then have a number of updates so that $u_2=\left(\frac{1}{\eps},\frac{1}{\eps}\right)$. 
Then $F_2(v+u_2)=\frac{2}{\eps^2}+\frac{2}{\eps}+1$ and $F_2(v+u_2)-F_2(u_2)=\frac{2}{\eps}+1\ge\frac{\eps}{2}\cdot F_2(v+u_2)$ for sufficiently small constant $\eps>0$. 
Now we have $\gamma'\ge\frac{\eps}{2}$, so we should allocate space proportional to $\frac{\gamma'}{\eps^2}\approx\frac{1}{\eps}$ for the difference estimator for $v$, but we already only allocated space proportional to $\frac{\gamma}{\eps}\approx 1$. 
Therefore, we do not have any guarantee that our difference estimator will be accurate!
Moreover, there can be $\O{\frac{1}{\eps^2}}$ such blocks $v$, so we cannot afford to allocate $\frac{\gamma'}{\eps^2}$ space for each block, which would result in space dependency $\frac{1}{\eps^3}$. 

\begin{figure*}[!htb]
\centering
\begin{tikzpicture}[scale=1]
\filldraw[shading=radial, inner color = white, outer color = green!50!, opacity=1] (0,0.25) rectangle+(2,0.25);
\draw (0,0.25) rectangle+(2,0.25);
%\draw [decorate,decoration={brace}] (0,0.6) -- (2,0.6);
%\node at (2,1.1){\tiny{Streaming algorithm,}};
%\node at (2,0.75){\tiny{accuracy $\eps$}};
\draw [decorate,decoration={brace}] (2-0.1,-0.1) -- (0,-0.1);
\node at (1,-0.5){\tiny{$v$}};
\filldraw[shading=radial, inner color = white, outer color = red!50!, opacity=1] (2,0.25) rectangle+(2,0.25);
\draw (2,0.25) rectangle+(2,0.25);
%\draw [decorate,decoration={brace}] (2,0.6) -- (4,0.6);
%\node at (2,1.1){\tiny{Streaming algorithm,}};
%\node at (2,0.75){\tiny{accuracy $\eps$}};
\draw [decorate,decoration={brace}] (4,-0.1) -- (2+0.1,-0.1);
\node at (3,-0.5){\tiny{$u_1$}};

\filldraw[shading=radial, inner color = white, outer color = green!50!, opacity=1] (5,0.25) rectangle+(2,0.25);
\draw (5,0.25) rectangle+(2,0.25);
%\draw [decorate,decoration={brace}] (5,0.6) -- (7,0.6);
%\node at (2,1.1){\tiny{Streaming algorithm,}};
%\node at (2,0.75){\tiny{accuracy $\eps$}};
\draw [decorate,decoration={brace}] (7-0.1,-0.1) -- (5,-0.1);
\node at (6,-0.5){\tiny{$v$}};
\filldraw[shading=radial, inner color = white, outer color = red!50!, opacity=1] (7,0.25) rectangle+(4,0.25);
\draw (7,0.25) rectangle+(4,0.25);
%\draw [decorate,decoration={brace}] (7,0.6) -- (11,0.6);
%\node at (2,1.1){\tiny{Streaming algorithm,}};
%\node at (2,0.75){\tiny{accuracy $\eps$}};
\draw [decorate,decoration={brace}] (11,-0.1) -- (7+0.1,-0.1);
\node at (9,-0.5){\tiny{$u_2$}};
\end{tikzpicture}
\caption{For $v=(0,1)$, $u_1=\left(\frac{1}{\eps},0\right)$, and $u_2=\left(\frac{1}{\eps},\frac{1}{\eps}\right)$, we have $F_2(v+u_1)-F_2(u_1)=\O{\eps^2}\cdot F_2(v+u_1)$ but $F_2(v+u_2)-F_2(u_2)=\Omega(\eps)\cdot F_2(v+u_2)$, so it is unclear which granularity should be assigned to the difference estimator for $v$.}
\figlab{fig:sw:bad}
\end{figure*}

\paragraph{Stronger suffix-pivoted difference estimators.}
We thus define the suffix-pivoted difference estimators to give an $\eps\cdot F(v+u_2)$ approximation to an $(\eps,\eps^q)$-smooth function $F$, provided that $F(v+u_1)-F(u_1)\le\eps^q\cdot F(v+u_1)$ at \emph{some point} in the stream with $u_1\preceq u_2$, even if $F(v+u_2)-F(u_2)=\eps\cdot F(v+u_2)$. 
We give constructions of suffix-pivoted difference estimators under this stronger definition. 
The main point is that we can now use roughly $2^k$ difference estimators with space dependence $\frac{1}{2^k\eps^2}$ at level $k$ for each $k$ up to roughly $\beta=\log\frac{1}{\eps^q}$, even if the difference eventually becomes a much larger fraction. 
Moreover, the stronger guarantees of the suffix-pivoted difference estimators ensure that the additive error is still $\eps\cdot F(v+u_2)$, even if the difference eventually becomes a much larger fraction of what it was originally. Also, additive error $\eps \cdot F(v + u_2)$ is fine, because we only add up $\O{\log\frac{1}{\eps}}$ difference estimators, and can rescale $\eps$ by a $\log\frac{1}{\eps}$ factor. 

This stronger definition may sound unachievable if $k=\log\frac{1}{\eps^q}$ since the space dependence is now $\frac{\eps^q}{\eps^2}$, which is even smaller than $1$. But for an important
class of functions we consider, namely $F_p$-moment estimation, where $q = p$, this
$\frac{\eps^p}{\eps^2}$ multiplies a much larger term, namely, $n^{1-2/p}$, and so the space
is still much larger than $1$. Interestingly, we are able to beat the usual $n^{1-2/p}$ space 
lower bound for this problem by additionally removing global heavy hitters from all difference
estimators at a given level, which is a key idea. 

In more detail, we give a suffix-pivoted difference estimator for $F_p$-moment estimation with $p>2$ where the space dependency is $\tO{\frac{\gamma}{\eps^2} n^{1-2/p}}$, which can still be much larger than $1$ if $\gamma=\eps^p$ and $n$ is large. 
We still incur a $\frac{1}{\eps^p}$ space dependence due to running $\frac{1}{\gamma}=\frac{1}{\eps^p}$ different difference estimators at level $k=\log\frac{1}{\eps^p}$, but this multiplies  $\tO{\frac{\gamma}{\eps^2} n^{1-2/p}}$ which then becomes $\tO{\frac{1}{\eps^2} n^{1-2/p}}$, as desired. 
As mentioned, we achieve space $\tO{\frac{\eps^p}{\eps^2} n^{1-2/p}}$ for each difference estimator at this level by running a single algorithm in parallel that uses $\tO{\frac{1}{\eps^2} n^{1-2/p}+\frac{1}{\eps^p}}$ space to remove $\O{\frac{1}{\eps^p}}$ global $L_p$ heavy hitters in the sliding window, i.e., the elements $i\in[n]$ with frequency $f_i^p\ge\eps^p\cdot F_p$, and can be seen as a variance reduction technique. 
Thus after the global heavy hitters are removed, the usual $\frac{1}{\eps^2}\cdot n^{1-2/p}$ lower bounds no longer apply. 
%We then show that the usual $\frac{1}{\eps^2}\cdot n^{1-2/p}$ lower bounds no longer apply. 
% because after removing the heavy hitters, the variance of a single instance of our estimator is upper bounded by $\tO{\frac{\gamma}{\eps^2}n^{1-2/p}}\cdot (F_p)^2$, even when $\gamma$ is as small as $\eps^p$. 
%Since each estimator uses polylogarithmic space, then it suffices to take the mean of $\tO{\frac{\gamma}{\eps^2}n^{1-2/p}}$ such estimators. 
%
%the $n^{1-2/p}$ lower bound that is required of any difference estimator for
%$F_p$ for $p > 2$ {\it only holds if there are heavy hitters}. 

Thus, we suffer an additive $\O{\frac{1}{\eps^p}}$ term for finding heavy hitters, which is necessary of any algorithm, but it does not multiply the dominant $n^{1-2/p}$ term in our space bound. 
For more details but still a high level overview, see the construction details later in \secref{sec:overview:de}. 
We thus obtain an additive $\O{\frac{\eps}{\beta}}\cdot F(v)$ error at each level, and summing the errors across all $\beta$ levels gives $\eps\cdot F(v)$ total additive error, i.e., $(1+\eps)$ multiplicative approximation to $F(v)$. 

%Since the definition of suffix-pivoted difference estimators is quite strong, one might wonder whether they are even possible. 
%More generally, the main property utilized by a difference estimator to $F(v+u_1)-F(u_1)$ is that the relationship $F(v+u_1)-F(u_1)=\gamma\cdot F(v+u_1)$ should imply that $F(v)$ is small (this is true for the functions that we consider). 
%We then develop each instance of a difference estimator to have variance that depends on $F(v)$. 
%Thus for example if $F(v)=\gamma\cdot F(v+u_1)$, a single instance of our difference estimator may have variance roughly $S\cdot F(v)\cdot F(v+u_1)=S\cdot\gamma\cdot (F(v+u_1))^2$ for some quantity $S$. 
%Then it suffices to take the mean of $\frac{\gamma}{\eps^2}\cdot S$ such instances, which is larger than $1$ even if $\gamma=\eps^q$. 
%Now the point is that if $F(v+u_1)-F(u_1)=\gamma\cdot F(v+u_1)$ implies $F(v)$ is small, then our difference estimator for $F(v+u_2)-F(u_2)$ will still have ``small'' variance due to its dependency on $F(v)$. 
%Hence, we can still avoid taking the mean of $\frac{1}{\eps^2}$ such instances. 
%We describe the suffix-pivoted difference estimators in more detail below. 

%{\bf Difference estimator, $p\in[0,2]$:}
\subsection{Difference Estimators}
\seclab{sec:overview:de}
We now outline the ideas involved in developing our new difference estimators. 
We describe our difference estimators for $F_p$-moment estimation, for both the adversarially robust and sliding window models; such estimators have applications to other streaming problems and can be used for estimating entropy and for finding the $L_2$-heavy hitters. 
The high-level intuition for why difference estimators should be possible is that we leverage the fact that the difference is small into an estimator with significantly smaller variance. 
Thus we can take the mean of a smaller number of independent estimators. 

\paragraph{Difference estimators for $F_2$ moment estimation.} 
We approximate $F_2(v)-F_2(u)$ by noting that $F_2(v)=F_2((v-u)+u)=F_2(v-u)+2\ip{v-u}{u}+F_2(u)$. 
Thus if $F_2(v)-F_2(u)\le\gamma F_2(u)$, then a $\left(1+\frac{\eps}{\gamma}\right)$-approximation to $F_2(v)-F_2(u)$ translates to additive error $\eps\cdot F_2(u)$. 
We then show that if streaming algorithm $\calA$ gives a $(1+\eps)$ approximation to $F_2(u)$ and a $\left(1+\frac{\eps}{\sqrt{\gamma}}\right)$ approximation to $F_2(v-u)$, then it gives an approximation to $\ip{v-u}{u}$ with additive error at most $\frac{\eps}{\sqrt{\gamma}}\cdot\|v-u\|_2\|u\|_2\le\eps\cdot F_2(u)$, for $F_2(v-u)\le\gamma\cdot F_2(u)$, where recall $F_2(u) = \|u\|_2^2$.  
Hence $\calB$ can be directly built from the sketches maintained by $\calA$ in the case of $p=2$, with a space dependence of $\frac{\gamma}{\eps^2}$.  

To adjust our construction to handle suffix-pivoted difference estimators, the crucial observation is that if $F_2(v+u_1)-F_2(u_1)=\eps^2\cdot F_2(v+u_1)$, then for any vector $u_2\succeq u_1$, we still have $F_2(v)\le\eps^2\cdot F_2(v+u_2)$. 
Thus if we decompose $F_2(v+u_2)-F_2(u_2)=F_2(v)+2\langle v,u_2\rangle$, then we can still obtain an additive $\eps\cdot F_2(v+u_2)$ estimation to the difference by using a constant number of rows in the sketch to estimate $\langle v,u_2\rangle$ because $F_2(v)$ is so small. 
Intuitively the main idea is that although the difference has increased from $F_2(v+u_1)-F_2(u_1)$ to $F_2(v+u_2)-F_2(u_2)$, the difference between these differences is just the dot product $\ip{v}{u_2-u_1}$ and we can still use the same space to efficiently estimate the dot product because $\|v\|_2$ remains the same. 
%We use similar observations to build suffix-pivoted difference estimators from other fixed-prefix difference estimators for $F_p$-moment/entropy estimation as well. 
%We defer those details to \secref{sec:sliding}. 

\paragraph{Difference estimators for $F_0$ estimation.} 
We remark that the $F_0$ estimator of \cite{Bar-YossefJKST02} similarly maintains sketches that can be used to build a difference estimator for $F_0$ with a space dependence of $\frac{\gamma}{\eps^2}$. 
The algorithm first subsamples each item in the universe with probability $\frac{1}{2^k}$ at a level $k$ so that the expected number of distinct items appearing throughout the stream at level $k$ is $\frac{1}{2^k}\cdot F_0$. 
The algorithm then determines a level $k$ that contains $\Theta\left(\frac{1}{\eps^2}\right)$ distinct items throughout the stream, which the algorithm then rescales by a factor of $2^k$ to obtain an unbiased estimator to $F_0$; it can be shown the variance of this estimator is roughly $\eps^2\cdot(F_0)^2$. 

To obtain a difference estimator for $F_0(v)-F_0(u)=\gamma F_0(u)$ with $v\succeq u$, we instead count the number of items in a level with $\Theta\left(\frac{\gamma}{\eps^2}\right)$ distinct items that appear in the prefix $u$ of the stream. 
It then follows that the expected number of the distinct items in $v$ but not in $u$ is $\Theta\left(\frac{\gamma^2}{\eps^2}\right)$. 
Hence, we obtain a $\left(1+\frac{\eps}{\gamma}\right)$-approximation to $F_0(v)-F_0(u)$ by first running the subsampling procedure on $u$ and counting the number of distinct items at some level $k$ with $\Theta\left(\frac{\gamma}{\eps^2}\right)$ items. 
We then run the same subsampling procedure with the same randomness on $v-u$ by only counting the additional items that are occupied at level $k$, and rescaling this number by $2^k$. Note that additional items must only appear in $v$ but not $u$, which is exactly $F_0(v)-F_0(u)$. 
Although level $k$ does not necessarily give a $(1+\eps)$-approximation to $F_0(v)-F_0(u)$, it does give a $\left(1+\frac{\eps}{\gamma}\right)$-approximation to $F_0(v)-F_0(u)$, which translates to an additive $\eps\cdot F_0(u)$ approximation to $F_0(v)-F_0(u)$ since $F_0(v)-F_0(u)\le\gamma F(u)$. 

\paragraph{Difference estimator, $p\in(0,2)$.} 
Our difference estimators for $p\in(0,2)$ are more involved, and we obtain an optimal dependence on $\log n$ and $\frac{1}{\eps}$. 
A natural starting point is the identity $F_p(v)-F_p(u)=F_p(u+(v-u))-F_p(u)$ and the expansion $F_p(u+(v-u))=\sum_{k=0}^{\infty} \binom{p}{k}\cdot\langle u^{p-k},v^k\rangle$, where $v^k$ denotes the coordinate-wise exponent of $u$. 
We would like to approximate each term in the series to approximate the overall difference. 
However, this approach fails for non-integer $p<2$, since coordinates of the frequency vector $u$ may be zero, so that $u^{p-k}$ is undefined for $k=2$, and there does not seem to be a simple fix.  

Another natural approach is the $L_p$ strong tracking algorithm of \cite{BlasiokDN17}, which generates $\O{\frac{1}{\eps^2}}$ random vectors $z_1,z_2,\ldots$ of $p$-stable random variables. 
The estimate of $L_p(v)$ is then the median of the absolute values of the inner products $\langle z_1,v\rangle,\langle z_2,v\rangle,\ldots$. 
However, it seems challenging to recover the difference $F_p(v)-F_p(u)$ from the inner products $\langle z_1,v\rangle,\langle z_2,v\rangle,\ldots$ and $\langle z_1,u\rangle,\langle z_2,u\rangle,\ldots$. 

Instead, we consider a variant of Li's geometric mean estimator~\cite{Li08}. 
As in the strong $L_p$-tracker of \cite{BlasiokDN17}, we generate $\O{\frac{\gamma\log n}{\eps^2}}$ random vectors $z_1,z_2,\ldots$ of $p$-stable variables. 
For a fixed positive integer constant $q$, we form $y_1,y_2,\ldots$ so that each $y_i$ is the geometric mean of the absolute values of $q$ consecutive inner products $\langle z_{1+(i-1)q},v\rangle,\ldots,\langle z_{iq},v\rangle$. 
The final output is the average of the estimators $y_1,y_2,\ldots$. 
The key property of Li's geometric mean estimator is that $Z:=\prod_{j=1+(i-1)q}^{iq}|\langle z_j,v\rangle|^{p/q}-\prod_{j=1+(i-1)q}^{iq}|\langle z_j,u\rangle|^{p/q}$ is an unbiased estimate to $F_p(v)-F_p(u)$ whose behavior we can analyze. 
Namely, we relate $Z$ to $\prod_{j=1+(i-1)q}^{iq}\left(1+\left(\frac{|\langle z_j,v-u\rangle|}{|\langle z_j,u\rangle|}\right)^{p/q}\right)$, which has small variance given $\frac{F_p(v-u)}{F_p(u)}\le\gamma$. 
Thus we require a smaller number of independent estimators and obtain a more efficient difference estimator for $F_p(v)-F_p(u)$ for fixed $u$ and $v$ with constant probability. 

To achieve an optimal dependence on $\log n$, we cannot afford to union bound over $\poly(n)$ possible frequency vectors $u$. 
Instead we use an approach similar to the strong trackers of \cite{BravermanCIW16, BravermanCINWW17, BlasiokDN17}. 
We set a constant $q$ and split the stream into roughly $\frac{1}{\eps^{q/p}}$ times $r_1,r_2,\ldots$ between which the difference estimator increases by $\eps^{q/p}\cdot(F_p(v)-F_p(u))$ and apply a union bound to argue correctness at these times $\{r_i\}$, incurring a $\log\frac{1}{\eps}$ term. 
To bound the difference estimator between times $r_i$ and $r_{i+1}$ for a fixed $i$, we note that the difference estimator only increases by $\eps^{q/p}\cdot(F_p(v)-F_p(u))$ from $r_i$ to $r_{i+1}$, so that we still obtain a $(1+\eps)$-approximation to $F_p(v)-F_p(u)$ with a $\eps^{1-q/p}$-approximation to the difference. 
We then bound the probability that the supremum of the error between times $r_i$ and $r_{i+1}$ is bounded by $\eps^{1-q/p}$ by applying known results from chaining \cite{BravermanCIW16, BravermanCINWW17, BlasiokDN17}, thus avoiding additional $\log n$ factors. 
To derandomize our algorithm, we give a generalization of a pseudorandom generator that fools half-space queries shown by~\cite{GopalanKM18, JayaramW18}. 
Our result shows a high-probability derandomization of any algorithm that stores the product of a number of frequency vectors induced by the stream between fixed times, with a matrix of i.i.d. entries. 
We then union bound over all possible stopping times to argue the correctness guarantees of our difference estimator still hold under the derandomization. 
Interestingly, we do not guarantee the same output distribution as that of using independent $p$-stable random variables, yet we still have a correct algorithm in low memory. 
%which suffices for the requirements needed from the difference estimator from our framework. 
We also show that Li's geometric estimator can be used as a strong tracker for $L_p$. 
% presenting an alternative option to \cite{BlasiokDN17} for a strong tracking $L_p$ algorithm. 

Our construction for the suffix-pivoted difference estimator is the same. 
We note that upon a change of notation, the key property to efficiently estimate $F_p(v+u_1)-F_p(u_1)$ is that $F_p(v)$ is small. 
Moreover, we can express the variance of Li's geometric estimator in terms of $F_p(v)$ so that smaller values of $F_p(v)$ correspond to smaller variance for the estimators and thus we can use a smaller number of independent copies. 
For the suffix-pivoted difference estimator, it may be true that $F_p(v+u_2)-F_p(u_2)$ could be much larger than $F_p(v+u_1)-F_p(u_1)$ for some $u_2\succeq u_1$, but $F_p(v)$ will still be ``small'' provided that $F_p(v+u_1)-F_p(u_1)$ is ``small''. 
Hence, the variance for Li's geometric estimator will still be small despite the difference potentially being much larger. 

%{\bf Difference estimator, integer $p>2$:}
\paragraph{Fixed-prefix difference estimator, $p>2$.} 
Our difference estimator $\calB$ for $F_p$ for integers $p>2$ does use the expansion $F_p(u+(v-u))-F_p(u)=\sum_{k=0}^{p-1}\binom{p}{k}\langle u^k,(v-u)^{p-k}\rangle$, where $u^k$ denotes the coordinate-wise $k$-th power of $u$. 
To obtain an unbiased estimator of $\langle u^k,(v-u)^{p-k}\rangle$, suppose we sample a coordinate $a\in[n]$ with probability $\frac{u_a^k}{\|u\|_k^k}$ and set $Z$ to be the $a$-th coordinate of $(v-u)^{p-k}$, which we can explicitly track since the updates to the frequency vector $v-u$ arrive in the stream after the updates to $u$. 
If we could obtain an unbiased estimate $Y$ to $\|u\|_k^k$, then the expected value of $YZ$ would be exactly $\langle u^k,(v-u)^{p-k}\rangle$. 
One could try to output a coordinate $a\in[n]$ with probability $(1 \pm \varepsilon)\frac{u_a^k}{\|u\|_k^k}$. 
Such ``approximate'' $L_p$-samplers exist for $p\le 2$ \cite{MonemizadehW10,JowhariST11,JayaramW18}, though this would lead to an overall suboptimal dependence on $\varepsilon$. 
We could also potentially use an efficient ``perfect'' $L_p$-sampler, which outputs a coordinate $a\in[n]$ with probability $\frac{u_a^p}{\|u\|_p^p}$. 
Perfect $L_p$-samplers are known for $p\le 2$~\cite{JayaramW18}. 
Corresponding samplers for $p>2$ are currently unknown, though it may be possible to extend the $L_p$ samplers of \cite{JayaramW18} to $p>2$ and potentially obtain a difference estimator algorithm.  
%their constructions are based on duplicating each stream update $\poly(n)$ times; thus adapting the techniques of \cite{JayaramW18} to build perfect $L_p$-samplers for $p > 2$ would be space-inefficient, since the space dependence is $\Omega(n^{1-2/p})$ for $p>2$, rather than $\polylog(n)$ for $p\le 2$, and so after duplication the space required would be larger than $n$. 
Instead, we use the perfect $L_2$-sampler of \cite{JayaramW18} to return a coordinate $a\in[n]$ with probability $\frac{u_a^2}{\|u\|_2^2} \pm \frac{1}{\poly(n)}$. 
We also obtain unbiased estimates $X$ and $Y$ of $u_a^{k-2}$ and $\|u\|_2^2$ with low variance, respectively, by using a $\countsketch$ algorithm and our $F_2$ moment estimation algorithm. 
%(although a standard moment estimation algorithm such as $\ams$ would also suffice). 
Given $a\in[n]$, we then track the $a$-th coordinate of $(v-u)^{p-k}$ exactly. 
We show that the product of these terms $X$, $Y$, and $(v-u)^{p-k}$ forms an unbiased estimate to $\langle u^k,(v-u)^{p-k}\rangle$. 
Unfortunately, the variance of these estimates is still too large. 
We thus also find approximate frequencies to the heavy-hitters of $u$, i.e., we estimate each $u_a$ such that $u_a\ge\frac{\eps}{16\gamma^{1-1/p}}\|u\|_p$ and track the corresponding $a$-th coordinate of $(v-u)^{p-k}$ exactly. 
We remove these terms before we perform the perfect $L_2$ sampling, so that the resulting unbiased estimate to $\langle u^k,(v-u)^{p-k}\rangle$ has smaller variance. 
We then show that taking the mean of enough repetitions gives a $(1+\eps)$-approximation to $\langle u^k,(v-u)^{p-k}\rangle$. 
By repeating the estimator for each summand in $\sum_{k=0}^{p-1}\binom{p}{k}\langle u^k,(v-u)^{p-k}\rangle$, it follows that we obtain a $(\gamma,\eps,\delta)$-difference estimator for $F_p$. 

\paragraph{Suffix-pivoted difference estimator, $p>2$.} 
Our construction for the suffix-pivoted difference estimator requires a bit more care. 
We can again write $F_p(v)-F_p(u)=\sum_{k=1}^p\langle(v-u)^{p-k},u^k\rangle$. 
However, we require $\O{\frac{1}{\eps^p}}$ difference estimators due to the smoothness parameter of $F_p$. 
Namely, we implement suffix-pivoted difference estimators with $\gamma$ as small as $\eps^p$ whereas we previously implemented prefix-fixed difference estimators with $\gamma$ roughly as small as $\eps^2$. 
Thus, we can no longer necessarily run a heavy-hitter algorithm that uses space proportional to $\O{\frac{\gamma^{2-2/p}}{\eps^2}}$ for each block of the stream corresponding to a difference estimator, since the quantity is no longer well-defined for $\gamma=\eps^p$. 

On the other hand, the necessity to run separate heavy-hitter algorithms for each block in the fixed-prefix difference estimator originated from the adversarial robustness requiring fresh randomness in each block. 
The sliding window model requires no such restriction and thus we instead we run a single heavy-hitter algorithm in parallel that uses space $\tO{\frac{1}{\eps^2}n^{1-2/p}}$ and \emph{simultaneously} finds a list $\calH$ of all heavy-hitters in the sliding window. 
In particular, $\calH$ will report an approximate frequency of each heavy-hitter in each block assigned to a difference estimator. 
The main point is that $\tO{\frac{1}{\eps^2}n^{1-2/p}}$ space is used to identify the heavy-hitters but only $\O{\frac{1}{\eps^{p}}\log n}$ space is used to track their approximate frequency in each block. 
Hence, the $\frac{1}{\eps^p}$ dependency is still present but it now multiplies a lower-order term. 
We can then decompose the difference $\sum_{k=1}^p\langle(v-u)^{p-k},u^k\rangle$ into the amount contributed by the heavy items $\sum_{a\in\calH}\sum_{k=1}^p(v_a-u_a)^{p-k}u_a^k$ and the amount contributed by the items that are not heavy $\sum_{a\notin\calH}\sum_{k=1}^p(v_a-u_a)^{p-k}u_a^k$. 
We show that to estimate $\sum_{a\in\calH}\sum_{k=1}^p(v_a-u_a)^{p-k}u_a^k$, it again suffices to estimate the value of $v_a-u_a$ and then subsequently read off the corresponding coordinate $u_a$. 

It remains to estimate $\sum_{a\notin\calH}\sum_{k=1}^p(v_a-u_a)^{p-k}u_a^k$. 
We use the same approach of perfect $L_2$ sampling coordinates from $v-u$ and reading off the corresponding coordinates in $u$, so that each perfect $L_2$ sample and the corresponding coordinate in $u$ forms an estimate of the inner product. 
Crucially we can use a more refined level set analysis to upper bound the variance of the inner product estimate by $\tO{\frac{\gamma}{\eps^2}n^{1-2/p}}\cdot (F_p(v))^2$, even when $\gamma$ is as small as $\eps^p$. 
Thus as before, smaller values of $F_p(v)$ correspond to smaller variance for the estimators and thus we can use a smaller number of independent copies of the inner product estimates. 
Since each inner product estimate requires polylogarithmic space, it suffices to take the mean of $\tO{\frac{\gamma}{\eps^2}n^{1-2/p}}$ such inner product estimators to form our difference estimator. 

%{\bf Heavy hitters:}
\paragraph{Heavy hitters.} 
Our framework extends to an adversarially robust streaming algorithm for finding $L_2$ heavy hitters. 
We first run the same $F_2$ moment estimation algorithm that stitches together sketches of different granularities, which partitions the stream of length $m$ is partitioned into $\beta$ blocks. 
Observe that a heavy-hitter $i\in[n]$ with $f_i>\eps\cdot L_2(1,m)$ must have frequency at least $\frac{\eps}{\beta}\cdot L_2(1,m)$ in one of these $\beta$ blocks. 
Now if $F_2(t_j,t_{j+1})=\frac{1}{2^k}\cdot F_2(1,m)$ for a block corresponding to time interval $[t_j,t_{j+1}]$ in which $i$ has frequency at least $\frac{\eps}{\beta}\cdot L_2(1,m)$, then $f_i^2\ge\frac{2^k\eps^2}{\beta^2} F_2(t_j,t_{j+1})$. 
We can identify $i$ using a much higher threshold and thus much less space, i.e., space proportional to $\frac{\beta^2}{2^k\eps^2}$ rather than $\frac{\beta^2}{\eps^2}$. 
Thus in parallel with our $F_2$ moment estimation algorithm, each time we run a difference estimator $\calB$ and an $F_2$-estimation algorithm $\calA$ with a specific granularity, we also run a heavy hitters algorithm with threshold corresponding to the same granularity of $\calB$. 
For example, if the difference estimator on a block contributes $\frac{1}{2^k}\cdot F_2(1,m)$ to the overall $F_2$ moment, then we use a level $k$ difference estimator $\calB$ with effective accuracy $1+\frac{2^k\eps}{\beta}$ and we similarly run a heavy hitters algorithm to detect items that are $\frac{2^k\eps^2}{\beta^2}$-heavy with respect to $F_2$ on the block. 
%The crucial observation is that a heavy hitter across the entire stream must have at least a $\frac{1}{\beta}$-fraction of its contribution to $F_2$ in one of the $\beta$ disjoint contiguous substreams of stream updates, which are partitioned using the sketch stitching technique by their change in $F_2$ value. 
%However, the level $k$ difference estimator only contributes a $\frac{1}{2^k}$ fraction of the $F_2$ value, so the initially $\frac{\eps^2}{\beta^2}$-heavy item must be $\frac{2^k\eps^2}{\beta^2}$-heavy with respect to $F_2$ in the corresponding block, and will be reported by the heavy hitters algorithm in that block. 
%Note that we can afford to lose $\poly(\beta) = \polylog(1/\eps)$ factors. 
%By further decreasing the threshold, we can also ensure that only a small fraction of the heavy hitter is missed by the algorithm. 
%Similarly, our optimized space $F_2$ algorithm can only miss a $\poly(\eps)$-fraction of the contribution of any heavy hitter towards the total $F_2$ value. 
We then explicitly count each reported item to remove all items that may be heavy with respect to a single block, but not the overall stream, and obtain an accurate count for each heavy hitter. 

\subsection{Preliminaries}
\seclab{sec:prelim}
We use the notation $[n]$ to denote the set $\{1,\ldots,n\}$ for any positive integer $n$. 
We use $\poly(n)$ to denote a constant degree polynomial in $n$. 
When an event has probability $1-\frac{1}{\poly(n)}$ of occurring, we say the event occurs with high probability. 
We use $\polylog(n)$ to denote a polynomial in $\log n$. 
 
For a function $F$, let $F(t_1,t_2)$ denote the value of $F$ on the underlying frequency vector induced between times $t_1$ and $t_2$, inclusive. 
In the basic insertion-only stream, the $t$-th stream update $u_t$ is an integer in $[n]$, where $t\in[m]$ for a stream of length $m$. 
The updates implicitly define a frequency vector $f\in\mathbb{R}^n$ so that the effect of each update is to increase the value of a coordinate of $f$. 
Thus, we have $f_i=|\{t\,:\,u_t=i\}|$ for each $i \in [n] = \{1, 2, \ldots, n\}$. 
In the sliding window model, the frequency vector is induced by only the most recent $W$ updates in an insertion-only stream. 
Hence if $W\ge m$ for a stream of length $m$, the model reduces to the insertion-only streaming model and if $W<m$, then $f_i=|\{t\,:\,u_t=i \textrm{ and } t\ge m-W+1\}|$ for each $i\in[n]$. 
In turnstile streams, each update $u_t=(\Delta_t, i_t)$, where $\Delta_t\in\{-1,+1\}$ and $i_t\in[n]$. 
Hence, each update either increases or decreases a coordinate of $f$. 
%More generally, updates in each of these streaming models can, in which case the coordinates of the frequency vector $f$ are implicitly defined by the sum of the updates to the coordinate rather than the count; our algorithms can be adapted appropriately. 
For vectors $u$ and $v$ of length $n$, we use the notation $u\succeq v$ to denote that $u$ coordinate-wise dominates $v$, i.e., $u_i\ge v_i$ for all $i\in[n]$. 

Let $\lsb(x,i)$ be the $i$-th least significant bit in the binary representation of $x$, when using $1$-indexing, e.g., we start counting indices from $1$ rather than $0$.  
Let $\numbits(x)$ be the number of nonzero bits in the binary representation of $x$.  
For example, $\lsb(6,1)=2$, $\lsb(6,2)=3$, and $\numbits(x)=2$, since $6=110_2$. 
As another example, $\lsb(8,1)=4$ and $\numbits(x)=1$, since $8=1000_2$.

\paragraph{Frequency moments.} 
The space complexity of approximating the frequency moments $F_p$ is one of the oldest and most fundamental problems on streams, originating with the work of Alon, Matias, and Szegedy \cite{AlonMS99}. 
For example, the $F_2$-estimation problem corresponds to the size of a self-join \cite{AlonMS99} and tracking it with limited storage is essential \cite{AlonGMS02}. 
Given $p,\eps>0$, the $p$-th frequency moment estimation problem is to approximate $F_p = \sum_{i\in[n]} f_i^p$ within a $(1\pm\eps)$ factor for a frequency vector $f\in\mathbb{R}^n$. 
The complexity of this problem and the related $L_p$-norm estimation problem differs greatly for different values of $p$.  
For $p>2$, \cite{Bar-YossefJKS04,ChakrabartiKS03} show that for the streaming model, the memory required for $F_p$ estimation requires polynomial factors in $n$, whereas polylogarithmic space is achievable for $p\le 2$~\cite{AlonMS99,Indyk06,Li08,KaneNW10,KaneNPW11,BlasiokDN17}. 
%In this paper, we focus on $p\in[0,2]$, which has a number of applications in the streaming and sliding window models, as well as integers $p>2$. 
We note that $F_p$ estimation for fractional $p$ near $1$ is used for empirical entropy estimation~\cite{HarveyNO08}, and $p=0.25$ and $p=0.5$ are useful for mining tabular data~\cite{CormodeIKM02}. 
$L_1$-norm estimation is used in network traffic monitoring~\cite{FeigenbaumKSV02}, low-rank approximation and linear regression~\cite{FeldmanMSW10}, dynamic earth-mover distance approximation~\cite{Indyk04}, and cascaded norm estimation~\cite{JayramW09}. 
$L_2$-norm estimation is used for estimating join and self-join sizes~\cite{AlonGMS02}, randomized numerical linear algebra \cite{ClarksonW09}, and for detecting network anomalies~\cite{KrishnamurthySZC03,ThorupZ04}. 
The values $p=3$ and $p=4$ also respectively correspond to the skewness and the kurtosis of a stream~\cite{AlonMS99}.

For a frequency vector $f$ with length $n$, we define $F_p(f)=\sum_{i=1}^n|f_i|^p$. 
The $L_p$ norm\footnote{Observe that $L_p$ does not satisfy the triangle inequality for $0<p<1$ and thus is not a norm, but is still a well-defined quantity.} of $f$ is defined by $L_p(f)=(F_p(f))^{1/p}$. 
We also use the notation $\|f\|_p$ to denote the $L_p$ norm of $f$. 
The $F_p$-moment estimation problem is also often stated as the norm estimation problem, since an algorithm that outputs a $(1+\eps)$-approximation to one of these problems can be modified to output a $(1+\eps)$-approximation to the other using a constant rescaling of $\eps$ (for constant $p$); thus we use these two equivalent problems interchangeably. 
The $L_2$-heavy hitters problem is to output all coordinates $i$ such that $f_i\ge\eps\cdot L_2$. 
The problem permits the output of coordinates $j$ with $f_j\le\eps\cdot L_2$ provided that $f_j\ge\frac{1}{2}\cdot L_2$. 
Moreover, each coordinate output by the algorithm also requires an estimate to its frequency, with additive error at most $\frac{1}{2}\cdot L_2$.

%For any $x\in\mathbb{R}$ and $\eps>0$, define $[x]_{\eps}\in\mathbb{R}$ so that $[x]_{\eps}=0$ for $x=0$ and $[x]_{\eps}=y$, where $y>0=(1+\eps)^k$ for the integer $k$ that minimizes $\max\{y/x,x/y\}$. 

%\begin{definition}[$\eps$-rounding]
%An $\eps$-rounding of a sequence $y_1,\ldots,y_m\in\mathbb{R}$ is a sequence $y'_1,\ldots,y'_m$ constructed so that $y'_1=[y_1]_{\eps}$ and for $i\ge 1$, $y'_{i+1}=y'_i$ if $(1-\eps)y_{i+1}\le y'_i\le (1+\eps)y_{i+1}$ and $y'_{i+1}=[y_{i+1}]_{\eps}$ otherwise. 
%\end{definition}
%
%\begin{definition}[$\eps$-rounding of algorithms]
%\cite{Ben-EliezerJWY20}
%The $\eps$-rounding of a streaming algorithm $A$ is an algorithm $A'$ that simulates an instance of $A$ and runs as follows: 
%For the first element $z_1$ output by $A$, $A'$ returns $z'_1=[z_1]_{\eps}$. 
%For each $i\ge2$ and element $z_i$ output by $A$, then $A'$ outputs $z'_i$, where $z'_i=z'_{i-1}$ if $(1-\eps)z_i\le z'_{i-1}\le(1+\eps)z_i$ and $z'_i=[z_i]_{\eps}$ otherwise. 
%\end{definition}

\begin{definition}[$(\eps,m)$-flip number]
\cite{Ben-EliezerJWY20}
\deflab{def:flipnumber}
For $\eps>0$ and an integer $m>0$, the $(\eps,m)$-flip number of a sequence $y_1,\ldots,y_m$ is the largest integer $k$ such that there exist $0\le i_1<\ldots<i_k\le m$ with $y_{i_{j-1}}\notin[(1-\eps)y_{i_j},(1+\eps)y_{i_j}]$ for all integers $j\in[2,k]$. 
\end{definition}

\begin{definition}[$(\eps,m)$-twist number]
\deflab{def:twistnumber}
For $\eps>0$ and an integer $m>0$, the $(\eps,m)$-twist number for a non-negative function $F$ on a stream of updates $1=u_1,\ldots,u_m$ is the largest integer $k$ such that there exist $0\le i_1<\ldots<i_k\le m$ with for each $j\in[2,k]$, either $F(u_1:u_{j-1})\notin[(1-\eps)F(u_1:u_j),(1+\eps)F(u_1:u_j)]$ or $F(u_{j-1}:u_j)\ge\eps F(1:u_{j-1})$. 
\end{definition}

\begin{definition}[Strong tracking]
Let $f^{(1)},\ldots,f^{(m)}$ be the frequency vectors induced by a stream of length $m$ and let $g:\mathbb{R}^n\to\mathbb{R}$ be a function on frequency vectors. 
An algorithm $A$ provides $(\eps,\delta)$-strong $g$-tracking if at each time step $t\in[m]$, $A$ outputs an estimate $X_t$ such that
\[|X_t-g(f^{(t)})|\le\eps|g(f^{(t)})|\]
for all $t\in[m]$ with probability at least $1-\delta$. 
\end{definition}

\begin{observation}
\cite{Ben-EliezerJWY20}
\obslab{obs:flip:fp}
For $p>0$, the $(\eps,m)$-twist number of $\|x\|_p^p$ in the insertion-only model is $\O{\frac{1}{\eps}\log m}$ for $p\le 2$ and $\O{\frac{p}{\eps}\log m}$ for $p>2$. 
\end{observation}

\section{Framework for Adversarially Robust Streaming Algorithms}
\seclab{sec:framework}
In this section, we describe a general framework for adversarially robust streaming algorithms, using the sketch stitching and granularity changing techniques. 
We first require the following specific form of a difference estimator. 
\begin{definition}[Fixed-Prefix Difference Estimator]
\deflab{def:diff:est:pre}
Given a stream $\calS$, a fixed time $t_1$, and a splitting time $t_2$ that is only revealed at time $t_2$, let frequency vector $v$ be induced by the updates of $\calS$ from time $t_1$ to $t_2$ and frequency vector $w_t$ be induced by updates from time $t_2$ to $t$ exclusive.  
Given an accuracy parameter $\eps>0$ and a failure probability $\delta\in(0,1)$, a streaming algorithm $\calB(t_1,t_2,t,\gamma,\eps,\delta)$ is a $(\gamma,\eps,\delta)$-\emph{difference estimator} for a function $F$ if, with probability at least $1-\delta$, it outputs an additive $\eps\cdot F(v)$ approximation to $F(v+w_t)-F(v)$ simultaneously for all $t\ge t_2$ with $F(v+w_t)-F(v)\le\gamma\cdot F(v)$ and $F(w_t)\le\gamma F(v)$ for a ratio parameter $\gamma\in(0,1]$.  
\end{definition}
We shall use the shorthand ``difference estimator'' terminology to refer to \defref{def:diff:est} until \secref{sec:sliding}; in \secref{sec:sliding} we will introduce an additional notion of a difference estimator for when $F(v)$ can change. 

It is easy to see that difference estimators offer an immediate means to approximation algorithms. 
\framestitch*
\begin{proof}
Since $F(u+v_1+\ldots+v_i)=F(u+v_1+\ldots+v_{i-1})+(F(u+v_1+\ldots+v_{i-1})-F(u+v_1+\ldots+v_i))$ and the difference is bounded by $\gamma_i\cdot F(u+v_1+\ldots+v_{i-1})$ by the definition of difference estimator, then we have the invariant $F(u+v_1+\ldots+v_i)\le(1+\gamma_i)F(u+v_1+\ldots+v_{i-1})\le 2F(u)$ for all $i\in\{1,\ldots,k\}$ by induction. 
Thus the sum of the additive errors from each difference estimator is at most $\sum_{i=1}^k\eps_i\cdot F(u+v_1+\ldots+v_{i-1})\le\frac{\eps}{4}\cdot 2F(u)=\frac{\eps}{2}\cdot F(u)$. 
Similarly the additive error due to the $\left(1+\frac{\eps}{2}\right)$-approximation to $F(u)$ is also $\frac{\eps}{2}\cdot F(u)$, so that by monotonicity of $F$ and a telescoping argument the sum gives a $(1+\eps)$ multiplicative approximation to $F(u+v_1+\ldots+v_i)$. 
Since each difference estimator fails with probability at most $\frac{1}{3k}$, then by a union bound, the total probability of failure is at most $\frac{1}{3}$. 
\end{proof}

\framestitchspace*
\begin{proof} 
The total space is $\sum_{i=1}^k N_i\cdot\frac{\gamma_i}{\eps^2}\cdot S$. 
Since $N_i\le 2^i$, $\gamma_i\le\frac{1}{2^i}$, and $k=\O{\log n}$, then the total space is at most $k\cdot\frac{S}{\eps^2}\le\O{S\log n}{\eps^2}$. 
\end{proof}

\subsection{Algorithm}
We first describe a simplified version of our adversarially robust framework that adapts the usage of difference estimators. 
To achieve a robust $(1+\O{\eps})$-approximation, \cite{Ben-EliezerJWY20} used a ``switch-a-sketch'' technique that maintains $(\eps,m)$-flip number $\xi$ (recall \defref{def:flipnumber}) independent subroutines that each provide a $(1+\eps)$ to a function $F$ evaluated on the frequency vector induced by the stream, with high probability. 
For the remainder of the discussion, we assume the $\xi=\Omega\left(\frac{1}{\eps}\log n\right)$, which is true for many important functions $F$, especially the important $F_p$ moments and moreover, that the twist number $\lambda=\xi$, which is true for insertion-only streams. 
%The output of each subroutine is rounded to a power of $(1+\eps)$ so that the rounding of each subroutine is still within a multiplicative $(1+\O{\eps})$ factor of the value of $F$ applied to the frequency vector induced by the stream. 
To prevent an adversary from affecting the output of the algorithm, each subroutine is effectively only used once. 
The output of the $i$-th subroutine is only used the first time the true output of the subroutine is at least $(1+\eps)^i$. 
The algorithm of \cite{Ben-EliezerJWY20} then repeatedly outputs this value until the $(i+1)$-st subroutine is at least $(1+\eps)^{i+1}$, at which point the algorithm switches to using the output of the $(i+1)$-st subroutine instead. 
Hence, the adversary information-theoretically knows nothing about the internal randomness of the $i$-th subroutine until the output is at least $(1+\eps)^i$. 
However, due to monotonicity of $F$ and correctness of the oblivious $(i+1)$-st instance, whatever knowledge the adversary gains about the $i$-th instance does not impact the internal randomness of future instances. 
Intuitively, the switch-a-sketch approach uses a sketch once and switches to another sketch once the estimated $F$ has increased by $(1+\eps)$. 
$F$ can only increase $\lambda$ times by definition of the $(\eps,m)$-twist number. 
Since $\lambda=\Omega\left(\frac{1}{\eps}\log n\right)$, this approach generally achieves $\frac{1}{\eps^3}$ space dependency. 

\begin{algorithm}[!htb]
\caption{Framework for Robust Algorithms on Insertion-Only Streams}
\alglab{alg:framework}
\begin{algorithmic}[1]
\Require{Stream $u_1,\ldots,u_m\in[n]$ of updates to coordinates of an underlying frequency vector, accuracy parameter $\eps\in(0,1)$, $(\gamma,\eps,\delta)$-difference estimator $\calB$ for $F$ with space dependency $\frac{\gamma^C}{\eps^2}$ for $C\ge 1$, oblivious strong tracker $\calA$ for $F$}
\Ensure{Robust $(1+\eps)$-approximation to $F$}
\State{$\delta\gets\frac{1}{\poly\left(\frac{1}{\eps}, \log n\right)}$, $\zeta\gets\frac{2}{2^{(C-1)/4}-1}$, $\eta\gets\frac{\eps}{64\zeta}$, $\beta\gets\ceil{\log\frac{8}{\eps}}$}
\State{$a\gets 0$, $\varphi\gets 2^{(C-1)/4}$, $\gamma_j\gets 2^{j-1}\eta$}
\State{For $j\in[\beta]$, $\eta_j\gets\frac{\eta}{\beta}$ if $C=1$, $\eta_j\gets\frac{\eta}{\varphi^{\beta-j}}$ if $C>1$.}
\Comment{Accuracy for each difference estimator}
\For{each update $u_t\in[n]$, $t\in[m]$}
\State{$X\gets\calA_{a+1}(1,t,\eta,\delta)$}
\If{$X>2^a$}
\linlab{lin:frame:init}
\Comment{Switch sketch at top layer}
\State{$a\gets a+1$, $b\gets 0$, $Z_a\gets X$, $t_{a,j}\gets t$ for $j\in[\beta]$.}
%\State{Let $a$ be the smallest integer $i$ such that $2^i\ge X$.}
%\State{$b\gets 0$, $Z_a\gets X$, $t_{a,j}\gets t$ for $j\in[\beta]$.}
\EndIf
\State{$X\gets\estimateF$}
\Comment{Compute estimator $X$ for $F$ using unrevealed sketch}
\If{$X>\left(1+\frac{(b+1)\eps}{8}\right)\cdot Z_a$}
\Comment{Switch sketch at lower layer}
\State{$b\gets b+1$, $k\gets\lsb(b,1)$, $j\gets\flr{\frac{b}{2^k}}$}
\State{$Z_{a,k}\gets\calB_{a,j}(1,t_{a,k},t,\gamma_k,\eta_k,\delta)$}
\Comment{Freeze old sketch}
\State{$t_{a,j}\gets t$ for $j\in[k]$.}
\Comment{Update difference estimator times}
\linlab{lin:frame:reset}
\EndIf
\State{\Return $\left(1+\frac{b\eps}{8}\right)\cdot Z_a$}
\linlab{lin:final:round}
\Comment{Output estimate for round $t$}
\EndFor
\end{algorithmic}
\end{algorithm}

\begin{algorithm}[!htb]
\caption{Subroutine $\estimateF$ of \algref{alg:framework}}
\alglab{alg:estimatef}
\begin{algorithmic}[1]
\State{$X\gets Z_a$, $k\gets\numbits(b+1)$, $z_i\gets\lsb(b+1,k+1-i)$ for $i\in[k]$.\\}
\Comment{$z_1>\ldots>z_k$ are the nonzero bits in the binary representation of $b+1$.}
\For{$1\le j\le k-1$}
\Comment{Compile previous frozen components for estimator $X$}
\State{$X\gets X+Z_{a,j}$}
\EndFor
\State{$j\gets\flr{\frac{b+1}{2^{z_k}}}$}
\State{$X\gets X+\calB_{a,j}\left(1,t_{a,z_k},t,\gamma_{z_k},\eta_{z_k},\delta\right)$}
\Comment{Use unrevealed sketch for last component}
\State{\Return $X$}
\end{algorithmic}
\end{algorithm}

We first observe that if we instead use the switch-a-sketch technique each time $F$ increases by a power of $2$, then we only need to switch $\O{\log n}$ sketches. 
Effectively, this follows from setting $\eps=\O{1}$ in the value of $\lambda$.  
Let $t_i$ be the first time $F$ of the stream surpasses $2^i$. 
The challenge is then achieving a $(1+\eps)$-approximation to $F$ at the times between each $2^i$ and $2^{i+1}$. 
Let $u$ be the underlying frequency vector at time $t_i$, so that $F(u)\ge 2^i$. 
If $v$ is the underlying frequency vector at some time between $t_i$ and $t_{i+1}$, then we can decompose $F(v)=F(u)+\sum_{j=1}^{\beta}(F(u_j)-F(u_{j-1}))$, where we use the convention that $u_0=u$ and $u_{\beta}=j$. 
Moreover, we assume that $F(u_j)-F(u_{j-1})\le\gamma\cdot F(v)$ for $\gamma\le\frac{1}{2^j}$. 

Our key observation is that because we only care about a $(1+\eps)$-approximation to $F(v)$, we do not need a $(1+\eps)$-approximation to each of the differences $F(u_j)-F(u_{j-1})$, which may be significantly smaller than $F(v)$. 
For example, note that a $(2^j\cdot\eps)$-approximation to $F(u_j)-F(u_{j-1})$ only equates to an additive $\O{\eps\cdot F(v)}$ error, since $F(u_j)-F(u_{j-1})=\O{\frac{1}{2^j}}\cdot F(v)$. 
We require $\frac{1}{\gamma}$ instances of algorithms with such accuracies, to account for the various possible vectors $v$. 
Thus if there exists an algorithm that uses $\frac{\gamma}{\eps^2}S(n)$ bits of space to output additive $\eps\cdot F(v)$ error, then the space required across the level $j$ estimators is $\frac{1}{\eps^2}S(n)$ bits of space. 
Since there are at most $\O{\log\frac{1}{\eps}}$ levels, then we do not incur any additional factors in $\frac{1}{\eps}$. 
Recall that a difference estimator (\defref{def:diff:est}) to $F$ serves exactly this purpose! 

It is not obvious how to obtain a difference estimator for various functions $F$. 
However, for the purposes of a general framework, the theoretical assumption of such a quantity suffices; we shall give explicit difference estimators for specific functions $F$ of interest. 
The framework appears in full in \algref{alg:framework}. 

\paragraph{Interpretation of \algref{alg:framework}.}
We now translate between the previous intuition and the pseudocode of \algref{alg:framework}. 
The streaming algorithm $\calA_{a+1}$ attempts to notify the algorithm whenever the value of $F$ on the underlying frequency vector induced by the stream has increased by a power of two. 
The role of $a$ serves as a counter, ensuring that a new instance of $\calA$ is used each time the value of $a$ has increased. 
Since $a$ increases when $\calA_{a+1}>2^a$, i.e., the value of $F$ has increased by roughly a factor of two, then the instances $\calA$, indexed by $a$, perform sketch switching with the finest accuracy. 
Thus, \algref{alg:framework} defines $Z$ to be the value of $\calA_a$ when it is revealed and never uses $\calA_a$ again. 

It remains for \algref{alg:framework} to stitch the sketches between the times when $a$ increases. 
Suppose $t_a$ is the time when $Z$ was most recently defined. 
Then \algref{alg:framework} uses the variables $Z_{a,k}$ to stitch together the sketches of the difference estimators $\calB$ for $F(1:t)-F(1:t_a)$. 
Each variable $Z_{a,k}$ refers to a separate granularity of the difference estimator $\calB$. 
Like the streaming algorithm $\calA$, the difference estimators $\calB$ are indexed by the counter $a$ to distinguish between the instances required for the sketch switching technique. 
Additionally, the difference estimators $\calB$ for the variable $Z_{a,k}$ are indexed by a counter $j$ that represents the number of times a level $k$ difference estimator is required throughout the times $t$ that we roughly have $2^a<F(1:t)<2^{a+1}$. 
The subroutine $\estimateF$ then stitches together the values $Z_{a,k}$ to form an estimate $X$ of $F(1:t)$. 
Similarly, \algref{alg:framework} stitches together the values $Z_{a,k}$ to form an estimate $X$ of $F(1:t)$, which it then rounds as the output. 
The actual stitching process requires partitioning $F(1:t)-F(1:t_a)$, so that the appropriate granularities may be used for each difference estimator. 
This is monitored by the nonzero bits in the binary representation of the counter $b$, which roughly denotes the ratio of $F(1:t)-F(1:t_a)$ to $F(1:t_a)$.  
Now of course an inconvenience in the actual analysis is that because $\calB$ and $\calA$ are approximation algorithms, then we do not exactly have $2^a<F(1:t)<2^{a+1}$, although this is true up to a multiplicative $(1+\O{\eps})$-factor. 

\subsection{Analysis}
We first show correctness of \algref{alg:framework} at times where the counter $a$ increases. 
This allows us to subsequently focus on times $t\in(t_i,t_{i+1})$ between points in the stream where the counter increases. 
\begin{lemma}[Correctness when $F$ doubles]
\lemlab{lem:constant}
Let $F$ be a monotonic function with $(\eps,m)$-twist number $\lambda=\O{\frac{1}{\eps}\log n}$ for $\log m=\O{\log n}$. 
For any integer $i>0$, let $t_i$ be the update time during which the counter $a$ in \algref{alg:framework} is first set to $i$. 
Then \algref{alg:framework} outputs a $\left(1+\frac{\eps}{32}\right)$-approximation to $F$ at all times $t_i$, with probability at least $1-\O{\frac{\delta\log n}{\eps}}$. 
\end{lemma}
\begin{proof}
Note that by definition, the counter $a$ in \algref{alg:framework} is incremented from $i-1$ to $i$ in round $t_i$. 
Thus all previous outputs of \algref{alg:framework} have not used the subroutine $\calA_i$. 
Hence, the adversarial input is independent of the randomness of $\calA_i$. 
Thus by the correctness guarantee of the algorithm, $\calA_i$ outputs a $\left(1+\frac{\eps}{32}\right)$-approximation to $F$ at time $t_i$, with probability $1-\delta$. 
Since $F$ has twist number $\lambda=\O{\frac{1}{\eps}\log n}$, then there are at most $\O{\frac{1}{\eps}\log n}$ update times $t_i$ in which the counter $a$ is increased. 
By a union bound, \algref{alg:framework} outputs a $\left(1+\frac{\eps}{32}\right)$-approximation to $F$ at all times $t_i$, with probability at least $1-\O{\frac{\delta\log n}{\eps}}$. 
\end{proof}

For any integer $i>0$, let $t_i$ be the update time during which the counter $a$ in \algref{alg:framework} is first set to $i$. 
Since \lemref{lem:constant} shows correctness at all times $t_i$, it remains to analyze the behavior of \algref{alg:framework} between each of the times $t_i$ and $t_{i+1}$. 

\begin{lemma}[Geometric bounds on splitting times]
\lemlab{lem:time:geometric}
For each integer $j\ge 0$ and a fixed integer $i$, let $u_{i,j}$ be the last round for which $Z_{i,j}$ in \algref{alg:framework} is defined. 
%With probability at least $1-\delta\cdot\poly\left(\frac{1}{\eps},\log n\right)$, we have $F(1,u_{i,j})-F(1,t_{i,j})\le\frac{1}{2^{\beta-j-3}}F(1,t_i)$ for each positive integer $j>1$.
With probability at least $1-\delta\cdot\poly\left(\frac{1}{\eps},\log n\right)$, we have $F(1,u_{i,j})-F(1,t_{i,j})\le\frac{2^j\eps}{64}\cdot F(1,t_i)$ for each positive integer $j>1$.
\end{lemma}
\begin{proof}
As before, let $t_i$ be the update time during which the counter $a$ in \algref{alg:framework} is first set to $i$ for any integer $i>0$. 
Let $t_{i,j}$ be the first round such that the counter $b$ in \algref{alg:framework} is first set to $j$, and we use the convention that $t_{i,0}=t_i$.
Let $\mathcal{E}_1$ be the event that \algref{alg:framework} outputs a $\left(1+\frac{\eps}{32}\right)$-approximation to the value of $F$ at all times $t_i$, so that by \lemref{lem:constant}, $\mathcal{E}_1$ holds with probability at least $1-\O{\frac{\delta\log n}{\eps}}$. 
For the remainder of the proof, we fix the integer $i$ and show correctness between $t_i$ and $t_{i+1}$. 
Note that conditioning on $\mathcal{E}_1$, we have correctness at times $t_i$ and $t_{i+1}$. 

Let $\mathcal{E}_2$ be the event that the first $\poly\left(\frac{1}{\eps},\log n\right)$ instances of $\calB$ and $\calA$ in \algref{alg:framework} are correct, so that $\PPr{\mathcal{E}_2}\ge1-\delta\cdot\poly\left(\frac{1}{\eps},\log n\right)$ by a union bound. 
We show by induction on $b$ that conditioned on $\mathcal{E}_1$ and $\mathcal{E}_2$, we have that $F(1,u_{i,j})-F(1,t_{i,j})\le\frac{1}{2^{\beta-j-3}}F(1,t_i)$ for $j\ge 1$. 

For the base case $b=1$, only $Z_{i,1}$ is defined. 	
Suppose by way of contradiction that $F(1,u_{i,1})-F(1,t_{i,1})>\frac{1}{2^{\beta-j-3}}\cdot F(1,t_i)$. 
Since $\beta=\ceil{\log\frac{8}{\eps}}$, then $F(1,u_{i,1})-F(1,t_{i,1})>\frac{\eps}{2}\cdot F(1,t_i)$. 
Observe that $\calB_{i,1}(1,t_{i,1},u_{i,1},\gamma_1,\eta_1,\delta)$ gives an $\eta_1\cdot F(1,t_i)$ additive approximation to $F(1,u_{i,1})-F(1,t_{i,1})$. 
Hence, the estimate $Z_{i,1}$ of $F(1,u_{i,1})-F(1,t_{i,1})$ from $\estimateF$ is at least 
\[Z_{i,1}\ge F(1,u_{i,1})-F(1,t_{i,1})-\eta_1\cdot F(1,t_i)>\frac{\eps}{3}\cdot F(1,t_i),\]
since $\eta_1=\frac{\eta}{\varphi^{\beta-1}}<\frac{\eps}{3}$ for $\eta=\frac{\eps}{64\zeta}$ and $\varphi>1$. 
But this contradicts the fact that the counter $b$ is set to $j=1$, since 
\[Z_i+Z_{i,1}>\left(1+\frac{\eps}{4}\right)Z_i,\]
since $Z_i\le(1+\eps)\cdot F(1,t_i)$ conditioned on $\mathcal{E}_1$. 
Thus we have $F(1,u_{i,1})-F(1,t_{i,1})>\frac{1}{2^{\beta-j-3}}\cdot F(1,t_i)$, which completes the base case $j=1$. 

For a fixed value $\ell>0$ of the counter $b$, let $k=\numbits(\ell)$ and for $x\in[k]$, let $z_x=\lsb(\ell,k-x+1)$, so that $z_1>\ldots>z_k$ are the nonzero bits of $\ell$. 
Let $u=\sum_{x=1}^{k-1} 2^{z_x-1}\cdot z_x$ be simply the decimal representation of the binary number resulting from changing the least significant bit of $\ell$ to zero. 
%\[F(1,t_{i,u_{z_x}})-F(t_{i,u_{z_x}},t_{i,z_x})\le\frac{1}{2^{\beta-{z_x}-3}}F(1,t_i).\]
Moreover, if $X_u$ is the output at time $t^{(u)}_i$, then we have $X_u\ge\left(1+\frac{u\eps}{8}\right)\cdot Z_i$, since the counter $b$ is first set to $u$ at time $t^{(u)}_i$. 
Note that $Z_{i,z_x}$ has not been changed from time $t^{(u)}_i$ to $t^{(\ell)}_i$ for all $x\in[k-1]$. 
Since $X_{\ell}=\sum_{x=1}^{k}Z_{i,z_x}$, then we have 
\[X_{\ell}=Z_{i,z_k}+X_u\ge Z_{i,z_k}+\left(1+\frac{u\eps}{8}\right)\cdot Z_i.\]
Suppose by way of contradiction that $F(1,u_{i,z_k})-F(1,t_{i,z_k})>\frac{1}{2^{\beta-z_k-3}}\cdot F(1,t_i)$. 
Recall that $Z_{i,z_k}$ is an additive $2\eta_{z_k}F(1,t_i)$ approximation to $F(1,u_{i,z_k})-F(1,t_{i,z_k})$. 
Hence we have
\[Z_{i,z_k}>\frac{1}{2^{\beta-z_k-3}}\cdot F(1,t_i)-2\eta_{z_k}F(1,t_i)>\frac{2^{z_k}\cdot\eps}{3}\cdot F(1,t_i),\]
since $\beta=\ceil{\log\frac{8}{\eps}}$ and $\eta_{z_k}=\frac{\eta}{\varphi^{\beta-z_k}}<\frac{\eps}{3}$ for $\varphi>1$. 
Then we have 
\[X_{\ell}=Z_{i,z_k}+X_u>\frac{2^{z_k}\cdot\eps}{3}\cdot F(1,t_i)+\left(1+\frac{u\eps}{8}\right)\cdot Z_i>\left(1+\frac{\ell\eps}{8}\right)\cdot Z_i,\]
for sufficiently small constant $\eps>0$, since $2^{z_k}\ge(u-\ell)$ and $F(1,t_i)\ge(1-\eps)Z_i$ conditioned on $\mathcal{E}_1$. 
But this contradicts the fact that the counter $b$ is set to $\ell$ at time $t^{(\ell)}_i$, since 
\[X_{\ell}>\left(1+\frac{\ell\eps}{8}\right)\cdot Z_i.\]
Thus we have $F(1,u_{i,z_k})-F(1,t_{i,z_k})\le\frac{1}{2^{\beta-z_k-3}}\cdot F(1,t_i)$, which completes the inductive step. 
The result then follows from observing that for $\beta=\ceil{\log\frac{8}{\eps}}$, we have that $\frac{1}{2^{\beta-j-3}}F(1,t_i)\le\frac{2^j\eps}{64}$. 
\end{proof}

\begin{lemma}[Correctness on non-adaptive streams]
\lemlab{lem:frame:oblivious}
With probability at least $1-\delta\cdot\poly\left(\frac{1}{\eps},\log n\right)$, \algref{alg:framework} outputs a $(1+\eps)$-approximation to $F$ at all times on a non-adaptive stream.
\end{lemma}
\begin{proof}
As before, let $t_i$ be the update time during which the counter $a$ in \algref{alg:framework} is first set to $i$ for any integer $i>0$. 
For each integer $j\ge 0$, let $u_{i,j}$ be the last round for which $Z_{i,j}$ in \algref{alg:framework} is defined. 
Let $t_{i,\ell}$ be the first round such that the counter $b$ in \algref{alg:framework} is first set to $\ell$, and we use the convention that $t_{i,0}=t_i$. 
Let $\mathcal{E}_1$ be the event that \algref{alg:framework} outputs a $\left(1+\frac{\eps}{32}\right)$-approximation to the value of $F$ at all times $t_i$, so that by \lemref{lem:constant}, $\mathcal{E}_1$ holds with probability at least $1-\O{\frac{\delta\log n}{\eps}}$. 
For the remainder of the proof, we fix the integer $i$ and show correctness between $t_i$ and $t_{i+1}$. 
Note that conditioning on $\mathcal{E}_1$, we have correctness at times $t_i$ and $t_{i+1}$. 
By the monotonicity of $F$, it suffices to show correctness at all times $t_{i,\ell}$ for each value $\ell$ obtained by the counter $b$ between times $t_i$ and $t_{i+1}$. 

Let $\ell$ be a fixed value of the counter $b$, $k=\numbits(\ell)$, and $z_x=\lsb(\ell,x)$ for $x\in[k]$ so that $z_1>\ldots>z_k$ are the nonzero bits in the binary representation of $\ell$. 
By \lemref{lem:time:geometric}, we have for all $x\in[k]$
\[F(1,u_{i,z_x})-F(1,t_{i,z_x})\le\frac{1}{2^{\beta-x-3}}F(1,t_i),\]
with probability at least $1-\delta\cdot\poly\left(\frac{1}{\eps},\log n\right)$. 
Since $u_{i,z_x}$ is the last round for which $Z_{i,z_x}$ in \algref{alg:framework} is defined and all times $t_{i,j}$ for $j\le z_x$ are reset at that round, then we have $u_{i,z_x}=t_{i,z_{x+1}}$ for all $x\in[k]$. 
Thus since $u_{i,z_k}=t_{i,\ell}$, we can decompose 
\[F(1,t_{i,\ell})=\sum_{x=1}^k\left(F(1,u_{i,z_x})-F(1,u_{i,z_{x-1}})\right)=\sum_{x=1}^k\left(F(1,u_{i,z_x})-F(1,t_{i,z_x})\right).\]
Recall that the value $X_{\ell}$ of $X$ output by $\estimateF$ at time $t_{i,\ell}$ satisfies
\[X_{\ell}=Z_i+\sum_{x=1}^k Z_{i,z_x}.\]
By \lemref{lem:constant}, $Z_i$ incurs at most $\left(1+\frac{\eps}{8}\right)$ multiplicative error to $F(1,t_{i,\ell})$ since $F(1,t_{i,\ell})\le 4F(t_i)$. 
Moreover, $Z_{i,z_k}$ is an additive $2\eta_{z_k}\cdot F(1,t_i)$ approximation to $F(1,u_{i,z_k})-F(1,t_{i,z_k})$. 
Hence, the total additive error of $X_{\ell}$ to $F(1,t_{i,\ell})$ is at most
\[\frac{\eps}{4}\cdot F(1,t_{i,\ell})+\sum_{x=1}^k 2\eta_{z_x}\cdot F(1,t_i).\]
Now if the $(\gamma,\eps,\delta)$-difference estimator uses space dependency $\frac{\gamma}{\eps^2}$, then $\eta_x=\frac{\eta}{\beta}$ for all $x\in[\beta]$, so that 
\[\sum_{x=1}^k 2\eta_{z_x}\cdot F(1,t_i)\le\sum_{x\in[\beta]}\eta_x\cdot F(1,t_i)\le\frac{\eps}{32}\cdot F(1,t_i)\]
and the error is at most $\frac{\eps}{2}\cdot F(1,t_{i,\ell})$. 
On the other hand, if the $(\gamma,\eps,\delta)$-difference estimator uses space dependency $\frac{\gamma^C}{\eps^2}$ for $C>1$, then $\eta_x=\frac{\eta}{\varphi^{\beta-x}}$ for $\varphi=2^{(C-1)/4}$ implies that the error is at most
\[\sum_{x=1}^k 2\eta_{z_x}\cdot F(1,t_i)\le\frac{\eps}{2}\cdot F(1,t_{i,\ell}).\] 
Therefore, $X_{\ell}$ is a $\left(1+\frac{\eps}{2}\right)$-approximation to $F(1,t_{i,\ell})$. 

Note that the final output at time $t_{i,\ell}$ by \algref{alg:framework} further incurs additive error at most $\frac{\eps}{8}\cdot Z_i$ due to the rounding $\left(1+\frac{b\eps}{8}\right)\cdot Z_a$ in the last step. 
By \lemref{lem:constant}, $Z_i\le\left(1+\frac{\eps}{8}\right)\cdot F(1,t_{i,\ell})$. 
Hence for sufficiently small constant $\eps>0$, we have that \algref{alg:framework} outputs a $(1+\eps)$-approximation to $F(1,t_{i,\ell})$, with probability at least $1-\delta\cdot\poly\left(\frac{1}{\eps},\log n\right)$.
\end{proof}

\begin{theorem}[Framework for adversarially robust algorithms on insertion-only streams]
\thmlab{thm:framework}
Let $\eps,\delta>0$ and $F$ be a monotonic function with $(\eps,m)$-twist number $\lambda=\O{\frac{\log n}{\eps}}$ on a stream of length $m$, with $\log m=\O{\log n}$. 
Suppose there exists a $(\gamma,\eps,\delta)$-difference estimator for $F$ that uses $\O{\frac{\gamma^C}{\eps^2}\cdot S_1(n,\delta,\eps)+S_2(n,\delta,\eps)}$ bits of space for some constant $C\ge 1$ and a strong tracker for $F$ that use $\O{\frac{1}{\eps^2}\cdot S_1(n,\delta,\eps)+S_2(n,\delta,\eps)}$ bits of space and functions $S_1,S_2$ that depend on $F$. 
Then there exists an adversarially robust streaming algorithm that outputs a $(1+\eps)$-approximation for $F$ that succeeds with constant probability. 
For $C>1$, the algorithm uses 
\[\O{\frac{1}{\eps^2}\log n\cdot S_1(n,\delta',\eps)+\frac{1}{\eps}\log n\log\frac{1}{\eps}\cdot S_2(n,\delta',\eps)+\frac{1}{\eps^2}\log^2 n}\]
bits of space, where $\delta'=\O{\frac{1}{\poly\left(\frac{1}{\eps},\,\log n\right)}}$. 
For $C=1$, the algorithm uses
\[\O{\frac{1}{\eps^2}\log n\log^3\frac{1}{\eps}\cdot S_1(n,\delta',\eps)+\frac{1}{\eps}\log n\log\frac{1}{\eps}\cdot S_2(n,\delta',\eps)+\frac{1}{\eps^2}\log^2 n}\]
bits of space. 
\end{theorem}
\begin{proof}
Consider \algref{alg:framework}. 
\lemref{lem:frame:oblivious} proves correctness over non-adaptive streams. 
We now argue that \lemref{lem:frame:oblivious} holds over adversarial streams. 
The claim follows from a similar argument to the switch-a-sketch technique of \cite{Ben-EliezerJWY20}. 
Observe that each time either counter $a$ or $b$ increases, the subroutine $\estimateF$ causes \algref{alg:framework} to use new subroutines $\calB_{i,j}$ and $\calA_i$ that have not previously been revealed to the adversary. 
Thus the adversarial input is independent of the randomness of $\calB_{i,j}$ and $\calA_i$. 
Hence by the correctness guarantee over non-adaptive streams, \algref{alg:framework} outputs a $(1+\eps)$-approximation at all times. 

We now analyze the space complexity of \algref{alg:framework} for a difference estimator with space dependency $\frac{\gamma^C}{\eps^2}$ with $C>1$.   
Recall that $\estimateF$ uses instance $j=\flr{\frac{b+1}{2^k}}$ for $\calB_{i,j}$ for each ratio parameter $\gamma_k=2^{k-1}\eta$ and accuracy parameter $\eta_k=\frac{\eta}{\varphi^{\beta-k}}$, where $k\in[\beta]$ for $\beta=\O{\log\frac{1}{\eps}}$, $\varphi=2^{(C-1)/4}$, and $\eta=\frac{\eps}{64\zeta}$. 

By assumption, a single instance of $\calB$ uses 
\[\O{\frac{\gamma^C}{\eps^2}\cdot S_1(n,\delta,\eps)+S_2(n,\delta,\eps)}\]
bits of space. 
Since we require correctness of $\poly\left(\log n,\frac{1}{\eps}\right)$ instances of $\calB$, it suffices to set $\delta'=\O{\frac{1}{\poly\left(\frac{1}{\eps},\log n\right)}}$. 
For granularity $k$, we have $j=\flr{\frac{b+1}{2^k}}$ for $\calB_{i,j}$ so that there are at most $\frac{C_2}{2^k\eps}$ instances of $\calB_{i,j}$, for some absolute constant $C_2>0$. 
Thus for fixed $a$, $\beta\le\log\frac{8}{\eps}+1$, $\gamma_k=2^{k-1}\eta$, and accuracy parameter $\eta_k=\frac{\eta}{\varphi^{\beta-k}}$, the total space for each $\calB_{a,j}$ across the $\beta$ granularities is 
\begin{align*}
\sum_{k=1}^\beta\frac{C_3\gamma_k^C}{\eta_k^2 2^k\eps}\cdot S_1(n,\delta',\eps)+\frac{C_3}{2^k\eps}S_2(n,\delta',\eps)&=\sum_{k=1}^\beta\frac{C_3\gamma_{\beta-k}^C}{\eta_{\beta-k}^2 2^{\beta-k}\eps}\cdot S_1(n,\delta',\eps)+\frac{C_3}{2^k\eps}S_2(n,\delta',\eps)\\
&\le\sum_{k=1}^\beta\frac{C_4 (2^{\beta-k}\eps)^C\varphi^{2k}}{\eps^2\cdot2^{\beta-k}\eps}\cdot S_1(n,\delta',\eps)+\frac{C_4}{2^k\eps}\cdot S_2(n,\delta',\eps)\\
&\le\sum_{k=1}^\beta\frac{C_5 2^{-Ck}\varphi^{2k}}{\eps^2\cdot2^{-k}}\cdot S_1(n,\delta',\eps)+\frac{C_4}{2^k\eps}\cdot S_2(n,\delta',\eps),
\end{align*}
for some absolute constants $C_3,C_4,C_5>0$. 
Since $\varphi=2^{(C-1)/4}$ for $C>1$, then 
\begin{align*}
\sum_{k=1}^\beta\frac{C_3\gamma_k^C}{\eta_k^2 2^k\eps}\cdot S_1(n,\delta',\eps)+\frac{C_3}{2^k\eps}S_2(n,\delta',\eps)&\le\sum_{k=1}^\beta\frac{C_5 2^{k-Ck}2^{(Ck-k)/2}}{\eps^2}\cdot S_1(n,\delta',\eps)+\frac{C_4}{2^k\eps}\cdot S_2(n,\delta',\eps)\\
&\le\O{\frac{1}{\eps^2}\cdot S_1(n,\delta',\eps)+\frac{1}{\eps}\log\frac{1}{\eps}\cdot S_2(n,\delta',\eps)}. 
\end{align*}
Observe that for a stream with $(\eps,m)$-twist number $\lambda=\O{\frac{\log n}{\eps}}$, we have $a=\O{\log n}$, since the value of $F$ increases by a constant factor each time $a$ increases. 
Hence, accounting for $a=\O{\log n}$, then \algref{alg:framework} uses 
\[\O{\frac{1}{\eps^2}\log n\cdot S_1(n,\delta',\eps)+\frac{1}{\eps}\log n\log\frac{1}{\eps}\cdot S_2(n,\delta',\eps)}\]
bits of space in total for $C>1$ across each of the strong trackers and difference estimators. 

Finally, we observe that because we run $\O{\frac{1}{\eps^2}\log n}$ instances of the subroutines, it takes $\O{\frac{1}{\eps^2}\log^2 n}$ total additional bits of space to maintain the splitting times for each of the instances over the course of a stream of length $m$, with $\log m=\O{\log n}$, so that the total space is
\[\O{\frac{1}{\eps^2}\log n\cdot S_1(n,\delta',\eps)+\frac{1}{\eps}\log n\log\frac{1}{\eps}\cdot S_2(n,\delta',\eps)+\frac{1}{\eps^2}\log^2 n}.\]

We now analyze the space complexity for \algref{alg:framework} for a difference estimator with space dependency $\frac{\gamma^C}{\eps^2}$ with $C=1$. 
Instead of the accuracy parameter $\eta_k=\frac{\eta}{\varphi^{\beta-k}}$ for $C>1$, we now have $\eta_k=\frac{\eta}{\beta}$ for $C=1$
Thus for fixed $a$, $\beta\le\log\frac{8}{\eps}+1$, and $\gamma_k=2^{k-1}\eta$, there exist constants $C_6,C_7>0$ such that the total space for each $\calB_{a,j}$ across the $\beta$ granularities is 
\begin{align*}
\sum_{k=1}^\beta\frac{C_3\gamma_k^C}{\eta_k^2 2^k\eps}\cdot S_1(n,\delta',\eps)+\frac{C_3}{2^k\eps}S_2(n,\delta',\eps)&\le\sum_{k=1}^\beta\frac{C_6 (2^{k-1}\eps)}{\eta^2\beta^{-2}\cdot 2^k\eps}\cdot S_1(n,\delta',\eps)+\frac{C_6}{2^k\eps}\cdot S_2(n,\delta',\eps)\\
&\le\sum_{k=1}^\beta\frac{C_7\beta^2}{\eta^2}\cdot S_1(n,\delta',\eps)+\frac{C_7}{2^k\eps}\cdot S_2(n,\delta',\eps)\\
&\le\O{\frac{1}{\eps^2}\log^3\frac{1}{\eps}\cdot S_1(n,\delta',\eps)+\frac{1}{\eps}\log\frac{1}{\eps}\cdot S_2(n,\delta',\eps)}. 
\end{align*}
Again accounting for $a=\O{\log n}$ and the space it takes to maintain the splitting times, then \algref{alg:framework} uses 
\[\O{\frac{1}{\eps^2}\log n\log^3\frac{1}{\eps}\cdot S_1(n,\delta',\eps)+\frac{1}{\eps}\log n\log\frac{1}{\eps}\cdot S_2(n,\delta',\eps)+\frac{1}{\eps^2}\log^2 n}.\]
bits of space in total for $C>1$ across each of the strong trackers and difference estimators. 
\end{proof}

\section{Robust \texorpdfstring{$F_2$}{F2} Estimation}
\seclab{sec:F2}
In this section, we use the previous framework of \secref{sec:framework} to give an adversarially robust streaming algorithm for $F_2$ moment estimation. 
Recall that to apply \thmref{thm:framework}, we require an $F_2$ strong tracker and an $F_2$ difference estimator. 
We present these subroutines in this section. 
We further optimize our algorithm beyond the guarantees of \thmref{thm:framework} specifically for $F_p$ moments, so that our final space guarantees in \thmref{thm:robust:opt:F2} matches the best known $F_2$ algorithm on insertion-only streams, up to lower order $\polylog\frac{1}{\eps}$ terms.  
Finally, we show that our algorithm naturally extends to the problem of finding the $L_2$-heavy hitters, along with producing an estimate for the frequency of each heavy-hitter up to an additive $\O{\eps}\cdot L_2$ error. 

We first recall the following $F_2$ strong tracker. 
\begin{theorem}[Oblivious $F_2$ strong tracking]
\thmlab{thm:strong:F2}
\cite{BlasiokDN17}
Given an accuracy parameter $\eps>0$ and a failure probability $\delta\in(0,1)$, let $d=\O{\frac{1}{\eps^2}\left(\log\frac{1}{\eps}+\log\frac{1}{\delta}+\log\log n\right)}$. 
There exists an insertion-only streaming algorithm $\estimator$ that uses $\O{d\log n}$ space to provide $(\eps,\delta)$-strong $F_2$ tracking of an underlying frequency vector $f$. 
%There exists an insertion-only streaming algorithm $\estimator$ that uses $\O{d\log n}$ space to output a vector $v$ of length $d$ that provides $(\eps,\delta)$-strong $F_2$ tracking of an underlying frequency vector $f$. 
\end{theorem}
To define our difference estimator, we first note that ``good'' $F_2$ approximation to two vectors $u$ and $v$ also gives a ``good'' approximation to their inner product $\ip{u}{v}$. 

\begin{lemma}[AMS $F_2$ approximation gives inner product approximation]
\lemlab{lem:approx:ip}
%\lemlab{lem:approx:ip:one}
Let vectors $u,v\in\mathbb{R}^n$ and $M\in\mathbb{R}^{d\times n}$ be a sketching matrix so that each entry $M_{i,j}$ of $M$ is a four-wise independent random sign scaled by $\frac{1}{\sqrt{d}}$ for $d=\O{\frac{1}{\eps^2}}$. 
Then
\[\Pr{|\langle u,v\rangle-\langle Mu, Mv\rangle|\le\eps\|u\|_2\|v\|_2}\ge\frac{2}{3}.\]
\end{lemma}
\begin{proof}
The proof is a standard variance reduction argument that follows immediately from \lemref{lem:approx:ip:all} and \lemref{lem:ams:ex} and Chebyshev's inequality, since $M$ can be viewed as taking the arithmetic mean of $d$ scaled sign vectors $s_1,\ldots,s_d$. 
\end{proof}

\begin{lemma}[Strong tracking of AMS inner product approximation]
\lemlab{lem:approx:ip:all}
Given vectors $0^n\preceq u_1\preceq u_2\preceq\ldots\preceq u_m\in\mathbb{R}^n$ whose entries are bounded by a polynomial in $n$, there exists an algorithm that uses a sketching matrix $M\in\mathbb{R}^{d\times n}$ with $d=\O{\frac{1}{\eps^2}\left(\log\frac{1}{\eps}+\log\frac{1}{\delta}+\log\log n\right)}$ such that for $m=\poly(n)$ and a fixed $v\in\mathbb{R}^n$ with $v\succeq 0^n$, 
\[|\ip{u_i}{v}-\ip{Mu_i}{Mv}|\le\eps\|u_i\|_2\|v\|_2,\]
simultaneously for all $i\in[m]$ with probability at least $1-\delta$. 
\end{lemma}
\begin{proof}
Consider a sequence of indices $t_1<t_2<\ldots<t_q$ such that $t_{i+1}$ is the minimal index such that $\langle u_{t_{i+1}},v\rangle\ge\left(1+\frac{\eps}{2}\right)\langle u_{t_i},v\rangle$ and observe that $q=\O{\frac{1}{\eps}\log n}$ for $m=\poly(n)$ since each vector $u_i$ has polynomially bounded entries. 
Since $0^n\preceq u_1\preceq u_2\preceq\ldots\preceq u_m\in\mathbb{R}^n$ and thus both $\|u_1\|_2\le\|u_2\|_2\le\ldots\|u_m\|_2$ and $\langle u_1,v\rangle\le\langle u_2,v\rangle\le\ldots\le\langle u_m,v\rangle$, it suffices to show that 
\[|\ip{u_{t_i}}{v}-\ip{Mu_{t_i}}{Mv}|\le\eps\|u_{t_i}\|_2\|v\|_2,\]
for all $i\in[q]$. 
By \lemref{lem:ams:ex} and \lemref{lem:ams:var}, we have that for any 
\[\Pr{|\ip{u_{t_i}}{v}-\ip{Mu_{t_i}}{Mv}|\le\eps\|u_{t_i}\|_2\|v\|_2}\ge\frac{2}{3}\]
for a sketching matrix with $\O{\frac{1}{\eps^2}}$ rows. 
Thus by taking a sketching matrix with $d$ rows, for $d=\O{\frac{1}{\eps^2}\left(\log\frac{1}{\eps}+\log\frac{1}{\delta}+\log\log n\right)}$, and using a standard median-of-means approach, we can take a union bound over all $i\in[q]$. 
Therefore,
\[|\ip{u_i}{v}-\ip{Mu_i}{Mv}|\le\eps\|u_i\|_2\|v\|_2,\]
simultaneously for all $i\in[m]$ with probability at least $1-\delta$. 
\end{proof}

\begin{lemma}
\lemlab{lem:ams:ex}
Let $s\in\mathbb{R}^n$ be a sign vector $\{-1,+1\}^n$ with four-wise independent entries. 
For any $u,v\in\mathbb{R}^n$, we have
\[\Ex{\langle s,u\rangle\cdot\langle s,v\rangle}=\langle u,v\rangle.\]
\end{lemma}
\begin{proof}
By linearity of expectation, we have that
\begin{align*}
\Ex{\langle s,u\rangle\cdot\langle s,v\rangle}&=\Ex{\left(\sum_{i=1}^ns_iu_i\right)\left(\sum_{i=1}^ns_iv_i\right)}\\
&=\Ex{\sum_{i=1}^n\sum_{j=1}^n s_iu_is_jv_j}=\sum_{i=1}^n\sum_{j=1}^n\Ex{s_iu_is_jv_j}.
\end{align*}
Since the variables $s_i$ are $4$-wise independent sign variables $\{-1,+1\}$, then $\Ex{s_is_j}=1$ for $i=j$ and $\Ex{s_is_j}=0$ for $i\neq j$. 
Hence, 
\[\Ex{\langle s,u\rangle\cdot\langle s,v\rangle}=\sum_{i=1}^nu_iv_i=\langle u,v\rangle.\]
%By linearity of expectation, we have that
%\begin{align*}
%\Ex{\langle\M\u,\M\v\rangle}&=\Ex{(\M\u)_1(\M\v)_1+(\M\u)_2(\M\v)_2+\ldots+(\M\u)_d(\M\v)_d}\\
%&=\sum_{i=1}^d\Ex{\left(\sum_{j=1}^n M_{i,j} u_j\right)\left(\sum_{j=1}^n M_{i,j} v_j\right)}\\
%&=\frac{1}{d}\sum_{i=1}^d\Ex{\left(\sum_{j=1}^n s_{i,j} u_j\right)\left(\sum_{j=1}^n s_{i,j} v_j\right)},
%\end{align*}
%where the variables $s_{i,j}$ are $4$-wise independent sign variables $\{-1,+1\}$. 
%Thus, $\Ex{s_{i,j}s_{i,k}}=1$ for $j=k$ and $\Ex{s_{i,j}s_{i,k}}=0$ for $j\neq k$. 
%Hence, we have
%\[\Ex{\langle\M\u,\M\v\rangle}=\frac{1}{d}\sum_{i=1}^d\sum_{j=1}^n u_jv_j=\frac{1}{d}\sum_{i=1}^d \langle\u,\v\rangle=\langle\u,\v\rangle.\]
\end{proof}

\begin{lemma}
\lemlab{lem:ams:var}
Let $s\in\mathbb{R}^n$ be a sign vector $\{-1,+1\}^n$ with four-wise independent entries. 
For any $u,v\in\mathbb{R}^n$, we have
\[\Var\left(\langle s,u\rangle\cdot\langle s,v\rangle\right)\le(\langle u,v\rangle)^2.\]
%\[\Var{\ip{\M\u_i}{\M\v}}\le\frac{1}{d}\,(\ip{\u_i}{\v})^2.\]
\end{lemma}
\begin{proof}
Since $\Var{\langle s,u\rangle\cdot\langle s,v\rangle}\le\Ex{\left(\langle s,u\rangle\cdot\langle s,v\rangle\right)^2}$, it follows by linearity of expectation that
\begin{align*}
\Var\left(\langle s,u\rangle\cdot\langle s,v\rangle\right)&\le\Ex{\left(\sum_{i=1}^ns_iu_i\right)^2\left(\sum_{i=1}^ns_iv_i\right)^2}\\
&=\Ex{\sum_{i=1}^n\sum_{j=1}^n\sum_{k=1}^n\sum_{\ell=1}^n s_is_js_ks_\ell u_iv_ju_kv_\ell}\\
&=\sum_{i=1}^n\sum_{j=1}^n\sum_{k=1}^n\sum_{\ell=1}^n\Ex{s_is_js_ks_\ell u_iv_ju_kv_\ell}.
\end{align*}
Since the variables $s_i$ are $4$-wise independent sign variables $\{-1,+1\}$, then $\Ex{s_is_js_ks_\ell}=1$ if $i,j,k,\ell$ consists of two (possibly not distinct) pairs of indices, e.g., $i=j$ and $k=\ell$, and $\Ex{s_is_js_ks_\ell}=0$ otherwise. 
Therefore, 
\[\Var\left(\langle s,u\rangle\cdot\langle s,v\rangle\right)\le\sum_{i=1}^n\sum_{j=1}^n u_iv_iu_jv_j=\left(\langle u,v\rangle\right)^2.\]
%Since $\Var{\langle\M\u,\M\v\rangle}\le\Ex{(\langle\M\u,\M\v\rangle)^2}$, it follows that
%\begin{align*}
%\Var{\langle\M\u,\M\v\rangle}&\le\Ex{\left((\M\u)_1(\M\v)_1+(\M\u)_2(\M\v)_2+\ldots+(\M\u)_d(\M\v)_d\right)^2}\\
%&=\Ex{\sum_{i=1}^d\sum_{j=1}^d(\M\u)_i(\M\v)_i(\M\u)_j(\M\v)_j}\\
%&=\sum_{i=1}^d\sum_{j=1}^d\Ex{\left(\sum_{k=1}^n M_{i,k} u_k\right)\left(\sum_{k=1}^n M_{i,k} v_k\right)\left(\sum_{j=1}^n M_{j,k} u_k\right)\left(\sum_{j=1}^n M_{j,k} v_k\right)}\\
%&=\frac{1}{d^2}\sum_{i=1}^d\sum_{j=1}^d\Ex{\left(\sum_{k=1}^n s_{i,k} u_k\right)\left(\sum_{k=1}^n s_{i,k} v_k\right)\left(\sum_{k=1}^n s_{j,k} u_k\right)\left(\sum_{k=1}^n s_{j,k} v_k\right)},
%\end{align*}
%where the variables $s_{i,k}$ are again $4$-wise independent sign variables $\{-1,+1\}$. 
%Thus, $\Ex{s_{i,k}s_{i,\ell}s_{j,k}s_{j,\ell}}=1$ for either $i=j$ or $k=\ell$ and $\Ex{s_{i,j}s_{i,k}}=0$ otherwise. 
%Hence, we have
%\[\Var{\langle\M\u,\M\v\rangle}\le\frac{1}{d^2}\sum_{i=1}^d\sum_{j=1}^d u_iv_iu_jv_j=\left(\langle\u,\v\rangle\right)^2.\]
\end{proof}

%\begin{lemma}[$F_2$ approximation is inner product approximation]
%\lemlab{lem:approx:ip}
%Given $n$ points $u_1,\ldots,u_m\in\mathbb{R}^n$, let $M:\mathbb{R}^n\to\mathbb{R}^d$ be a function such that
%\[\left(1-\frac{\eps}{4}\right)\|u_i\|_2\le\|M(u_i)\|_2\le\left(1+\frac{\eps}{4}\right)\|u_i\|_2,\]
%for all $i\in[n]$. 
%Then $i,j\in[n]$:
%\[|\ip{u_i}{u_j}-\ip{M(u_i)}{M(u_j)}| \le \eps\|u_i\|_2\|u_j\|_2.\]
%\end{lemma}
%\begin{proof}
%Observe that $(1-\eps)\|u_i\|_2\le\|M(u_i)\|_2\le(1+\eps)\|u_i\|_2$ implies
%\begin{align*}\label{eq:approx}
%4\ip{M(u_i)}{M(u_j)}&=\|M(u_i)+M(u_j)\|_2^2-\|M(u_i)-M(u_j)\|_2^2\\
%&\ge(1-\eps)\|u_i+u_j\|_2^2-(1+\eps)\|u_i-u_j\|_2^2\\
%&=4\ip{u_i}{u_j}-2\eps(\|u_i\|_2^2+\|u_j\|_2^2)\\
%&\ge4\ip{u_i}{u_j}-4\eps(\|u_i\|_2\|u_j\|_2).
%\end{align*}
%The argument then follows from rescaling $\eps$. 
%Similarly, we have
%\begin{align*}
%4\ip{M(u_i)}{M(u_j)}&=\|M(u_i)+M(u_j)\|_2^2-\|M(u_i)-M(u_j)\|_2^2\\
%&\le(1+\eps)\|u_i+u_j\|_2^2-(1-\eps)\|u_i-u_j\|_2^2\\
%&=4\ip{u_i}{u_j}+2\eps(\|u_i\|_2^2+\|u_j\|_2^2)\\
%&\le4\ip{u_i}{u_j}+4\eps(\|u_i\|_2\|u_j\|_2).
%\end{align*}
%\end{proof}

\noindent
We now give our $F_2$ difference estimator using the inner product approximation property. 
\begin{lemma}[$F_2$ difference estimator]
\lemlab{lem:diff:est:F2}
There exists a $(\gamma,\eps,\delta)$-\emph{difference estimator} for $F_2$ that uses space
\[\O{\frac{\gamma\log n}{\eps^2}\left(\log\frac{1}{\eps}+\log\frac{1}{\delta}\right)}.\]
\end{lemma}
\begin{proof}
Let $\eps'=\frac{\min(1,\eps)}{32}$ and $\delta'=\frac{\delta}{2}$. 
Given an oblivious stream $\calS$, let $v$ be the frequency vector induced by the updates of $\calS$ from time $t_1$ to $t_2$ and $w_t$ be the frequency vector induced by updates from time $t_2$ to $t$ exclusive, with $\gamma\cdot F_2(v)\ge F_2(w_t)$. 
Let $M(x)$ represent the output of the algorithm $\estimator$ with accuracy parameter $\eps'$ and failure rate $\delta'$ on an input vector $x$ induced by a stream. 
Hence we have $M(x)\in\mathbb{R}^d$, where $d=\O{\frac{\gamma}{\eps^2}\left(\log\frac{1}{\eps}+\log\frac{1}{\delta}\right)}$. 
We remark that the $\log\log n$ term in the dimension of $d$ from \thmref{thm:strong:F2} results from the analysis of \cite{BlasiokDN17} breaking the stream into $\O{\log n}$ times when the value of $F_2$ on the stream doubles. 
We do not need this $\log\log n$ term since $\gamma\cdot F_2(v)\ge F_2(w_t)$ for an absolute constant $\gamma>0$ implies that the value of the difference can only double a constant number of times. 

%Let $\mathcal{E}_1$ be the event that $M$ succeeds on the substream from time $t_1$ to $t_2$ and $\mathcal{E}_2$ be the event that $M$ succeeds on the substream from time $t_2$ to $t$. 
%We condition on $\mathcal{E}_1$ and $\mathcal{E}_2$, which each hold with probability at least $1-\delta'$. 
%By \thmref{thm:strong:F2}, we have 
%\[\left(1-\frac{\eps'}{\sqrt{\gamma}}\right)\|v\|_2\le\|M(v)\|_2\le\left(1+\frac{\eps'}{\sqrt{\gamma}}\right)\|v\|_2.\]
%Similarly by \thmref{thm:strong:F2}, we have 
%\[\left(1-\frac{\eps'}{\sqrt{\gamma}}\right)\|w_t\|_2\le\|M(w_t)\|_2\le\left(1+\frac{\eps'}{\sqrt{\gamma}}\right)\|w_t\|_2,\]
%for all $t\ge t_2$. 
%Thus by \lemref{lem:approx:ip}, we have 
By \lemref{lem:approx:ip}, we have that with probability at least $1-\delta'$,
\[|\ip{v}{w_t}-\ip{M(v)}{M(w_t)}|\le\frac{4\eps'}{\sqrt{\gamma}}\|v\|_2\|w_t\|_2.\]
Recalling the identity 
\[F_2(v+w_t)-F_2(v)=2\ip{v}{w_t}+\|w_t\|_2^2,\]
we thus have that
\begin{align*}
\big|\left[F_2(v+w_t)-F_2(v)\right]-&\left[2\ip{M(v)}{M(w_t)}+\|M(w_t)\|_2^2\right]\big|\\
&=\big|\left[2\ip{v}{w_t}-2\ip{M(v)}{M(w_t)}\right]+\left[\|w_t\|_2^2-\|M(w_t)\|_2^2\right]\big|\\
&\le8\eps'\|v\|_2\|w_t\|_2+3\eps'\|w_t\|_2^2,
\end{align*}
since $\eps'=\frac{\min(1,\eps)}{32}$. 
Because $\gamma\cdot F(v)\ge F(w_t)$, then we have $\|w_t\|_2\le\sqrt{\gamma}\|v\|_2$, so that the total error of the estimator is at most
\[\frac{8\eps'}{\sqrt{\gamma}}\|v\|_2\|w_t\|_2+3\eps'\|w_t\|_2^2\le 8\eps'\cdot F(v)+ 3\eps'\gamma\cdot F(v).\]
Hence for $\gamma\le 2$, the error is at most $14\eps'\cdot F(v)\le\eps\cdot F(v)$. 
%Since $\mathcal{E}_1$ and $\mathcal{E}_2$ each hold with probability at least $1-\delta'$ for $\delta'=\frac{\delta}{2}$, then the total probability of success is at least $1-\delta$. 
Since \thmref{thm:strong:F2} and \lemref{lem:approx:ip} each hold with probability at least $1-\delta'$ for $\delta'=\frac{\delta}{2}$, then the total probability of success is at least $1-\delta$.

By \thmref{thm:strong:F2}, each of the two instances of $\estimator$ uses
\[\O{\frac{\gamma\log n}{(\eps')^2}\left(\log\frac{1}{\eps'}+\log\frac{1}{\delta'}+\log\log n\right)}\]
bits of space, for $\eps'=\frac{\min(1,\eps)}{32}$ and $\delta'=\frac{\delta}{2}$. 
However, the $\log\log n$ term in the space requirement of $\estimator$ results from the need to union bound over $\log n$ locations where the value of $F_2$ on the stream increases by a factor of $2$. 
Since $\gamma$ is upper bounded by a constant, we only need a union bound over a constant number of locations. 
Hence, the space complexity of the $(\gamma,\eps,\delta)$-difference estimator for $F_2$ is 
\[\O{\frac{\gamma\log n}{\eps^2}\left(\log\frac{1}{\eps}+\log\frac{1}{\delta}\right)}.\]
\end{proof}

\begin{theorem}
Given $\eps>0$, there exists an adversarially robust streaming algorithm that outputs a $(1+\eps)$-approximation for $F_2$ that uses $\O{\frac{1}{\eps^2}\log^2 n\log^3\frac{1}{\eps}\left(\log\frac{1}{\eps}+\log\log n\right)}$ bits of space and succeeds with probability at least $\frac{2}{3}$. 
\end{theorem}
\begin{proof}
$F_2$ is a monotonic function with $(\eps,m)$-twist number $\lambda=\O{\frac{1}{\eps}\log n}$, by \obsref{obs:flip:fp}. 
By \lemref{lem:diff:est:F2} and \thmref{thm:strong:F2}, there exists a $(\gamma,\eps,\delta)$-difference estimator that uses space
\[\O{\frac{\gamma\log n}{\eps^2}\left(\log\frac{1}{\eps}+\log\frac{1}{\delta}\right)}\]
and there exists an oblivious strong tracker for $F_2$ that uses space
\[\O{\frac{\log n}{\eps^2}\left(\log\frac{1}{\eps}+\log\frac{1}{\delta}+\log\log n\right)}.\]
Therefore by using the $(\gamma,\eps,\delta)$-difference estimator and the oblivious strong tracker for $F_2$ in the framework of \algref{alg:framework}, then \thmref{thm:framework} with $S_1=\log n\left(\log\frac{1}{\eps}+\log\frac{1}{\delta}+\log\log n\right)$, $C=1$, and $S_2=0$ proves that there exists an adversarially robust streaming algorithm that outputs a $(1+\eps)$-approximation for $F_2$ that uses space $\O{\frac{1}{\eps^2}\log^2 n\log^3\frac{1}{\eps}\left(\log\frac{1}{\eps}+\log\log n\right)}$ and succeeds with constant probability.    
\end{proof}

\paragraph{Optimized $F_2$ Algorithm.}
Instead of maintaining all sketches $\calB_a$ and $\calA_{a,j}$ simultaneously, it suffices to maintain sketches $\calB_a$ and $\calA_{a,j}$ for $\O{\log\frac{1}{\eps}}$ values of $a$ at a time. 
Namely, the counter $a$ in \algref{alg:framework} tracks the active instances $\calB_a$ and $\calA_{a,j}$ output by the algorithm. 
Because the output increases by a factor of $2$ each time $a$ increases, it suffices to maintain the sketches $\calB_i$ and $\calA_{i,c}$ for only the smallest $\O{\log\frac{1}{\eps}}$ values of $i$ that are at least $a$. 
Any larger index will have only missed $\O{\eps}$ fraction of the $F_2$ of the stream and thus still output a $(1+\eps)$-approximation. 

\begin{theorem}[Adversarially robust $F_2$ streaming algorithm]
\thmlab{thm:robust:opt:F2}
Given $\eps>0$, there exists an adversarially robust streaming algorithm that outputs a $(1+\eps)$-approximation for $F_2$ that succeeds with probability at least $\frac{2}{3}$ and uses $\O{\frac{1}{\eps^2}\log n\log^4\frac{1}{\eps}\left(\log\frac{1}{\eps}+\log\log n\right)}$ bits of space. 
\end{theorem}
\begin{proof}
At any point, there are $\O{\log\frac{1}{\eps}}$ active indices $a$ and $\O{1}$ active indices $c$ corresponding to sketches $\calB_{a}$ and $\calA_{a,c}$. 
Recall from \thmref{thm:framework} that for fixed $a$, the total space for each $\calA_{a,j}$ across the $\beta$ granularities is 
\[\O{\frac{1}{\eps^2}\log^3\frac{1}{\eps}\cdot S_1(n,\delta',\eps)+\frac{1}{\eps}\log\frac{1}{\eps}\cdot S_2(n,\delta',\eps)},\]
where we have $S_1(n,\delta',\eps)=\log n\left(\log\frac{1}{\eps}+\log\frac{1}{\delta'}+\log\log n\right)$ and $S_2=0$ for our $F_2$ difference estimator and $F_2$ strong tracker. 

Whereas \thmref{thm:framework} bounds the counter $a$ by $\O{\log n}$ so that we maintain $\O{\log n}$ simultaneous instances of $i$ for the algorithms $\calA_{i,j}$, we observed that it suffices to maintain at most $\O{\log\frac{1}{\eps}}$ simultaneous active values of $i$. 
Since there are still $\O{\log n}$ total indices $a$ over the course of the stream, then each sketch must have failure probability $\frac{\delta}{\poly\left(\log n,\frac{1}{\eps}\right)}$ for the entire algorithm to have failure probability $\delta=\frac{2}{3}$. 
Finally, there are $\O{\log\frac{1}{\eps}}$ active indices of $a$, consisting of $\O{\frac{1}{\eps}}$ subroutines. 
Thus, it takes $\O{\frac{1}{\eps}\log n\log\frac{1}{\eps}}$ additional bits of space to store the splitting times for each of the $\O{\frac{1}{\eps}\log\frac{1}{\eps}}$ active subroutines across a stream of length $m$, with $\log m=\O{\log n}$. 
Thus the total space required is $\O{\frac{1}{\eps^2}\log n\log^4\frac{1}{\eps}\left(\log\frac{1}{\eps}+\log\log n\right)}$. 
\end{proof}

\paragraph{Heavy-hitters.}
As a simple corollary, note that our framework also solves the $L_2$-heavy hitters problem. 
By running separate $L_2$-heavy hitters algorithms corresponding to each difference estimator $\calA$ and strong tracker $\calB$, with the heavy-hitter threshold corresponding to the accuracy of each procedure, we obtain a list containing all the possible heavy-hitters along with an estimated frequency of each item in the list. 

\begin{theorem}
\cite{BravermanCINWW17}
\thmlab{thm:bptree}
For any $\eps>0$ and $\delta\in[0,1)$, there exists a streaming algorithm $(\eps,\delta)-\bptree$, that with probability at least $1-\delta$, returns a set of $\frac{\eps}{2}$-heavy hitters containing every $\eps$-heavy hitter and an approximate frequency for every item returned with additive error $\frac{\eps}{4}\cdot L_2$. 
The algorithm uses $\O{\frac{1}{\eps^2}\left(\log\frac{1}{\delta\eps}\right)(\log n+\log m)}$ bits of space.
\end{theorem}

\begin{theorem}[Adversarially robust $L_2$-heavy hitters streaming algorithm]
\thmlab{thm:robust:opt:HH}
Given $\eps>0$, there exists an adversarially robust streaming algorithm $\heavyhitters$ that solves the $L_2$-heavy hitters problem with probability at least $\frac{2}{3}$ and uses $\O{\frac{1}{\eps^2}\log n\log^4\frac{1}{\eps}\left(\log\frac{1}{\eps}+\log\log n\right)}$ bits of space. 
\end{theorem}
\begin{proof}
For each integer $j\ge 0$ and a fixed integer $i$, let $u_{i,j}$ be the last round for which $Z_{i,j}$ in \algref{alg:framework} is defined. 
As before, let $t_i$ be the update time during which the counter $a$ in \algref{alg:framework} is first set to $i$ for any integer $i>0$. 
Let $t^{(b)}_i$ be the first round such that the counter $b$ in \algref{alg:framework} is first set to $j$, and we use the convention that $t^{(0)}_i=t_i$.

Let $\mathcal{E}_1$ be the event that $X_a$ is a $\left(1+\frac{\eps}{32}\right)$-approximation to $F_2(1,t_a)$ and $\mathcal{E}_2$ be the event that $F(1,u_{a,j})-F(1,t_{a,j})\le\frac{1}{2^{\beta-j-3}}F(1,t_a)$ for each positive integer $j>1$. 
By \lemref{lem:constant} and \lemref{lem:time:geometric}, we have $\PPr{\mathcal{E}_1\wedge\mathcal{E}_2}\ge 1-\O{\frac{\delta\log n}{\eps}}$.

Suppose there exists $r\in[n]$ such that $(f_r)^2\ge\eps^2 F_2(1,m)$. 
Then coordinate $r$ must either be $\frac{\eps^2}{32^2}$-heavy with respect to $F_2(1,t_a)$ or $\frac{2^k\cdot\eps^2}{32^2\beta^2}$-heavy with respect to $F_2(t_{a,k},t_{a,k-1})$ for some integer $k\in[\beta]$, where we use the convention that $t_{a,0}=t_a$. 
%$\eps\cdot(\sqrt{2})^{k-10}$
Hence $\bptree$ with threshold $\frac{2^k\cdot\eps^2}{32^2\beta^2}$ with respect to $F_2(t_{a,k},t_{a,k-1})$ will detect that $r$ is heavy. 
Moreover, note that if coordinate $r$ is only reported as heavy in an algorithm corresponding to the level $k$ difference estimator, then $\O{\eps^2}\cdot F_2$ of the contribution of coordinate $r$ towards the overall $F_2$ still appears after $r$ is reported as a heavy-hitter. 
Thus by keeping track of the frequency of $r$ after it is reported, we obtain an estimate of the frequency of $r$ up to an additive $\O{\eps}\cdot L_2$. 
Since this is also the accuracy parameter of $\calA_{a,k}$, then the space complexity follows.  
\end{proof}

\section{Robust \texorpdfstring{$F_p$}{Fp} Estimation for \texorpdfstring{$0<p<2$}{0<p<2}}
\seclab{sec:smallp}
In this section, we use the framework of \secref{sec:framework} to give an adversarially robust streaming algorithm for $F_p$ moment estimation, where $p\in(0,2)$. 
As before, we require an $F_p$ strong tracker and an $F_p$ difference estimator in order to apply \thmref{thm:framework}. 
Whereas the $F_2$ difference estimator was straightforward from the inner product approximation via $F_2$ approximation property, it is not clear that such a formulation exists for $F_p$ approximation. 
Instead, we use separate sketches for our $F_p$ strong tracker and our $F_p$ difference estimator. 
We present these subroutines in this section. 
Finally, we again optimize our algorithm beyond the guarantees of \thmref{thm:framework} so that our final space guarantees in \thmref{thm:robust:opt:Fp:smallp} matches the best known $F_p$ algorithm on insertion-only streams, up to $\poly\left(\log\frac{1}{\eps},\log n\right)$ terms. 

We first require the following definition for $p$-stable distributions, which will be integral to both our $F_p$ strong tracker and our $F_p$ difference estimator. 
\begin{definition}[$p$-stable distribution]
\cite{Zolotarev89}
\deflab{def:pstable}
For $0<p\le 2$, there exists a probability distribution $\calD_p$ called the $p$-stable distribution so that for any positive integer $n$ with $Z_1,\ldots,Z_n\sim\calD_p$ and vector $x\in\mathbb{R}^n$, then $\sum_{i=1}^n Z_ix_i\sim\|x\|_pZ$ for $Z\sim\calD_p$. 
\end{definition}
The probability density function $f_X$ of a $p$-stable random variable $X$ satisfies $f_X(x)=\Theta\left(\frac{1}{1+|x|^{1+p}}\right)$ for $p<2$, while the normal distribution corresponds to $p=2$. 
Moreover, \cite{Nolan03} details standard methods for generating $p$-stable random variables by taking $\theta$ uniformly at random from the interval $\left[-\frac{\pi}{2},\frac{\pi}{2}\right]$, $r$ uniformly at random from the interval $[0,1]$, and setting
\[X=f(r,\theta)=\frac{\sin(p\theta)}{\cos^{1/p}(\theta)}\cdot\left(\frac{\cos(\theta(1-p))}{\log\frac{1}{r}}\right)^{\frac{1}{p}-1}.\]
These $p$-stable random variables are crucial to obtaining a strong $F_p$ tracking algorithm.
\begin{theorem}[Oblivious $F_p$ strong tracking for $0<p<2$]
\thmlab{thm:strong:Fp:smallp}
\cite{BlasiokDN17}
For $0<p<2$, there exists an insertion-only streaming algorithm $\psketch(1,t,\eps,\delta)$ that uses $\O{\frac{\log n}{\eps^2}\left(\log\log n+\log\frac{1}{\eps}+\log\frac{1}{\delta}\right)}$ bits of space and provides $(\eps,\delta)$-strong $F_p$ tracking. 
\end{theorem}
Unfortunately, $\psketch$ is based on the $p$-stable sketch of~\cite{Indyk06}, which offers a conceptual promise for the existence of the quantile estimators needed to guarantee provable bounds, but not an explicit computation. 
Thus adapting the analysis of $\psketch$ in \cite{Indyk06, BlasiokDN17} for the purposes of our difference estimator seems to be a challenge. 
Even for the case $p=1$, it does not seem evident how to adapt the median estimator of $\psketch$ to obtain a difference estimator for $F_p$. 
 
Instead, we describe a formulation of Li's geometric mean estimator~\cite{Li08}, which also provides a streaming algorithm for $F_p$, but was not previously known to offer strong tracking. 
For a positive integer $q\ge3$, let $d$ be a multiple of $q$ and let $A\in\mathbb{R}^{d\times n}$ have independent $p$-stable random variables for the entries of $A$. 
Then for a vector $x\in\mathbb{R}^n$ and $y=Ax$, let $z_i:=C_{q,p}\cdot\left(\prod_{j=q(i-1)+1}^{qi}|y_j|^{p/q}\right)$ be the geometric mean of the inner products of $q$ random $p$-stable vectors with the vector $x$, where 
\[C_{q,p}=\left[\frac{2}{\pi}\cdot\Gamma\left(1-\frac{1}{q}\right)\cdot\Gamma\left(\frac{p}{q}\right)\cdot\sin\left(\frac{\pi p}{2q}\right)\right]^{-q}.\]
We have the following characterization of the asymptotic behavior of $C_{q,p}$:
\begin{observation}[Characterization of $C_{q,p}$]
\obslab{obs:c:behave}
\cite{Li08}
$C_{q,p}=\O{\exp(\gamma_e(q-1))}$, for the Euler-Mascheroni constant $\gamma_e\approx 0.57721$. 
\end{observation}
The value of $C_{q,p}$ is chosen so that each $z_i$ is an unbiased estimate of $\|x\|_p^p$, as a result of the following statement. 
We can then bound the expectation and the variance of each random variable.
\begin{lemma}[Expectation and variance of Li's geometric mean estimator, Lemma 2.2 in \cite{Li08}]
\lemlab{ligeoest:var}
Let $d$ be an integer multiple of an integer $q\ge 3$ and let $A\in\mathbb{R}^{d\times n}$ have independent $p$-stable random variables for the entries of $A$. 
For a vector $x\in\mathbb{R}^n$ and $y=Ax$, let $z_i:=C_{q,p}\cdot\left(\prod_{j=q(i-1)+1}^{qi}|y_j|^{p/q}\right)$ be the geometric mean of the inner products of $q$ random $p$-stable vectors with the vector $x$. 
Then $\Ex{z_i}=\|x\|_p^p$ and for $\xi:=\frac{C_{\frac{q}{2},p}}{C_{q,p}}$, we have $\Ex{z_i^2}\le\left(\xi^2-1\right)\cdot\|x\|_p^{2p}$.
\end{lemma}
Thus by taking the arithmetic mean of $\O{\frac{1}{\eps^2}}$ variables $z_i$, we obtain a $(1+\eps)$-approximation to $\|x\|_p^p$ with constant probability by Chebyshev's inequality. 

\subsection{\texorpdfstring{$F_p$}{Fp} Difference Estimator for \texorpdfstring{$0<p<2$}{0<p<2}}
We now describe our $F_p$ difference estimator and give the high-level details of the analysis. We use Li's geometric mean estimator to maintain $A(v+w_t)$ and $Av$, where $A$ is the sketching matrix for Li's geometric mean estimator, and $v$ and $w_t$ are frequency vectors. 
Observe that if we computed $A(v+w_t)-Av$, then we would obtain $A(w_t)$, which is a sketch that allows us to recover $F_p(w_t)$. 
However, we want to estimate $F_p(v+w_t)-F_p(v)$ rather than $F_p(w_t)$. 
Instead, we use the sketches $A(v+w_t)$ and $Av$ to compute terms $z_1,z_2,\ldots,z'_1,z'_2,\ldots$, where each $z_i$ is the geometric mean of $q$ consecutive entries in $A(v+w_t)$ and similarly $z'_i$ is the geometric mean of $q$ consecutive entries in $Av$. 
Since $z_i$ is an unbiased estimator of $F_p(v+w_t)$ and $z_i$ is an unbiased estimator of $F_p(v)$, it follows that $z_i-z'_i$ is an unbiased estimator of $F_p(v+w_t)-F_p(v)$. 
We take the average of the values $z_i-z'_i$ across $\O{\frac{1}{\eps^2}}$ indices of $i$ to obtain a single estimate and take the median of all estimates. 
The challenge is achieving both the variance bounds on $z_i-z'_i$ while also obtaining a strong tracking property. 
To bound the variance, we expand $z_i-z'_i$ to be a sum of $2^q-1$ geometric means of $q$ terms, each with at least one term $(\ip{A_j}{w_t})^{p/q}$. 
Since $A_j$ is a vector of $p$-stable entries, then $(\ip{A_j}{w_t})^{p/q}$ has the same distribution as $(\|w_t\|_p\cdot X)^{p/q}$ for a $p$-stable random variable $X$. 
We also have $F_p(w_t)=\gamma\cdot F(v)$, where $\gamma$ is bounded by some absolute constant. 
Thus we can bound the probability that $(\ip{A_j}{w_t})^{p/q}\ge\|v\|^{p/q}_p$. 

\begin{figure*}
\begin{mdframed}
\begin{enumerate}
\item
Let $A$ be a $d\times n$ random matrix whose entries are i.i.d from the $p$-stable distribution $\calD_p$, for $d=\O{\frac{\gamma^{2/p}}{\eps^2}\left(\log\frac{1}{\eps}+\log\log n\right)}$
\item
For a parameter $q=3$, let each $z_i=\prod_{j=q(i-1)+1}^{qi}(Av+Aw_t)_j^{p/q}$ and $z'_i=\prod_{j=q(i-1)+1}^{qi}(Av)_j^{p/q}$. 
\item
Output the arithmetic mean of $(z_1-z'_1),(z_2-z'_2),\ldots,(z_{d/q}-z'_{d/q})$. 
\end{enumerate}
\end{mdframed}
\caption{Difference estimator for $F_p(v+w_t)-F_p(w_t)$ with $0<p<2$}
\figlab{fig:diff:est:smallp}
\end{figure*}

\begin{lemma}[Expectation and variance of difference estimator terms]
\lemlab{lem:exp:var:smallp}
Let $p\in(0,2)$ and $\calS$ be an oblivious stream on which $v$ is the frequency vector induced by the updates of $\calS$ from time $t_1$ to $t_2$ and let $u$ is the frequency vector induced by the updates of $\calS$ from time $t_2$. 
Suppose $F_p(v+u)-F_p(v)\le\gamma F_p(v)$ and $F_p(u)\le\gamma F_p(v)$. 
Then for each $i\in[d/q]$, $z_i$, and $z'_i$ as defined in \figref{fig:diff:est:smallp}, we have $\Ex{z_i-z'_i}=F_p(v+u)-F_p(u)$ and 
\[\Var(z_i-z'_i)\le2^{2q}\gamma^{2/p}\|v\|_p^{2p}=2^{2q}\gamma^{2/p}(F_p(v))^2.\]
\end{lemma}
\begin{proof}
For an oblivious stream $\calS$, let $v$ be the frequency vector induced by the updates of $\calS$ from time $t_1$ to $t_2$ and let $u$ be the frequency vector induced by the updates of $\calS$ from time $t_2$ so that $F_p(v+u)-F_p(v)\le\gamma F_p(v)$ and $F_p(u)\le\gamma F_p(v)$. 
For any integer $d>0$, let $A\in\mathbb{R}^{d\times n}$ be a matrix whose entries are drawn i.i.d. from the $p$-stable distribution $\calD_p$. 

Recall that Li's geometric mean estimator first takes a geometric mean $z_i$ of $q$ separate inner products, where $i\in\left[\frac{d}{q}\right]$. 
Explicitly, each variable $z_i$ in the estimator takes the form 
\begin{align*}
z_i:&=C_{q,p}\cdot\prod_{j=q(i-1)+1}^{qi}\left(\ip{A_j}{v+u}\right)^{p/q}\\
&=C_{q,p}\cdot\left(\prod_{j=q(i-1)+1}^{qi}\left(1+\frac{\ip{A_j}{u}}{\ip{A_j}{v}}\right)^{p/q}\right)\cdot\left(\prod_{j=q(i-1)+1}^{qi}\left(\ip{A_j}{v}\right)^{p/q}\right),
\end{align*}
where $A_j$ is the $j$-th column of $A$. 
We define
\[T_i:=\prod_{j=q(i-1)+1}^{qi}\left(1+\frac{\ip{A_j}{u}}{\ip{A_j}{v}}\right)^{p/q}.\]
Similarly, Li's geometric mean estimator for $v$ is a $(1+\eps)$-approximation for $\|v\|_p^p=F_p(v)$ and each variable $z'_i$ in the estimator takes the form
\[z'_i:=C_{q,p}\cdot\prod_{j=q(i-1)+1}^{qi}\left(\ip{A_j}{v}\right)^{p/q}.\]
By \lemref{ligeoest:var}, $z_i$ and $z'_i$ are unbiased estimators for $F_p(v+u)$ and $F_p(v)$ respectively, so that $z_i-z'_i$ is an unbiased estimator for $F_p(v+u)-F_p(v)$,
\[\Ex{z_i-z'_i}=F_p(v+u)-F_p(v).\]
We can also expand $z_i-z'_i$ so that
\[z_i-z'_i=C_{q,p}\cdot(T_i-1)\cdot\prod_{j=q(i-1)+1}^{qi}\left(\ip{A_j}{v}\right)^{p/q}.\]
To bound the variance of $z_i-z'_i$, we first require the following structural inequality to handle the terms inside of the product of $T_i$:
\begin{claim}
\claimlab{clm:ab}
For $a,b\ge 0$ and $p\le 2$, we have $(a+b)^{p/q}\le a^{p/q}+2b^{p/q}$.
\end{claim}
\begin{proof}
For $0\le a\le b$ and $p\le 2\le q$, we have 
\[(a+b)^{p/q}\le(2b)^{p/q}\le 2b^{p/q}\le a^{p/q}+2b^{p/q}.\]
For $0\le b\le a$ and $p\le 2\le q$, we have from Bernoulli's inequality that
\[(a+b)^{p/q}=a^{p/q}\left(1+\frac{b}{a}\right)^{p/q}\le a^{p/q}\left(1+\frac{pb}{qa}\right)\le a^{p/q}+b^{p/q}.\]
\end{proof}
\noindent
By \claimref{clm:ab}, we have
\[T_i\le\prod_{j=q(i-1)+1}^{qi}\left(1+\left(\frac{\ip{A_j}{u}}{\ip{A_j}{v}}\right)^{p/q}\right).\]
Hence, $T_i-1$ is a sum of $2^q-1$ terms and similarly, $z_i-z'_i$ is a sum of $2^q-1$ product of $q$ terms, each with at least one term $(\ip{A_j}{u})^{p/q}$. 
For example, one term in $z_1-z'_1$ is 
\[C_{q,p}(\ip{A_1}{u})^{p/q}(\ip{A_2}{v})^{p/q}(\ip{A_3}{v})^{p/q}\ldots(\ip{A_q}{v})^{p/q},\]
while another term is
\[C_{q,p}(\ip{A_1}{v})^{p/q}(\ip{A_2}{u})^{p/q}(\ip{A_3}{v})^{p/q}\ldots(\ip{A_q}{v})^{p/q}.\]

By \lemref{ligeoest:var}, each of these terms has second moment at most $(\xi^2-1)\gamma^{2/p}\|v\|_p^{2p}$ for $\xi=\frac{C_{\frac{q}{2},p}}{C_{q,p}}$.  
Thus the sum of the $2^q-1$ terms has second moment at most $(2^q)^2(\xi^2-1)\gamma^{2/p}\|v\|_p^{2p}$, so that the variance of the sum is at most $\O{2^{2q}\gamma^{2/p}\|v\|_p^{2p}}$. 
That is,
\[\Var(z_i-z'_i)\le2^{2q}\gamma^{2/p}\|v\|_p^{2p}=2^{2q}\gamma^{2/p}(F_p(v))^2.\]
\end{proof}

\begin{corollary}[Pointwise $F_p$ difference estimator for $0<p<2$]
\corlab{cor:diff:est:Fp:smallp:once}
Let $p\in(0,2)$, $\delta,\eps\in(0,1)$, and $\calS$ be an oblivious stream on which $v$ is the frequency vector induced by the updates of $\calS$ from time $t_1$ to $t_2$ and let $u$ is the frequency vector induced by the updates of $\calS$ from time $t_2$. 
Suppose $F_p(v+u)-F_p(u)\le\gamma F_p(v)$ and $F_p(u)\le\gamma F_p(v)$.  
Then there exists an algorithm that uses a sketch of dimension $d=\O{\frac{\gamma^{2/p}}{\eps^2}\left(\log\frac{1}{\delta}\right)}$ and outputs an additive $\eps\cdot F_p(v)$ approximation to $F_p(v+u)-F_p(u)$ with probability at least $1-\delta$. 
%For $0<p<2$, there exists a $(\gamma,\eps,\delta)$-\emph{difference estimator} for $F_p$ that uses space
%\[\O{\frac{\gamma^{2/p}\log n}{\eps^2}\left(\log\frac{1}{\eps}+\log\frac{1}{\delta}\right)}.\]
\end{corollary}
\begin{proof}
For each $i\in[d/q]$, with  $z_i$, and $z'_i$ defined in \figref{fig:diff:est:smallp}, we have from \lemref{lem:exp:var:smallp} that $\Ex{z_i-z'_i}=F_p(v+u)-F_p(u)$ and 
\[\Var(z_i-z'_i)\le2^{2q}\gamma^{2/p}\|v\|_p^{2p}=2^{2q}\gamma^{2/p}(F_p(v))^2.\]
Thus by Chebyshev's inequality, the arithmetic mean of $\O{\frac{2^{2q}\gamma^{2/p}}{\eps^2}}$ differences $z_i-z'_i$ suffices to obtain additive $\eps F_p(v)$ error of the difference $F_p(v+u)-F_p(v)$ with probability at least $\frac{2}{3}$. 
The probability of success can then be boosted to $1-\delta$ by taking the median of $\O{\log\frac{1}{\delta}}$ such instances.  
\end{proof}

To obtain the strong tracking property and avoid extra $\log n$ factors, this argument needs to be interleaved with carefully chosen values of $w_t$ and bounding the supremum of $(\ip{A_j}{w_t})^{p/q}$ across all values of $w_t$, rather than considering each inner product separately and taking a union bound. 
We split the stream into $\O{\frac{1}{\eps^{16q/p^2}}}$ times $u_1,u_2,\ldots$ between which the difference estimator increases by roughly $\eps^{16q/p^2}\cdot F_p(v)$ and apply a union bound to argue correctness at these times $\{r_i\}$, incurring a $\log\frac{1}{\eps}$ term. 
To bound the difference estimator between times $u_i$ and $u_{i+1}$ for a fixed $i$, we note that the difference estimator only increases by roughly $\eps^{16q/p^2}\cdot F_p(v)$ from $u_i$ to $u_{i+1}$, so that we still obtain a $(1+\eps)$-approximation to $F_p(v+w_t)-F_p(v)$ with a $\eps^{1-16q/p^2}$-approximation to the difference, for a frequency vector $w_t$ induced by updates between $u_i$ and $u_{i+1}$. 
We then bound the probability that the supremum of the error between times $u_i$ and $u_{i+1}$ is bounded by $\eps^{1-16q/p}$ by applying chaining results from \cite{BravermanCIW16, BravermanCINWW17, BlasiokDN17} that bound the supremum of a random process. 

We first require the following structural property that bounds the supremum of the inner product of a random process with a vector of independent $p$-stable random variables. 
\begin{lemma}
\lemlab{lem:chain:lp}
\cite{BlasiokDN17}
Let $x^{(1)},x^{(2)},\ldots,x^{(m)}\in\mathbb{R}^n$ satisfy $0\preceq x^{(1)}\preceq\ldots\preceq x^{(m)}$. 
Let $Z\in\mathbb{R}^n$ be a vector of entries that are i.i.d. sampled from the $p$-stable random distribution $\calD_p$. 
Then for some constant $C_p$ depending only on $p$, we have
\[\PPr{\underset{k\le m}{\sup}|\ip{Z}{x^{(k)}}|\ge\lambda\|x^{(m)}\|_p}\le C_p\left(\frac{1}{\lambda^{2p/(2+p)}}+n^{-1/p}\right).\]
\end{lemma}

\begin{lemma}[$F_p$ difference estimator for $0<p<2$]
\lemlab{lem:sketching:Fp:smallp}
For $0<p<2$, it suffices to use a sketching matrix $A\in\mathbb{R}^{d\times n}$ of i.i.d. entries drawn from the $p$-stable distribution $\calD_p$, with dimension
\[d=\O{\frac{\gamma^{2/p}}{\eps^2}\left(\log\frac{1}{\eps}+\log\frac{1}{\delta}\right)}\] 
%\[d=\O{\frac{\gamma^{2/p}}{\eps^{2}}\left(\log^2\frac{1}{\eps}+\log\frac{1}{\delta}\right)}\]
to obtain a $(\gamma,\eps,\delta)$-difference estimator for $F_p$. 
\end{lemma}
\begin{proof}
For an oblivious stream $\calS$, let $v$ be the frequency vector induced by the updates of $\calS$ from time $t_1$ to $t_2$ and let $u_i$ be the frequency vector induced by the updates of $\calS$ from time $t_2+1$ to the last time when 
\[\max(F_p(u_i),F_p(v+u_1+\ldots+u_i)-F_p(v+u_1+\ldots+u_{i-1}))\le\frac{1}{2^{q^2}}\eps^{16q/p^2}\cdot F_p(v),\]
for some fixed constant $q\ge 3$. 
We first show that our difference estimator gives an $\eps\cdot F_p(v)$ approximation to $F_p(v+u_i)-F_p(v)$ for each $i=\O{\frac{1}{\eps^{16q/p^2}}}$. 
Let $\mathcal{E}$ be the event that we obtain a $(1+\O{\eps})$-approximation to both $F_p(v+u_i)$ and $F_p(v)$ for all $\O{\frac{1}{\eps^{16q/p^2}}}$ vectors $u_i$. 
By \corref{cor:diff:est:Fp:smallp:once}, we have $\PPr{\mathcal{E}}\ge 1-\frac{\poly(\eps,\delta)}{2}$ for an instance that uses $d=\O{\frac{\gamma^{2/p}}{\eps^2}\log\frac{1}{\eps\delta}}$ rows in the sketching matrix $A$ whose entries are i.i.d. sampled from the $p$-stable distribution $\calD_p$. 
We thus condition on the event $\mathcal{E}$ and assume we have additive $\O{\eps}\cdot F(v)$ approximations at all times $u_i$. 
Let $w_t$ be a frequency vector induced by updates from time $u_i$ and $u_{i+1}$, so that $F_p(w_t)\le\frac{1}{2^{q^2}}\eps^{16q/p^2}\cdot F_p(v)$. 
It then suffices to show that the output of the difference estimator of $F_p(v+w_t)-F_p(v)$ at most $\O{\eps}\cdot F_p(v)$ for all values of $w_t$ over the stream, since such a result proves that the difference estimator changes by at most $\O{\eps}\cdot F_p(v)$ between times $u_i$ and $u_{i+1}$ and conditioning on $\mathcal{E}$, we already have additive $\O{\eps}\cdot F(v)$ approximations at all times $u_i$. 
The claim would then follow by monotonicity of $F_p$. 

We use the same computation as in \corref{cor:diff:est:Fp:smallp:once} to reduce the differences of the geometric means. 
That is, our difference estimator first takes a geometric mean $z_i$ of $q$ separate inner products, where $i\in\left[\frac{d}{q}\right]$, formed from the entries of $A(v+w_t)$. 
Explicitly, each variable $z_i$ in the estimator takes the form 
\begin{align*}
z_i:&=C_{q,p}\cdot\prod_{j=q(i-1)+1}^{qi}\left(\ip{A_j}{v+w_t}\right)^{p/q}\\
&=C_{q,p}\cdot\left(\prod_{j=q(i-1)+1}^{qi}\left(1+\frac{\ip{A_j}{w_t}}{\ip{A_j}{v}}\right)^{p/q}\right)\cdot\left(\prod_{j=q(i-1)+1}^{qi}\left(\ip{A_j}{v}\right)^{p/q}\right),
\end{align*}
where $A_j$ is the $j$-th column of $A$. 
Similarly, each variable $z'_i$ in the estimator formed from the entries of $Av$ takes the form 
\[z'_i:=C_{q,p}\cdot\prod_{j=q(i-1)+1}^{qi}\left(\ip{A_j}{v}\right)^{p/q}.\]

As in \corref{cor:diff:est:Fp:smallp:once} we can write $z_i-z'_i$ as a sum of $2^q-1$ geometric means of $q$ terms, each with at least one term $(\ip{A_j}{w_t})^{p/q}$, e.g., one term in $z_1-z'_1$ is 
\[C_{q,p}(\ip{A_1}{w_t})^{p/q}(\ip{A_2}{v})^{p/q}(\ip{A_3}{v})^{p/q}\ldots(\ip{A_q}{v})^{p/q},\]
while another term is
\[C_{q,p}(\ip{A_1}{v})^{p/q}(\ip{A_2}{w_t})^{p/q}(\ip{A_3}{v})^{p/q}\ldots(\ip{A_q}{v})^{p/q}.\]
Thus by \lemref{lem:chain:lp} with $\lambda=\frac{1}{\eps^{8/p}}$, we have 
\[\PPr{\underset{t\le t_2}{\sup}|\ip{A_i}{w_t}|\ge\lambda\|w_t\|_p}\le C_p\left(\eps^{16/(2+p)}+n^{-1/p}\right)\]
and similarly 
\[\PPr{\underset{t\le t_2}{\sup}|\ip{A_i}{v}|\ge\lambda\|v\|_p}\le C_p\left(\eps^{16/(2+p)}+n^{-1/p}\right).\]
Thus with probability at least $1-\O{\eps^4}$, none of the $2^q-1$ terms exceeds  
\[\lambda^p\|w_t\|_p^{p/q}\|v\|^{(q-1)p/q}_p\le\frac{1}{\eps^8}\|w_t\|_p^{p/q}\|v\|^{(q-1)p/q}_p\le\frac{\eps^8}{2^q}\cdot F_p(v),\]
since $F_p(w_t)\le\frac{1}{2^{q^2}}\eps^{64q/p^2}\cdot F_p(v)$. 
Hence, the sum of the $2^q-1$ terms is at most $\eps^8\cdot F(v)$, with probability at least $1-\O{2^q\eps^4}$. 
In other words, some particular difference of geometric means is at most $\eps^8\cdot F(v)$, with probability at least $1-\O{2^q\eps^4}$. 
Taking a union bound over $\O{\frac{1}{\eps^2}}$ geometric means, then the difference estimator $F_p(v+w_t)-F_p(v)$ outputs a value that is at most $\eps^8\cdot F(v)$, with constant probability, for sufficiently small $\eps$. 
Finally, if we take the median of $\O{\log\frac{1}{\delta}}$ such estimators, then we can increase the probability of success to at least $1-\delta$. 
As previously noted, this shows that the difference estimator provides strong tracking over the entire stream. 

Because each entry takes $\O{\log n}$ bits of space, then for constant $q$ and a sketch with $d=\O{\frac{\gamma^{2/p}}{\eps^2}\log\frac{1}{\eps}+\log\frac{1}{\delta}}$ rows, the total space required is 
\[\O{\frac{\gamma^{2/p}\log n}{\eps^2}\left(\log\frac{1}{\eps}+\log\frac{1}{\delta}\right)},\]
since we require a rescaling of $\delta'=\frac{\delta}{\poly\left(\frac{1}{\eps}\right)}$ to union bound over the $\poly\left(\frac{1}{\eps}\right)$ times $\{u_i\}$. 
\end{proof}

\paragraph{Derandomization of $p$-stable random variables.} 
To handle the generation and storage of the entries of $A$ drawn i.i.d. from the $p$-stable distribution $\calD_p$, we first note the following derandomization. 
\begin{lemma}[Lemma 8 in \cite{JayaramW18}]
\lemlab{lem:derandom:once}
Let $\calA$ be any streaming algorithm that stores only a linear sketch $A\cdot f$ on a stream vector $f\in\{-M,\ldots,M\}^n$ for some $M=\poly(n)$, such that the entries of $A\in\mathbb{R}^{k\times n}$ are i.i.d., and can be sampled using $\O{\log n}$ bits. 
Then for any fixed constant $c\ge 1$, $\calA$ can be implemented using a random matrix $A'$ using $\O{k\log n(\log\log n)^2}$ bits of space, such that for all $y\in\mathbb{R}^k$ with entry-wise bit complexity of $\O{\log n}$,
\[|\PPr{A\cdot f=y}-\PPr{A'\cdot f=y}|<n^{-ck}.\]
\end{lemma}
We cannot immediately apply \lemref{lem:derandom:once} because our difference estimator actually stores $A\cdot(v+w_t)$ and $A\cdot v$ and requires the entries of $A$ to be i.i.d. We thus require the following generalization:
\begin{corollary}
\corlab{lem:derandom:many}
For a constant $q\ge 1$, let $f_1,\ldots,f_q\in\{-M,\ldots,M\}^n$ for some $M=\poly(n)$ be vectors defined by a stream so that $f_i$ is defined by the updates of the stream between given times $t_{i,1}$ and $t_{i,2}$, for each $i\in[q]$. 
Let $\calA$ be any streaming algorithm that stores linear sketches $A\cdot f_1,\ldots,A\cdot f_q$ , such that the entries of $A\in\mathbb{R}^{k\times n}$ are i.i.d. that can be sampled using $\O{\log n}$ bits, and outputs $g(A\cdot f_1,\ldots,A\cdot f_q)$ for some composition function $g:\mathbb{R}^q\to\mathbb{R}$. 
Then for any fixed constant $c\ge 1$, $\calA$ can be implemented using a random matrix $A'$ using $\O{k\log n(\log\log n)^2}$ bits of space, such that for all $y\in\mathbb{R}^k$ with entry-wise bit complexity of $\O{\log n}$,
\[|\PPr{g(A\cdot f_1,\ldots,A\cdot f_q)=y}-\PPr{g(A'\cdot f_1,\ldots,A'\cdot f_q) = y}|<n^{-ck}.\]
\end{corollary}
\begin{proof}
We first assume without loss of generality that all entries of $A\cdot f_i$ are integers bounded by $\poly(n)$ for each $i\in[q]$, due to the bit complexity of $A$ and $f_i$. 
Let $w_i=A\cdot f_i$ for each $i\in[q]$ and suppose the maximum entry, denoted $\|A\|_\infty$, satisfies $\|A\|_\infty\le n^\alpha$ for some constant $\alpha$. 
Let $N=M\cdot n^\alpha$ so that $\|A\cdot f_i\|_\infty\le N$ for all $i\in[q]$. 
Define the vector $v=\sum_{i=1}^q N^{2i} w_i$ so that $\|v\|_{\infty}<N'$ for $N'=N^{3q}$. 
Note that $N'=\poly(n)$ so that all entries of $v$ can be stored in $\O{\log n}$ bits. 
Moreover, observe that each coordinate $v_j$ with $j\in[n]$ has a unique $N^2$-ary representation. 
Thus if $v_j=\sum_{i=1}^q N^{2i} w_{i,j}$ where $w_{i,j}$ represents the $j$-th coordinate of vector $w_i$, then the system $v_j=\sum_{i=1}^q N^{2i}\alpha_i$ constrained to the condition that $|\alpha_i|\le N$ for all $i\in[q]$ has a unique solution. 
Hence, the vectors $w_i=A\cdot f_i$ can be computed from the vector $v$. 

By \lemref{lem:derandom:once}, we can use a random matrix $A'$ with $\O{k\log n(\log\log n)^2}$ bits of space and have for all $y \in \mathbb{R}^k$ with entry-wise bit complexity of $O(\log n)$, that: 
\[|\PPr{A\cdot v=y}-\PPr{A'\cdot v=y}|<n^{-ck}.\]
Since any $\mathcal{A}$ that computes $A\cdot v$ can compute the vectors $A\cdot f_1,\ldots,A\cdot f_q$ by the above argument, $\mathcal{A}$ can then compute the composition $g(A\cdot f_1,\ldots, A\cdot f_q)$. 
Hence it follows that
\[|\PPr{g(A\cdot f_1,\ldots,A\cdot f_q)=y}-\PPr{g(A'\cdot f_1,\ldots,A'\cdot f_q)}|<n^{-ck}.\]
\end{proof}
We now argue that the pseudorandom generator of \corref{lem:derandom:many} suffices to derandomize the correctness guarantees of our difference estimator. 
\begin{lemma}[$F_p$ difference estimator for $0<p<2$]
\lemlab{lem:diff:est:Fp:smallp}
For $0<p<2$, there exists a $(\gamma,\eps,\delta)$-\emph{difference estimator} for $F_p$ that uses space
\[\O{\frac{\gamma^{2/p}\log n}{\eps^{2}}(\log\log n)^2\left(\log\frac{1}{\eps}+\log\frac{1}{\delta}\right)}.\]
\end{lemma}
\begin{proof}
The difference estimator for $F_p(v+u)-F_p(v)$ is given an input splitting time $t_1$. 
We then need to argue correctness over all possible stopping times $t$ with $t>t_1$, such that $F_p(v+u)-F_p(v)\le\gamma F_p(v)$ for $p<1$ or $F_p(u)\le\gamma F_p(v)$ for $p\ge 1$. 
By \corref{lem:derandom:many} applied with $q = 2$, there exists a pseudorandom generator that succeeds with high probability $1-\frac{1}{\poly(n)}$ over a specific value of $t$. 
Since the stream has length $m=\poly(n)$, by taking a union bound over all $\poly(n)$ choice of the stopping times $t$, we have that the pseudorandom generator is also correct with high probability over all possible stopping times $t$. 
Thus from \lemref{lem:sketching:Fp:smallp}, we obtain the following guarantees of our difference estimator. 

In particular, we remark that the guarantees of the difference estimator only require that the marginal distribution is correct at all times, rather than the joint distribution is correct. 
Recall that the difference estimator is used to choose a final stopping time for the purposes of the framework in \thmref{thm:framework} when its output is sufficiently large. 
Specifically, the framework takes a set of outputs $s_{t_1+1},\ldots, s_t$ from the difference estimator over the course of the stream and chooses a final stopping time $t_2$ based on the first output that exceeds a certain threshold $T$. 
We emphasize that our argument does not rely on fooling the final stopping time $t_2$ chosen by the framework, since it requires a conditional statement on the set $s_{t_1+1},\ldots, s_t$ of outputs from the difference estimator not exceeding the threshold $T$. 

For example, consider a scenario where the probability over the choice of independent $p$-stable random variables that $s_{t_1+1}$ exceeds $T$ is $\frac{1}{2}$ and the probability that $s_{t_1+2}$ exceeds $T$ is $\frac{1}{2}$, i.e., $\PPr{s_{t_1+1}\ge T}=\frac{1}{2}$ and $\PPr{s_{t_1+2}\ge T} = \frac{1}{2}$. 
However, suppose that conditioned on the event that $s_{t_1+1}<T$, we have that $s_{t_1+2}$ cannot exceed $T$, so that $\PPr{s_{t_1+2}\ge T|s_{t_1+1}<T}=0$. 
Then since \thmref{thm:framework} chooses the first time that exceeds $T$, the framework can never choose $t_1+2$, i.e., $\PPr{t_2=t_1+2}=0$. 
Since our derandomization does not use independent $p$-stable random variables and only fools the marginal probabilities (of a pair of times $(t, t_1)$ for each $t > t_1$), we need not fool the joint distribution, and thus not fool the choice of stopping time. 
Hence, it could be that \thmref{thm:framework} using our derandomized difference estimator selects $t_2=t_1+2$. 
On the other hand, $t_1+2$ is still a valid stopping time if the difference estimator is correct at time $t_1+1$, and the output does not exceed $T$. 

In other words, while the geometric mean estimator can be fooled by fooling a constant a number of half-space queries, we may need to condition on $\poly(n)$ intermediate stream positions. 
Instead, we note that the derandomization of our difference estimator provides the same correctness guarantees. 
Thus, even though the difference estimator with the pseudorandom generator could induce a different distribution on the final stopping time, the final stopping time $t'_2$ output by our derandomized algorithm and the algorithm itself can be used in our framework because $t'_2$ corresponds to the first output of the algorithm that exceeds the threshold $T$. We also have correctness at all intermediate times between $t_1$ and $t'_2$.  
\end{proof}

\paragraph{Bit complexity and rounding of $p$-stable random variables.}
We remark that for each inner product in \lemref{lem:diff:est:Fp:smallp} to be stored in $\O{\log n}$ bits of space, then each randomly generated $p$-stable random variable must also be rounded to $\O{\log n}$ bits of precision. 
Due to the rounding, that each summand changes additively by $\frac{1}{\poly(n)}$, so that the total estimate also changes by an additive $\frac{1}{\poly(n)}$, so that the total error for the estimate of $F_p(v+u)-F_p(v)$ is $\eps\cdot F_p(v)+\frac{1}{\poly(n)}$. 
The additive $\frac{1}{\poly(n)}$ can be absorbed into the $\eps\cdot F(v)$ term with a rescaling of $\eps$ unless $F_p(v)=0$. 
However, in the case $F_p(v)=0$, then our estimator will output $0$ anyway, so that the rounding of the $p$-stable random variables still gives an additive $\eps\cdot F(v)$. 

\subsection{\texorpdfstring{$F_p$}{Fp} Estimation Algorithm}
We now give an adversarially robust streaming algorithm for $F_p$ moment estimation with $p\in(0,2)$ by using \thmref{thm:framework}. 
\begin{theorem}
Given $\eps>0$ and $p\in(0,2)$, there exists an adversarially robust streaming algorithm that outputs a $(1+\eps)$-approximation for $F_p$ that uses $\O{\frac{1}{\eps^{2}}\log^2 n(\log\log n)^2\left(\log\frac{1}{\eps}+\log\log n\right)}$ bits of space and succeeds with probability at least $\frac{2}{3}$. 
\end{theorem}
\begin{proof}
For $p\in(0,2)$, $F_p$ is a monotonic function with $(\eps,m)$-twist number $\lambda=\O{\frac{1}{\eps}\log n}$, by \obsref{obs:flip:fp}. 
By \lemref{lem:diff:est:Fp:smallp} and \thmref{thm:strong:Fp:smallp}, there exists a $(\gamma,\eps,\delta)$-difference estimator that uses space
\[\O{\frac{\gamma^{2/p}\log n}{\eps^2}(\log\log n)^2\left(\log\frac{1}{\eps}+\log\frac{1}{\delta}\right)}\]
and an oblivious strong tracker $\psketch$ for $F_p$ that uses space
\[\O{\frac{\log n}{\eps^2}\left(\log\log n+\log\frac{1}{\eps}+\log\frac{1}{\delta}\right)}.\]
Therefore by using the $(\gamma,\eps,\delta)$-difference estimator and the oblivious strong tracker for $F_p$ in the framework of \algref{alg:framework}, then consider applying \thmref{thm:framework} with parameters $C=2/p\ge 1$, $S_1(n,\delta,\eps)=\frac{1}{\eps}\log n(\log\log n)^2\left(\log\log n+\log\frac{1}{\eps}+\log\frac{1}{\delta}\right)$, and $S_2=0$. 
Namely, \thmref{thm:framework} gives an adversarially robust streaming algorithm that outputs a $(1+\eps)$ for $F_p$ that succeeds with constant probability and uses space 
\[\O{\frac{1}{\eps^{2}}\log^2 n(\log\log n)^2\left(\log\log n+\log\frac{1}{\eps}\right)}.\]    
\end{proof}

\paragraph{Optimized $F_p$ Algorithm.}
To improve the space dependency, we use the same technique as in \secref{sec:F2}. 
That is, recall that the counter $a$ in \algref{alg:framework} tracks the active instances $\calA_a$ and $\calB_{a,c}$ output by the algorithm. 
Then we only maintain the most sketches $\calA_i$ and $\calB_{i,c}$ for the smallest $\O{\log\frac{1}{\eps}}$ values of $i$ that are at least the value of the counter $a$, instead of maintaining $\O{\log n}$ total sketches.  
Because the output increases by a factor of $2$ each time $a$ increases, any larger index will have only missed $\O{\eps}$ fraction of the $F_p$ of the stream and thus still output a $(1+\eps)$-approximation. 

\begin{theorem}[Adversarially robust $F_p$ streaming algorithm for $p\in(0,2)$]
\thmlab{thm:robust:opt:Fp:smallp}
Given $\eps>0$ and $p\in(0,2)$, there exists an adversarially robust streaming algorithm that outputs a $(1+\eps)$ for $F_p$ that uses $\O{\frac{1}{\eps^{2}}\log n(\log\log n)^2\log\frac{1}{\eps}\left(\log\log n+\log\frac{1}{\eps}\right)}$ bits of space and succeeds with probability at least $\frac{2}{3}$. 
\end{theorem}
\begin{proof}
Using the above optimization, there are simultaneously $\O{\log\frac{1}{\eps}}$ active indices $a$ and $\O{1}$ active indices $c$ corresponding to sketches $\calA_{a}$ and $\calB_{a,c}$. 
Recall from \thmref{thm:framework} that for fixed $a$ and $C>1$, the total space for each $\calB_{a,j}$ across the $\beta$ granularities is 
\[\O{\frac{1}{\eps^2}\cdot S_1(n,\delta',\eps)+\frac{1}{\eps}\log\frac{1}{\eps}\cdot S_2(n,\delta',\eps)},\]
where for our purposes, $S_1(n,\delta',\eps)=\log n(\log\log n)^2\left(\log\frac{1}{\eps}+\log\frac{1}{\delta'}+\log\log n\right)$ and $S_2=0$ for our $F_p$ strong tracker and $F_p$ difference estimator. 
Since there are $\O{\log n}$ total indices $a$ over the course of the stream, then each sketch must have failure probability $\frac{\delta}{\poly\left(\log n,\frac{1}{\eps}\right)}$ for the entire algorithm to have failure probability $\delta=\frac{2}{3}$. 
However, since we maintain at most $\O{\log\frac{1}{\eps}}$ simultaneous active values of $a$, then the $\O{\frac{1}{\eps^{2}}\log n(\log\log n)^2\log\frac{1}{\eps}\left(\log\log n+\log\frac{1}{\eps}\right)}$ bits of space are required to perform the sketching stitching and granularity changing framework.  
Since there are $\O{\log\frac{1}{\eps}}$ active indices of $a$, consisting of $\O{\frac{1}{\eps}}$ subroutines, then it takes $\O{\frac{1}{\eps}\log n\log\frac{1}{\eps}}$ additional bits of space to store the splitting times for each of the $\O{\frac{1}{\eps}\log\frac{1}{\eps}}$ active subroutines across a stream of length $m$, with $\log m=\O{\log n}$. 
Hence, the total space required is $\O{\frac{1}{\eps^{2}}\log n(\log\log n)^2\log\frac{1}{\eps}\left(\log\log n+\log\frac{1}{\eps}\right)}$. 
\end{proof}
Similarly, we have:
\begin{theorem}[Adversarially robust $F_p$ streaming algorithm for $p\in(0,1)$]
\thmlab{thm:robust:opt:Fp:tinyp}
Given $\eps>0$ and $p\in(0,1]$, there exists an adversarially robust streaming algorithm that outputs a $(1+\eps)$ for $F_p$ that uses $\O{\frac{1}{\eps^{2}}\log\frac{1}{\eps}\left(\log\log n+\log\frac{1}{\eps}\right)+\frac{1}{\eps}\log\frac{1}{\eps}\log n}$ bits of space and succeeds with probability at least $\frac{2}{3}$. 
\end{theorem}

Finally, we remark that our analysis for \lemref{lem:diff:est:Fp:smallp} can be repeated to show that Li's geometric mean estimator provides strong $L_p$ tracking for $p\in(0,2)$. 
Namely by first union bounding over all $\O{\log n}$ times in the stream where the value of $F_p$ on the underlying frequency vector induced by the stream roughly doubles, we again reduce the problem to maintaining an accurate estimation of $F_p$ between two times in which the value of $F_p$ roughly doubles. 
We can subsequently apply the techniques of \lemref{lem:sketching:Fp:smallp} by further breaking down the stream into times when the value of $F_p$ increases by a factor of $(1+\poly(\eps))$ and bounding the same terms $z_i$, as defined in \figref{fig:diff:est:smallp}. 
This analysis incurs an extra factor of $\log\log n$ due to an initial union bound over $\poly\left(\log n, \frac{1}{\eps}\right)$ times in the stream rather than $\poly\left(\frac{1}{\eps}\right)$. 
Hence, the space required to implement Li's geometric mean estimator as a strong $L_p$ tracker in the random oracle model is $\O{\frac{\log n}{\eps^2}\left(\log\log n+\log\frac{1}{\eps}+\log\frac{1}{\delta}\right)}$, which matches the strong $L_p$ tracker of \cite{BlasiokDN17}.  
To derandomize the algorithm, we incur an additional $(\log\log n)^2$ factor due to \corref{lem:derandom:many}. 

\begin{theorem}
\thmlab{thm:smallp:strong:track}
For $p\in(0,2)$, there exists a one-pass streaming algorithm that uses total space $\O{\frac{1}{\eps^2}\log n(\log\log n)^2\left(\log\log n+\log\frac{1}{\eps}+\log\frac{1}{\delta}\right)}$ bits and provides $(\eps,\delta)$-strong tracking for the $F_p$ moment estimation problem. 
\end{theorem}
\begin{proof}
Consider an instance of Li's geometric mean estimator with $d=\O{\frac{1}{\eps^2}\left(\log\frac{1}{\delta'}\right)}$ rows. 
We define a sequence of times $t_1,t_2,\ldots,t_w$ throughout an oblivious stream $\calS$ so that each time $t_i$ is implicitly defined as the last time that $F_p(1:t_i)\le2^{i-1}$ and thus $w=\O{\log n}$. 
For each $i\in[w]$, we further consider a sequence of times $t_{i,1},t_{i,2},\ldots$ so that time $t_{i,j}$ is defined as the last time $t$ such that that 
\[\max(F_p(t_{i,j-1}+1:t),F_p(1:t)-F_p(1:t_{i,j-1}))\le\frac{1}{2^{q^2}}\eps^{16q/p^2}\cdot F_p(1:t_i).\]
Let $\mathcal{E}_1$ be the event that Li's geometric mean estimator gives a $\left(1+\frac{\eps}{2}\right)$-approximation to $F_p(1:t_i)$ for all $i\in[w]$ and let $\mathcal{E}_2$ be the event that Li's geometric mean estimator gives a $\left(1+\frac{\eps}{2}\right)$-approximation to $F_p(1:t_{i,j})$ for all $t_{i,j}$. 
Since there are $\poly\left(\frac{1}{\eps},\log n\right)$ times $\{t_i\}\cup\{t_{i,j}\}$, then we require $\delta'=\delta/\poly\left(\frac{1}{\eps},\log n\right)$ to obtain success probability $1-\delta$ after taking a union bound over $\poly\left(\frac{1}{\eps},\log n\right)$ times. 
By \lemref{ligeoest:var}, the arithmetic mean of $\O{\frac{1}{\eps^2}}$ variables $z_i$ achieves a $(1+\eps)$-approximation to $F_p(1:t)$ with constant probability for a fixed time $t$.  
Thus by taking the median of $\O{\log\log n+\log\frac{1}{\eps}+\log\frac{1}{\delta}}$ such means and taking a union bound, we have that $\PPr{\mathcal{E}_1\wedge\mathcal{E}_2}\ge 1-\frac{\delta}{2}$. 

For times $a,b\in[m]$, let $\calL(a:b)$ be the output of Li's geometric mean estimator for the frequency vector induced by the updates of $\calS$ from times $a$ to $b$.  
From the chaining argument in \lemref{lem:sketching:Fp:smallp} and a sketching matrix with dimension $d=\O{\frac{1}{\eps^2}\left(\log\frac{1}{\delta'}\right)}$ rows, we have that for a fixed $i\in[w]$, there does not exist $t\in[t_{i,j}+1,t_{i,j+1}]$ such that $|\calL(1:t)-\calL(1:t_{i,j}+1)|\ge\frac{\eps}{2}F_p(1:t_{i,j})$, with probability at least $1-\delta'$. 
Hence conditioned on $\mathcal{E}_1$ and $\mathcal{E}_2$, $\calL(1:t)$ outputs a $(1+\eps)$ approximation to $F_p(1:t)$ for all times $t\in[t_{i,j}+1,t_{i,j+1}]$, with probability at least $1-\delta'$. 
We again take a union bound over $\poly\left(\frac{1}{\eps},\log n\right)$ times $\{t_i\}\cup\{t_{i,j}\}$, so that setting $\delta'=\delta/\poly\left(\frac{1}{\eps},\log n\right)$ implies that the $\calL(1:t)$ simultaneously gives a $(1+\eps)$-approximation to all $F_p(1:t)$ with probability $1-\delta$. 
Hence, the total number of rows required in $D$ is $d=\O{\frac{1}{\eps^2}\left(\log\frac{1}{\delta'}\right)}$, which translates to $\O{\frac{1}{\eps^2}\left(\log\log n+\log\frac{1}{\eps}+\log\frac{1}{\delta}\right)}$ total rows. 
Each entry in the sketch uses $\O{\log n}$ bits of space to store. 
Moreover, we incur an additional $(\log\log n)^2$ factor due to \corref{lem:derandom:many} to derandomize the algorithm. 
Therefore, the total space in bits is
\[\O{\frac{1}{\eps^2}\log n(\log\log n)^2\left(\log\log n+\log\frac{1}{\eps}+\log\frac{1}{\delta}\right)}.\]
\end{proof}

\paragraph{Applications to Entropy Estimation.}
Finally, we give an application to adversarially robust entropy estimation, similar to \thmref{thm:sw:entropy}. 
\begin{theorem}[Adversarially robust entropy streaming algorithm]
\thmlab{thm:robust:entropy}
Given $\eps>0$, there exists an adversarially robust streaming algorithm that outputs an additive $\eps$-approximation to Shannon entropy and uses $\tO{\frac{1}{\eps^2}\log^3 n}$ bits of space and succeeds with probability at least $\frac{2}{3}$. 
\end{theorem}
\begin{proof}
By \obsref{obs:entropy:addmult} and \lemref{lem:entropy:reduction}, it suffices to obtain adversarially robust $(1+\eps')$-approximation algorithms to $F_{y_i}(v)$ for all $y_i\in(0,2)$ in the set $\{y_0,\ldots,y_k\}$, where $k=\log\frac{1}{\eps}+\log\log m$ and $\eps'=\frac{\eps}{12(k+1)^3\log m}$. 
By \thmref{thm:robust:opt:Fp:smallp} with accuracy parameter $\eps'$, we can obtain such robust algorithms approximating each $F_{y_i}(v)$, using space 
\[\tO{\frac{1}{(\eps')^{2}}\log n(\log\log n)^2\log\frac{1}{\eps}\left(\log\log n+\log\frac{1}{\eps}\right)}.\] 
Moreover, the failure probability of each algorithm is $1-\poly\left(\eps,\frac{1}{\log n}\right)$, due to the rescaling of the failure probability $\delta'=\frac{\delta}{\poly\left(\frac{1}{\eps},\log n\right)}$ in both \lemref{lem:diff:est:Fp:smallp} and \thmref{thm:strong:Fp:smallp}. 
Hence for $\eps'=\frac{\eps}{12(k+1)^3\log m}$ and $\log m=\O{\log n}$, each of the $\O{k}$ algorithms use space $\tO{\frac{1}{\eps^2}\log^3 n}$. 
Since $k=\log\frac{1}{\eps}+\log\log m$, the overall space complexity follows. 
\end{proof}
\section{Robust \texorpdfstring{$F_p$}{Fp} Estimation for Integer \texorpdfstring{$p>2$}{p>2}}
In this section, we use the framework of \secref{sec:framework} to give an adversarially robust streaming algorithm for $F_p$ moment estimation where $p>2$ is an integer. 
We again require both an $F_p$ strong tracker and an $F_p$ difference estimator to use \thmref{thm:framework}. 
Recall that for integer $p>2$, the dominant space factor is $n^{1-2/p}$ and thus we will not try to optimize the $\log n$ factors. 
Hence we obtain an $F_p$ strong tracker by adapting an $F_p$ streaming algorithm and a union bound over $m$ points in the stream. 
The main challenge of the section is to develop the $F_p$ difference estimator, since we have the following $F_p$ strong tracker:
\begin{theorem}[\cite{Ganguly11}, Theorem 22 in \cite{GangulyW18}]
\thmlab{thm:Fp:largep}
For integer $p>2$, there exists an insertion-only streaming algorithm $\ghss(1,t,\eps,\delta)$ that uses $\O{\frac{1}{\eps^2}n^{1-2/p}\log^2 n\log\frac{1}{\delta}}$ bits of space that gives a $(1+\eps)$-approximation to the $F_p$ moment.  
\end{theorem}
By setting $\delta'=\frac{\delta}{\poly(n)}$, we can apply \thmref{thm:Fp:largep} across all $m=\poly(n)$ times on a stream, thus obtaining a strong tracker. 
\begin{theorem}[Oblivious $F_p$ strong tracking for integer $p>2$]
\thmlab{thm:strong:Fp:largep}
\cite{Ganguly11, GangulyW18}
For integer $p>2$, there exists an insertion-only streaming algorithm $\ghss(1,t,\eps,\delta)$ that uses $\O{\frac{1}{\eps^2}n^{1-2/p}\log\frac{n}{\delta}\log^2 n}$ bits of space and provides $(\eps,\delta)$-strong $F_p$ tracking. 
\end{theorem}

\noindent
To acquire intuition for our $F_p$ estimation algorithm, we first recall the following definition of perfect $L_p$ sampling. 
\begin{definition}[$L_p$ sampling]
Let $f\in\mathbb{R}^n$ and failure probability $\delta\in(0,1)$. 
A \emph{perfect $L_p$ sampler} is an algorithm that either outputs a failure symbol $\bot$ with probability at most $\delta$ or an index $i\in[n]$ such that for each $j\in[n]$,
\[\PPr{i=j}=\frac{u_j^p}{F_p(u)}+\O{n^{-c}},\]
for some arbitrarily large constant $c\ge 1$. 
\end{definition}

We use the expansion $F_p(v+u)-F_p(v)=\sum_{k=1}^{p}\binom{p}{k}\langle v^k,u^{p-k}\rangle$ for our difference estimator for $F_p$ moment estimation for integer $p>2$, where $u^k$ denotes the coordinate-wise $k$-th power of $u$. 
Suppose we use a perfect $L_p$ sampler to sample a coordinate $a\in[n]$ with probability $\frac{v_a^k}{\|v\|_k^k}$. 
We then set $Z$ to be the $a$-th coordinate of the frequency vector $v^{p-k}$. 
Observe that $u$ arrives completely after the splitting time denotes the end of the updates to frequency vector $v$. 
Thus we can sample $a$ at the splitting time and then explicitly compute $Z$. 
We can then obtain an unbiased estimate $Y$ to $\|v\|_k^k$ with low variance, so that the expected value of $YZ$ would be exactly $\langle v^k,u^{p-k}\rangle$ and moreover, $YZ$ is a ``good'' approximation to $\langle v^k,u^{p-k}\rangle$. 

The downfall of this approach is that it requires a perfect $L_p$-sampler for $p>2$, which is not known. 
Perfect $L_p$-samplers are known for $p\le 2$~\cite{JayaramW18}, but their constructions are based on duplicating each stream update $\poly(n)$ times. 
Hence adapting these constructions to build perfect $L_p$-samplers for $p>2$ would be space-inefficient, since the space dependence is $\Omega(n^{1-2/p})$ for $p>2$, rather than $\polylog(n)$ for $p\le 2$. 
Thus after duplication the space required would be $\Omega((\poly(n))^{1-2/p})$ rather than $\polylog(\poly(n))$. 
Note that the former requires space larger than $n$ while the latter remains $\polylog(n)$. 
One possible approach would be to use approximate $L_p$ samplers and their variants~\cite{MonemizadehW10, JowhariST11, AndoniKO11, MahabadiRWZ20}, but these algorithms already have $\frac{1}{\eps^2}$ space dependency for each instance, which prohibits using $\Omega\left(\frac{1}{\eps}\right)$ instances to reduce the variance of each sampler. 
Instead, we use the following perfect $L_2$-sampler of \cite{JayaramW18} to return a coordinate $a\in[n]$ with probability $\frac{v_a^2}{\|v\|_2^2} \pm \frac{1}{\poly(n)}$. 

\begin{theorem}[Perfect $L_2$ sampler]
\thmlab{thm:perfect:sampler}
\cite{JayaramW18}
Given failure probability $\delta\in(0,1)$, there exists a one-pass streaming algorithm $\sampler$ that is a perfect $L_2$ sampler and uses $\O{\log^3 n\log\frac{1}{\delta}}$ bits of space.
\end{theorem}

We also obtain unbiased estimates $X$ and $Y$ of $v_a^{k-2}$ and $\|v\|_2^2$, respectively. 
Given $a\in[n]$, we then track the $a$-th coordinate of $u^{p-k}$ exactly. 
We show that the product of these terms $X$, $Y$, and $u^{p-k}$ forms an unbiased estimate to $\langle v^k,u^{p-k}\rangle$. 
We then analyze the variance and show that taking the mean of enough repetitions gives a $(1+\eps)$-approximation to $\langle v^k,u^{p-k}\rangle$. 
By repeating the estimator for each summand in $\sum_{k=1}^{p}\binom{p}{k}\langle v^k,u^{p-k}\rangle$, it follows that we obtain a $(\gamma,\eps,\delta)$-difference estimator for $F_p$. 

We first require the well-known $\countsketch$ algorithm for identifying heavy-hitters, which consists of a table with $\log\frac{n}{\delta}$ rows, each consisting of $\O{\frac{1}{\eps^2}}$ buckets, to identify $\eps\cdot L_2$ heavy hitters. 
For each row, each item $i\in[n]$ in the universe is hashed to one of the $\O{\frac{1}{\eps^2}}$ buckets along with a random sign. 
The signed sum of all items assigned to each bucket across all rows is tracked by the data structure and the estimated frequency of each item $i$ is the median of the values associated with each bucket that $i$ is hashed to, across all rows. 
%To identify $L_p$ heavy hitters with $p>2$, $\countsketch$ instead uses $\O{n^{1-2/p}}$ buckets in each row. 

\begin{figure*}
\begin{mdframed}
\begin{enumerate}
\item
Find a list $\calH$ that includes all $i\in[n]$ with $v_i\ge\frac{\gamma^{1/p}}{16}\|v\|_p$.
\item
Using $\countsketch$, obtain an estimate $\widehat{v_i}$ to $v_i$ with additive error $\frac{\eps\gamma^{1/p}}{64\gamma}\|v\|_p$ for each $i\in\calH$ and let $h\in\mathbb{R}^n$ be the vector such that $h_i=\widehat{v_i}$ if $i\in\calH$ and zero otherwise. 
\item
Perform perfect $L_2$ sampling on $v-h$ to obtain a set $\calS$ of size $k=\O{\frac{\gamma}{\eps^2}n^{1-2/p}}$. 
\item
Obtain an estimate $\widehat{s_i}$ to $v_i-h_i$ for each $i\in\calS$. 
\item
Let $W$ be a $(1+\eps)$-approximation to $\|v-h\|_2^2$. 
 %and $U$ be a $\left(1+\frac{\eps}{\gamma}\right)$-approximation to $\|u\|_p^p$. 
\item
Output $W+\sum_{k=1}^{p-1}\binom{p}{k}\left(\sum_{i\in\calH}\widehat{v_i}^k,u_i^{p-k}+W\cdot\sum_{i\in\calS}\widehat{s_i}^{k-2},u_i^{p-k}\right)$. 
\end{enumerate}
\end{mdframed}
\caption{$F_p$ difference estimator for $F_p(v+u)-F_p(v)$ with integer $p>2$.}
\figlab{fig:diff:est:largep}
\end{figure*}

\begin{lemma}[Moment estimation of sampled items]
\lemlab{lem:coor:exp:var}
For integer $p>2$ and failure probability $\delta\in(0,1)$, there exists a one-pass streaming algorithm that outputs an index $i\in[n]$ with probability $\frac{u_i^2}{F_2(u)}+\frac{1}{\poly(n)}$, as well as an unbiased estimate to $u_i^p$ with variance $u_i^{2p}$. 
The algorithm uses $\O{\log^3 n\log\frac{1}{\delta}}$ bits of space.  
\end{lemma}
\begin{proof}
Recall that $\sampler$ first duplicates each coordinate $u_i$ for $i\in[n]$ a total number of $n^c$ times and for each $(i,j)\in[n]\times[n^c]$, $\sampler$ scales $u_i$ to obtain random variables $z_{i,j}=\frac{u_i}{\sqrt{e_{i,j}}}$ for an exponential random variable $e_{i,j}$ (the exponential random variables in \cite{JayaramW18} are then derandomized).  
This linear transformation to the frequency vector $u\in\mathbb{R}^n$ results in a vector $z\in\mathbb{R}^{n^{c+1}}$, which can also be interpreted as having dimensions $n\times n^c$. 
Moreover, $\sampler$ only outputs an index $i\in[n]$ if there exists an index $j\in[n^c]$ if $|z_{i,j}|\ge\O{1}|z|_2$. 
We then run an instance of $\countsketch$ with constant factor approximation to obtain an estimate $\widehat{z_{i,j}}$ for the frequency of $z_{i,j}$ and set $\widehat{u_i}=\sqrt{e_{i,j}}\cdot\widehat{z_{i,j}}$. 
Recall that each estimate $\widehat{z_{i,j}}$ of $z_{i,j}$ is
\[\sum s_{i,j}\bfone_{h(a,b)=h(i,j)}s_{a,b}z_{a,b},\]
where $s_{a,b}$ is a uniformly random sign and $\bfone_{h(a,b)=h(i,j)}$ is the indicator variable of whether $h(a,b)=h(i,j)$. 
That is, we define $\bfone_{h(a,b)=h(i,j)}=1$ if $h(a,b)=h(i,j)$ and $\bfone_{h(a,b)=h(i,j)}=0$ if $h(a,b)\neq h(i,j)$. 
Hence, $\widehat{z_{i,j}}$ is an unbiased estimate of $z_{i,j}$ so that $\widehat{u_i}$ is an unbiased estimate of $u_i$. 
Moreover, the variance of $\widehat{z_{i,j}}$ is at most $\O{\|z\|_2^2}=\O{z_{i,j}^2}$ so that the variance of $\widehat{u_i}$ is $\O{u_i^2}$. 
Thus if we use $p$ independent instances of $\countsketch$ with estimates $\widehat{u_i}^{(1)},\ldots,\widehat{u_i}^{(p)}$, their product is an unbiased estimate to $u_i^p$ with variance $u_i^{2p}$, as desired. 

Each of the $p$ instances of $\countsketch$ with constant factor approximation uses $\O{\log n\log\frac{n}{\delta}}$ bits of space. 
Since the perfect $L_2$ $\sampler$ uses $\O{\log^3 n\log\frac{1}{\delta}}$ bits of space by \thmref{thm:perfect:sampler}, then the total space used is $\O{\log^3 n\log\frac{1}{\delta}}$.
%Then we have
%\begin{align*}
%\Ex{(\widehat{z_{i,j}})^p}&\ge z_{i,j}^p,\\
%\Ex{(\widehat{z_{i,j}})^p}&\le\Ex{p\left(\bfone_{h(a,b)=h(i,j)}z_{a,b}^p\right)},
%\end{align*}
%Since $\countsketch$ is an unbiased estimator, then $\widehat{u_i}$ is an unbiased estimate to $\widehat{u_i}$. 
%
\end{proof}

Unfortunately, perfect $L_2$ sampling coordinates of $v$ alone is not enough; the variance of the resulting procedure is too high to obtain space dependency $\frac{\gamma}{\eps^2}$. 
Thus we also run a subroutine that removes a set of ``heavy'' coordinates $\calH$ of $v$ and tracks the corresponding coordinates of $u$. 
Although we have the exact values of $u_a$ for $a\in\calH$, we still do not have exact values of $v_a$; instead, we have estimates $\widehat{v_a}$ for each $v_a$ with $a\in\calH$. 
Setting $h$ to be the sparse vector that contains the estimates $\widehat{v_a}$ for each $a\in\calH$ and $w:=v-h$, our algorithm perfect $L_2$ samples from $L_2$ sample from. 
To show correctness, we decompose 
\[F_p(v+u)-F_p(v)=\sum_{a\in\calH}\sum_{i=1}^{p}\binom{p}{i}v_a^iu_a^{p-i}+\sum_{a\notin\calH}\sum_{i=1}^{p}\binom{p}{i}v_a^iu_a^{p-i}.\]
The heavy-hitter subroutine allows accurate estimation to the first term and the perfect $L_2$ sampling subroutines allows accurate estimation to the second term. 
Finally, we remark that perfect $L_2$ sampling incurs a term of $v_a^2$. 
Thus we further need to split the estimation of the second term into the cases where $i=1$ and $i>1$. 

\paragraph{Estimation of $\langle v^i,u^{p-i}\rangle$ for $i>1$.}
We first show that our difference estimator gives an additive $\eps\cdot F_p(v)$ approximation to $\sum_{a\notin\calH}\binom{p}{i}v_a^iu_a^{p-i}$ with $i\ge 2$
\begin{lemma}
\lemlab{lem:diff:est:largep:generali}
For integer $p>2$, there exists an algorithm that uses space $\O{\frac{\gamma}{\eps^2}n^{1-2/p}\log^3 n\log\frac{n}{\delta}}$ and outputs an additive $\eps\cdot F_p(v)$ approximation to $\sum_{a\notin\calH}\binom{p}{i}v_a^iu_a^{p-i}$ with $i\ge 2$. 
\end{lemma}
\begin{proof}
For an oblivious stream $\calS$, let $v$ be the frequency vector induced by the updates of $\calS$ from time $t_1$ to $t_2$ and $u$ be the frequency vector induced by updates from time $t_2$ to $t$ exclusive, with $F_p(u)\le\gamma\cdot F_p(v)$. 
By \thmref{thm:robust:opt:HH}, there exists an algorithm $\heavyhitters$ that uses $\O{\frac{\gamma^{2-2/p}\log^2 n}{\eps^2}n^{1-2/p}}$ space and outputs a list $\calH$ of coordinates and estimates $\widehat{v_i}$ with additive error $\frac{\eps\gamma^{1/p}}{64\gamma}\|v\|_p$. 
Moreover, $\calH$ includes all coordinates $i\in[n]$ such that $v_i\ge\frac{\gamma^{1/p}}{16}\|v\|_p$ and no coordinate $j$ such that $v_j\le\frac{\gamma^{1/p}}{32}\|v\|_p$. 
Let $h$ be the vector consisting of the estimates $\widehat{v_i}$ for each $i\in[n]$ with $i\in\calH$ and $0$ in the positions $i$ for which $i\notin\calH$. 
We define the vector $w:=v-h$. 

We use $\sampler$ to sample indices $j_1,\ldots,j_k\in[n]$ with so that each sample is a coordinate $a\in[n]$ with probability $\frac{w_a^2}{\|w\|_2^2}+\frac{1}{\poly(n)}$. 
We obtain unbiased estimates to $\widehat{w_{j_1}^{p-i-2}},\ldots,\widehat{w_{j_k}^{p-i-2}}$ of $w_{j_1}^{p-i-2},\ldots,w_{j_k}^{p-i-2}$ through \lemref{lem:coor:exp:var}. 
We also obtain a $(1+\O{\eps})$-approximation unbiased estimate $W$ of $\|w\|_2^2$. 
Thus for each $b\in[k]$, the product $\widehat{w_{j_b}^{p-i-2}}\cdot W\cdot u_{j_b}^i$ satisfies
\begin{align*}
\Ex{\widehat{w_{j_b}^{p-i-2}}\cdot W\cdot u_{j_b}^i}&=\sum_{a\in\calH}\ip{\widehat{v_a}^{p-i}}{u_a^i}+\sum_{a\in\calH}\left(\frac{w_a^2}{\|w\|_2^2}+\frac{1}{\poly(n)}\right)\cdot w_a^{p-i-2}\cdot(1\pm\O{\eps})\|w\|_2^2\cdot u_a^i\\
&+\sum_{a\in[n]\setminus\calH}\left(\frac{w_a^2}{\|w\|_2^2}+\frac{1}{\poly(n)}\right)\cdot w_a^{p-i-2}\cdot(1\pm\O{\eps})\|w\|_2^2\cdot u_a^i.
\end{align*}
Observe that since $\widehat{v_a}$ is a $(1+\eps)$-approximation to $v_a$ for $a\in\calH$, then $|w_a|\le\eps|v_a|$, so that the second summation is at most $\O{\eps}\ip{v^{p-i}}{u^i}$. 
Moreover, we have $w_a=v_a$ for $a\in[n]\setminus\calH$, so that 
\begin{align*}
\Ex{\widehat{w_{j_b}^{p-i-2}}\cdot W\cdot u_{j_b}^i}&\in(1\pm\O{\eps})|\ip{w^{p-i}}{u^i}|+\frac{1}{\poly(n)}.
\end{align*}
Similarly, the variance is at most
\begin{align*}
\Var\left(\widehat{w_{j_b}^{p-i-2}}\cdot W\cdot u_{j_b}^i\right)&\le\sum_{a\in[n]}\left(\frac{w_a^2}{\|w\|_2^2}+\frac{1}{\poly(n)}\right)\cdot w_a^{2p-2i-4}\cdot(1\pm\O{\eps})\|w\|_2^4\cdot u_a^{2i}\\
&\le\sum_{a\in[n]}(1\pm\O{\eps})w_a^{2p-2i-2}\cdot\|w\|_2^2\cdot u_a^{2i}\\
&\le\sum_{a\in[n]}(1\pm\O{\eps})w_a^{2p-4}\cdot\|w\|_2^2\cdot u_a^{2}.
\end{align*}
By H\"{o}lder's inequality, the variance is at most
\begin{align*}
(1\pm\O{\eps})\|w\|_2^2\sum_{a\in[n]}w_a^{2p-4}u_a^2&\le(1\pm\O{\eps})\|w\|_2^2\left(\sum_{a\in[n]}u_a^p\right)^{2/p}\left(\sum_{a\in[n]}w_a^{2p}\right)^{1-2/p}\\
&=(1\pm\O{\eps})\|w\|_2^2\cdot\|u\|_p^2\cdot\|w\|_{2p}^{2p-4}.
\end{align*}
Since $(1+\eps)\|w\|_2^2\le2\|v\|_2^2$ for $\eps\le 1$ and $\|u\|_p^p\le\gamma\|v\|_p^p$, then the variance is at most
\begin{align*}
2\|v\|_2^2\cdot\gamma^{2/p}\|v\|_p^2\cdot\|w\|_{2p}^{2p-4}. 
\end{align*}
Recall that the vector $w$ is formed by removing from $v$ the coordinates $i\in[n]$ such that $v_i\ge\frac{\gamma^{1/p}}{16}\|v\|_p$. 
Thus, $|w_i|\le\frac{\gamma^{1/p}}{16}\|v\|_p$ for all $i\in[n]$. 
Subject to these constraints, we have that 
\[\|w\|_{2p}^{2p}\le\frac{16^p}{\gamma}\cdot\frac{\gamma^2}{16^{2p}}\|v\|_p^{2p}.\]
Thus we have
\[\|w\|_{2p}^{2p-4}\le\frac{\gamma^{1-2/p}}{16^{p-2}}\|v\|_p^{2p-4}.\]
Hence, the variance is at most 
\begin{align*}
\|v\|_2^2\cdot\gamma^{2/p}\|v\|_p^2\cdot\gamma^{1-2/p}\|v\|_p^{2p-4}\le\frac{\gamma}{n^{1-2/p}}\|v\|_p^2\cdot\|v\|_p^2\cdot\|v\|_p^{2p-4}. 
\end{align*}
Thus by setting $k=\O{\frac{\gamma}{\eps^2}n^{1-2/p}}$, we obtain an additive $\frac{\eps}{p}\cdot F_p(v)$ approximation to $F_p(v+u)-F_p(v)$ with constant probability. 
We can then boost this probability to $1-\frac{\delta}{\poly(n)}$ by repeating $\O{\log\frac{n}{\delta}}$ times. 
Hence, we have a $(\gamma,\eps,\delta)$-difference estimator for $F_p$. 

To analyze the total space complexity, first observe that we use $\O{\frac{\gamma}{\eps^2}n^{1-2/p}\log n}$ instances of $\sampler$, $\countsketch$, and an $F_2$ moment estimation algorithm. 
For $\countsketch$, we only require constant factor approximation, so that each instance uses space $\O{\log^2 n}$. 
Similarly for an $F_2$ moment estimation algorithm, we only require constant factor approximation, so that $\O{\log^2 n}$ space suffices, by \thmref{thm:robust:opt:F2}. 
Each instance of $\sampler$ uses space $\O{\log^3 n}$. 
Hence, the total space is $\O{\frac{\gamma}{\eps^2}n^{1-2/p}\log^3 n\log\frac{n}{\delta}}$. 
\end{proof}

\paragraph{Estimation of $\langle v,u^{p-1}\rangle$.}
We first duplicate each coordinate of $u$ and $v$ a total number of $M=n^c$ times, for some sufficiently large constant $c>0$. 
For each $i\in[n]$ and $j\in[M]$, we scale $u_i$ and $v_i$ by $1/e_{i,j}^{1/2}$, where $e_{i,j}$ is a random exponential variable truncated at $N^2$, where $N=nM$. 
Thus we have that the cdf of $\frac{1}{e_{i,j}}$ satisfies $\PPr{\frac{1}{e_{i,j}}\ge x}\le\Theta\left(\frac{1}{x}\right)$ for $x>0$, so that $\Ex{\frac{1}{e_{i,j}}}=\int_0^{N^2}1-\left(1-\Theta\left(\frac{1}{x}\right)\right)\,dx=\O{\log n}$. 
Let $U\in\mathbb{R}^N$ and $V\in\mathbb{R}^N$ be scaled and duplicated vector representations of $u$ and $v$, respectively. 
%Observe that $\Ex{\|V\|_2^2}=M\cdot\|v\|_2^2\cdot\Ex{\frac{1}{e_{i,j}}}=\O{\|v\|_2^2\cdot M\log n}$. 
%Similarly for any $i\in[n]$, we have that $\max_{j\in[M]}\frac{1}{e_{i,j}}\ge\O{\frac{M}{\log n}}$ with probability at least $1-\frac{1}{\poly(n)}$. 

\begin{figure*}
\begin{mdframed}
\begin{enumerate}
\item
Use a set of exponential random variables to form a vector $V$ of duplicated and scaled coordinates of $v$. 
\item
Hash the coordinates of $V$ into a $\countsketch$ data structure with $\O{\log n}$ buckets. 
\item
Use the same exponential random variables to perform perfect $L_2$ sampling on $u$ to obtain a coordinate $(i,j)$ and an unbiased estimate $\widehat{u_{i,j}}$ to $u_{i,j}$. 
\item
Let $\widehat{U}$ be an unbiased estimate of $\|u\|_2^2$ with second moment $\O{\|u\|_2^4}$.  
\item
Query $\countsketch$ for an unbiased estimate $\widehat{v_{i,j}}$ to $v_{i,j}$ and set an estimator as $\widehat{U}\cdot\widehat{v_{i,j}}\left(\widehat{u_{i,j}}\right)^{p-3}$. 
\item
Output the mean of $\O{\frac{\gamma}{\eps^2}\cdot n^{1-2/p}}$ such estimators. 
\end{enumerate}
\end{mdframed}
\caption{$F_p$ difference estimator for $\langle v,u^{p-1}\rangle$ with integer $p>2$.}
\figlab{fig:diff:est:largep:indexone}
\end{figure*}

\begin{lemma}
\lemlab{lem:diff:est:largep:one}
For integer $p>2$, there exists an algorithm that uses space $\O{\frac{\gamma}{\eps^2}n^{1-2/p}\log^3 n\log\frac{n}{\delta}}$ and outputs an additive $\eps\cdot F_p(v)$ approximation to $\langle v,u^{p-1}\rangle$. 
\end{lemma}
\begin{proof}
For an oblivious stream $\calS$, let $v$ be the frequency vector induced by the updates of $\calS$ from time $t_1$ to $t_2$ and $u$ be the frequency vector induced by updates from time $t_2$ to $t$ exclusive, with $F_p(u)\le\gamma\cdot F_p(v)$. 
We use $\sampler$ to sample indices $j_1,\ldots,j_k\in[n]$ with so that each sample is a coordinate $a\in[n]$ with probability $\frac{u_a^2}{\|u\|_2^2}+\frac{1}{\poly(n)}$. 
For each $b\in[k]$, we also obtain an unbiased estimate $\widehat{v_{j_b}}$ of $v_{j_b}$ from the $\countsketch$ data structure. 
We also obtain an unbiased estimate $\widehat{U}$ of $\|u\|_2^2$. 
Observe that
\[\Ex{\widehat{U}\cdot\widehat{v_{j_b}}\cdot\widehat{u_{j_b}^{p-3}}}=\|u\|_2^2\sum_{a\in[n]}\left(\frac{u_a^2}{\|u\|_2^2}+\frac{1}{\poly(n)}\right)v_a\cdot u_a^{p-3}=\langle v,u^{p-1}\rangle.\]
To analyze the variance of $\widehat{U}\cdot\widehat{v_{j_b}}\left(\widehat{u_{j_b}}\right)^{p-3}$, we first observe that by the min-stability of exponential random variables, the maximum coordinate of the scaled vector $V$ is distributed as $\|V\|_2/E^{1/2}$ for an exponential random variable $E$. 
If the coordinate $a\in[n]$ from the unscaled vector $v$ is sampled, this maximum coordinate is also $v_a/e_{a,j}^{1/2}$ for an exponential random variable $e_{a,j}$. 
We have that $v_a/e_{a,j}^{1/2}\ge\|V_{-a,j}\|_2/E^{1/2}$, where $V_{-a}$ indicates the vector $V$ with the $a$-th entry set to zero. 
Since $E=\O{\log n}$ with probability $1-\frac{1}{\poly(n)}$, then we have $e_{a,j}\le\frac{\O{\log n}\cdot v_a^2}{\|V\|_2^2}$. 
Hence $e_{a,j}$ has expectation $\frac{\O{\log n}\cdot v_a^2}{\|V\|_2^2}$ so that by using $\O{\log n}$ buckets, $\countsketch$ outputs an estimate $\widehat{v_{j_b}}$ with variance $v_a^2$ if $a\in[n]$ is sampled in iteration $b$.  
%we first let $\mathcal{E}$ be the event that for any $i\in[n]$, we have that $\max_{j\in[M]}\frac{1}{e_{i,j}}\ge\O{\frac{M}{\log n}}$. 
%Hence, $\PPr{\mathcal{E}}\ge 1-\frac{1}{\poly(n)}$. 
%Moreover, $\Ex{\|V\|_2^2}=M\cdot\|v\|_2^2\cdot\Ex{\frac{1}{e_{i,j}}}=\O{\|v\|_2^2\cdot M\log n}$. 
%Thus the variance of $v_{j_b}$ is at most $\Ex{\|V\|_2^2\cdot\frac{\log n}{M}}\le\O{\|v\|_2^2\cdot\log^2 n}$. 
%Therefore conditioned on $\mathcal{E}$, we have
Thus, there exists a constant $C>0$ such that the variance of $\widehat{U}\cdot\widehat{v_{j_b}}\left(\widehat{u_{j_b}}\right)^{p-3}$ is at most
\begin{align*}
\Var\left(\widehat{U}\cdot\widehat{v_{j_b}}\cdot\widehat{u_{j_b}^{p-3}}\right)&\le\sum_{a\in[n]}\left(\frac{u_a^2}{\|u\|_2^2}+\frac{1}{\poly(n)}\right)\cdot C\|u\|_2^4\cdot v_a^2\cdot u_a^{2p-6}\\
&\le\sum_{a\in[n]} 2C\cdot u_a^{2p-4}v_a^2\cdot\|u\|_2^2.
\end{align*}
By H\"{o}lder's inequality, the variance is at most
\begin{align*}
\Var\left(\widehat{U}\cdot\widehat{v_{j_b}}\cdot\widehat{u_{j_b}^{p-3}}\right)&\le2C\cdot\|u\|_2^2\cdot\left(\sum_{a\in[n]}v_a^p\right)^{2/p}\left(\sum_{a\in[n]}u_a^{2p}\right)^{1-2/p}\\
&=2C\cdot\|u\|_2^2\cdot\|v\|_p^2\cdot\|u\|_{2p}^{2p-4}\\
&\le2C\cdot n^{1-2/p}\cdot\|u\|_p^2\cdot\|v\|_p^2\cdot\|u\|_{2p}^{2p-4}\\
&\le2C\cdot n^{1-2/p}\cdot\|u\|_p^{2p-2}\|v\|_p^2\\
&\le2C\cdot\gamma n^{1-2/p}\cdot n^{1-2/p}\|v\|_p^2,
\end{align*}
where the last inequality results from the fact that $F_p(u)\le\gamma F_p(v)$ and $p>2$. 
Thus by setting $k=\O{\frac{\gamma}{\eps^2}n^{1-2/p}}$, we obtain an additive $\frac{\eps}{p}\cdot F_p(v)$ approximation to $\langle v,u^{p-1}\rangle$ with constant probability. 
We can then boost this probability to $1-\frac{\delta}{\poly(n)}$ by repeating $\O{\log\frac{n}{\delta}}$ times. 
\end{proof}

\paragraph{Putting it all together.}
We now give our $F_p$ difference estimator for integer $p>2$ using the above subroutines to approximate $F_p(v+u)-F_p(v)=\sum_{i=1}^{p}\binom{p}{i}\langle v^i,u^{p-i}\rangle$. 
\begin{lemma}[$F_p$ difference estimator for integer $p>2$]
\lemlab{lem:diff:est:Fp:largep}
For integer $p>2$, there exists a $(\gamma,\eps,\delta)$-difference estimator for $F_p$ that uses space $\O{\frac{\gamma}{\eps^2}n^{1-2/p}\log^3 n\log\frac{n}{\delta}}$. 
\end{lemma}
\begin{proof}
For an oblivious stream $\calS$, let $v$ be the frequency vector induced by the updates of $\calS$ from time $t_1$ to $t_2$ and $u$ be the frequency vector induced by updates from time $t_2$ to $t$ exclusive, with $F_p(u)\le\gamma\cdot F_p(v)$. 
Observe that for integer $p$, we can expand 
\[F_p(v+u)-F_p(v)=\sum_{i=1}^{p}\binom{p}{i}\langle v^i,u^{p-i}\rangle,\]
where we write $v^p$ as the coordinate-wise $p$-th power of $v$. 
Let $\calH$ be the set of coordinates output by $\heavyhitters$. 
Moreover,
\[F_p(v+u)-F_p(v)=\sum_{a\in\calH}\sum_{i=1}^{p}\binom{p}{i}v_a^iu_a^{p-i}+\sum_{a\notin\calH}\sum_{i=1}^{p}\binom{p}{i}v_a^iu_a^{p-i}.\]

By \thmref{thm:robust:opt:HH}, there exists an algorithm $\heavyhitters$ that uses $\O{\frac{\gamma^{2-2/p}\log^2 n}{\eps^2}n^{1-2/p}}$ space and outputs a list $\calH$ of coordinates and estimates $\widehat{v_a}$ with additive error $\frac{\eps\gamma^{1/p}}{64\gamma}\|v\|_p$. 
Thus using the estimates $\widehat{v_a}$ along with the corresponding coordinates $u_a$ gives an additive $\O{\eps}F_p(v)$ approximation to $\sum_{a\in\calH}\sum_{i=1}^{p}\binom{p}{i}v_a^iu_a^{p-i}$ using $\tO{\frac{\gamma\log^2 n}{\eps^2}\,n^{1-2/p}}$ space, for $p>2$. 
By \lemref{lem:diff:est:largep:one} and \lemref{lem:diff:est:largep:generali}, we similarly obtain an additive $\O{\eps}F_p(v)$ approximation to $\sum_{a\notin\calH}\sum_{i=1}^{p}\binom{p}{i}v_a^iu_a^{p-i}$ using $\O{\frac{\gamma}{\eps^2}n^{1-2/p}\log^3 n\log\frac{n}{\delta'}}$ space, where $\delta'=\frac{\delta}{p}$. 
Thus by rescaling $\eps$, we obtain an additive $\eps\,F_p(v)$ approximation to $F_p(v+u)-F_p(v)$ using $\O{\frac{\gamma}{\eps^2}n^{1-2/p}\log^3 n\log\frac{n}{\delta}}$ bits of space. 
\end{proof}

Using our difference estimator, we obtain a robust algorithm for $F_p$ moment estimation for integer $p>2$. 
\begin{theorem}[Adversarially robust $F_p$ streaming algorithm for integer $p>2$]
\thmlab{thm:robust:opt:Fp:largep}
Given $\eps>0$ and integer $p>2$, there exists an adversarially robust streaming algorithm that outputs a $(1+\eps)$-approximation for $F_p$ that succeeds with probability at least $\frac{2}{3}$ and uses $\O{\frac{1}{\eps^2}n^{1-2/p}\log^5 n\log^3\frac{1}{\eps}}$ bits of space. 
\end{theorem}
\begin{proof}
For $p>2$, $F_p$ is a monotonic function with $(\eps,m)$-twist number $\O{\frac{1}{\eps}\log n}$, by \obsref{obs:flip:fp}. 
By \lemref{lem:diff:est:Fp:largep} and \thmref{thm:strong:Fp:largep}, there exists a $(\gamma,\eps,\delta)$-difference estimator that uses space
\[\O{\frac{\gamma}{\eps^2}n^{1-2/p}\log^3 n\log\frac{n}{\delta}}\]
and a strong tracker for $F_p$ that each uses space 
\[\O{\frac{1}{\eps^2}n^{1-2/p}\log^2 n\log\frac{n}{\delta}}.\]
We use these subroutines in the framework of \algref{alg:framework}. 
Therefore, \thmref{thm:framework} with $S_1(n,\delta,\eps)=n^{1-2/p}\log^3 n\log\frac{n}{\delta}$, $C=1$, and $S_2=0$ proves that there exists an adversarially robust streaming algorithm that outputs a $(1+\eps)$ for $F_p$ that succeeds with constant probability and uses space $\O{\frac{1}{\eps^2}n^{1-2/p}\log^5 n\log^3\frac{1}{\eps}}$. 
\end{proof}

\section{Robust \texorpdfstring{$F_0$}{F0} Estimation}
In this section, we use the framework of \secref{sec:framework} to give an adversarially robust streaming algorithm for the distinct elements problem or equivalently, the $F_0$ moment estimation. 
We again require an $F_0$ strong tracker and an $F_0$ difference estimator so that we can apply \thmref{thm:framework}. 
Fortunately, we can again use similar sketches for both the $F_0$ strong stracker and the $F_0$ difference estimator, similar to the approach for $F_2$ moment estimation. 
We can then optimize our algorithm beyond the guarantees of \thmref{thm:framework} so that our final space guarantees in \thmref{thm:robust:opt:F0} match the bounds in \thmref{thm:robust:opt:F2}. 
Note that this does not quite match the best known $F_0$ strong-tracking algorithm on insertion-only streams~\cite{Blasiok20}, which uses $\O{\frac{\log\log n}{\eps^2}+\log n}$ bits of space: 
\begin{theorem}[Oblivious $F_0$ strong tracking]
\thmlab{thm:strong:F0}
\cite{Blasiok20}
There exists an insertion-only streaming algorithm $\zeroestimate(1,t,\eps,\delta)$ that uses $\O{\frac{1}{\eps^2}\log\frac{1}{\delta}+\log n}$ bits of space and provides $(\eps,\delta)$-strong $F_0$ tracking. 
\end{theorem}
%Interestingly, $\zeroestimate$ also gives an $F_0$ difference estimator, which is immediate from the proof of \cite{Blasiok20}. 
\thmref{thm:strong:F0} and other $F_0$ approximation algorithms use a balls-and-bins argument with increasing levels of sophistication~\cite{Bar-YossefJKST02, KaneNW10b, Blasiok20}. 
We use a similar balls-and-bins argument, where each item is subsampled at a level $k$ with probability $\frac{1}{2^k}$, to obtain an $F_0$ difference estimator. 
By counting the number of items in a level with $\Theta\left(\frac{\gamma}{\eps^2}\right)$ items that survive the subsampling process for the stream $u$, it follows that the expected number of the distinct items in $v$ but not $u$ is $\Theta\left(\frac{\gamma^2}{\eps^2}\right)$. 
Thus we obtain a $\left(1+\frac{\eps}{\gamma}\right)$-approximation to $F_0(v)-F_0(u)$ from this procedure by first running the balls-and-bins experiment on $u$ and counting the number of bins that are occupied at some level $k$ with $\Theta\left(\frac{\gamma}{\eps^2}\right)$ survivors. 
We then run the same balls-and-bins experiment on $v-u$ by only counting the additional bins that are occupied at level $k$, and rescaling this number by $2^k$. 
Note that additional bins only correspond to items in $v$ but not $u$, which is exactly $F_0(v)-F_0(u)$. 
Although level $k$ does not necessarily give a $(1+\eps)$-approximation to $F_0(v)-F_0(u)$, it does give a $\left(1+\frac{\eps}{\gamma}\right)$-approximation to $F_0(v)-F_0(u)$, which translates to an additive $\eps\cdot F_0(u)$ approximation to $F_0(v)-F_0(u)$ and is exactly the requirement for the $F_0$ difference estimator, since $F_0(v)-F_0(u)\le\gamma F(u)$. 
%\begin{lemma}[$F_0$ difference estimator]
%\lemlab{lem:diff:est:F0}
%\cite{Blasiok20}
%There exists a $(\gamma,\eps,\delta)$-\emph{difference estimator} for $F_0$ that uses $\O{\frac{\gamma}{\eps^2}\log\frac{1}{\delta}+\log n}$ bits of space. 
%\end{lemma}

\begin{lemma}[$F_0$ difference estimator]
\lemlab{lem:diff:est:F0}
There exists a $(\gamma,\eps,\delta)$-\emph{difference estimator} for $F_0$ that uses $\O{\frac{\gamma}{\eps^2}\left(\log\frac{1}{\eps}+\log\log n+\log\frac{1}{\delta}\right)+\log n}$ bits of space. 
%$\O{\frac{\gamma\log n}{\eps^2}\left(\log\frac{1}{\eps}+\log\frac{1}{\delta}\right)}$
\end{lemma}
\begin{proof}
Suppose we would like to estimate $F_0(v)-F_0(u)=\gamma\cdot F_0(u)$ for $\gamma=\Omega(\eps)$ and we subsample with probability $\frac{1}{2^k}$ so that $\frac{F_0(u)}{2^k}=\Theta\left(\frac{\gamma}{\eps^2}\right)$. 
Then in expectation, the number of sampled items in $v$ that are not in $u$ is $\frac{F_0(v)-F_0(u)}{2^k}=\Theta\left(\frac{\gamma^2}{\eps^2}\right)$. 
Hence if $X$ is the number of survivors from $F_0(v)-F_0(u)$ at level $k$, then $\Ex{2^k\cdot X}=F_0(v)-F_0(u)$ and the variance of $2^k\cdot X$ is at most
\[2^k\cdot(F_0(v)-F_0(u))\le\frac{\eps^2\cdot F_0(u)}{\gamma}\cdot(F_0(v)-F_0(u))\le\eps^2(F_0(u))^2.\]
Thus, the $2^k\cdot X$ gives an additive $\eps\cdot F_0$ approximation to $F_0(v)-F_0(u)$ with constant probability. 
The probability can then be boosted to at least $1-\delta$ by taking the median of $\log\frac{1}{\delta}$ parallel instances. 
Since the algorithm maintains $\frac{\gamma}{\eps^2}$ sampled items, the total space na\"{i}vely required is $\O{\frac{\gamma\log n}{\eps^2}\log\frac{1}{\delta}}$. 
By hashing to $P=\poly\left(\frac{1}{\eps},\log n,\frac{1}{\delta}\right)$ buckets, the space bound can be further improved to $\O{\frac{\gamma}{\eps^2}\left(\log\frac{1}{\eps}+\log\log n+\log\frac{1}{\delta}\right)+\log n}$, e.g., by first composing a hash function $h_1:[n]\to[P^3]$ and then a hash function $h_2:[P^3]\to[P]$, along the lines of \cite{KaneNW10b,Blasiok20}.  
This suffices for a difference estimator at a single point in time; to obtain the strong tracking property, we note that both the estimator and $F_0(v)-F_0(u)$ are monotonic. 
Hence, it suffices to take a union bound over $\log\frac{1}{\eps}$ times when the difference increases by a factor of $(1+\O{\eps})$. 
Therefore, the total space required is $\O{\frac{\gamma}{\eps^2}\left(\log\frac{1}{\eps}+\log\log n+\log\frac{1}{\delta}\right)+\log n}$. 
%$\O{\frac{\gamma\log n}{\eps^2}\left(\log\frac{1}{\eps}+\log\frac{1}{\delta}\right)}$
\end{proof}
Note that the difference estimator in \lemref{lem:diff:est:F0} only requires pairwise independence and thus can be derandomized using a hash function that can be stored using $\O{\log n}$ bits of space. 
We now use our difference estimator to get an adversarially robust streaming algorithm for the distinct elements problem. 
\begin{theorem}
\thmlab{thm:framework:F0}
Given $\eps>0$, there exists an adversarially robust streaming algorithm that outputs a $(1+\eps)$-approximation to $F_0$ that succeeds with probability at least $\frac{2}{3}$ and uses space
\[\O{\frac{1}{\eps^2}\log n\log^3\frac{1}{\eps}\cdot\left(\log\frac{1}{\eps}+\log\log n\right)+\frac{1}{\eps^2}\log^2 n}.\]
%\[\O{\frac{1}{\eps^2}\log^2 n\log^3\frac{1}{\eps}\left(\log\log n+\log\frac{1}{\eps}\right)}.\]
\end{theorem}
\begin{proof}
First note that $F_0$ is a monotonic function with $(\eps,m)$-twist number $\lambda=\O{\frac{1}{\eps}\log n}$ for $m=\poly(n)$, by \obsref{obs:flip:fp}. 
By \thmref{thm:strong:F0} and \lemref{lem:diff:est:F0}, there exist a $(\gamma,\eps,\delta)$-difference estimator that uses $\O{\frac{\gamma}{\eps^2}\left(\log\frac{1}{\eps}+\log\log n+\log\frac{1}{\delta}\right)+\log n}$  bits of space and an oblivious strong tracker for $F_0$ that uses $\O{\frac{1}{\eps^2}\log\frac{1}{\delta}+\log n}$ bits of space, which can be used in the framework of \algref{alg:framework}. 
We thus apply \thmref{thm:framework} with $S_1(n,\delta,\eps)=\log\frac{1}{\eps}+\log\log n+\log\frac{1}{\delta}$ , $C=1$, and $S_2(n,\delta,\eps)=\log n$. 
Thus the total space is 
\[\O{\frac{1}{\eps^2}\log n\log^3\frac{1}{\eps}\cdot\left(\log\frac{1}{\eps}+\log\log n\right)+\frac{1}{\eps^2}\log^2 n}.\]
%\[\O{\frac{1}{\eps^2}\log^2 n\log^3\frac{1}{\eps}\left(\log\log n+\log\frac{1}{\eps}\right)}.\]
\end{proof}

\paragraph{Optimized $F_0$ Algorithm.}
To improve the space requirements, we can again observe that it suffices to maintain sketches $\calA_a$ and $\calB_{a,c}$ for $\O{\log\frac{1}{\eps}}$ values of $a$ at a time, instead of maintaining all sketches $\calA_a$ and $\calB_{a,c}$ simultaneously. 
By the same argument as before, it suffices to maintain only the most sketches $\calA_i$ and $\calB_{i,c}$ for only the smallest $\O{\log\frac{1}{\eps}}$ values of $i$ that are at least $a$ since the output increases by a factor of $2$ each time $a$ increases and thus any larger index will have only missed $\O{\eps}$ fraction of the $F_0$ of the stream. 
Hence, any larger index still outputs a $(1+\eps)$-approximation once it becomes initialized.  
\begin{theorem}
\thmlab{thm:robust:opt:F0}
Given $\eps>0$, there exists an adversarially robust streaming algorithm that outputs a $(1+\eps)$-approximation for $F_0$ that succeeds with probability at least $\frac{2}{3}$ and uses total space
%\[\O{\frac{1}{\eps^2}\log n\log^4\frac{1}{\eps}\left(\log\log n+\log\frac{1}{\eps}\right)}.\]
\[\O{\frac{1}{\eps^2}\log^4\frac{1}{\eps}\left(\log\log n+\log\frac{1}{\eps}\right)+\frac{1}{\eps}\log n\log\frac{1}{\eps}}.\]
\end{theorem}
\begin{proof}
At any point, there are $\O{\log\frac{1}{\eps}}$ active indices $a$ and $\O{1}$ active indices $c$ corresponding to sketches $\calA_{a}$ and $\calB_{a,c}$. 
Because there are $\O{\log n}$ total indices $a$ over the stream, we set $\delta'=\frac{\delta}{\poly\left(\log n,\frac{1}{\eps}\right)}$. 
Thus from \thmref{thm:framework} and \thmref{thm:framework:F0} that for fixed $a$, the total space that sketches $\calA_{a}$ and $\calB_{a,c}$ use across the $\beta$ granularities is 
%\[\O{\frac{1}{\eps^2}\log n\log^3\frac{1}{\eps}\left(\log\log n+\log\frac{1}{\eps}\right)},\]
\[\O{\frac{1}{\eps^2}\left(\log\log n+\log\frac{1}{\eps}\right)+\log n},\]
%since we have $S_1(n,\delta',\eps)=\log n\left(\log\frac{1}{\eps}+\log\frac{1}{\delta'}\right)$, $C=1$, and $S_2(n,\delta',\eps)=0$. 
since we have $S_1(n,\delta',\eps)=\left(\log\frac{1}{\eps}+\log\log n+\log\frac{1}{\delta'}\right)$, $C=1$, and $S_2(n,\delta',\eps)=\log n$. 
However, there are only $\O{\log\frac{1}{\eps}}$ active indices $a$ simultaneously maintained by the optimized algorithm. 
For each active index $a$, the sketches use $\O{\frac{1}{\eps^2}\log^3\frac{1}{\eps}\left(\log\log n+\log\frac{1}{\eps}\right)+\log n}$ bits of space. 
%For each active index $a$, the sketches use $\O{\frac{1}{\eps^2}\log n\log^3\frac{1}{\eps}\left(\log\log n+\log\frac{1}{\eps}\right)}$ bits of space. 
As there are $\O{\log\frac{1}{\eps}}$ active indices of $a$, consisting of $\O{\frac{1}{\eps}}$ subroutines, then it takes $\O{\frac{1}{\eps}\log n\log\frac{1}{\eps}}$ additional bits of space to store the splitting times for each of the $\O{\frac{1}{\eps}\log\frac{1}{\eps}}$ active subroutines across a stream of length $m$, with $\log m=\O{\log n}$. 
Therefore, the total space required is 
%\[\O{\frac{1}{\eps^2}\log n\log^4\frac{1}{\eps}\left(\log\log n+\log\frac{1}{\eps}\right)}.\]
\[\O{\frac{1}{\eps^2}\log^4\frac{1}{\eps}\left(\log\log n+\log\frac{1}{\eps}\right)+\frac{1}{\eps}\log n\log\frac{1}{\eps}}.\]
\end{proof}

\section{Framework for Streams with Bounded Twist Number}
Recall that the previous framework of \thmref{thm:framework} in \secref{sec:framework} was only suited for insertion-only streams. 
In this section, we modify the framework of \algref{alg:framework} to handle $F_p$ moment estimation on turnstile streams. 

The main difference for turnstile streams is that the value of $F$ on the underlying frequency vector induced by the stream can decrease as well as increase. 
Thus if we assign an algorithm $\calA_j$ to be revealed when the value of $F$ exceeds $2^j$, this may happen multiple times. 
However, once $\calA_j$ is revealed for the first time, then the input to the algorithm is no longer independent to the randomness of the algorithm, due to the adversary. 
Hence, we can no longer use $\calA_j$ to estimate the value of $F$ if it decreases and then exceeds $2^j$ again. 

Nevertheless, note that each time the value of $F$ increases or decreases by a factor of roughly $(1+\eps)$, the value of the $(\eps,m)$-twist number $\lambda$ increases. 
Therefore, $\lambda$ governs the number of times the value of $F$ can exceed $2^j$, going from $2^{j-1}$. 
Similarly, $\lambda$ governs the number of times the value of $F$ can exceed $(1+\eps)^{Ci}2^j$, going from $(1+\eps)^{C(i-1)}2^j$ for any constant $C$.  
Thus if $\lambda$ is given in advance, the algorithm can maintain the appropriate number of independent instances of the difference estimators $\calB$ and the streaming algorithms $\calA$ to handle each time the value of $F$ crosses a threshold. 
We give the algorithm in full in \algref{alg:framework:flip}. 

\begin{algorithm}[!htb]
\caption{Framework for Robust Algorithms on Turnstile Streams}
\alglab{alg:framework:flip}
\begin{algorithmic}[1]
\Require{Stream $u_1,\ldots,u_m\in[n]$ of updates to coordinates of an underlying frequency vector, accuracy parameter $\eps\in(0,1)$, difference estimator $\calB$ for $F$, oblivious strong tracker $\calA$ for $F$, $\left(\frac{\eps}{32},m\right)$-twist number $\lambda$}
\Ensure{Robust $(1+\eps)$-approximation to $F$}
\State{$\delta\gets\frac{1}{\poly\left(\frac{1}{\eps}, \log n\right)}$, $\zeta\gets\frac{2}{2^{(C-1)/4}-1}$, $\eta\gets\frac{\eps}{64\zeta}$, $\beta\gets\ceil{\log\frac{8}{\eps}}$, $\tau\gets 0$}
\State{$a\gets 0$, $\varphi\gets 2^{(C-1)/4}$, $\gamma_j\gets 2^{j-1}\eta$}
\State{For $j\in[\beta]$, $\eta_j\gets\frac{\eta}{\beta}$ if $C=1$, $\eta_j\gets\frac{\eta}{\varphi^{\beta-j}}$ if $C>1$.}
\Comment{Accuracy for each difference estimator}
\State{$t_{i,j}\gets\infty$ for $i\le\O{\eps\lambda}$, $j\in[\beta]$}
\Comment{Counters/Timestamps}
\For{each update $u_t\in[n]$, $t\in[m]$}
\State{$j\gets\ceil{\eta\tau}$}
\State{$X\gets\calA_{a+1,j}(1,t,\eta,\delta)$}
\If{$X>2^a$}
\Comment{Switch sketch at top layer}
\State{$a\gets a+1$, $b\gets 0$, $Z_a\gets X$, $t_{a,j}\gets\min(t,t_{a,j})$ for $j\in[\beta]$.}
\EndIf
\If{$X<2^{a-1}$}
\Comment{Switch sketch at top layer}
\State{$a\gets a-1$, $b\gets\flr{\frac{8}{\eps}\cdot(X-Z_a)}$}
%\State{$a\gets a-1$, $b\gets 0$, $Z_a\gets X$, $t^{(y)}_{a,j}\gets t$ for $j\in[\beta]$, $\flr{\frac{\tau}{2^j}}\le y\le\ceil{\frac{\lambda}{2^j}}$.}
\EndIf
\State{$X\gets\estimateF$}
\Comment{Compute estimator $X$ for $F$ using unrevealed sketch}
\If{$X>\left(1+\frac{(b+1)\eps}{8}\right)\cdot Z_a$}
\Comment{Switch sketch at lower layer}
\State{$\tau\gets\tau+1$, $b\gets b+1$, $k\gets\lsb(b,1)$, $j\gets\flr{\frac{\tau}{2^k}}$}
\State{$Z_{a,k}\gets\calB_{a,j}(1,t_{a,k},t,\eta_k,\delta)$}
\Comment{Freeze old sketch}
\State{$t_{a,j}\gets t$ for $j\in[k]$.}
\Comment{Update difference estimator times}
\ElsIf{$X<\left(1+\frac{(b-1)\eps}{8}\right)\cdot Z_a$}
\Comment{Switch sketch at lower layer}
\State{$\tau\gets\tau+1$, $b\gets b-1$, $k\gets\lsb(b,1)$, $j\gets\flr{\frac{\tau}{2^k}}$}
\State{$Z_{a,k}\gets\calB_{a,j}(1,t_{a,k},t,\gamma_k,\eta_k,\delta)$}
\Comment{Freeze old sketch}
\State{$t_{a,j}\gets t$ for $j\in[k]$.}
\Comment{Update difference estimator times}
\ElsIf{if there exist $j,k$ with unfrozen $\calB_{a,j}(1,t_{a,k},t,\gamma_k,\eta_k,\delta)>\frac{\gamma_j}{2} X$}
\State{$\tau\gets\tau+1$, $j\gets\flr{\frac{\tau}{2^k}}$}
\State{$Z_{a,k}\gets Z_{a,k}+\calB_{a,j}(1,t_{a,k},t,\gamma_k,\eta_k,\delta)$}
\Comment{Freeze old sketch}
\State{$t_{a,j}\gets t$ for $j\in[k]$.}
\EndIf
\State{\Return $\left(1+\frac{b\eps}{8}\right)\cdot Z_a$}
\Comment{Output estimate for round $t$}
\EndFor
\end{algorithmic}
\end{algorithm}

\begin{algorithm}[!htb]
\caption{Subroutine $\estimateF$, modified for \algref{alg:framework:flip}}
\alglab{alg:estimatef:flip}
\begin{algorithmic}[1]
\State{$X\gets Z_a$, $k\gets\numbits(b+1)$, $z_i\gets\lsb(b+1,k+1-i)$ for $i\in[k]$.\\}
\Comment{$z_1>\ldots>z_k$ are the nonzero bits in the binary representation of $b+1$.}
\For{$1\le j\le k-1$}
\Comment{Compile previous frozen components for estimator $X$}
\State{$X\gets X+Z_{a,j}$}
\EndFor
%\State{$j\gets\flr{\frac{b+1}{2^{z_k}}}+\ceil{\frac{\tau}{2^k}}$}
\State{$j\gets\flr{\frac{\tau}{2^{z_k}}}$}
\State{$X\gets X+\calB_{a,j}\left(1,t^{(j)}_{a,z_k},t,\gamma_{z_k},\eta_{z_k},\delta\right)$}
\Comment{Use unrevealed sketch for last component}
\State{\Return $X$}
\end{algorithmic}
\end{algorithm}

The proof of correctness on non-adaptive streams follows similarly from the same argument as \lemref{lem:frame:oblivious}. 
The main difference is that because the value of $F$ on the stream can both increase and decrease, we must be a little more careful in defining the times for which the difference estimators reveal their outputs. 
In addition, to ensure the correctness of the difference estimator, we must facilitate the conditions required for the difference estimator. 
Thus we restart a difference estimator each time the value of the suffix has become sufficiently large even if the difference itself has not become large. 
This is a subtle issue for which we use the definition of twist number, rather than the simpler definition of flip number. 

\begin{lemma}[Correctness on non-adaptive streams]
\lemlab{lem:frame:oblivious:flip}
With probability at least $1-\delta\cdot\poly\left(\frac{1}{\eps},\log n\right)$, \algref{alg:framework:flip} outputs a $(1+\eps)$-approximation to $F$ at all times on a non-adaptive stream.
\end{lemma}
\begin{proof}
Let $t^{(y)}_i$ be the time in the stream at which the counter $a$ in \algref{alg:framework:flip} is set to $i$ for the $y$-th time, for integers $i,y>0$. 
For each integer $j\ge 0$, let $u_{i,j}$ be the last round for which $Z_{i,j}$ in \algref{alg:framework:flip} is defined. 
Let $t_{i,\ell}$ be the last round such that the counter $b$ in \algref{alg:framework:flip} is first set to $\ell$. 
Let $\mathcal{E}_1$ be the event that \algref{alg:framework:flip} outputs a $\left(1+\frac{\eps}{32}\right)$-approximation to the value of $F$ at all times $t^{(y)}_i$, so that by \lemref{lem:constant}, $\mathcal{E}_1$ holds with probability at least $1-\O{\frac{\lambda\delta\log n}{\eps}}$, since there are at most $\lambda$ times $\{t^{(y)}_i\}$ by the definition of the twist number.  
For the remainder of the proof, we fix the integers $i$ and $y$ and show correctness between $t^{(y)}_i$ and the first time $t>t^{(y)}_i$ such that either $t=t^{(y')}_{i+1}$ or $t=t^{(y')}_{i-1}$ for some $y'$. 
Without loss of generality, suppose $t=t^{(y')}_{i+1}$. 
We then show correctness at all times $t^{(\ell,r)}_{i}$ for each value $\ell$ obtained by the counter $b$ between times $t^{(y)}_i$ and $t^{(y')}_{i+1}$. 
The remainder of the argument proceeds similarly to \lemref{lem:frame:oblivious}. 

Let $\ell$ be a fixed value of the counter $b$, $k=\numbits(\ell)$, and $z_x=\lsb(\ell,x)$ for $x\in[k]$ so that $z_1>\ldots>z_k$ are the nonzero bits in the binary representation of $\ell$. 
By similar reasoning to \lemref{lem:time:geometric}, we have for all $x\in[k]$
\[F(1,u_{i,z_x})-F(1,t_{i,z_x})\le\frac{1}{2^{\beta-x-3}}F(1,t_i),\]
with probability at least $1-\delta\cdot\poly\left(\frac{1}{\eps},\log n\right)$. 
Since $u_{i,z_x}$ is the last round for which $Z_{i,z_x}$ in \algref{alg:framework:flip} is defined and all times $t_{i,j}$ for $j\le z_x$ are reset at that round, then we have $u_{i,z_x}=t_{i,z_{x+1}}$ for all $x\in[k]$. 
Thus for $u_{i,z_k}=t_{i,\ell}$, we can decompose 
\[F(1,t_{i,\ell})=\sum_{x=1}^k\left(F(1,u_{i,z_x})-F(1,u_{i,z_{x-1}})\right)=\sum_{x=1}^k\left(F(1,u_{i,z_x})-F(1,t_{i,z_x})\right).\]
Recall that the value $X_{\ell}$ of $X$ output by $\estimateF$ at time $t_{i,\ell}$ satisfies $X_{\ell}=Z_i+\sum_{x=1}^k Z_{i,z_x}$. 
By \lemref{lem:constant}, $Z_i$ incurs at most $\left(1+\frac{\eps}{8}\right)$ multiplicative error to $F(1,t_{i,\ell})$ since $F(1,t_{i,\ell})\le 4F(t_i)$. 
Moreover, $Z_{i,z_k}$ is an additive $2\eta_{z_k}F(1,t_i)$ approximation to $F(1,u_{i,z_k})-F(1,t_{i,z_k})$. 
Hence, the total additive error of $X_{\ell}$ to $F(1,t_{i,\ell})$ is at most $\frac{\eps}{4}\cdot F(1,t_{i,\ell})+\sum_{x=1}^k 2\eta_{z_k}F(1,t_i)$. 
Now if the $(\gamma,\eps,\delta)$-difference estimator has space dependency $\frac{\gamma^C}{\eps^2}$ with $C=1$, then $\eta_{z_k}=\frac{\eta}{\beta}$. 
Otherwise if $C>1$, then $\eta_{z_k}=\frac{\eta}{\varphi^{\beta-z_k}}$, with $\varphi>1$. 
In both cases, we have that the total additive error is at most $\frac{\eps}{2}\cdot F(1,t_{i,\ell})$, since $\beta=\ceil{\log\frac{8}{\eps}}$. 
Therefore, $X_{\ell}$ is a $\left(1+\frac{\eps}{2}\right)$-approximation to $F(1,t_{i,\ell})$. 

Note that the final output at time $t_{i,\ell}$ by \algref{alg:framework:flip} further incurs additive error at most $\frac{\eps}{8}\cdot Z_i$ due to the rounding $\left(1+\frac{b\eps}{8}\right)\cdot Z_a$ in the last step. 
By \lemref{lem:constant}, $Z_i\le\left(1+\frac{\eps}{8}\right)\cdot F(1,t_{i,\ell})$. 
Hence for sufficiently small constant $\eps>0$, we have that \algref{alg:framework:flip} outputs a $(1+\eps)$-approximation to $F(1,t_{i,\ell})$, with probability at least $1-\delta\cdot\poly\left(\frac{1}{\eps},\log n\right)$.
\end{proof}

\begin{theorem}[Framework for adversarially robust algorithms on turnstile streams]
\thmlab{thm:framework:flip}
Let $\eps,\delta>0$ and $F$ be a monotonic function with $(\eta,m)$-twist number $\lambda$, where $\eta=\frac{\eps}{32}$.  
Suppose there exists a $(\gamma,\eps,\delta)$-difference estimator with space dependency $\frac{\gamma^C}{\eps^2}$ for $C\ge 1$ and a strong tracker for $F$ on dynamic streams that use $\O{\frac{1}{\eps^2} S_F(n,\delta,\eps)}$ bits of space. 
Then there exists an adversarially robust dynamic streaming algorithm that outputs a $(1+\eps)$ approximation for $F$ that succeeds with constant probability. 
For $C>1$, the algorithm uses $\O{\frac{\lambda}{\eps}\cdot S_F(n,\delta',\eps)}$ bits of space, where $\delta'=\O{\frac{1}{\poly\left(\lambda,\frac{1}{\eps}\right)}}$. 
For $C=1$, the algorithm uses $\O{\frac{\lambda}{\eps}\log^3\frac{1}{\eps}\cdot S_F(n,\delta',\eps)}$ bits of space.
\end{theorem}
\begin{proof}
Consider \algref{alg:framework:flip} and note that \lemref{lem:frame:oblivious:flip} proves correctness over non-adaptive streams. 
We now argue that \lemref{lem:frame:oblivious:flip} holds with over adversarial streams, using the same argument as \thmref{thm:framework} and the switch-a-sketch proof by \cite{Ben-EliezerJWY20}. 
Observe that each time either counter $a$ or $b$ either increases or decreases, the modified subroutine $\estimateF$ compels \algref{alg:framework:flip} to use new subroutines $\calB_{i,j}$ and $\calA_i$ that have not previously been revealed to the adversary. 
In particular, since $\estimateF$ sets $j=\flr{\frac{b+1}{2^{z_k}}}+\ceil{\frac{\tau}{2^k}}$, then since the $\eta$ twist number has only reached $\tau$, the value of $j$ ensures that $\calB_{i,j}$ has not been previously revealed. 
Hence, the input is independent of the internal randomness of $\calB_{i,j}$ until the counter $b$ reaches $b=j$ for the $\ceil{\frac{\tau}{2^k}}$-th time, at which point $\calB_{i,j}$ is revealed once and not used again. 

It remains to analyze the space complexity of \algref{alg:framework:flip}. 
We first analyze the space complexity for a $(\gamma,\eps,\delta)$-difference estimator with space dependency $\frac{\gamma^C}{\eps^2}$ and $C>1$. 
Observe that $a=\O{\eps\lambda}$ and across the $\beta$ granularities, each value of $\eta_k$ is invoked at most $\O{\frac{\lambda}{2^k}}$ times. 
Since we require correctness of $\O{\eps\lambda}$ instances of $\calB$, it suffices to set $\delta'=\O{\frac{1}{\poly\left(\lambda,\frac{1}{\eps}\right)}}$. 
By assumption, a single instance of $\calB$ with granularity $\eta_k$ uses at most $\frac{C_1}{(\eta_k)^2} S_F(n,\delta',\eps)$ bits of space, for some constant $C_1>0$.  
Thus for fixed $a$, $\beta\le\log\frac{8}{\eps}+1$, $\gamma_k=2^{k-1}\eta$, and accuracy parameter $\eta_k=\frac{\eta}{\varphi^{\beta-k}}$, the total space for each $\calB_{a,j}$ across the $\beta$ granularities and the values of the counter $a$ is 
\begin{align*}
\sum_{k=1}^\beta\frac{C_1\lambda\gamma_k^C}{\eta_k^2 2^k}\cdot S_F(n,\delta',\eps)&=\sum_{k=1}^\beta\frac{C_1\lambda\gamma_{\beta-k}^C}{\eta_{\beta-k}^2 2^{\beta-k}}\cdot S_F(n,\delta',\eps)\\
&\le\sum_{k=1}^\beta\frac{C_2\lambda(2^{\beta-k}\eps)^C\varphi^{2k}}{\eps^2\cdot2^{\beta-k}}\cdot S_F(n,\delta',\eps)\\
&\le\sum_{k=1}^\beta\frac{C_3\lambda 2^{-Ck}\varphi^{2k}}{\eps\cdot2^{-k}}\cdot S_F(n,\delta',\eps),
\end{align*}
for some absolute constants $C_1,C_2,C_3>0$. 
Since $\varphi=2^{(C-1)/4}$ for $C>1$, then there exists a constant $C_4$ such that
\begin{align*}
\sum_{k=1}^\beta\frac{C_1\lambda\gamma_k^C}{\eta_k^2 2^k}\cdot S_F(n,\delta',\eps)&\le\sum_{k=1}^\beta\frac{C_4\lambda 2^{k-Ck}2^{(Ck-k)/2}}{\eps}\cdot S_F(n,\delta',\eps)\\
&\le\O{\frac{\lambda}{\eps}\cdot S_F(n,\delta',\eps)}. 
\end{align*}
Moreover, since each strong tracker $\calA$ is used when the value of $F$ increases by a constant factor, then $\O{\eps\lambda}$ instances of $\calA$ are required. 
Hence, the total space required for $C>1$ is $\O{\frac{\lambda}{\eps}\cdot S_F(n,\delta',\eps)}$. 

For $C=1$, we recall that $\eta_k=\frac{\eta}{\beta}$, with $\beta\le\log\frac{8}{\eps}+1$. 
Thus the total space is
\begin{align*}
\sum_{k=1}^\beta\frac{C_1\lambda\gamma_k^C}{\eta_k^2 2^k}\cdot S_F(n,\delta',\eps)&\le\sum_{k=1}^\beta\frac{C_5\lambda(2^{k-1}\eps)}{\eta^2\beta^{-2}\cdot 2^k}\cdot S_F(n,\delta',\eps)\\
&\le\sum_{k=1}^\beta\frac{C_6\lambda\beta^2\eps}{\eta^2}\cdot S_F(n,\delta',\eps)\\
&\le\O{\frac{\lambda}{\eps}\log^3\frac{1}{\eps}\cdot S_F(n,\delta',\eps)}. 
\end{align*}
\end{proof}

As a direct application of \thmref{thm:framework:flip}, we obtain the following:
\begin{theorem}
\thmlab{thm:robust:turnstile:Fp}
Given $\eps>0$, there exists an adversarially robust streaming algorithm on turnstile streams with $(\eps,m)$ twist number $\lambda$ that outputs a $(1+\eps)$-approximation for $F_p$ with $p\in[0,2]$ that succeeds with probability at least $\frac{2}{3}$ and uses $\tO{\frac{\lambda}{\eps}\log^2 n}$ bits of space. 
\end{theorem}
\section{Framework for Sliding Window Algorithms}
\seclab{sec:sliding}
In this section, we describe a general framework for norm estimation in the sliding window model, using the sketch stitching and granularity changing technique. 
We first require the following background on sliding windows algorithms.

\begin{definition}[Smooth function]
Given adjacent substreams $A$, $B$, and $C$, a function $f$ is $(\alpha,\beta)$-smooth if $(1-\beta)f(A\cup B)\le f(B)$, then $(1-\alpha)f(A\cup B\cup C)\le f(B\cup C)$ for some parameters $0<\beta\le\alpha<1$. 
\end{definition}
Intuitively, once a suffix of a data stream becomes a $(1\pm\beta)$-approximation for a smooth function, then it is \emph{always} a $(1\pm\alpha)$-approximation, regardless of the subsequent updates that arrive in the stream. 
Smooth functions are integral to the smooth histogram framework for sliding window algorithms. 

\paragraph{Smooth Histograms.}
Braverman and Ostrovsky introduced the \emph{smooth histogram}, an elegant framework that solves a large number of problems in the sliding window model, such as $L_p$ norm estimation, longest increasing subsequence, geometric mean estimation, or other weakly additive functions~\cite{BravermanO07}. 
Thus the smooth histogram data structure maintains a number of timestamps throughout the data stream, along with a streaming algorithm for each timestamp that stores a sketch of all the elements seen from the timestamp. 
The timestamps maintain the invariant that at most three checkpoints produce values that are within $(1-\beta)$ of each other, since any two of the sketches would always output values that are within $(1-\alpha)$ afterwards. 
Hence if the function is polynomially bounded, then the smooth histogram data structure only needs a logarithmic number of timestamps. 

We now define the following variant of a difference estimator. 
\begin{definition}[Suffix-Pivoted Difference Estimator]
\deflab{def:diff:est:suf}
Given a stream $\calS$ and fixed times $t_1$, $t_2$, and $t_3$, let frequency vectors $u$ and $v$ be induced by the updates of $\calS$ between times $[t_1,t_2)$ and $[t_2, t_3)$. 
Given an accuracy parameter $\eps>0$ and a failure probability $\delta\in(0,1)$, a streaming algorithm $\calC(t_1,t_2,t,\gamma,\eps,\delta)$ is a $(\gamma,\eps,\delta)$-\emph{suffix difference estimator} for a function $F$ if, with probability at least $1-\delta$, it outputs an additive $\eps\cdot F(v+w_t)$ approximation to $F(u+v+w_t)-F(v+w_t)$ for all frequency vectors $w_t$ induced by $[t_3,t)$ for times $t>t_3$, given $\min(F(u), F(u+v)-F(v))\le\gamma\cdot F(v)$ for a ratio parameter $\gamma\in(0,1]$. 
\end{definition}
Observe that the difference between \defref{def:diff:est:pre} and \defref{def:diff:est:suf} is that the fixed-prefix difference estimator approximates $F(v+w_t)-F(v)$ when the contribution to $F$ of the first frequency vector $v$ that arrives in the stream is much larger than that of $w_t$, while the suffix-pivoted difference estimator approximates $F(v+w_t)-F(w_t)$ when $F(w_t)$ is larger than $F(v)$. 
Nevertheless, we prove that our fixed-prefix difference estimators for $F_p$ moment estimation with $p\in(0,2]$ can also be adjusted to form suffix-pivoted difference estimators. 
Unless otherwise noted, the difference estimators in this section refer to the suffix-pivoted difference estimators rather than the fixed-prefix difference estimators. 

We adapt our sketch stitching and granularity changing technique to the sliding window model by focusing on the suffix of the stream, since prefixes of the stream may expire. 
We thus run the highest accuracy algorithms, the separate streaming algorithms $\calA$, on various suffixes of the stream similar to the smooth histogram framework. 
It follows from smoothness that we maintain an instance of $\calA$ starting at some time $t_0\le m-W+1$, whose output is within a factor $2$ of the value of $F$ on the sliding window. 
Our task is then to remove the extraneous contribution of the updates between times $t_0$ and $m-W+1$, i.e., the starting time of the sliding window. 

We partition the substream of these updates into separate blocks based on their contribution to the value of $F$ by guessing $\O{\log n}$ values for the final value of $F$ on the sliding window and forming new difference estimators at level $j$ when the value of $F$ on each block has exceeded a $\frac{1}{2^j}$ fraction of the corresponding guess. 
We terminate a guess when there are more than $100\cdot 2^j$ blocks in that level, indicating that the guess is too low. 
%Since $F$ is sub-additive, the contribution of each block will also upper bound their contributions to the difference.  
We can maintain separate sketches for these blocks, with varying granularities, and stitch these sketches together at the end. 
We give our algorithm in full in \algref{alg:sliding}.

\paragraph{Interpretation of \algref{alg:sliding}.} 
We now translate between the previous intuition and the pseudocode of \algref{alg:sliding}. 
At each time, \algref{alg:sliding} only runs two subroutines: $\guessupdatesw$ and $\mergesw$. 
The first subroutine $\guessupdatesw$ creates new instances of each algorithm (both a streaming algorithm approximating $F$ on a suffix of the stream and a difference estimator starting at each time) at each time in the stream. 
Moreover, $\guessupdatesw$ partitions the stream into blocks for the difference estimator by using exponentially increasing guesses for the value of $F$ at the end of the stream to assist with appropriate granularity for each block.  
The second subroutine $\mergesw$ performs maintenance on the data structure to ensure that there are not too many instances that have been created by the first subroutine that are simultaneously running. 
Namely, $\mergesw$ deletes algorithms running on suffixes of the stream that output a similar value, so that the number of remaining algorithms is logarithmic rather than linear. 
Similarly, $\mergesw$ merges two difference estimators when it is clear their combined contribution is still too small. 
Finally at the end of the stream, $\stitchsw$ creates an estimate for the value of $F$ on the stream by stitching together the estimates output by each difference estimator. 
Although \algref{alg:framework} is notationally heavy, each timestamp $t^{(k)}_{i,j,\ell}$ should be associated with (1) a guess $k\in[C\log n]$ for the value of $F$ at the end of the stream, (2) an index $i\in\O{\log n}$ roughly associated with the number of times $F$ has \emph{actually} doubled in the stream so far, (3) a granularity $j$, and (4) the number of the block $\ell$ in granularity $j$.  

\begin{algorithm}[!htb]
\caption{Moment Estimation in the Sliding Window Model}
\alglab{alg:sliding}
\begin{algorithmic}[1]
\Require{Stream $u_1,\ldots,u_m$ of updates, an $(\eps,\eps^q)$-smooth function $F$, accuracy parameter $\eps\in(0,1)$, window parameter $W>0$}
\Ensure{Robust $(1+\eps)$-approximation to $F$}
\State{$\delta\gets\frac{1}{\poly(m)}$, $\eta\gets\frac{\eps}{2^{20}q\log\frac{1}{\eps}}$, $\varphi\gets\sqrt{2}$ be a parameter}
\State{$\beta\gets\ceil{\log\frac{100\cdot 4^q}{\eps^q}}$, $\gamma_j\gets 2^{3-j}$ for all $j\in[\beta]$}
\For{each update $u_t\in[n]$, $t\in[m]$}
\State{$\guessupdatesw$}
\Comment{Create new subroutines for each update}
\State{$\mergesw$}
\Comment{Removes extraneous subroutines}
\EndFor
\State{\Return $Z\gets\stitchsw$}
\Comment{Estimate $F$ on sliding window}
\end{algorithmic}
\end{algorithm}

\begin{algorithm}[!htb]
\caption{Subroutine $\guessupdatesw$ of \algref{alg:sliding}: create new subroutines for each update}
\alglab{alg:updatesw:frame}
\begin{algorithmic}[1]
\State{Let $s$ be the number of instances of $\calA$.}
\State{$t_{s+1}\gets t$}
\State{Start a new instance $\calA(t_{s+1},t,\eta,\delta)$.}
\For{$j\in[\beta]$}
\Comment{Maintain instances of each granularity}
\For{$k\in[C\log n]$}
\Comment{$n^C$ is upper bound on value of $F$}
\State{Let $r_k$ be the number of instances of timestamps $t^{(k)}_{s+1,j,*}$.}
\If{$\calA(t^{(k)}_{s+1,j,r_k},t-1,1,\delta)\in[n^C/2^{j+k+11},n^C/2^{j+k+10}]$, $\calA(t^{(k)}_{s+1,j,r_k},t,1,\delta)>n^C/2^{j+k+10}$, and $r<100\cdot 2^{j+10}$}
\For{$\ell>k$}
\State{$t^{(\ell)}_{s+1,j,r_{\ell+1}}\gets t$ .}
\State{Demarcate $\sdiffest(t^{(\ell)}_{s+1,j,r_\ell},t^{(\ell)}_{s+1,j,r_{\ell+1}},t,\gamma_j,\eta,\delta)$.}
\Comment{Update splitting time}
\State{Start a new instance $\calA(t^{(\ell)}_{s+1,j,r_{\ell+1}},t,1,\delta)$.}
\EndFor
\EndIf
\EndFor
\EndFor
\end{algorithmic}
\end{algorithm}

\begin{algorithm}[!htb]
\caption{Subroutine $\mergesw$ of \algref{alg:sliding}: removes extraneous subroutines}
\alglab{alg:mergesw:frame}
\begin{algorithmic}[1]
\State{Let $s$ be the number of instances of $\calA$.}
\For{$i\in[s]$, $j\in[\beta]$, and $k\in[C\log n]$}
\Comment{Difference estimator maintenance}
\State{Let $r_k$ be the number of instances of timestamps $t^{(k)}_{s+1,j,*}$.}
\For{$\ell\in[r_k-1]$}
\Comment{Merges two algorithms with ``small'' contributions}
\If{$\calA(t^{(k)}_{i,j,\ell-1},t^{(k)}_{i,j,\ell+1},1,\delta)\le n^C/2^{k+j+10}$}
\State{Merge (add) the sketches for $\calA(t_{i,j,k-1},t_{i,j,k},1,\delta)$ and $\calA(t_{i,j,k},t_{i,j,k+1},1,\delta)$.}
\State{Merge (add) the sketches for $\sdiffest(t_{i,j,k-1},t_{i,j,k},t,\gamma_j,\eta,\delta)$ and $\sdiffest(t_{i,j,k},t_{i,j,k+1},t,\gamma_j,\eta,\delta)$.}
\State{Relabel the times $t_{i,j,*}$.}
\EndIf
\EndFor
\For{$i\in[s-2]$}
\Comment{Smooth histogram maintenance}
\If{$\calA(t_{i+2},t,\eta,\delta)\ge(1-1/8^q)\calA(t_i,t,\eta,\delta)$}
\For{$j\in[\beta]$ and $k\in[C\log n]$}
\State{Append the times $t^{(k)}_{i+1,j,*}$ to $\{t^{(k)}_{i,j,*}\}$.}
\EndFor
\State{Delete $t_{i+1}$ and all times $t_{i+1,*,*}$.}
\State{Relabel the times $\{t_i\}$ and $\{t^{(k)}_{i,j,*}\}$.}
\EndIf
\EndFor
\EndFor
\end{algorithmic}
\end{algorithm}

\begin{algorithm}[!htb]
\caption{Subroutine $\stitchsw$ of \algref{alg:sliding}: output estimate of $F_p$ on the sliding window}
\alglab{alg:stitchsw:frame}
\begin{algorithmic}[1]
\State{Let $i$ be the largest index such that $t_i\le m-W+1$.}
\State{Let $k$ be the smallest integer such that $n^C/2^k\le\calA(t_i,m,1,\delta)$.}
\State{$c_0\gets t_i$}
\State{$X\gets\calA(t_i,m,\eta,\delta)$}
\For{$j\in[\beta]$}
\Comment{Stitch sketches}
\State{Let $a$ be the smallest index such that $t^{(k)}_{i,j,a}\ge c_{j-1}$}
\State{Let $b$ be the largest index such that $t^{(k)}_{i,j,b}\le m-W+1$}
\State{$c_j\gets t^{(k)}_{i,j,b}$}
\State{$Y_j\gets\sum_{k=a}^{b-1}\sdiffest(t^{(k)}_{i,j,k},t^{(k)}_{i,j,k+1},m,\gamma_j,\eta,\delta)$}
\EndFor
\State{\Return $Z:=X-\sum_{j=1}^\beta Y_j$}
\end{algorithmic}
\end{algorithm}

Recall that we abuse notation so that $W$ represents the active updates in the stream. 
Thus we use $F_p(W)$ to denote the $F_p$ moment of the active elements in the sliding window, i.e., if $m<W$, then $F_p(W)=F_p(1:m)$ and if $m\ge W$, then $F_p(W)=F_p(1:m-W+1)$.  
We first show that the timestamps $t_1,\ldots,t_s$ mark suffixes of the stream in which the moment roughly double, so that $F_p(t_1,m)\ge\ldots\ge F_p(t_s,m)$ and $2^{s-i}\cdot F_p(t_s,m)\ge F_p(t_i,m)$ for each $i\in[s]$. 
Since $W$ is oblivious to the algorithim until the $\stitchsw$ subroutine, it suffices to show that if $t_i\le m-W+1\le t_{i+1}$, then $F_p(W)\le F_p(t_i,m)\le 2F_p(W)$ and $\frac{1}{2}F_p(W)\le F_p(t_{i+1},m)\le F_p(W)$. 

We first show that the top level gives a constant factor approximation to $F(W)$. 
The intuition behind the lemma is that the maintenance procedure $\mergesw$ only deletes an instance of an algorithm if the instance is sandwiched between two other algorithms whose values are within a factor of $(1-1/8^q)$, where the function is $(\eps,\eps^q)$-smooth. 
Thus by the definition of smoothness, we should expect the other two algorithms to always produce values that are within a factor of two, regardless of the subsequent updates of the stream. 
This is important because it allows us to translate between the values of various suffixes of the stream, e.g., $F(t_i,m)\le 2F(W)$ implies that additive $\frac{\eps}{2}\cdot F(t_i,m)$ error gives a $(1+\eps)$-multiplicative error to $F(W)$. 
\begin{lemma}[Constant factor partitions in top level]
\lemlab{lem:sw:top:constant}
Let $i$ be the largest index such that $t_i\le m-W+1$. 
Let $k$ be the smallest integer such that $n^C/2^k\le\calA(t_i,m,1,\delta)$. 
Let $\calE$ be the event that all subroutines $\calA$ in \algref{alg:sliding} succeed.
Then conditioned on $\calE$, $F(W)\le F(t_i,m)\le 2F(W)$ and $\frac{1}{2}F(W)\le F(t_{i+1},m)\le F(W)$. 
\end{lemma}
\begin{proof}
Observe that in the $\mergesw$ subroutine, we only delete a timestamp $t_{i'}$ if $\calA(t_{i'+2},t,\eta,\delta)\ge(1-1/8^q)\calA(t_{i'},t,\eta,\delta)$ at some point $t$. 
Since $\eta=\frac{\eps}{2^{20}q\log\frac{1}{\eps}}\le\frac{1}{1024}$, we have that conditioned on $\calE$, $F(t_{i'+2},t)\left(1+\frac{1}{1024}\right)\ge(1-1/8^q)\left(1-\frac{1}{1024}\right)F(t_{i'},t)$. 
Thus, $F(t_{i'+2},t)\ge(1-1/4^q)F(t_{i'},t)$. 

We have either $t_i=m-W+1$ and $t_{i+1}=m-W+2$ in which case the statement is trivially true or timestamps $t_i$ and $t_{i+1}$ at some point $t\le m$ must have satisfied $F(t_{i+1},t)\ge(1-1/4^q)F(t_i,t)$ so that $t_i+1$ was removed from the set of timestamps. 
Thus, $F(t_{i+1},t)\ge(1-1/4^q)F(t_i,t)$ and recall that we assume $F$ is $(\eps,\eps^q)$-smooth. 
By the definition of smoothness, $F(t_{i+1},t)\ge(1-1/4^q)F(t_i,t)$ implies $F(t_{i+1},m)\ge\frac{1}{2}\cdot F(t_i,m)$ for all $m\ge t$. 
Since $t_i\le m-W+1$, then $F(t_i,m)\ge F(m-W+1,m)=F(W)\ge F(t_{i+1},m)$ and the conclusion follows. 
\end{proof}

We now show that our difference estimators are well-defined and thus give good approximations to the differences that they estimate. 
Although this sounds trivial, we remark that a pre-condition to a $(\gamma,\eps,\delta)$-difference estimator to succeed on the approximation of a difference $F(v+u_1)-F(u_1)$, we must have $F(v+u_1)-F(u_1)\le\gamma\cdot F(v+u_1)$, which requires proving properties about the specific partition of the stream. 
On the other hand, since $\guessupdatesw$ partitions the stream into blocks for each difference estimator based on their contribution, we show that the pre-condition holds by construction.  
\begin{lemma}[Accuracy of difference estimators]
\lemlab{lem:diff:est:well:defined}
Let $i$ be the largest index such that $t_i\le m-W+1$. 
Let $k$ be the smallest integer such that $n^C/2^k\le\calA(t_i,m,1,\delta)$.
Let $\calE$ be the event that all subroutines $\calA$ and $\sdiffest$ in \algref{alg:sliding} succeed.
Then conditioned on $\calE$, we have that for each $j\in[\beta]$ and all $\ell$, $\sdiffest(t^{(k)}_{i,j,\ell-1},t_{i,j,\ell},t,\gamma_j,\eta,\delta)$ gives an additive $\eta\cdot F(t^{(k)}_{i,j,\ell}:t)$ approximation to $F(t^{(k)}_{i,j,\ell-1}:t)-F(t^{(k)}_{i,j,\ell}:t)$. 
\end{lemma}
\begin{proof}
Recall that for a suffix-pivoted difference estimator to be well-defined, we require $F(t^{(k)}_{i,j,\ell-1}:t^{(k)}_{i,j,\ell})\le\gamma_j\cdot F(t^{(k)}_{i,j,\ell}:x)$ 
% or $F(t^{(k)}_{i,j,\ell-1}:x)-F(t^{(k)}_{i,j,\ell-1}:t_{i,j,\ell})\le\gamma_j\cdot F_p(t_{i,j,\ell}:x)$ 
at some time $x$. 
Since $\gamma_j=2^{3-j}$ and $\guessupdatesw$ creates a new partition at level $j$ when $\calA(t^{(k)}_{s+1,j,r_k},t-1,1,\delta)\le n^C/2^{j+k+10}$, then it follows that $F(t^{(k)}_{i,j,\ell-1}:t^{(k)}_{i,j,\ell})\le\gamma_j\cdot F(t^{(k)}_{i,j,\ell}:x)$. 
Thus the conditions for a suffix-pivoted difference estimator hold, so $\sdiffest(t^{(k)}_{i,j,\ell-1},t^{(k)}_{i,j,\ell},m,\gamma_j,\eta,\delta)$ gives an additive $\eta\cdot F(t^{(k)}_{i,j,\ell}:t)$ approximation to $F(t^{(k)}_{i,j,\ell-1}:m)-F(t^{(k)}_{i,j,\ell}:m)$ by the guarantees of the difference estimator. 
\end{proof}

We now lower bound the values of each difference. 
This is important for bounding the number of difference estimators that are used to stitch together the estimate at the end of the stream. 
For example, without a lower bound on the difference estimators at each level, it could be possible that we use a linear number of difference estimators to form our estimate of $F(W)$, in which case the space would be too large. 
Using the following lower bound on the difference estimator, we ultimately show only a logarithmic number of difference estimators is used to form our estimate of $F(W)$. 
The intuition behind the proof of \lemref{lem:sliding:delta:lower} is that if the contributions of the difference estimators of two consecutive blocks are \emph{both} small, then $\mergesw$ will end up merging the blocks. 
Hence, we should expect the contribution of any two consecutive blocks to be sufficiently large. 
\begin{lemma}[Geometric lower bounds on splitting times]
\lemlab{lem:sliding:delta:lower}
Let $i$ be the largest index such that $t_i\le m-W+1$. 
Let $k$ be the smallest integer such that $n^C/2^k\le\calA(t_i,m,1,\delta)$. 
Let $\calE$ be the event that all subroutines $\calA$ and $\sdiffest$ in \algref{alg:sliding} succeed. 
For any $\ell\in[r_k-2]$, let $u$ denote the frequency vector between times $t^{(k)}_{i,j,\ell}$ and $t^{(k)}_{i,j,\ell+2}$ and let $w$ denote the frequency vector between times $t^{(k)}_{i,j,\ell+2}$ and $t$. 
Finally, let $w$ denote the frequency vector between times $t_i$ and $t$. 
Then conditioned on $\calE$, we have that
\[F(u+v)-F(v)\ge 2^{-j-8}\cdot F(w).\]
\end{lemma}
\begin{proof}
Let $u_1$ denote the frequency vector between times $t^{(k)}_{i,j,\ell}$ and $t^{(k)}_{i,j,\ell+1}$ and $u_2$ denote the frequency vector between times $t^{(k)}_{i,j,\ell+1}$ and $t^{(k)}_{i,j,\ell+2}$. 
Recall that $\mergesw$ merges two instances of $\sdiffest$ when 
\begin{align*}
\sdiffest(t^{(k)}_{i,j,\ell-1},t^{(k)}_{i,j,\ell},t,\gamma_j,\eta,\delta)+\sdiffest(t^{(k)}_{i,j,\ell},t^{(k)}_{i,j,\ell+1},t,\gamma_j,\eta,\delta)\le 2^{-j-10}\cdot\calA_{i}(t_{i},t,\eta,\delta). 
\end{align*}
Since we do not remove $t^{(k)}_{i,j,\ell+1}$ from the list of timestamps, then we have
\[\sdiffest(t^{(k)}_{i,j,\ell-1},t^{(k)}_{i,j,\ell},t,\gamma_j,\eta,\delta)+\sdiffest(t^{(k)}_{i,j,\ell},t^{(k)}_{i,j,\ell+1},t,\gamma_j,\eta,\delta)>2^{-j-10}\cdot\calA_{i}(t_{i},t,\eta,\delta).\]
Since the difference estimators are well-defined by \lemref{lem:diff:est:well:defined}, then $\sdiffest(t^{(k)}_{i,j,\ell},t^{(k)}_{i,j,\ell+1},t,\gamma_j,\eta,\delta)$ is at most a $2$-approximation to $F(u_1+u_2+v)-F(u_2+v)$ and we similarly obtain a $2$-approximation to $F(u_2+v)-F(v)$. 
Therefore, 
\[F(u)-F(v)=F(u_1+u_2+v)-F(v)\ge 2^{-j-8}\cdot F(w).\] 
\end{proof}
\noindent
We next bound the number of level $j$ difference estimators that can occur from the end of the previous level $j-1$ difference estimator to the time when the sliding window begins. 
We say a difference estimator $\sdiffest$ is active if $\ell\in[a_j,b_j]$ for the indices $a_j$ and $b_j$ defined in \algref{alg:sliding}. 
The active difference estimators in each level will be the algorithms whose output are subtracted from the initial rough estimate to form the final estimate of $F_2(W)$. 
The intuition behind the proof of \lemref{lem:sliding:active:indices} is that \lemref{lem:sliding:delta:lower} lower bounds the contributions of the difference estimators of two consecutive blocks. 
Hence there cannot be too many active difference estimators in a level or else their contributions will be too large for the level. 
\begin{lemma}[Number of active level $j$ difference estimators]
\lemlab{lem:sliding:active:indices}
Let $i$ be the largest index such that $t_i\le m-W+1$.  
Let $k$ be the smallest integer such that $n^C/2^k\le\calA(t_i,m,1,\delta)$. 
Let $\calE$ be the event that all subroutines $\calA$ and $\sdiffest$ in \algref{alg:sliding} succeed. 
For each $j\in[\beta]$, let $a_j$ be the smallest index such that $t_{i,j,a_j}\ge c_{j-1}$ and let $b_j$ be the largest index such that $t_{i,j,b_j}\le m-W+1$. 
Then conditioned on $\calE$, we have $b_j-a_j\le 512$. 
\end{lemma}
\begin{proof}
Conditioned on $\calE$, we have that for each $j\in[\beta]$ and all $\ell$, 
\[2^{-j-8}\cdot F(W)<F(t^{(k)}_{i,j,\ell},t)-F(t^{(k)}_{i,j,\ell+2},t),\]
by \lemref{lem:sliding:delta:lower}. 
By \lemref{lem:sw:top:constant}, $F_p(t_i,m)\le 2F_p(W)$. 
Thus by a telescoping argument, we have that for $j=1$, $a_j-b_j<512$. 

For $j>1$, suppose by way of contradiction that $b_j-a_j\ge 512$. 
Let $k=b_j-a_j$ and for $x\in[k]$, let $u_x$ be the frequency vector induced by the updates of the stream from time $a_j+{x-1}$ to $a_j+x$. 
Let $u$ be the frequency vector induced by the updates of the stream from time $t_{i,j,k}$ to $m$ and $v$ be the frequency vector induced by the updates from $t_i$ to $m$. 
By \lemref{lem:sliding:delta:lower},
\[F\left(u+\sum_{\ell=1}^k u_{\ell}\right)-F(u)\ge8\cdot2^{-j-8}\cdot F(v),\]
since there are at least $8$ disjoint pairs of tuples $(i,j,\ell)$ and $(i,j,\ell+1)$ if $k>512$. 

Thus if $b_j-a_j>512$, then the sum of the outputs of the active difference estimators at level $j$ is more than $8\cdot2^{-j-8}\cdot F(v)$. 
In particular since $t^{(k)}_{i,j,a_j}\ge c_{j-1}$, the sum of the outputs of the active difference estimators at level $j$ after $c_{j-1}$ is at least $8\cdot2^{-j-8}\cdot F(v)$. 
However, by \lemref{lem:sliding:delta:lower}, each difference estimator at level $j-1$ has output at least $2^{-j-7}\cdot F(u)>8\cdot 2^{-j-8}\cdot F(v)$. 
Specifically, the difference estimator from times $t^{(k)}_{i,j,c_{j-1}}$ to $t^{(k)}_{i,j,c_{j}}$ must have output at least $2^{-j-7}\cdot F(u)>8\cdot 2^{-j-8}\cdot F(v)$. 
Therefore, there exists some other $z>c_{j-1}$ such that $t^{(k)}_{i,j-1,z}\le m-W+1$, which contradicts the maximality of $c_{j-1}$ at level $j-1$. 
\end{proof}
We next upper bound the value of each difference estimator. 
This is important for upper bounding the error associated with the sliding window possibly not beginning exactly where a difference estimator begins. 
For example, if a difference estimator exactly computes $F(1:10)-F(4:10)$ but the sliding window consists of the last eight updates, i.e., $F(3:10)$, then the error of our estimator is still only lower bounded by the difference. 
Thus it is crucial to not only upper bound the error of the difference estimator, but also the value of the difference estimator.  
The intuition behind the proof of \lemref{lem:sliding:delta:upper} is that the subroutine $\mergesw$ only creates a new block for each difference estimator when the contribution is in a particular range. 
Specifically, the contribution cannot be too large or else a new block would have already been formed. 
Then by smoothness of the function, the contribution of difference estimator cannot grow to be too large regardless of the subsequent updates in the stream. 
\begin{lemma}[Geometric upper bounds on splitting times]
\lemlab{lem:sliding:delta:upper}
Let $i$ be the largest index such that $t_i\le m-W+1$. 
Let $k$ be the smallest integer such that $n^C/2^k\le\calA(t_i,m,1,\delta)$. 
Let $\calE$ be the event that all subroutines $\calA$ and $\sdiffest$ in \algref{alg:sliding} succeed. 
Then conditioned on $\calE$, we have that for each $j\in[\beta]$ and each $\ell$ that either $t^{(k)}_{i,j,\ell+1}=t^{(k)}_{i,j,\ell}+1$ or $F(t^{(k)}_{i,j,\ell},t)-F(t^{(k)}_{i,j,\ell+1},t)\le 2^{-j/q-7/q}\cdot F(t^{(k)}_{i,j,\ell+1},t)$. 
\end{lemma}
\begin{proof}
Suppose $t^{(k)}_{i,j,\ell+1}\neq t^{(k)}_{i,j,\ell}+1$. 
Then at some time $x$, the timestamp $t^{(k)}_{i,j,\ell}+1$ must have been merged by subroutine $\mergesw$. 
Thus 
\[\sdiffest(t^{(k)}_{i,j,\ell},t^{(k)}_{i,j,\ell'},x,\gamma_j,\eta,\delta)+\sdiffest(t^{(k)}_{i,j,\ell'},t^{(k)}_{i,j,\ell+1},x,\gamma_j,\eta,\delta)\le 2^{-j-10}\cdot\calA_{i}(t_{i},x,\eta,\delta).\]
Let $u_1$ be frequency vector representing the updates from time $t^{(k)}_{i,j,\ell}$ to $t^{(k)}_{i,j,\ell'}-1$, $u_2$ represent times $t^{(k)}_{i,j,\ell'}$ to $t^{(k)}_{i,j,\ell+1}$, and $u$ represent times $t^{(k)}_{i,j,\ell}$ to $t^{(k)}_{i,j,\ell+1}$, so that $u_1+u_2=u$. 
Let $x_1$ represent the updates between time $t^{(k)}_{i,j,\ell+1}$ and $x$ and let $x_2$ represent the updates between time $x+1$ and $t$. 
Since the difference estimators are well-defined by \lemref{lem:diff:est:well:defined}, we have that $\sdiffest(t^{(k)}_{i,j,\ell},t^{(k)}_{i,j,\ell'},x,\gamma_j,\eta,\delta)$ is at most a $2$-approximation to $F(u_1+u_2+x_1)-F(u_2+x_1)$ and we similarly obtain a $2$-approximation to $F(u_2+x_1)-F(x_1)$. 
Thus, 
\[F(u+x_1)-F(x_1)=F(u_1+u_2+x_1)-F(x_1)\le 2^{-j-7}\cdot F(x_1).\]
By the $(\eps,\eps^q)$ smoothness of $F$, we thus have that for any vector $x_2$,
\[F(u+v)-F(v)=F(u+x_1+x_2)-F(x_1+x_2)\le 2^{-j/q-7/q}\cdot F(x_1+x_2)=2^{-j/q-7/q}\cdot F_p(v).\]
\end{proof}
Finally, we show the correctness of \algref{alg:sliding}. 
The main intuition behind \lemref{lem:sw:frame:correct} is that there are two sources of error. 
The first source of error originates from the boundaries of the blocks corresponding to the difference estimator not aligning with the beginning of the sliding window. 
This error, resulting from no difference estimator being assigned to compute the exactly correct value, corresponds to the elements marked by blue in \figref{fig:sliding} and cannot be accounted for even if all difference estimators have zero error.  
On the other hand, this error is upper bounded by the contribution of a difference estimator at the bottom level, which is bounded by \lemref{lem:sliding:delta:upper}. 
The second source of error stems from the approximation error caused by each of the difference estimators. 
Since \lemref{lem:sliding:active:indices} bounds the total number of difference estimators being used in our output, we can also upper bound the total approximation error due to the difference estimators. 
\begin{lemma}[Correctness of framework]
\lemlab{lem:sw:frame:correct}
With high probability, \algref{alg:sliding} gives a $(1+\eps)$-approximation to the value of $F(W)$. 
\end{lemma}
\begin{proof}
Let $\calE$ be the event that all algorithms throughout the stream succeed and note that $\PPr{\calE}\ge1-\frac{1}{\poly(n)}$ since $\delta=\frac{1}{\poly(m)}$. 
Thus we condition on $\calE$ throughout the remainder of the proof. 

Let $f$ be the frequency vector induced by the window and let $i$ be the largest index such that $t_i\le m-W+1$ and let $c_0=t_i$. 
Let $u$ be the frequency vector induced by the updates of the stream from time $t_i$ to $m$. 
Let $k$ be the smallest integer such that $n^C/2^k\le\calA(t_i,m,1,\delta)$. 
For each $j\in[\beta]$, let $a_j$ be the smallest index such that $t^{(k)}_{i,j,a_j}\ge c_{i-1}$ and let $b_j$ be the largest index such that $t^{(k)}_{i,j,b_j}\le m-W+1$. 
Note that by construction of the $\guessupdatesw$ forming a new timestamp for all levels $\ell>k$ each time a new timestamp at level $k$ is formed, then it follows that $t_{i,j,a_j}=c_{i-1}$. 
Let $u_j$ be the frequency vector induced by the updates of the stream from time $a_j$ to $b_j$. 
Let $v$ be the frequency vector induced by the updates of the stream from $c_\beta$ to $m-W+1$. 
Thus, we have $u=\sum_{j=1}^\beta u_j+v+f$, so it remains to show that:
\begin{enumerate}
\item
$F(v+f)-F(f)\le\frac{\eps}{2}\cdot F(f)$
\item
We have an additive $\frac{\eps}{2}\cdot F(f)$ approximation to $F(v+f)$. 
\end{enumerate}
Combined, these two statements show that we have a multiplicative $(1+\eps)$-approximation of $F_p(f)$. 

To show that $F(v+f)-F(f)\le\frac{\eps}{2}\cdot F(f)$, let $\ell$ be the index such that $t^{(k)}_{i,\beta,\ell}=c_\beta$, so that $t^{(k)}_{i,\beta,\ell}\le m-W+1\le t^{(k)}_{i,\beta,\ell+1}$. 
Note that if $t^{(k)}_{i,\beta,\ell}=m-W+1$ then $v=0$ and so the claim is trivially true. 
Otherwise by the definition of the choice of $c_\beta$, we have that $t^{(k)}_{i,\beta,\ell+1}>m-W+1$ so that $t^{(k)}_{i,\beta,\ell+1}>t^{(k)}_{i,\beta,\ell}+1$. 

By \lemref{lem:sliding:delta:upper}, it follows that $F(t^{(k)}_{i,\beta,\ell},m)-F(t^{(k)}_{i,\beta,\ell+1},m)\le 2^{-\beta/q}\cdot F_p(t_{i,\beta,k+1},m)$. 
Since $\beta=\ceil{\log\frac{100\cdot 4^q}{\eps^q}}$, then it follows that $F(t^{(k)}_{i,\beta,\ell},m)-F(t^{(k)}_{i,\beta,\ell+1},m)\le\frac{\eps}{4}\cdot F(t^{(k)}_{i,\beta,\ell+1},m)$. 
By \lemref{lem:sw:top:constant}, we have that $\frac{\eps}{4}\cdot F(t^{(k)}_{i,\beta,\ell+1},m)\le 2F(W)=2F(f)$. 
Thus, $F(t^{(k)}_{i,\beta,\ell},m)-F(t^{(k)}_{i,\beta,\ell+1},m)\le\frac{\eps}{2}\cdot F(f)$. 
Since the updates from the times between $t^{(k)}_{i,\beta,\ell}$ and $m$ form the vector $v+f$, we have
$F(v+f)-F(t^{(k)}_{i,\beta,\ell+1},m)\le\frac{\eps}{2}\cdot F(f)$. 
Since $t^{(k)}_{i,\beta,\ell+1}>m-W+1$ then by the monotonicity of $F$, we have that $F(v+f)-F(f)\le\frac{\eps}{2}\cdot F(f)$, as desired.

We next show that we have an additive $\frac{\eps}{2}\cdot F(f)$ approximation to $F(v+f)$.
If we define 
\[\Delta_j:=F\left(u-\sum_{k=1}^{j-1} u_k\right)-F\left(u-\sum_{k=1}^j u_k\right),\]
then we have
\[\sum_{j=1}^\beta\Delta_j=F(u)-F(v+f).\]

Observe that $\Delta_j$ is approximated by the active level $j$ difference estimators. 
By \lemref{lem:sliding:active:indices}, there are at most $512$ active indices at level $j$. 
Since the difference estimators are well-defined by \lemref{lem:diff:est:well:defined} and each difference estimator $\sdiffest$ uses accuracy parameter $\eta=\frac{\eps}{2^{20}q\log\frac{1}{\eps}}$, then the additive error in the estimation of $F(v+f)$ incurred by $Y_j$ is at most 
\[512\cdot\frac{\eps}{2^{20}q\log\frac{1}{\eps}}\cdot F(u)=\frac{\eps}{2^{11}q\log\frac{1}{\eps}}\cdot F(u).\]
Summing across all $\beta\le2^9q\log\frac{1}{\eps}$ levels, then the total error in the estimation $Z$ of $F(v+f)$ across all levels $Y_j$ with $j\in[\beta]$ is at most 
\[\frac{\eps}{4}\cdot F(u)\le\frac{\eps}{2}\cdot F(f).\]
Thus, we have an additive $\frac{\eps}{2}\cdot F(f)$ approximation to $F(v+f)$ as desired. 
Since $F(v+f)-F(f)\le\frac{\eps}{2}\cdot F(f)$, then \algref{alg:sliding} outputs a $(1+\eps)$-approximation to $F(W)$. 
\end{proof}

\begin{theorem}[Framework for sliding window algorithms]
\thmlab{thm:sw:framework}
Let $\eps,\delta\in(0,1)$ be constants and $F$ be a monotonic and polynomially bounded function that is $(\eps,\eps^q)$-smooth for some constant $q\ge 0$. 
Suppose there exists a $(\gamma,\eps,\delta)$-suffix pivoted difference estimator that uses space $\frac{\gamma}{\eps^2} S_F(m,\delta,\eps)$ and a streaming algorithm for $F$ that uses space $\frac{1}{\eps^2} S_F(m,\delta,\eps)$, where $S_F$ is a monotonic function in $m$, $\frac{1}{\delta}$, and $\frac{1}{\eps}$. 
Then there exists a sliding window algorithm that outputs a $(1+\eps)$ approximation to $F$ that succeeds with constant probability and uses $\frac{1}{\eps^2}\cdot S_F(m,\delta',\eps)\cdot\poly\left(\log m,\log\frac{1}{\eps}\right)$ space, where $\delta'=\O{\frac{1}{\poly(m)}}$. 
\end{theorem}
\begin{proof}
We first observe that correctness follows from \lemref{lem:sw:frame:correct} for $\delta'=\O{\frac{1}{\poly(m)}}$. 
To analyze the space complexity, note that we maintain $\O{\log n}$ guesses for the value of $F$ at the end of the stream, since $F$ is monotonic and polynomially bounded. 
For each guess, we maintain $\O{\log\frac{1}{\eps}}$ granularities whose space complexity in total forms a geometric series that sums to $S_F(m,\delta',\eps)\cdot\polylog\frac{1}{\eps}$. 
Finally, we repeat this procedure for each of the $\O{\log n}$ indices $i$ corresponding to the timestamps maintained at the top level by the smooth histogram, again because $F$ is monotonic and polynomially bounded. 
Therefore, the total space used is $\frac{1}{\eps^2}\cdot S_F(m,\delta',\eps)\cdot\poly\left(\log m,\log\frac{1}{\eps}\right)$. 
\end{proof}

\subsection{Moment Estimation for \texorpdfstring{$p\in(0,2]$}{p in (0,2]}}
\seclab{sec:sliding:smallp}
To improve \algref{alg:sliding} for $F_p$ moment estimation with $p\in(0,2]$, we remove the additional $\O{\log n}$ overhead associated with making $\O{\log n}$ guesses for the value of $F$ at the end of the stream. 
Instead, we note that since we maintain a constant factor approximation to the value of $F_p$ due to the smooth histogram, it suffices to partition the substream into blocks for the difference estimators based on the ratio of the difference to the constant factor approximation to the value of $F_p$. 
We also recall the following useful characterization of the smoothness of the $F_p$ moment function. 
\begin{lemma}
\lemlab{lem:fp:smooth}
\cite{BravermanO07}
The $F_p$ function is $\left(\eps,\frac{\eps^p}{p}\right)$-smooth for $p\ge 1$ and $(\eps,\eps)$-smooth for $0<p\le 1$.  
\end{lemma}
Moreover, we require constructions of suffix-pivoted difference estimators for $F_p$ moment estimation with $p\in(0,2]$. 
However, we claim that the previous constructions, i.e., the fixed-prefix difference estimator based on Li's geometric estimator for $p\in(0,2)$ and the fixed-prefix difference estimator based on the inner product sketch for $p=2$ are already valid suffix-pivoted difference estimators. 
This is because the key property to approximating the difference $F_p(v+u_1)-F_p(u_1)$ in ``small'' space is that $F_p(v)$ is small. 
We can express the variance of both Li's geometric estimator and the inner product sketch in terms of $F_p(v)$ so that smaller values of $F_p(v)$ correspond to smaller variance for the estimators. 
Hence even if $F_p(v+u_2)-F_p(u_2)$ is much larger than $F_p(v+u_1)-F_p(u_1)$ for some $u_2\succeq u_1$, as long as $F_p(v+u_1)-F_p(u_1)\le\gamma\cdot F_p(v+u_1)$, then $F_p(v)\le\gamma\cdot F_p(v+u_1)\le F_p(v+u_2)$. 
Therefore, the variance for our difference estimators is at most $\gamma F_p(v+u_2)^2$, so we only need to run $\O{\frac{\gamma}{\eps^2}}$ independent copies of the difference estimators. 

The remaining aspects of our algorithms are the same as \algref{alg:sliding}; for completeness, we give the modified algorithms in \algref{alg:sliding:smallp}.

\begin{algorithm}[!htb]
\caption{Moment Estimation in the Sliding Window Model}
\alglab{alg:sliding:smallp}
\begin{algorithmic}[1]
\Require{Stream $u_1,\ldots,u_m\in[n]$ of updates to coordinates of an underlying frequency vector, accuracy parameter $\eps\in(0,1)$, window parameter $W>0$}
\Ensure{Robust $(1+\eps)$-approximation to $F_p$}
\State{$\delta\gets\frac{1}{\poly(n,m)}$, $\eta\gets\frac{\eps}{1024\log\frac{1}{\eps}}$, $\varphi\gets\sqrt{2}$ be a parameter}
\State{$\beta\gets\ceil{\log\frac{100\max(p,1)}{\eps^{\max(p,1)}}}$, $\gamma_j\gets 2^{3-j}$ for all $j\in[\beta]$}
\For{each update $u_t\in[n]$, $t\in[m]$}
\State{$\updatesw$}
\Comment{Create new subroutines for each update}
\State{$\mergesw$}
\Comment{Removes extraneous subroutines}
\EndFor
\State{\Return $Z\gets\stitchsw$}
\Comment{Estimate $F_p$ on sliding window}
\end{algorithmic}
\end{algorithm}

\begin{algorithm}[!htb]
\caption{Subroutine $\updatesw$ for moment estimation}
\alglab{alg:updatesw:smallp}
\begin{algorithmic}[1]
\State{Let $s$ be the number of instances of $\calA$.}
\State{$t_{s+1}\gets t$}
\State{Start a new instance $\calA(t_{s+1},t,\eta,\delta)$.}
\State{Start a new instance $\hhest(t_{s+1},t,\eta,\delta)$.}
\For{$j\in[\beta]$}
\Comment{Start instances of each granularity}
\State{$t_{s+1,j,0}\gets t$}
\State{Start a new instance $\sdiffest(t_{s+1},t_{s+1,j,1},t,\gamma_j,\eta,\delta)$.}
\State{Start a new instance $\calA(t_{s+1,j,0},t_{s+1,j,1},1,\delta)$.}
\EndFor
\end{algorithmic}
\end{algorithm}

\begin{algorithm}[!htb]
\caption{Subroutine $\mergesw$ for moment estimation with $p\in[1,2]$}
\alglab{alg:mergesw:smallp}
\begin{algorithmic}[1]
\State{Let $s$ be the number of instances of $\calA$ and $p\in[1,2]$.}
\State{$\beta\gets\ceil{\log\frac{100p}{\eps^p}}$}
\For{$i\in[s]$, $j\in[\beta]$}
\Comment{Difference estimator maintenance}
\State{Let $r$ be the number of times $t_{i,j,*}$}
\For{$k\in[r-1]$}
\Comment{Merges two algorithms with ``small'' contributions}
\If{$\calA(t_{i,j,k-1},t_{i,j,k+1},1,\delta)\le 2^{-j-10}\cdot\calA(t_{i},t,\eta,\delta)$}
\If{$t_{i,j,k}\notin\{t_{i,j-1,*}\}$}
\State{Merge (add) the sketches for $\calA(t_{i,j,k-1},t_{i,j,k},1,\delta)$ and $\calA(t_{i,j,k},t_{i,j,k+1},1,\delta)$.}
\State{Merge (add) the sketches for $\sdiffest(t_{i,j,k-1},t_{i,j,k},t,\gamma_j,\eta,\delta)$ and $\sdiffest(t_{i,j,k},t_{i,j,k+1},t,\gamma_j,\eta,\delta)$.}
\State{Relabel the times $t_{i,j,*}$.}
\EndIf
\EndIf
\EndFor
\For{$i\in[s-2]$}
\Comment{Smooth histogram maintenance}
\If{$\calA(t_{i+2},t,\eta,\delta)\ge(1-1/4^q)\calA(t_i,t,\eta,\delta)$}
\For{$j\in[\beta]$}
\State{Append the times $t_{i+1,j,*}$ to $\{t_{i,j,*}\}$.}
\EndFor
\State{Delete $t_{i+1}$ and all times $t_{i+1,*,*}$.}
\State{Relabel the times $\{t_i\}$ and $\{t_{i,j,*}\}$.}
\EndIf
\EndFor
\EndFor
\end{algorithmic}
\end{algorithm}

\begin{algorithm}[!htb]
\caption{Subroutine $\stitchsw$ of \algref{alg:sliding} for moment estimation}
\alglab{alg:stitchsw:smallp}
\begin{algorithmic}[1]
\State{Let $i$ be the largest index such that $t_i\le m-W+1$.}
\State{$c_0\gets t_i$}
\State{$X\gets\calA(t_i,m,\eta,\delta)$}
\For{$j\in[\beta]$}
\Comment{Stitch sketches}
\State{Let $a$ be the smallest index such that $t_{i,j,a}\ge c_{j-1}$}
\State{Let $b$ be the largest index such that $t_{i,j,b}\le m-W+1$}
\State{$c_j\gets t_{i,j,b}$}
\State{$Y_j\gets\sum_{k=a}^{b-1}\sdiffest(t_{i,j,k},t_{i,j,k+1},m,\gamma_j,\eta,\delta)$}
\EndFor
\State{\Return $Z:=X-\sum_{j=1}^\beta Y_j$}
\end{algorithmic}
\end{algorithm}

We show \algref{alg:sliding:smallp} gives a $(1+\eps)$-approximation to $F_p(W)$. 
For $p\in(0,1]$, the algorithm is nearly identical to that of $p\in[1,2]$. 
However due to sub-additivity, the $\mergesw$ subroutine now merges two sketches when the sum of their contributions from their difference estimator is small, rather than the sum of their moments is small. 
We also show \algref{alg:sliding:smallp} with the modified $\mergesw$ subroutine gives a $(1+\eps)$-approximation to $F_p(W)$ for $p\in(0,1]$. 
\begin{lemma}[Correctness of sliding window algorithm]
\lemlab{lem:sliding:correctness:rangep}
For $p\in(0,2]$, \algref{alg:sliding:smallp} outputs a $(1+\eps)$-approximation to $F_p(W)$. 
\end{lemma}

\begin{theorem}
\thmlab{thm:sliding:main}
Given $\eps>0$ and $p\in(0,2]$, there exists a one-pass algorithm in the sliding window model that outputs a $(1+\eps)$-approximation to the $L_p$ norm with probability at least $\frac{2}{3}$. 
The algorithm uses $\O{\frac{1}{\eps^2}\log^3 n\log^3\frac{1}{\eps}}$ bits of space for $p=2$. 
For $p\in(0,2)$, the algorithm uses space $\O{\frac{1}{\eps^2}\log^3 n(\log\log n)^2\log^3\frac{1}{\eps}}$.
\end{theorem}
\begin{proof}
Consider \algref{alg:sliding} and note that it suffices to output a $(1+\eps)$-approximation to the $F_p$-moment of the window. 
Thus correctness follows from \lemref{lem:sliding:correctness:rangep}. 

To analyze the space complexity for $p=2$, we can use the strong tracker of \thmref{thm:strong:F2} and the difference estimator of \lemref{lem:diff:est:F2}. 
Observe that running a single instance of these subroutines with accuracy parameter $\eta$ and failure probability $\delta$ requires space $\O{\frac{\log n}{\eta^2}\left(\log\log n+\log\frac{1}{\eta}+\log\frac{1}{\delta}\right)}$. 
We set the probability of failure to be $\delta=\frac{1}{\poly(n,m)}$ so that by a union bound, all algorithms in the stream are simultaneously correct at all points in time, with probability $1-\frac{1}{\poly(n)}$. 
%There are $\poly\left(\log n,1/\eps\right)$ simultaneous algorithms running in the histogram; our intuition is that therefore it suffices to set 
%However, it should be noted that there can be $\poly(n)$ different algorithms in the smooth histogram over the course of the stream. 
%Fortunately, an argument of \cite{BravermanGLWZ18} notes that indeed we only require correctness of $\poly\left(\log n,1/\eps\right)$ simultaneous algorithms. 

By \lemref{lem:sliding:instances}, we run $\O{2^j}$ instances of the suffix difference estimator at level $j$ for each $i$, with accuracy $\eta=\O{\frac{\eps}{\log\frac{1}{\eps}}}$, ratio parameter $\gamma=2^{3-j}$, and $\delta=\frac{1}{\poly(n,m)}$. 
Hence, the total space used by level $j$ for a fixed $i$ is 
\[\O{2^j\cdot\frac{\gamma\,\log^2 n}{\eta^2}}=\O{\frac{1}{\eps^2}\log^2 n\log^2\frac{1}{\eps}}.\]
%\[\O{2^j\cdot\frac{\gamma}{\eta^2}\left(\log\log n+\log\frac{1}{\eps}\right)}=\O{\frac{\log^2\frac{1}{\eps}}{\eps^2}\left(\log\log n+\log\frac{1}{\eps}\right)}.\]
Summing across all $j\in[\beta]$ for $\beta=\O{\log\frac{1}{\eps}}$ levels, then the total space is for a fixed $i$ is $\O{\frac{1}{\eps^2}\log^2 n\log^3\frac{1}{\eps}}$. 
%\[\O{\frac{1}{\eps^2}\log n\log^3\frac{1}{\eps}\left(\log\log n+\log\frac{1}{\eps}\right)}.\]
Since $i=\O{\log n}$, then the total space is $\O{\frac{1}{\eps^2}\log^3 n\log^3\frac{1}{\eps}}$. 
%\[\O{\frac{1}{\eps^2}\log^2 n\log^3\frac{1}{\eps}\left(\log\log n+\log\frac{1}{\eps}\right)}.\]

For $p<2$, we can use the strong tracker of \thmref{thm:strong:Fp:smallp} and the difference estimator of \lemref{lem:diff:est:Fp:smallp}. 
Since a single instance of these subroutines with accuracy parameter $\eta=\O{\frac{\eps}{\log\frac{1}{\eps}}}$ and failure probability $\delta=\frac{1}{\poly(n,m)}$ requires space at most $\O{\frac{\log^2 n}{\eta^2}(\log\log n)^2}$, then the total space is $\O{\frac{1}{\eps^2}\log^3 n(\log\log n)^2\log^3\frac{1}{\eps}}$.
%\[\O{\frac{1}{\eps^2}\log^2 n(\log\log n)^2\log^3\frac{1}{\eps}\left(\log\log n+\log\frac{1}{\eps}\right)}.\]
\end{proof}
%Finally, we remark that since the difference estimator for $p\in(0,2)$ has space dependency $\gamma^C$ with $C>1$, then the $\log\frac{1}{\eps}$ dependencies in \thmref{thm:sliding:main} for $p\in(0,2)$ can be improved by parametrizing $\eta_k$ as in \algref{alg:framework}, rather than setting each $\eta_k=\eta$ as in \algref{alg:sliding}. 

\subsection{Entropy Estimation in the Sliding Window Model}
\seclab{sec:sw:entropy}
In this section, we show how our adversarially robust $F_p$ estimation algorithms gives a black-box algorithm for entropy estimation. 
For a frequency vector $v\in\mathbb{R}^n$, we define the Shannon entropy by $H=-\sum_{i=1}^n v_i\log v_i$. 
%For a frequency vector $\v\in\mathbb{R}^n$, we define the $\alpha$-th R\'{e}nyi entropy by $H_{\alpha}(\v)=\log(F_{\alpha}(\v))/(1-\alpha)=\log(\|\v\|_{\alpha}^\alpha)/(1-\alpha)$. 
%The Shannon entropy is defined by $H=H_1=-\sum_{i=1}^n v_i\log v_i$, where the relationship $H=\lim_{\alpha\rightarrow 1}H_\alpha$ results from L'H\^{o}pital's Rule. 
\begin{observation}
\obslab{obs:entropy:addmult}
Any algorithm that gives an $\eps$-additive approximation of the Shannon Entropy $H(v)$ also gives a $(1+\eps)$-multiplicative approximation of the function $h(v):=2^{H(v)}$ (and vice versa). 
\end{observation}
Thus for the remainder of the section, we focus on obtaining a $(1+\eps)$-multiplicative approximation of $h(v)=2^{H(v)}$. 
\begin{lemma}[Section 3.3 in~\cite{HarveyNO08}]
\lemlab{lem:entropy:reduction}
For $k=\log\frac{1}{\eps}+\log\log m$ and $\eps'=\frac{\eps}{12(k+1)^3\log m}$, there exists an efficiently computable set $\{y_0,\ldots,y_k\}$ with $y_i\in(0,2)$ for all $i$ and an efficiently computable deterministic function that takes $(1+\eps')$-approximations to $F_{y_i}(v)$ and outputs a $(1+\eps)$-approximation to $h(v)=2^{H(v)}$. 
\end{lemma}
The set $\{y_0,\ldots,y_k\}$ in \lemref{lem:entropy:reduction} can computed as follows, described by \cite{HarveyNO08}. 
Let $\ell=\frac{1}{2(k+1)\log m}$ and $f(z)=\frac{(k^2\ell)z-\ell(k^2+1)}{2k^2+1}$. 
Then for each $y_i$, we have $y_i=1+f(\cos(i\pi/k))$, so that the set $\{y_0,\ldots,y_k\}$ in \lemref{lem:entropy:reduction} can be computed in linear time. 
A $(1+\eps)$-multiplicative approximation to $h(v)=2^{H(v)}$ can then be computed from $2^{P(0)}$, where $P(x)$ is the degree $k$ polynomial interpolated at the points $y_0,\ldots,y_k$, so that $P(y_i)=F_{y_i}(v)$ for each $i$. 

\begin{theorem}
\thmlab{thm:sw:entropy}
Given $\eps>0$, there exists a sliding window algorithm that outputs an additive $\eps$-approximation to Shannon entropy and uses $\tO{\frac{\log^5 n}{\eps^2}}$ bits of space and succeeds with probability at least $\frac{2}{3}$. 
\end{theorem}
\begin{proof}
By \obsref{obs:entropy:addmult} and \lemref{lem:entropy:reduction}, it suffices to obtain adversarially robust $(1+\eps')$-approximation algorithms to $F_{y_i}(v)$ for all $y_i\in(0,2)$ in the set $\{y_0,\ldots,y_k\}$, where $k=\log\frac{1}{\eps}+\log\log m$ and $\eps'=\frac{\eps}{12(k+1)^3\log m}$. 
By \thmref{thm:sliding:main} with accuracy parameter $\eps'$, we can obtain such robust algorithms approximating each $F_{y_i}(v)$, using space $\tO{\frac{1}{(\eps')^{2}}\log^2 n}$. 
Moreover, with a rescaling of the failure probability $\delta'=\frac{\delta}{\poly(n,m)}$ in both \lemref{lem:diff:est:Fp:smallp} and \thmref{thm:strong:Fp:smallp}, then all algorithms simultaneously succeed with probability at least $1-\delta$. 
Hence for $\eps'=\frac{\eps}{12(k+1)^3\log m}$ and $\log m=\O{\log n}$, each of the $\O{k}$ algorithms use space $\tO{\frac{1}{\eps^2}\log^5 n}$. 
Since $k=\log\frac{1}{\eps}+\log\log m$, the overall space complexity follows. 
\end{proof}

\subsection{Moment Estimation for Integer \texorpdfstring{$p>2$}{p>2}}
In this section, we describe our algorithm to estimate the $F_p$ moment in the sliding window model, for integer $p>2$. 
Suppose that we have vectors $u$ and $v$ such that the vector $u$ arrives before the vector $v$ and $F_p(u)\le\gamma F_p(v)\le 2^p F_p(u)$, for some $\gamma\le 1$. 
A crucial subroutine in \algref{alg:sliding} is the $\mergesw$ subroutine, which controls the number of blocks in which the substream is partitioned into, at each granularity, and thus gives efficient bounds on the space of the algorithm.  
\algref{alg:mergesw:smallp} merges separate blocks based on their contribution to the value of the difference $F_p(u+v)-F_p(v)$. 
Rather than using such a merge subroutine in this section, we will again create $\O{\log n}$ parallel instances as in \algref{alg:mergesw:frame}, corresponding to exponentially increasing guesses $2^i$ for the value of $F_p(v)$ and incur the additional $\O{\log n}$ overhead in space.  
We will then partition blocks based on their $F_p$ values in comparison to the guess of $F_p(v)$. 
Now if each block $u$ is required to satisfy $F_p(u)\le\gamma F_p(v)\le 2^p F_p(u)$ for some $\gamma\approx\frac{1}{2^j}$, then we can assume our guess for $F_p(v)$ is incorrect if the partitioning creates more than $2^j$ blocks.
Hence, we can again assume that each block $u$ satisfies $F_p(u)\le\gamma F_p(v)\le 2^p F_p(u)$, at the cost of an additional $\O{\log n}$ overhead in space. 

As in \secref{sec:sliding:smallp}, it remains to find a difference estimator for $F_p(u+v)-F_p(v)$, i.e., an algorithm with space dependency $\tO{\frac{\gamma}{\eps^2}\,n^{1-2/p}}$. 
Observe that $F_p(u+v)-F_p(v)=\sum_{k=1}^p\binom{p}{k}\langle u^k,v^{p-k}\rangle$. 
We estimate each term $\langle u^k,v^{p-k}\rangle$ separately. 

We first use a heavy-hitter algorithm $\hhest$ to simultaneously find all heavy-hitters $a\in[n]$ such that $u_a\ge\eps\gamma^{1/p}\,L_p(v)$ across all vectors $u$ induced by the blocks. 
These coordinates form a set $\calH$ and we read off the corresponding coordinates of $v$ to estimate $\sum_{a\in\calH}\langle u^k,v^{p-k}\rangle$ for each $k$. 
For $a\notin\calH$, we analyze separate algorithms for estimating $\sum_{a\notin\calH}\langle u,v^{p-1}\rangle$ and $\sum_{a\notin\calH}\langle u^k,v^{p-k}\rangle$ for $k\ge 2$, though the analysis is similar for both algorithms. 
Thus we require the following modification to subroutine $\stitchsw$. 

\begin{algorithm}[!htb]
\caption{Subroutine $\stitchsw$ of \algref{alg:sliding}: output estimate of $F_p$ on the sliding window}
\alglab{alg:stitchsw}
\begin{algorithmic}[1]
\State{Let $i$ be the largest index such that $t_i\le m-W+1$.}
\State{$c_0\gets t_i$}
\State{$X\gets\calA(t_i,m,\eta,\delta)$}
\For{$j\in[\beta]$}
\Comment{Stitch sketches}
\State{Let $a$ be the smallest index such that $t_{i,j,a}\ge c_{j-1}$}
\State{Let $b$ be the largest index such that $t_{i,j,b}\le m-W+1$}
\State{$c_j\gets t_{i,j,b}$}
\State{$Y_j\gets\sum_{k=a}^{b-1}\sdiffest(t_{i,j,k},t_{i,j,k+1},m,\gamma_j,\eta,\delta)$}
\EndFor
\For{each heavy-hitter $k$ output by $\calH:=\hhest(t_i,m,\eta,\delta)$}
\State{Let $g_k$ be the estimated frequency of $k$ in times $[c_{\beta},m]$.}
\State{Let $h_k$ be the estimated frequency of $k$ in times $[m-W+1,m]$.}
\State{$W_k\gets(g_k)^p-(h_k)^p$}
\EndFor
\State{\Return $Z:=X-\sum_{j=1}^\beta Y_j-\sum_{k\in\calH} W_k$}
\end{algorithmic}
\end{algorithm}

Our analysis uses a level set argument. 
For a sufficiently large constant $\alpha>0$, let $i,j\in[\alpha\log n]$. 
Then we define the contribution $C_{i,j,k}$ of level set $L_{i,j}$ for $\langle u^k,v^{p-k}\rangle$ by
\[\sum_{a\in n}\left[u_a^kv_a^{p-k}\bfone\left(u_a\in\left[\frac{\gamma^{1/p} L_p(v)}{2^i},\frac{\gamma^{1/p} L_p(v)}{2^{i-1}}\right]\right)\bfone\left(v_a\in\left[\frac{L_p(v)}{2^j},\frac{L_p(v)}{2^{j-1}}\right]\right)\right].\]
Then we say contribution $C_{i,j,k}$ of level set $L_{i,j}$ is significant if $C_{i,j,k}\ge\frac{\eps}{p^{p+1}\log^2 n}\,F_p(v)$. 

\begin{lemma}
%If $C_{i,j,k}$ is significant, then $2^{|i-j|}\le p^{p+1}\log^2 n$.
If $C_{i,j,k}$ is significant, then $2^{|i-j|}\le\frac{\gamma^{k/(p(p-k))}p^{(p+1)/(p-k)}\log^{2/(p-k)} n}{\eps^{1/(p-k)}}$. 
\end{lemma}
\begin{proof}
Suppose without loss of generality that $i\le j$, so that $j=i+\ell$ for some $\ell\ge 0$. 
Because $F_p(u)\le\gamma F_p(v)$, then there can be at most $2^{ip}$ indices $a\in[n]$ such that $u_a\in\left[\frac{\gamma^{1/p} L_p(v)}{2^i},\frac{\gamma^{1/p} L_p(v)}{2^{i-1}}\right]$. 
Hence, we have 
\[C_{i,j,k}\le 2^{ip}\cdot\left(\frac{\gamma^{1/p} L_p(v)}{2^i}\right)^k\cdot\left(\frac{L_p(v)}{2^j}\right)^{p-k}=\frac{\gamma^{k/p}}{2^{\ell(p-k)}}\,F_p(v).\]
Since $C_{i,j,k}$ is significant, then $C_{i,j,k}\ge\frac{\eps}{p^{p+1}\log^2 n}\,F_p(v)$, which implies that $2^{\ell(p-k)}\le\frac{\gamma^{k/p}p^{p+1}\log^2 n}{\eps}$ and thus $2^j\le 2^i\,\frac{\gamma^{k/(p(p-k))}p^{(p+1)/(p-k)}\log^{2/(p-k)} n}{\eps^{1/(p-k)}}$. 
%$2^j\le 2^i p^{p+1}\log^2 n$. 
\end{proof}

We first run a heavy-hitters algorithm to find all indices $a\in[n]$ such that 
\[u_a\ge\frac{\eps}{100p^{p+1}\log^2 n}\gamma^{1/p} L_p(v)\ge\frac{\eps}{100p^{p+1}\log^2 n} L_p(u),\]
which takes $\tO{\frac{1}{\eps^2}\,n^{1-2/p}}$ space, since $F_p(u)\le\gamma F_p(v)$.  
Thus we bound the variance for a significant level set $L_{i,j}$ with $2^i\ge\frac{\eps}{100p^{p+1}\log^2 n}$.  

\begin{figure*}
\begin{mdframed}
\begin{enumerate}
\item
Find a list $\calH$ that includes all $a\in[n]$ with $u_a\ge\frac{\eps}{100p^{p+1}\log^2 n}\|u\|_p$.
\item
Using $\countsketch$, obtain an estimate $\widehat{u_a}$ to $u_a$ with additive error $\frac{\eps}{100p^{p+1}\log^2 n}\|u\|_p$ for each $a\in\calH$ and let $h\in\mathbb{R}^n$ be the vector such that $h_a=\widehat{u_a}$ if $a\in\calH$ and zero otherwise. 
\item
Perform perfect $L_2$ sampling on $w:=u-h$ to obtain a set $\calS$ of size $k=\O{\frac{\gamma}{\eps^2}n^{1-2/p}}$. 
\item
Obtain an estimate $\widehat{w_a}$ to $w_a$ for each $a\in\calS$. 
\item
Let $\widehat{W}$ be a $(1+\eps)$-approximation to $\|w\|_2^2$. 
 %and $U$ be a $\left(1+\frac{\eps}{\gamma}\right)$-approximation to $\|u\|_p^p$. 
\item
Output $W+\sum_{k=1}^{p-1}\binom{p}{k}\left(\sum_{a\in\calH}\widehat{v_a}^k,u_a^{p-k}+W\cdot\sum_{a\in\calS}\widehat{w_a}^{k-2},u_a^{p-k}\right)$. 
\end{enumerate}
\end{mdframed}
%\caption{$F_p$ difference estimator for $F_p(v+u)-F_p(v)$ with integer $p>2$.}
%\figlab{fig:diff:est:largep}
\end{figure*}

\begin{lemma}
For integer $p>2$, there exists an algorithm that uses space $\tO{\frac{\gamma}{\eps^2}n^{1-2/p}}$ and outputs an additive $\eps\cdot F_p(u)$ approximation to $\langle u^k,v^{p-k}\rangle$ with high probability for any $k\ge 2$. 
\end{lemma}
\begin{proof}
We first show that our algorithm gives an additive $\frac{\eps}{p^{p+1}} F_p(v)$ approximation to $\langle u^k,v^{p-k}\rangle$ with $k\ge 2$. 
To that end, we show that our algorithm gives an additive $\frac{\eps}{p^{p+1}\log^2 n} F_p(v)$ approximation to each contribution $C_{i,j,k}$ of level set $L_{i,j}$. 
Let $h\in\mathbb{R}^n$ be the vector such that $h_a=\widehat{u_a}$ if $a\in\calH$ and zero otherwise and let $w:=u-h$. 

We use $\sampler$ to sample indices $j_1,\ldots,j_\ell\in[n]$ with so that each sample is a coordinate $a\in[n]$ with probability $\frac{w_a^2}{\|w\|_2^2}+\frac{1}{\poly(n)}$. 
We obtain unbiased estimates to $\widehat{w_{j_1}^{k-2}},\ldots,\widehat{w_{j_\ell}^{k-2}}$ of $w_{j_1}^{k-2},\ldots,w_{j_\ell}^{k-2}$ through \lemref{lem:coor:exp:var}. 
We also obtain a $(1+\O{\eps})$-approximation unbiased estimate $W$ of $\|w\|_2^2$. 
Thus for each $b\in[\ell]$, the product $\widehat{w_{j_b}^{k-2}}\cdot W\cdot v_{j_b}^{p-k}$ satisfies
\begin{align*}
\Ex{\widehat{w_{j_b}^{k-2}}\cdot W\cdot v_{j_b}^{p-k}}&=\sum_{a\in\calH}\ip{\widehat{u_a}^{k}}{v_a^{p-k}}+\sum_{i,j}\sum_{a\in L_{i,j}\cap\calH}\left(\frac{w_a^2}{\|w\|_2^2}+\frac{1}{\poly(n)}\right)\cdot w_a^{k-2}\cdot(1\pm\O{\eps})\|w\|_2^2\cdot v_a^{p-k}\\
&+\sum_{i,j}\sum_{a\in L_{i,j}\setminus\calH}\left(\frac{w_a^2}{\|w\|_2^2}+\frac{1}{\poly(n)}\right)\cdot w_a^{k-2}\cdot(1\pm\O{\eps})\|w\|_2^2\cdot v_a^{p-k}.
\end{align*}
Observe that since $\widehat{u_a}$ is an additive $\eps\gamma^{1/p}\,L_p(v)$-approximation to $u_a$ for $a\in\calH$, then $|w_a|\le\eps|u_a|$, so that the second summation is at most $\O{\eps}\,\langle u^k, v^{p-k}\rangle$. 
Moreover, we have $w_a=u_a$ for $a\in[n]\setminus\calH$, so that 
\begin{align*}
\Ex{\widehat{w_{j_b}^{k-2}}\cdot W\cdot v_{j_b}^{p-k}}&\in(1\pm\O{\eps})|\ip{w^{k}}{v^{p-k}}|+\frac{1}{\poly(n)}.
\end{align*}
For $a\in L_{i,j}\setminus\calH$, we have $w_a\le u_a\le\eps\gamma^{1/p}\,L_p(v)$. 
Thus, the variance is at most
\begin{align*}
\Var\left(\widehat{w_{j_b}^{k-2}}\cdot W\cdot v_{j_b}^{p-k}\right)&\le\sum_{i,j}\sum_{a\in L_{i,j}\setminus\calH}\left(\frac{w_a^2}{\|w\|_2^2}+\frac{1}{\poly(n)}\right)\cdot w_a^{2k-4}\cdot(1\pm\O{\eps})\|w\|_2^4\cdot v_a^{2p-2k}\\
&\le\sum_{i,j}\sum_{a\in L_{i,j}\setminus\calH}(1\pm\O{\eps})w_a^{2k-2}\cdot\|w\|_2^2\cdot v_a^{2p-2k}\\
&\le\sum_{i,j}\sum_{a\in L_{i,j}\setminus\calH}(1\pm\O{\eps})\frac{\eps^p\gamma^2\,(L_p(v))^{2p-2}}{2^{2ip}}\frac{\gamma^{2k/p}p^{p(2p+2)}\log^4 n}{\eps^2}\,\|w\|_2^2\\
&\le\sum_{i,j}\sum_{a\in L_{i,j}\setminus\calH}\frac{Cn^{1-2/p}\eps^p\gamma\,(F_p(v))^2\log^4 n}{2^{2ip}\eps^2},
\end{align*}
for some constant $C>0$. 
Hence, we have
\begin{align*}
\Var\left(\widehat{w_{j_b}^{k-2}}\cdot W\cdot v_{j_b}^{p-k}\right)&\le\sum_{i,j}2^{ip}\cdot\frac{Cn^{1-2/p}\eps^p\gamma\,(F_p(v))^2\log^4 n}{2^{2ip}\eps^2}\\
&\le\O{\log^2 n}\cdot\frac{Cn^{1-2/p}\eps^p\gamma\,(F_p(v))^2\log^4 n}{\eps^2}.
\end{align*}
Thus by setting $\ell=\O{\frac{\gamma}{\eps^2}n^{1-2/p}}$, we obtain an additive $\frac{\eps}{p^{p+1}}\cdot F_p(v)$ approximation to $\langle u^k,v^{p-k}\rangle$ for $k\ge2$, with constant probability. 
We can then boost this probability to $1-\frac{\delta}{\poly(n)}$ by repeating $\O{\log\frac{n}{\delta}}$ times. 
\end{proof}

We use a similar approach to estimate $\langle u,v^{p-1}\rangle$. 
\begin{figure*}
\begin{mdframed}
\begin{enumerate}
\item
Use a set of exponential random variables to form a vector $W$ of duplicated and scaled coordinates of $w$. 
\item
Hash the coordinates of $W$ into a $\countsketch$ data structure with $\O{\log n}$ buckets. 
\item
Use the same exponential random variables to perform perfect $L_2$ sampling on $v$ to obtain a coordinate $(i,j)$ and an unbiased estimate $\widehat{v_{i,j}}$ to $v_{i,j}$. 
\item
Let $\widehat{V}$ be an unbiased estimate of $\|v\|_2^2$ with second moment $\O{\|v\|_2^4}$.  
\item
Query $\countsketch$ for an unbiased estimate $\widehat{w_{i,j}}$ to $w_{i,j}$ and set an estimator as $\widehat{V}\cdot\widehat{w_{i,j}}\left(\widehat{v_{i,j}}\right)^{p-3}$. 
\item
Output the mean of $\tO{\frac{\gamma}{\eps^2}\cdot n^{1-2/p}}$ such estimators. 
\end{enumerate}
\end{mdframed}
%\caption{$F_p$ difference estimator for $\langle v,u^{p-1}\rangle$ with integer $p>2$.}
%\figlab{fig:diff:est:largep:indexone}
\end{figure*}

\begin{lemma}
For integer $p>2$, there exists an algorithm that uses space $\tO{\frac{\gamma}{\eps^2}n^{1-2/p}}$ and outputs an additive $\eps\cdot F_p(u)$ approximation to $\langle u,v^{p-1}\rangle$ with high probability. 
\end{lemma}
\begin{proof}
We use $\sampler$ to sample indices $j_1,\ldots,j_\ell\in[n]$ with so that each sample is a coordinate $a\in[n]$ with probability $\frac{v_a^2}{\|v\|_2^2}+\frac{1}{\poly(n)}$. 
We obtain unbiased estimates to $\widehat{v_{j_1}^{p-1}},\ldots,\widehat{v_{j_\ell}^{p-1}}$ of $v_{j_1}^{p-1},\ldots,w_{j_\ell}^{p-1}$ through \lemref{lem:coor:exp:var}. 
We also obtain a $(1+\O{\eps})$-approximation unbiased estimate $V$ of $\|v\|_2^2$. 
Thus for each $b\in[\ell]$, the product $\widehat{w_{j_b}}\cdot V\cdot v_{j_b}^{p-3}$ satisfies
\begin{align*}
\Ex{\widehat{w_{j_b}}\cdot V\cdot v_{j_b}^{p-3}}&=\sum_{a\in\calH}\langle\widehat{u_a},v_a^{p-1}\rangle+\sum_{i,j}\sum_{a\in L_{i,j}\cap\calH}\left(\frac{v_a^2}{\|v\|_2^2}+\frac{1}{\poly(n)}\right)\cdot w_a\cdot(1\pm\O{\eps})\|v\|_2^2\cdot v_a^{p-3}\\
&+\sum_{i,j}\sum_{a\in L_{i,j}\setminus\calH}\left(\frac{v_a^2}{\|v\|_2^2}+\frac{1}{\poly(n)}\right)\cdot w_a\cdot(1\pm\O{\eps})\|w\|_2^2\cdot v_a^{p-3}.
\end{align*}
Observe that since $\widehat{u_a}$ is a $(1+\eps)$-approximation to $u_a$ for $a\in\calH$, then $|w_a|\le\eps|u_a|$, so that the second summation is at most $\O{\eps}\,\langle u, v^{p-1}\rangle$. 
Moreover, we have $w_a=u_a$ for $a\in[n]\setminus\calH$, so that 
\begin{align*}
\Ex{\widehat{w_{j_b}}\cdot V\cdot v_{j_b}^{p-3}}&\in(1\pm\O{\eps})|\langle w, v^{p-1}\rangle|+\frac{1}{\poly(n)}.
\end{align*}
For $a\in L_{i,j}\setminus\calH$, we have $w_a\le u_a\le\eps\gamma^{1/p}\,L_p(v)$. 
Thus, the variance is at most
\begin{align*}
\Var\left(\widehat{w_{j_b}}\cdot V\cdot v_{j_b}^{p-3}\right)&\le\sum_{i,j}\sum_{a\in L_{i,j}\setminus\calH}\left(\frac{v_a^2}{\|v\|_2^2}+\frac{1}{\poly(n)}\right)\cdot w_a^2\cdot(1\pm\O{\eps})\|v\|_2^4\cdot v_a^{2p-6}\\
&\le\sum_{i,j}\sum_{a\in L_{i,j}\setminus\calH}(1\pm\O{\eps})w_a^2\cdot\|v\|_2^2\cdot v_a^{2p-6}\\
&\le\sum_{i,j}\sum_{a\in L_{i,j}\setminus\calH}(1\pm\O{\eps})\frac{\eps^p\gamma^{2/p}\,(L_p(v))^{2p-2}}{2^{2ip}}\frac{\gamma^{2k/p}p^{p(2p+2)}\log^4 n}{\eps^2}\,\|v\|_2^2\\
&\le\sum_{i,j}\sum_{a\in L_{i,j}\setminus\calH}\frac{Cn^{1-2/p}\eps^p\gamma^{2/p}\,(F_p(v))^2\log^4 n}{2^{2ip}\eps^2},
\end{align*}
for some constant $C>0$. 
Therefore,
\begin{align*}
\Var\left(\widehat{w_{j_b}}\cdot V\cdot v_{j_b}^{p-3}\right)&\le\sum_{i,j}2^{ip}\cdot\frac{Cn^{1-2/p}\eps^p\gamma^{2/p}\,(F_p(v))^2\log^4 n}{2^{2ip}\eps^2}\\
&\le\O{\log^2 n}\cdot\frac{Cn^{1-2/p}\eps^p\gamma\,(F_p(v))^2\log^4 n}{\eps^2}.
\end{align*}
Hence for $\gamma\ge\O{\eps^p}$, we have that by setting $\ell=\O{\frac{\gamma}{\eps^2}n^{1-2/p}}$, we obtain an additive $\frac{\eps}{p^{p+1}}\cdot F_p(v)$ approximation to $\langle u,v^{p-1}\rangle$, with constant probability. 
We can then boost this probability to $1-\frac{\delta}{\poly(n)}$ by repeating $\O{\log\frac{n}{\delta}}$ times. 
\end{proof}

\begin{lemma}[$F_p$ difference estimator]
\lemlab{lem:sw:diff:est:Fp:largep}
For integer $p>2$, there exists a $(\gamma,\eps,\delta)$-\emph{difference estimator} for $F_p$ that uses space $\tO{\frac{\gamma n^{1-2/p}}{\eps^2}}$.
\end{lemma}

By applying the same argument as \lemref{lem:sw:frame:correct} using the difference estimators in this section, we have:
\begin{lemma}[Correctness of sliding window algorithm]
\lemlab{lem:sliding:correctness:largep}
For integer $p>2$, \algref{alg:sliding} outputs a $(1+\eps)$-approximation to $F_p(W)$. 
\end{lemma}

\begin{theorem}
\thmlab{thm:sliding:largep}
Given $\eps>0$ and integer $p>2$, there exists a one-pass algorithm in the sliding window model that outputs a $(1+\eps)$-approximation to the $L_p$ norm with probability at least $\frac{2}{3}$ and uses $\tO{\frac{1}{\eps^2}\,n^{1-2/p}}$ bits of space. 
\end{theorem}
\begin{proof}
Consider \algref{alg:sliding} and note that it suffices to output a $(1+\eps)$-approximation to the $F_p$-moment of the window. 
Thus correctness follows from \lemref{lem:sliding:correctness:largep}. 

We use the strong tracker of \thmref{thm:strong:Fp:largep} and the difference estimator of \lemref{lem:sw:diff:est:Fp:largep}. 
By \lemref{lem:sliding:instances}, we run $\O{2^j}$ instances of the suffix difference estimator at level $j$ for each $i$, with accuracy $\eta=\O{\frac{\eps}{\log\frac{1}{\eps}}}$, ratio parameter $\gamma=2^{3-j}$, and $\delta=\frac{1}{\poly(n,m)}$. 
Hence, the total space used by level $j$ for a fixed $i$ is 
\[\tO{2^j\cdot\frac{\gamma\,n^{1-2/p}}{\eta^2}}=\tO{\frac{n^{1-2/p}}{\eps^2}}.\]
Summing across all $j\in[\beta]$ for $\beta=\O{\log\frac{1}{\eps}}$ levels and all values of $i=\O{\log n}$, then the algorithm uses $\tO{\frac{1}{\eps^2}\,n^{1-2/p}}$ bits of space in total. 
\end{proof}

\section*{Acknowledgements}
We would like to thank Zhili Feng for discussions in the early stages of this work. 
This work was supported by National Science Foundation (NSF) Grant No. CCF-1815840, National Institute of Health (NIH) grant 5R01 HG 10798-2, and a Simons Investigator Award.

We would also like to thank Moshe Shechner and Uri Stemmer for pointing out an error in our bounded flip number algorithm for dynamic streams in a previous version of this paper. 
In this version, we have introduced the notion of twist number to resolve this issue and refer to their follow-up work~\cite{AttiasCSS21} for more details.

\def\shortbib{0}
\bibliographystyle{alpha}
\bibliography{references}
\appendix
\section{Appendix: Moment Estimation on Sliding Windows}
\applab{app:sliding}
In this section, we give the full details of the proofs from \secref{sec:sliding}. 
\subsection{Moment Estimation for \texorpdfstring{$p\in[1,2]$}{p in [1,2]}}
\begin{lemma}[Constant factor partitions in top level]
\lemlab{lem:sw:top:constant:medp}
Let $i$ be the largest index such that $t_i\le m-W+1$. 
Let $\calE$ be the event that all subroutines $\calA$ in \algref{alg:sliding} succeed.
Then conditioned on $\calE$, $F_p(W)\le F_p(t_i,m)\le 2F_p(W)$ and $\frac{1}{2}F_p(W)\le F_p(t_{i+1},m)\le F_p(W)$. 
\end{lemma}
\begin{proof}
Observe that in the $\mergesw$ subroutine, we only delete a timestamp $t_{i'}$ if $\calA(t_{i'+2},t,\eta,\delta)\ge\frac{9}{10}\calA(t_{i'},t,\eta,\delta)$ at some point $t$. 
Since $\eta=\frac{\eps}{1024\log\frac{1}{\eps}}\le\frac{1}{1024}$, we have that conditioned on $\calE$, $F_p(t_{i'+2},t)\left(1+\frac{1}{1024}\right)\ge\frac{9}{10}\left(1-\frac{1}{1024}\right)F_p(t_{i'},t)$. 
Thus, $F_p(t_{i'+2},t)\ge\frac{7}{8}F_p(t_{i'},t)$. 

We have either $t_i=m-W+1$ and $t_{i+1}=m-W+2$ in which case the statement is trivially true or timestamps $t_i$ and $t_{i+1}$ at some point $t\le m$ must have satisfied $F_p(t_{i+1},t)\ge\frac{7}{8}F_p(t_i,t)$ so that $t_i+1$ was removed from the set of timestamps. 
Thus, $F_p(t_{i+1},t)\ge\frac{7}{8}F_p(t_i,t)$. 
By \lemref{lem:fp:smooth} for $p\le 2$ and the definition of smoothness, $F_p(t_{i+1},t)\ge\frac{7}{8}F_p(t_i,t)$ implies $F_p(t_{i+1},m)\ge 2F_p(t_i,m)$ for all $m\ge t$. 
Since $t_i\le m-W+1$, then $F_p(t_i,m)\ge F_p(m-w+1,m)=F_p(W)\ge F_p(t_{i+1},m)$ and the conclusion follows. 
\end{proof}
We now show an upper bound on the moment of each substream whose contribution is estimated by the difference estimator, i.e., the difference estimators are well-defined.
\begin{lemma}
\lemlab{lem:sw:ub:chunk:medp}
Let $p\in[1,2]$ and $i$ be the largest index such that $t_i\le m-W+1$. 
Let $\calE$ be the event that all subroutines $\calA$ in \algref{alg:sliding:smallp} succeed.
Then conditioned on $\calE$, we have that for each $j\in[\beta]$ and all $k$, either $t_{i,j,k+1}=t_{i,j,k}+1$ or $F_p(t_{i,j,k},t_{i,j,k+1})\le2^{-j-7}\cdot F_p(W)$. 
\end{lemma}
\begin{proof}
Suppose $t_{i,j,k+1}\neq t_{i,j,k}+1$ and $i+1$ is the smallest index such that $t_{i+1}>m-W+1$. 
Then at some time during the stream, the timestamp $t_{i,j,k}+1$ must have been removed from the set of timestamps tracked by the algorithm. 
Thus at some time $t$, $\calA(t_{i,j,k},t_{i,j,k+1},1,\delta)\le 2^{-j-10}\cdot\calA(t_{i},t,\eta,\delta)$ so that conditioned on the correctness of the algorithms $\calA$ providing $2$-approximations to the respectively quantities, $\frac{1}{2}\cdot F_p(t_{i,j,k},t_{i,j,k+1})\le 2^{-j-9} F_p(t_i,t)$. 
By \lemref{lem:sw:top:constant:medp}, we have that $F_p(t_i,t)\le 2F_p(W)$. 
Hence, $F_p(t_{i,j,k},t_{i,j,k+1})\le2^{-j-7}\cdot F_p(W)$. 
\end{proof}

\begin{lemma}[Accuracy of difference estimators]
\lemlab{lem:diff:est:well:defined:medp}
Let $p\in[1,2]$ and $i$ be the largest index such that $t_i\le m-W+1$. 
Let $\calE$ be the event that all subroutines $\calA$ in \algref{alg:sliding:smallp} succeed.
Then conditioned on $\calE$, we have that for each $j\in[\beta]$ and all $k$, $\sdiffest(t_{i,j,k-1},t_{i,j,k},t,\gamma_j,\eta,\delta)$ gives an additive $\eta\cdot F_p(t_{i,j,k}:t)$ approximation to $F_p(t_{i,j,k-1}:t)-F_p(t_{i,j,k}:t)$ with probability at least $1-\delta$. 
\end{lemma}
\begin{proof}
Recall that for a suffix-pivoted difference estimator to be well-defined, we require $F_p(t_{i,j,k-1}:t_{i,j,k})\le\gamma_j\cdot F_p(t_{i,j,k}:m)$. 
By \lemref{lem:sw:ub:chunk:medp}, we have that $F_p(t_{i,j,k-1},t_{i,j,k})\le2^{-j-7}\cdot F_p(W)$. 
Since $F_p(W)\le F_p(t_{i,j,k}:m)$ and $\gamma_j=2^{-j+3}$, then we certainly have $F_p(t_{i,j,k-1}:t_{i,j,k})\le\gamma_j\cdot F_p(t_{i,j,k}:m)$. 
Hence by the definition of a (suffix-pivoted) difference estimator, we have that $\sdiffest(t_{i,j,k-1},t_{i,j,k},t,\gamma_j,\eta,\delta)$ gives an additive $\eta\cdot F_p(t_{i,j,k}:t)$ approximation to $F_p(t_{i,j,k-1}:t)-F_p(t_{i,j,k}:t)$ with probability at least $1-\delta$. 
\end{proof}

\begin{lemma}[Geometric upper bounds on splitting times]
\lemlab{lem:sliding:delta:upper:medp}
Let $p\in[1,2]$ and $i$ be the largest index such that $t_i\le m-W+1$.  
Let $\calE$ be the event that all subroutines $\calA$ in \algref{alg:sliding:smallp} succeed.
Then conditioned on $\calE$, we have that for each $j\in[\beta]$ and each $k$ that either $t_{i,j,k+1}=t_{i,j,k}+1$ or $F_p(t_{i,j,k},m)-F_p(t_{i,j,k+1},m)\le 2^{-j/p-2}p\cdot F_p(t_{i,j,k+1},m)$. 
\end{lemma}
\begin{proof}
Suppose $t_{i,j,r+1}\neq t_{i,j,r}+1$. 
Then conditioned on $\calE$, we have that $F_p(t_{i,j,k},t_{i,j,k+1})\le2^{-j-7}\cdot F_p(W)$ by \lemref{lem:sw:ub:chunk:medp}. 
Moreover, $F_p(W)\le F_p(t_{i,j,k+1},m)$. 
By the smoothness of $F_p$ in \lemref{lem:fp:smooth}, $F_p(t_{i,j,k},m)\le(1+2^{-(j-7)/p}p)\cdot F_p(t_{i,j,k+1},m)$. 
Hence for $p\in[1,2]$, we have $F_p(t_{i,j,k},m)-F_p(t_{i,j,k+1},m)\le 2^{-j/p-2}p\cdot F_p(t_{i,j,k+1},m)$. 
\end{proof}

\begin{lemma}[Geometric lower bounds on splitting times]
\lemlab{lem:sliding:delta:lower:medp}
Let $p\in[1,2]$ and $i$ be the largest index such that $t_i\le m-W+1$.  
Let $\calE$ be the event that all subroutines $\calA$ in \algref{alg:sliding:smallp} succeed.
Then conditioned on $\calE$, we have that for each $j\in[\beta]$ and each $k$ that 
\[F_p(t_{i,j,k-1},t_{i,j,k+1})>2^{-j-8}\cdot F_p(W).\]
\end{lemma}
\begin{proof}
Note that conditioned on $\calE$, we have that $F_p(t_{i,j,k-1},t_{i,j,k+1})\ge\frac{1}{2}\calA(t_{i,j,k-1},t_{i,j,k+1},1,\delta)$ and $2\cdot\calA(t_i,m,\eta,\delta)\ge F_p(t_i,m)\ge\frac{1}{2}\cdot F_p$. 
Since $\mergesw$ did not merge the timestamp $t_{i,j,k}$ then it follows that $\calA(t_{i,j,k-1},t_{i,j,k+1},1,\delta)>2^{-j-10}\cdot\calA(t_{i},t,\eta,\delta)$. 
Hence, it follows that $F_p(t_{i,j,k-1},t_{i,j,k+1})>2^{-j-8}\cdot F_p(W)$. 
\end{proof}

Next we bound the number of instances $k$ of the level $j$ difference estimators. 
\begin{lemma}[Number of level $j$ difference estimators]
\lemlab{lem:sliding:instances}
Let $p\in[1,2]$ and $i$ be the largest index such that $t_i\le m-W+1$.  
Let $\calE$ be the event that all subroutines $\calA$ in \algref{alg:sliding:smallp} succeed.
Then conditioned on $\calE$, we have that for each $j\in[\beta]$, that $k\le 2^{j+10}$ for any $t_{i,j,k}$.  
\end{lemma}
\begin{proof}
For a fixed $j$, let $r$ be the number of instances of $t_{i,j,k}$. 
For each $k\in[r]$, let $u_k$ denote the frequency vector between times $t_{i,j,k-1}$ and $t_{i,j,k}$. 
Let $u$ denote the frequency vector between times $t_{i+1}$ and $t$ and let $v$ denote the frequency vector between times $t_i$ and $t$. 
Then by \lemref{lem:sliding:delta:lower:medp}, we have $F_p(u_k+u_{k+1})>2^{-j-8}\cdot F_p(W)$. 
By \lemref{lem:sw:top:constant:medp}, $F_p(t_i,m)\le 2F_p(W)$. 
For $p\in[1,2]$, we have $F_p(t_i,m)\ge\sum_{\ell=1}^{r/2}F_p(u_{2\ell-1}+u_{\ell})$. 
Thus for $r>2^{j+10}$ implies that $F_p(t_i,m)>2F_p(W)$, which is a contradiction. 
Similarly for $p<1$, we have that 
\[F_p\left(u+\sum_{\ell=k}^r u_{\ell}\right)-F_p\left(u+\sum_{\ell=k+2}^r u_{\ell}\right)\ge2^{-j-8}\cdot F_p(v).\]
Hence $r>2^{j+10}$ implies 
\[F_p(v)=F_p\left(u+\sum_{\ell=1}^r u_{\ell}\right)>F_p(v),\]
which is a contradiction. 
Hence, it follows that $r\le 2^{j+10}$.  
\end{proof}

\noindent
We next bound the number of level $j$ difference estimators that can occur from the end of the previous level $j-1$ difference estimator to the time when the sliding window begins. 
We say a difference estimator $\calC_{i,j,k}$ is active if $k\in[a_j,b_j]$ for the indices $a_j$ and $b_j$ defined in \algref{alg:sliding:smallp}. 
The active difference estimators in each level will be the algorithms whose output are subtracted from the initial rough estimate to form the final estimate of $F_2(W)$. 
\begin{lemma}[Number of active level $j$ difference estimators]
\lemlab{lem:sliding:active:indices:medp}
Let $i$ be the largest index such that $t_i\le m-W+1$.  
Let $\calE$ be the event that all subroutines $\calA$ in \algref{alg:sliding:smallp} succeed.
For $p\in[1,2]$ and each $j\in[\beta]$, let $a_j$ be the smallest index such that $t_{i,j,a_j}\ge c_{j-1}$ and let $b_j$ be the largest index such that $t_{i,j,b_j}\le m-W+1$. 
Then conditioned on $\calE$, we have $b_j-a_j\le 512$. 
\end{lemma}
\begin{proof}
Conditioned on $\calE$, we have that for each $j\in[\beta]$ and all $k$, 
\[2^{-j-8}\cdot F_p(W)<F_p(t_{i,j,k},t_{i,j,k+1})\le2^{-j-7}\cdot F_p(W),\]
by \lemref{lem:sliding:delta:upper:medp} and \lemref{lem:sliding:delta:lower:medp}. 
By \lemref{lem:sw:top:constant:medp}, $F_p(t_i,m)\le 2F_p(W)$. 
Thus for $j=1$, we have $j\le 512$, so $b_j-a_j\le 1024$. 
For $j>1$, note that if $b_j-a_j>1024$, then we have $F_p(t_{i,j,a_j},t_{i,j,b_j})>2^{-j-6}\cdot F_p(W)$ for $p\in[1,2]$. 
Since $t_{i,j,a_j}=t_{i,j-1,b_{j-1}}=c_{j-1}$ and $b_j\le m-W+1$, then there exists an index $x>b_{j-1}$ such that $t_{i,j-1,x}>m-W+1$, which contradicts the maximality of $c_{j-1}=t_{i,j-1,b_{j-1}}$. 
\end{proof}

\begin{lemma}[Correctness of sliding window algorithm]
\lemlab{lem:sliding:correctness:medp}
For $p\in[1,2]$, \algref{alg:sliding:smallp} outputs a $(1+\eps)$-approximation to $F_p(W)$. 
\end{lemma}
\begin{proof}
Let $f$ be the frequency vector induced by the window and let $i$ be the largest index such that $t_i\le m-W+1$ and let $c_0=t_i$. 
Let $u$ be the frequency vector induced by the updates of the stream from time $t_i$ to $m$, so that we have $u\succeq f$ coordinate-wise, since $t_i\le m-W+1$. 
For each $j\in[\beta]$, let $a_j$ be the smallest index such that $t_{i,j,a_j}\ge c_{i-1}$ and let $b_j$ be the largest index such that $t_{i,j,b_j}\le m-W+1$. 
Since $\mergesw$ does not merge indices of $t_{i,j,k}$ that are indices of $t_{i,j-1,k'}$, it follows that $t_{i,j,a_j}=c_{i-1}$. 
Let $u_j$ be the frequency vector induced by the updates of the stream from time $a_j$ to $b_j$. 
Let $v$ denote the frequency vector induced by the updates of the stream from $c_\beta$ to $m-W+1$. 
Thus, we have $u=\sum_{j=1}^\beta u_j+v+f$, so it remains to show that:
\begin{enumerate}
\item
$F_p(v+f)-F_p(f)\le\frac{\eps}{2}\cdot F_p(f)$
\item
We have an additive $\frac{\eps}{2}\cdot F_p(f)$ approximation to $F_p(v+f)$. 
\end{enumerate}
Combined, these two statements show that we have a multiplicative $(1+\eps)$-approximation of $F_p(f)$. 

To show that $F_p(v+f)-F_p(f)\le\frac{\eps}{2}\cdot F_p(f)$, let $k$ be the index such that $t_{i,\beta,k}=c_\beta$, so that $t_{i,\beta,k}\le m-W+1\le t_{i,\beta,k+1}$. 
Note that if $t_{i,\beta,k}=m-W+1$ then $v=0$ and so the claim is trivially true. 
Otherwise by the definition of the choice of $c_\beta$, we have that either $t_{i,\beta,k+1}>m-W+1$ so that $t_{i,\beta,k+1}>t_{i,\beta,k}+1$. 
By \lemref{lem:sliding:delta:upper}, it follows that $F_p(t_{i,\beta,k},m)-F_p(t_{i,\beta,k+1},m)\le 2^{-\beta/p-2}p\cdot F_p(t_{i,\beta,k+1},m)$. 
Since $\beta=\ceil{\log\frac{100p}{\eps^p}}$, then it follows that $F_p(t_{i,\beta,k},m)-F_p(t_{i,\beta,k+1},m)\le\frac{\eps}{4}\cdot F_p(t_{i,\beta,k+1},m)$. 
By \lemref{lem:sw:top:constant}, we have that $\frac{\eps}{4}\cdot F_p(t_{i,\beta,k+1},m)\le 2F(W)=2F(f)$. 
Thus, $F_p(t_{i,\beta,k},m)-F_p(t_{i,\beta,k+1},m)\le\frac{\eps}{2}\cdot F_p(f)$. 
Since the updates from the times between $t_{i,\beta,k}$ and $m$ form the vector $v+f$, we have
$F_p(v+f)-F_p(t_{i,\beta,k+1},m)\le\frac{\eps}{2}\cdot F_p(f)$. 
Since $t_{i,\beta,k+1}>m-W+1$ then by the monotonicity of $F_p$, we have that $F_p(v+f)-F_p(f)\le\frac{\eps}{2}\cdot F_p(f)$, as desired.

We next show that we have an additive $\frac{\eps}{2}\cdot F_p(f)$ approximation to $F_p(v+f)$.
If we define 
\[\Delta_j:=F_p\left(u-\sum_{k=1}^{j-1} u_k\right)-F_p\left(u-\sum_{k=1}^j u_k\right),\]
then we have
\[\sum_{j=1}^\beta\Delta_j=F_p(u)-F_p(v+f).\]

Observe that $\Delta_j$ is approximated by the active level $j$ difference estimators. 
By \lemref{lem:sliding:active:indices}, there are at most $16$ active indices at level $j$. 
Since the difference estimators are well-defined by \lemref{lem:diff:est:well:defined} and each difference estimator $\sdiffest$ uses accuracy parameter $\eta=\frac{\eps}{1024\log\frac{1}{\eps}}$, then the additive error in the estimation of $F_p(v+f)$ incurred by $Y_j$ is at most 
\[16\cdot\frac{\eps}{1024\log\frac{1}{\eps}}\cdot F_p(u)=\frac{\eps}{64\log\frac{1}{\eps}}\cdot F_p(u).\]
Summing across all $\beta\le\log\frac{16}{\eps^2}$ levels, then the total error in the estimation $Z$ of $F_p(v+f)$ across all levels $Y_j$ with $j\in[\beta]$ is at most 
\[\frac{\eps}{32}\cdot F_p(u)\le\frac{\eps}{16}\cdot F_p(f).\]
Thus, we have an additive $\frac{\eps}{2}\cdot F_p(f)$ approximation to $F_p(v+f)$ as desired. 
Since $F_p(v+f)-F_p(f)\le\frac{\eps}{2}\cdot F_p(f)$, then \algref{alg:sliding} outputs a $(1+\eps)$-approximation to $F_p(W)$. 
\end{proof}

\subsection{Moment Estimation for \texorpdfstring{$p\in(0,1]$}{p in (0,1]}}
For $p\in(0,1]$, the algorithm is nearly identical to that of $p\in[1,2]$. 
However, the $\mergesw$ subroutine now merges two sketches when the sum of their contributions from their difference estimator is small, rather than the sum of their moments is small. 
In particular, note the merge condition of \algref{alg:mergesw:smallp} compared to \algref{alg:mergesw:tinyp}. 

\begin{algorithm}[!htb]
\caption{Subroutine $\mergesw$ of \algref{alg:sliding}: removes extraneous subroutines}
\alglab{alg:mergesw:tinyp}
\begin{algorithmic}[1]
\State{Let $s$ be the number of instances of $\calA$ and $p\in(0,1]$.}
\State{$\beta\gets\ceil{\log\frac{100}{\eps}}$}
\For{$i\in[s]$, $j\in[\beta]$}
\Comment{Difference estimator maintenance}
\State{Let $r$ be the number of times $t_{i,j,*}$}
\For{$k\in[r-1]$}
\Comment{Merges two algorithms with ``small'' contributions}
\If{$p<1$ and $\sdiffest(t_{i,j,k-1},t_{i,j,k},t,\gamma_j,\eta,\delta)+\sdiffest(t_{i,j,k},t_{i,j,k+1},t,\gamma_j,\eta,\delta)\le 2^{-j-10}\cdot\calA_{i}(t_{i},t,\eta,\delta)$}
\If{$t_{i,j,k}\notin\{t_{i,j-1,*}\}$}
\State{Merge (add) the sketches for $\calA(t_{i,j,k-1},t_{i,j,k},1,\delta)$ and $\calA(t_{i,j,k},t_{i,j,k+1},1,\delta)$.}
\State{Merge (add) the sketches for $\sdiffest(t_{i,j,k-1},t_{i,j,k},t,\gamma_j,\eta,\delta)$ and $\sdiffest(t_{i,j,k},t_{i,j,k+1},t,\gamma_j,\eta,\delta)$.}
\State{Relabel the times $t_{i,j,*}$.}
\EndIf
\EndIf
\EndFor
\For{$i\in[s-2]$}
\Comment{Smooth histogram maintenance}
\If{$\calA(t_{i+2},t,\eta,\delta)\ge\frac{9}{10}\calA(t_i,t,\eta,\delta)$}
\For{$j\in[\beta]$}
\State{Append the times $t_{i+1,j,*}$ to $\{t_{i,j,*}\}$.}
\EndFor
\State{Delete $t_{i+1}$ and all times $t_{i+1,*,*}$.}
\State{Relabel the times $\{t_i\}$ and $\{t_{i,j,*}\}$.}
\EndIf
\EndFor
\EndFor
\end{algorithmic}
\end{algorithm}

The analysis is mostly the same as the case for $p\in[1,2]$. 
However, rather than bounding the moments of the frequency vectors induced by each block, we instead bound their contributions to the difference estimators. 
\begin{lemma}[Accuracy of difference estimators]
\lemlab{lem:diff:est:well:defined:smallp}
Let $p\in(0,1]$ and $i$ be the largest index such that $t_i\le m-W+1$. 
Let $\calE$ be the event that all subroutines $\calA$ and $\sdiffest$ in \algref{alg:sliding} succeed.
Then conditioned on $\calE$, we have that for each $j\in[\beta]$ and all $k$, $\sdiffest(t_{i,j,k-1},t_{i,j,k},t,\gamma_j,\eta,\delta)$ gives an additive $\eta\cdot F_p(t_{i,j,k}:t)$ approximation to $F_p(t_{i,j,k-1}:t)-F_p(t_{i,j,k}:t)$. 
\end{lemma}
\begin{proof}
Recall that for a suffix-pivoted difference estimator to be well-defined, we require $F_p(t_{i,j,k-1}:t_{i,j,k})\le\gamma_j\cdot F_p(t_{i,j,k}:t)$ or $F_p(t_{i,j,k-1}:t)-F_p(t_{i,j,k-1}:t_{i,j,k})\le\gamma_j\cdot F_p(t_{i,j,k}:t)$. 
Observe that at the time $t$ when the splitting time $t_{i,j,k}$ is chosen and conditioned on $\calE$, we must have 
\[F_p(t_{i,j,k-1}:t)-F_p(t_{i,j,k}:t)\le 2^{-j-7}\cdot F_p(t_{i,j,k}:t)=2^{-10}\cdot\gamma_j\cdot F_p(t_{i,j,k}:t).\]
By smoothness of $F_p$ for $p\in(0,1]$ in \lemref{lem:fp:smooth}, it follows that
\[F_p(t_{i,j,k-1}:t)-F_p(t_{i,j,k}:t)\le 2^{-10}\cdot\gamma_j\cdot F_p(t_{i,j,k}:t).\]
Thus the conditions for a suffix-pivoted difference estimator hold, so $\sdiffest(t_{i,j,k-1},t_{i,j,k},m,\gamma_j,\eta,\delta)$ gives an additive $\eta\cdot F_p(t_{i,j,k}:t)$ approximation to $F_p(t_{i,j,k-1}:m)-F_p(t_{i,j,k}:m)$ by the guarantees of the difference estimator. 
\end{proof}

\begin{lemma}[Geometric upper bounds on splitting times]
\lemlab{lem:sliding:delta:upper:smallp}
Let $p\in(0,1]$ and $i$ be the largest index such that $t_i\le m-W+1$.  
Let $\calE$ be the event that all subroutines $\calA$ and $\sdiffest$ in \algref{alg:sliding} succeed. 
Then conditioned on $\calE$, we have that for each $j\in[\beta]$ and each $k$ that either $t_{i,j,k+1}=t_{i,j,k}+1$ or $F_p(t_{i,j,k},t)-F_p(t_{i,j,k+1},t)\le 2^{-j-7}\cdot F_p(t_{i,j,k+1},t)$. 
\end{lemma}
\begin{proof}
Suppose $t_{i,j,k+1}\neq t_{i,j,k}+1$. 
Then at some time $x$, the timestamp $t_{i,j,k}+1$ must have been merged by subroutine $\mergesw$ with $p\in(0,1]$. 
Thus 
\[\sdiffest(t_{i,j,k},t_{i,j,k'},x,\gamma_j,\eta,\delta)+\sdiffest(t_{i,j,k'},t_{i,j,k+1},x,\gamma_j,\eta,\delta)\le 2^{-j-10}\cdot\calA_{i}(t_{i},x,\eta,\delta).\]
Let $u_1$ be frequency vector representing the updates from time $t_{i,j,k}$ to $t_{i,j,k'}-1$, $u_2$ represent times $t_{i,j,k'}$ to $t_{i,j,k+1}$, and $u$ represent times $t_{i,j,k}$ to $t_{i,j,k+1}$, so that $u_1+u_2=u$. 
Let $x_1$ represent the updates between time $t_{i,j,r+1}$ and $x$ and let $x_2$ represent the updates between time $x+1$ and $t$. 
Since the difference estimators are well-defined by \lemref{lem:diff:est:well:defined:smallp}, we have that $\sdiffest(t_{i,j,k},t_{i,j,k'},x,\gamma_j,\eta,\delta)$ is at most a $2$-approximation to $F_p(u_1+u_2+x_1)-F_p(u_2+x_1)$ and we similarly obtain a $2$-approximation to $F_p(u_2+x_1)-F_p(x_1)$. 
Thus, 
\[F_p(u+x_1)-F_p(x_1)=F_p(u_1+u_2+x_1)-F_p(x_1)\le 2^{-j-7}\cdot F_p(x_1).\]
By the smoothness of $F_p$ in \lemref{lem:fp:smooth} for $p\in(0,1]$, we thus have that for any vector $x_2$,
\[F_p(u+v)-F_p(v)=F_p(u+x_1+x_2)-F_p(x_1+x_2)\le 2^{-j-7}\cdot F_p(x_1+x_2)=2^{-j-7}\cdot F_p(v).\]
\end{proof}

\begin{lemma}[Geometric lower bounds on splitting times]
\lemlab{lem:sliding:delta:lower:smallp}
Let $p\in(0,1]$ and $i$ be the largest index such that $t_i\le m-W+1$.  
Let $\calE$ be the event that all subroutines $\calA$ and $\sdiffest$ in \algref{alg:sliding} succeed. 
Let $u$ denote the frequency vector between times $t_{i,j,r}$ and $t_{i,j,r+2}$ and let $w$ denote the frequency vector between times $t_{i,j,r+2}$ and $t$. 
Finally, let $w$ denote the frequency vector between times $t_i$ and $t$. 
Then conditioned on $\calE$, we have that
\[F_p(u+v)-F_p(v)\ge 2^{-j-8}\cdot F_2(w).\]
\end{lemma}
\begin{proof}
Let $u_1$ denote the frequency vector between times $t_{i,j,r}$ and $t_{i,j,r+1}$ and $u_2$ denote the frequency vector between times $t_{i,j,r+1}$ and $t_{i,j,r+2}$,
Recall that $\mergesw$ for $p\in(0,1]$ merges two instances of $\sdiffest$ when 
\begin{align*}
\sdiffest(t_{i,j,k-1},t_{i,j,k},t,\gamma_j,\eta,\delta)+\sdiffest(t_{i,j,k},t_{i,j,k+1},t,\gamma_j,\eta,\delta)\le 2^{-j-10}\cdot\calA_{i}(t_{i},t,\eta,\delta).
\end{align*}
Thus we do not remove $t_{i,j,r+1}$ from the list of timestamps, then we have
\[\sdiffest(t_{i,j,k-1},t_{i,j,k},t,\gamma_j,\eta,\delta)+\sdiffest(t_{i,j,k},t_{i,j,k+1},t,\gamma_j,\eta,\delta)>2^{-j-10}\cdot\calA_{i}(t_{i},t,\eta,\delta).\]
Since the difference estimators are well-defined by \lemref{lem:diff:est:well:defined:smallp}, then $\sdiffest(t_{i,j,k},t_{i,j,k+1},t,\gamma_j,\eta,\delta)$ is at most a $2$-approximation to $F_p(u_1+u_2+v)-F_p(u_2+v)$ and we similarly obtain a $2$-approximation to $F_p(u_2+v)-F_p(v)$. 
Therefore, 
\[F_p(u)-F_p(v)=F_p(u_1+u_2+v)-F_p(v)\ge 2^{-j-8}\cdot F_p(w).\] 
\end{proof}

Next we bound the number of instances $k$ of the level $j$ difference estimators. 
\begin{lemma}[Number of level $j$ difference estimators]
\lemlab{lem:sliding:instances:smallp}
Let $p\in(0,1]$ and $i$ be the largest index such that $t_i\le m-W+1$.  
Let $\calE$ be the event that all subroutines $\calA$ and $\sdiffest$ in \algref{alg:sliding} succeed. 
Then conditioned on $\calE$, we have that for each $j\in[\beta]$, that $k\le 2^{j+10}$ for any $t_{i,j,k}$. 
\end{lemma}
\begin{proof}
Let $r$ be the number of instances of timestamps $t_{i,j,*}$ for a fixed $j\in[\beta]$.  
For each $k\in[r]$, let $u_k$ denote the frequency vector between times $t_{i,j,k-1}$ and $t_{i,j,k}$. 
Let $u$ denote the frequency vector between times $t_{i+1}$ and $t$ and let $v$ denote the frequency vector between times $t_i$ and $t$. 
Then by \lemref{lem:sliding:delta:lower:smallp}, we have 
\[F_p\left(u+\sum_{\ell=k}^r u_{\ell}\right)-F_p\left(u+\sum_{\ell=k+1}^r u_{\ell}\right)\ge2^{-j-8}\cdot F_p(v).\]
Thus if $r>2^{j+8}$, then we have 
\[F_p(v)=F_p\left(u+\sum_{\ell=1}^r u_{\ell}\right)>F_p(v),\]
which is a contradiction. 
Hence, it follows that $r\le 2^{j+8}$.  
\end{proof}

\noindent
We next bound the number of level $j$ difference estimators that can occur from the end of the previous level $j-1$ difference estimator to the time when the sliding window begins. 
We say a difference estimator $\sdiffest$ is active if $k\in[a_j,b_j]$ for the indices $a_j$ and $b_j$ defined in \algref{alg:sliding}. 
The active difference estimators in each level will be the algorithms whose output are subtracted from the initial rough estimate to form the final estimate of $F_2(W)$. 
\begin{lemma}[Number of active level $j$ difference estimators]
\lemlab{lem:sliding:active:indices:smallp}
Let $i$ be the largest index such that $t_i\le m-W+1$.  
Let $\calE$ be the event that all subroutines $\calA$ and $\sdiffest$ in \algref{alg:sliding} succeed. 
For $p\in(0,1]$ and each $j\in[\beta]$, let $a_j$ be the smallest index such that $t_{i,j,a_j}\ge c_{j-1}$ and let $b_j$ be the largest index such that $t_{i,j,b_j}\le m-W+1$. 
Then conditioned on $\calE$, we have $b_j-a_j\le 512$. 
\end{lemma}
\begin{proof}
Conditioned on $\calE$, we have that for each $j\in[\beta]$ and all $k$, 
\[2^{-j-8}\cdot F_p(W)<F_p(t_{i,j,k},t)-F_p(t_{i,j,k+2},t)\le 2^{-j-6}\cdot F_p(W)t_{i,j,k+1},t),\]
by \lemref{lem:sliding:delta:upper:smallp} and \lemref{lem:sliding:delta:lower:smallp}. 
By \lemref{lem:sw:top:constant}, $F_p(t_i,m)\le 2F_p(W)$. 
Thus by a telescoping argument, we have that for $j=1$, $a_j-b_j<512$. 

For $j>1$, suppose by way of contradiction that $b_j-a_j>512$. 
Let $k=b_j-a_j$ and for $x\in[k]$, let $u_x$ be the frequency vector induced by the updates of the stream from time $a_j+{x-1}$ to $a_j+x$. 
Let $u$ be the frequency vector induced by the updates of the stream from time $t_{i,j,k}$ to $m$ and $v$ be the frequency vector induced by the updates from $t_i$ to $m$. 
By \lemref{lem:sliding:delta:lower:smallp}
\[F_p\left(u+\sum_{\ell=1}^k u_{\ell}\right)-F_p(u)\ge8\cdot2^{-j-8}\cdot F_p(v),\]
since there are at least $8$ disjoint pairs of tuples $(i,j,\ell)$ and $(i,j,\ell+1)$ if $k>512$. 

Thus if $b_j-a_j>512$, then the sum of the outputs of the active difference estimators at level $j$ is more than $8\cdot2^{-j-8}\cdot F_p(v)$. 
In particular since $t_{i,j,a_j}\ge c_{j-1}$, the sum of the outputs of the active difference estimators at level $j$ after $c_{j-1}$ is at least $8\cdot2^{-j-8}\cdot F_p(v)=2^{-j-2}\cdot F_p(v)$. 
However, by \lemref{lem:sliding:delta:lower:smallp}, each difference estimator at level $j-1$ has output at least $2^{-j-7}\cdot F_p(u)>8\cdot 2^{-j-8}\cdot F_p(v)$. 
Specifically, the difference estimator from times $t_{i,j,c_{j-1}}$ to $t_{i,j,c_{j}}$ must have output at least $2^{-j-7}\cdot F_p(u)>8\cdot 2^{-j-8}\cdot F_p(v)$. 
Therefore, there exists some other $z>c_{j-1}$ such that $t_{i,j-1,z}\le m-W+1$, which contradicts the maximality of $c_{j-1}$ at level $j-1$. 
\end{proof}

By applying the same argument as \lemref{lem:sliding:correctness:medp} using the analogous statements in this section, we have:
\begin{lemma}[Correctness of sliding window algorithm]
\lemlab{lem:sliding:correctness:smallp}
For $p\in(0,1]$, \algref{alg:sliding:smallp} outputs a $(1+\eps)$-approximation to $F_p(W)$. 
\end{lemma}

Combined, \lemref{lem:sliding:correctness:smallp} and \lemref{lem:sliding:correctness:medp} give \lemref{lem:sliding:correctness:rangep}.
\end{document}